\documentclass[letterpaper,12pt]{amsart}
\usepackage[margin=2cm]{geometry}
\usepackage{amscd, amssymb}
\usepackage{amsmath,amscd}
\usepackage[driverfallback=hypertex]{hyperref}
\usepackage{nameref,zref-xr}                    
\zxrsetup{toltxlabel}
\zexternaldocument*[calmodulin-]{real_Moduli_20_11}
\usepackage{comment}
\usepackage{graphicx}
\usepackage[usenames,dvipsnames]{color}
\usepackage{bm}
\usepackage{enumitem}
\usepackage[all]{xy}   
\setlist[enumerate,1]{label={(\alph*)}}
\setlist[enumerate,2]{label={(\roman*)}}
\input xy
\xyoption{all}
\usepackage{epstopdf}

\newtheorem{thm}{Theorem}[section]
\newtheorem{prop}[thm]{Proposition}
\newtheorem{pr}[thm]{Proposition}

\newtheorem{lemma}[thm]{Lemma}
\newtheorem{cor}[thm]{Corollary}

\newtheorem{st}{Step}
\newtheorem{open}{Open problem}
\newtheorem{construction}{Construction$\backslash$Notation}

\newtheorem{conj}{Conjecture}
\newcommand{\p}{\partial}

\theoremstyle{definition}
\newtheorem{definition}[thm]{Definition}
\newtheorem{nn}[thm]{Notation}

\theoremstyle{remark}
\newtheorem{rmk}[thm]{Remark}
\newtheorem{ex}[thm]{Example}
\newtheorem{obs}[thm]{Observation}

\newcommand{\R}{\mathbb R}
\newcommand{\C}{\mathbb C}
\newcommand{\Q}{Q}
\newcommand{\q}{q}
\newcommand{\FFF}{\Psi}
\newcommand{\mmm}{\texttt{m}}
\newcommand{\PPP}{{P_0}}

\newcommand{\Z}{\mathbb Z}
\newcommand{\CM}{{\mathcal{M}}}
\newcommand{\oCM}{{\overline{\mathcal{M}}}}
\newcommand{\RCM}{{\mathcal{M}}^{\mathbb R}}
\newcommand{\oRCM}{{\overline{\mathcal{M}}}^{\mathbb{R}}}

\newcommand{\SC}{\text{Cont}}
\newcommand{\oCC}{{\overline{\mathcal{C}}}}

\newcommand{\CL}{{\mathbb{L}}}
\newcommand{\CCCL}{{\mathcal{L}}}
\newcommand{\SL}{{\mathbb{S}}}
\newcommand{\CF}{{\mathcal{F}}}
\newcommand{\CS}{{\mathcal{S}}}
\newcommand{\CB}{{\mathcal{B}}}

\newcommand{\CG}{{\mathcal{G}}}
\newcommand{\RCG}{{\mathcal{G}}^\mathbb{R}}

\newcommand{\pu}{{\partial^!}}
\newcommand{\blangle}{{\big\langle}}
\newcommand{\bblangle}{{\big\langle}{\big\langle}}
\newcommand{\brangle}{{\big\rangle}}
\newcommand{\bbrangle}{{\big\rangle}{\big\rangle}}

\newcommand{\D}{\mathcal{IT}}

\newcommand{\oR}{\mathcal{R}}
\newcommand{\F}{\mathcal{F}}
\newcommand{\Y}{\mathcal{X}}

\newcommand{\CCCG}{\mathcal{R}}

\newcommand{\oSRc}{{\mathcal{OSR}}}
\newcommand{\oSR}{\mathcal{SR}}
\newcommand{\oRc}{{\mathcal{OR}}}

\newcommand{\sig}{\sigma}
\newcommand{\K}{K}
\newcommand{\Or}{\text{alt}}
\newcommand{\tw}{\text{tw}}
\newcommand{\Norm}{\text{Norm}}
\newcommand{\NNN}{\text{Norm}}
\newcommand{\RGamma}{\Gamma^{\mathbb{R}}}
\newcommand{\B}{\mathcal{B}}
\newcommand{\I}{\mathcal{I}}

\newcommand{\BBB}{{0}}
\newcommand{\III}{{0}}
\newcommand{\Spin}{\text{Spin}}
\newcommand{\comb}{{\text{comb}}}
\newcommand{\Rcomb}{\text{comb}^\R}
\newcommand{\N}{\mathcal V}
\newcommand{\CBB}{{\widetilde{\CB}}}
\newcommand{\PP}{{\mathcal{P}_0}}
\newcommand{\pp}{{\mathbf{p}}}
\newcommand{\HN}{\text{HN}}

\DeclareMathOperator{\ior}{\text{int}}
\renewcommand{\mathring}{\ior}

\begin{document}
\title{The combinatorial formula for open gravitational descendents}
\author{Ran J. Tessler}
\address{Weizmann Institute for Science, \newline Rehovot, Israel}
\email{ran.tessler@weizmann.ac.il}

\begin{abstract}
In \cite{PST,ST},
descendent integrals on the moduli space of Riemann surfaces with boundary are defined.
It was conjectured in \cite{PST} that the generating function of these integrals satisfies the open KdV equations.
In this paper we prove a formula for these integrals in terms of sums of Feynman diagrams. This formula is a generalization of Kontsevich's combinatorial formula \cite{Kont} to the open setting. In order to overcome the main challenges of the open setting, which are orientation questions and the existence of boundary and boundary conditions, new techniques are developed. These techniques, which are interesting in their own right, include a characterization of graded spin structure in terms of open and nodal Kasteleyn orientations and a new formula for the angular form of $S^{2n-1}-$bundles.

Based on the work presented here, the conjecture of \cite{PST} was proved in \cite{BT}.
\end{abstract}

\maketitle

\pagestyle{plain}

\tableofcontents
\section{Introduction}
The study of the intersection theory on the moduli space of open Riemann surfaces was initiated in~\cite{PST}. The authors constructed a descendent theory in genus~$0$ and obtained a complete description of it. In all genera, they conjectured that the generating series of the descendent integrals satisfies the open KdV equations. This conjecture can be considered as an open analog of Witten's famous conjecture~\cite{Witten}.

The construction of the positive genus analog will appear in \cite{ST}, and is reviewed here. 
A physical interpretation of these constructions can be found in \cite{DijkWit}.

In this paper, after recalling the constructions of \cite{PST},\cite{ST}, we prove a formula for all the descendent integrals as sums over amplitudes of special Feynman diagrams which we call odd critical nodal ribbon graphs. With this formula one can effectively calculate all the open descendents.

Based on this formula, the conjecture of \cite{PST} is proved in \cite{BT}, and a calculation of finer invariants, is performed in \cite{ABT}.

\subsection{Witten's conjecture}\label{wittc}
\subsubsection{Intersection numbers}

Denote by $\CM_{g,l}$ the moduli space of compact connected Riemann surfaces with $l$ distinct marked points.
P.~Deligne and D.~Mumford defined a natural compactification of it via stable curves in~\cite{DM} in 1969.
Given $g,l,$ a stable curve is a compact connected complex curve with $l$ distinct marked points and finitely many singularities, all of which are simple nodes. We require the automorphism group of the surface to be finite, and the marked points and nodes are all distinct.
The moduli space of stable curves of fixed $g,l$ is denoted $\oCM_{g,l}.$
It is known that this space is a non-singular
complex orbifold of complex dimension $3g-3+l.$ For the basic theory the reader is referred to~\cite{DM,HM}.

In his seminal paper~\cite{Witten}, E.~Witten, motivated by theories of $2$-dimensional quantum gravity, initiated new directions in the study of $\oCM_{g,l}$. For each marking index $i$
he considered the tautological line bundles
$$\CL_i \rightarrow \oCM_{g,l}$$
whose fiber over a point
$$[\Sigma,z_1,\ldots,z_l]\in \oCM_{g,l}$$
is the complex cotangent space $T_{z_i}^*\Sigma$ of $\Sigma$ at $z_i$. Let
$$\psi_i\in H^2(\oCM_{g,l},\mathbb{Q})$$
denote the first Chern class of $\CL_i$, and write
\begin{equation}\label{products}
\blangle\tau_{a_1} \tau_{a_2} \cdots \tau_{a_l}\brangle_g^c:=\int
_{\overline {M}_{g,l}} \psi_1^{a_1} \psi_2^{a_2} \cdots \psi_l^{a_l}.
\end{equation}
The integral on the right-hand side of~\eqref{products} is well-defined, when the stability condition
$$2g-2+l >0$$
is satisfied, all the $a_i$'s are non-negative integers, and the dimension constraint
\begin{equation*}
3g-3+l=\sum_i a_i
\end{equation*}
holds. In all other cases $\blangle\prod_{i=1}^{l} \tau_{a_i}\brangle_g^c$ is defined to be zero.
The intersection products \eqref{products} are often called {\em descendent integrals} or {\em intersection numbers}.

Let $t_i$ (for $i\geq 0$) and $u$ be formal variables, and put
$$\gamma:=\sum_{i=0}^{\infty} t_i \tau_i.$$
Let
$$F^c_g(t_0, t_1,\ldots):=\sum_{n=0}^{\infty} \frac{\blangle\gamma^n\brangle^c_g}{n!}$$
be the generating function of the genus $g$ descendent integrals \eqref{products}. The bracket $\blangle\gamma^n\brangle^c_g$ is defined by the monomial expansion and the
multilinearity in the variables $t_i$.
The generating series
\begin{equation}\label{v34}
F^c:=\sum_{g=0}^{\infty} u^{2g-2} F_g^c
\end{equation}
is called the {\it (closed) free energy}. The exponent $\tau^c:=\exp(F^c)$ is called the {\it (closed) partition function}.

\subsubsection{KdV equations}

Set $\blangle\blangle\tau_{a_1} \tau_{a_2} \cdots \tau_{a_l}\brangle\brangle^c:=\frac{\p^l F^c}{\p t_{a_1}\p t_{a_2}\cdots\p t_{a_l}}$. Witten's conjecture~(\cite{Witten}) says that the closed partition function~$\tau^c$ becomes a tau-function of the KdV hierarchy after the change of variables~$t_n=(2n+1)!!T_{2n+1}$. In particular, it implies that the closed free energy~$F^c$ satisfies the following system of partial differential equations $(n\geq1)$:
\begin{align*}\label{eq:KdV equations}
(2n+1)u^{-2} &\blangle\blangle\tau_n \tau_0^2 \brangle\brangle^c = \\\notag
&\blangle\blangle\tau_{n-1} \tau_0\brangle\brangle^c\blangle\blangle\tau_0^3\brangle\brangle^c +
2\blangle\blangle\tau_{n-1}\tau_0^2\brangle\brangle^c\blangle\blangle\tau_0^2\brangle\brangle^c+
\frac{1}{4}\blangle\blangle\tau_{n-1} \tau_0^4\brangle\brangle^c.
\end{align*}
These equations are known in mathematical physics as the KdV equations. E.~Witten~(\cite{Witten}) proved that the intersection numbers~\eqref{products} satisfy the string equation
\begin{gather*}
\left\langle\tau_0 \prod_{i=1}^{l} \tau_{a_i}\right\rangle^c_g =
\sum_{j=1}^{l} \left\langle\tau_{a_j-1} \prod_{i\neq j} \tau_{a_i}\right\rangle^c_g,
\end{gather*}
for $2g-2+l>0$.
E.~Witten has shown that the KdV equations, together with the string equation determine the closed free energy $F^c$ completely.
R.~Dijkgraaf, E.~Verlinde and H.~Verlinde (\cite{DVV}) reformulated an alternative description to Witten's conjecture, in terms of the Virasoro algebra, and they have shown that the two descriptions are equivalent.

\subsection{Kontsevich's Proof}

Witten's conjecture was proved by M.~Kontsevich~\cite{Kont}.
The proof of~\cite{Kont} consisted of two parts.
The first part was to prove a combinatorial formula for the gravitational descendents. Let~$\CCCG_{g,n}$ be the set of isomorphism classes of trivalent ribbon graphs of genus~$g$ with~$n$ marked faces. Denote by~$V(G)$ the set of vertices of a graph $G\in \CCCG_{g,n}$. Introduce formal variables $\lambda_i$,~$i\in [n]$. For an edge $e\in E(G),$ let $\lambda(e):=\frac{1}{\lambda_i+\lambda_j},$ where~$i$ and~$j$ are the numbers of faces adjacent
to~$e$. The following formula holds
\begin{equation}\label{eq:Kontsevich's formula}
\sum_{a_1,\ldots,a_n\geq 0}\left<\prod_{i=1}^{n} \tau_{a_i}\right>_g^c\prod_{i=1}^n\frac{(2a_i-1)!!}{\lambda_i^{2a_i+1}} = \sum_{G\in
\CCCG_{g,n}}\frac{2^{|E(G)|-|V(G)|}}{|\text{Aut}(G)|}\prod_{e\in E(G)}\lambda(e).
\end{equation}
The second step of Kontsevich's proof was to translate the combinatorial formula into a matrix integral.
Then, by using non trivial analytical tools and the theory of the KdV hierarchy, he was able to prove that $F^c$ satisfies the KdV equations~\eqref{eq:KdV equations}.
Other proofs for Witten's conjecture were given, see for example ~\cite{Mir,OP}.
\subsection{Open intersection numbers and the open KdV equations}

\subsubsection{Open intersection numbers}

In \cite{PST} R. Pandharipande, J. Solomon and the author constructed an intersection theory on the moduli space of stable marked disks. Let~$\oCM_{0,k,l}$ be the moduli space of stable marked
disks with~$k$ boundary marked points and~$l$ internal marked points. This space carries a natural structure of a compact smooth oriented manifold with corners. One can easily define
the tautological line bundles $\CL_i,$ for an internal marking $i,$ as in the closed case.

In order to define gravitational descendents, we must specify boundary conditions. 
The main construction in~\cite{PST} is a construction of boundary conditions for $\CL_i\to\oCM_{0,k,l}.$ In~\cite{PST}, vector spaces $\CS_i = \CS_{i,0,k,l}$ of
\emph{multisections} of $\CL_i\to\partial\oCM_{0,k,l},$ which satisfy the following requirements, were defined. Suppose $a_1,\ldots,a_l$ are non-negative integers with $2\sum_i
a_i=\dim_\R\oCM_{0,k,l}=k+2l-3,$ then
\begin{enumerate}
\item\label{it:canonic1}
For any generic choice of multisections $s_{ij}\in\CS_i,$ for $1\leq j\leq a_i,$ the multisection
$$
s=\bigoplus_{\substack{i\in [l]\\1\leq j\leq a_i}} s_{ij}
$$
vanishes nowhere on $\partial\oCM_{0,k,l}.$
\item\label{it:canonic2}
For any two such choices $s$ and $s'$ we have
$$\int_{\oCM_{0,k,l}}e(E,s) = \int_{\oCM_{0,k,l}}e(E,s'),$$
where $E=\bigoplus_i\CL_i^{a_i}$ and $e(E,s)$ is the relative Euler class.
\end{enumerate}
The multisections $s_{ij}$, as above, are called \emph{canonical}. With this construction the open gravitational descendents in genus~$0$ are defined by
\begin{equation}
\label{openproducts}
\blangle\tau_{a_1} \tau_{a_2} \cdots \tau_{a_l}\sigma^k\brangle_0^o:=2^{-\frac{k-1}{2}}\int_{\oCM_{0,k,l}}e(E,s),
\end{equation}
where $E$ is as above and $s$ is canonical.

In a forthcoming paper \cite{ST}, J.~Solomon and R.T. define a generalization for all genera. Suppose $g,k,l$ are such that
\begin{equation}\label{eq:restrictions}
2g-2+k+2l>0, 2|g+k-1.
\end{equation}
In \cite{ST} a moduli space $\oCM_{g,k,l},$ which classifies stable surfaces with boundaries and some extra structure, is constructed (see Subsection \ref{subsec:graded} for a precise definition). The moduli space $\oCM_{g,k,l}$ is a smooth oriented compact orbifold with corners, of real dimension
\begin{gather}\label{open dimension}
3g-3+k+2l.
\end{gather}
Note that naively, without adding an extra structure, the moduli of real stable curves of positive genus is non orientable.

Again, on $\oCM_{g,k,l}$ one defines vector spaces $\CS_i=\CS_{i,g,k,l}$, for
$i\in[l],$ for which the genus $g$ analogs of requirements \ref{it:canonic1},\ref{it:canonic2} from above hold.
Write
\begin{equation}
\label{openproducts}
\blangle\tau_{a_1} \tau_{a_2} \cdots \tau_{a_l}\sigma^k\brangle_g^o:=2^{-\frac{g+k-1}{2}}\int_{\oCM_{g,k,l}}e(E,s),
\end{equation}
for the corresponding higher genus descendents. Introduce one more formal variable $s$. The \emph{open free energy} is the generating function
\begin{equation}\label{eq:gen_func}
F^o(s,t_0,t_1,\ldots;u) := \sum_{g=0}^\infty u^{g-1} \sum_{l=0}^\infty
\frac{\blangle {\gamma^l \delta^k}
\brangle_g^o}{n!k!},
\end{equation}
where $\gamma:=\sum_{i\ge 0} t_i\tau_i$, $\delta:=s\sigma$, and again we use the monomial expansion and the multilinearity in the variables $t_i,s.$

The descriptions of $\oCM_{g,k,l}$ and its construction, and of the boundary conditions and their construction are given in Section \ref{sec:geom}.
Throughout this article we shall write $\blangle {\cdots}\brangle$ for $\blangle {\cdots}\brangle_g^o,$ as closed descendents will not be considered, and the genus can be read from the numbers $k,l,a_1,\ldots,a_l.$
\subsubsection{Open KdV}

The following initial condition follows easily from the definitions (\cite{PST}):
\begin{gather}\label{eq:open initial condition}
\left.F^o\right|_{t_{\ge 1}=0}=u^{-1}\frac{s^3}{6}+u^{-1}t_0 s.
\end{gather}
In~\cite{PST} the authors conjectured the following equations:
\begin{align}\label{gvvt}
\frac{\partial F^o}{\partial t_0} =& \sum_{i=0}^{\infty}t_{i+1}\frac{\partial F^o}{\partial t_i} + u^{-1}s,\\
\frac{\partial F^o}{\partial t_1} =& \sum_{i=0}^{\infty}\frac{2i+1}{3}t_{i}\frac{\partial F^o}{\partial t_i} + \frac{2}{3}s\frac{\partial F^o}{\partial s}+\frac{1}{2}.\label{eq:open dilaton}
\end{align}
They were called the \emph{open string} and the \emph{open dilaton} equation correspondingly. These equation were geometrically proved in \cite{PST} for $g=0,$ and for all genera in \cite{ST}.

Put
$\blangle\blangle\tau_{a_1} \tau_{a_2} \cdots \tau_{a_l} \sigma^k
\brangle\brangle^o :=
\frac{\p^{l+k} F^o}{\p t_{a_1}\p t_{a_2}\cdots\p t_{a_l}\p s^k}.$ The main conjecture in \cite{PST} was

\begin{conj}[Open KdV conjecture]\label{conjecture:open KdV}
The following system of equations is satisfied:
\begin{multline}\label{eq:openkdv}
(2n+1)u^{-1}\bblangle \tau_n \bbrangle^o=u\bblangle
\tau_{n-1} \tau_0\bbrangle^c  \bblangle \tau_0\bbrangle^o -\frac{u}{2} \bblangle \tau_{n-1}\tau_0^2 \bbrangle^c\\
 + 2    \bblangle \tau_{n-1}\bbrangle^o
\bblangle\sigma\bbrangle^o + 2 \bblangle \tau_{n-1} \sigma \bbrangle^o,\quad n\ge 1.
\end{multline}
\end{conj}
\noindent In~\cite{PST} equations~\eqref{eq:openkdv} were called the \emph{open KdV equations}. It is easy to see that~$F^o$ is fully determined by the open KdV equations~\eqref{eq:openkdv}, the initial condition~\eqref{eq:open initial condition} and the closed free energy~$F^c$. They have also conjectured a Virasoro-type conjecture which also fully describes the open descendents. Both conjectures were proved in \cite{PST} for $g=0.$ In \cite{Bur} Buryak has proved the equivalence of the two conjectures.
Based on the work presented here, the conjecture was proven for all genus in \cite{BT}, see Subsection \ref{subsec:further} below for more details.

\subsection{The open combinatorial formula}\label{opencomb}
%
Here and below the genus of a Riemann surface with boundary $\Sigma,$ smooth or nodal, is defined as the usual genus of the doubled surface obtained from gluing two copies of $\Sigma$ along the common boundary $\partial\Sigma.$
\begin{definition}\label{def:smooth_trivalent}
Let $g,k,l$ be non-negative integers which satisfy conditions \ref{eq:restrictions},
and let $\B,\I$ be sets with $|\B|=k,|\I|=l,$ and let $(\Sigma,\{x_i\}_{i\in\B},\{z_i\}_{i\in\I}),$ be a genus $g$ surface with boundary, whose genus is $g,~\B$ is its set of boundary markings, and $\I$ its set of internal markings.
A \emph{$(g,\B,\I)$-smooth trivalent ribbon graph} is an embedding $\iota:G\to\Sigma$ of a connected graph~$G$ into $(\Sigma,\{x_i\}_{i\in\B},\{z_i\}_{i\in\I}),$ such that
\begin{enumerate}
\item $\{x_i\}_{i\in \B}\subseteq \iota(V(G)),$ where $V(G)$ is the set of vertices of~$G.$ We henceforth consider $\{x_i\}$ as vertices.
\item The degree of every $x_i$ is $2.$
\item The degree of any vertex $v\in V(G)\setminus\{x_i\}_{i\in\B}$ is $3.$
\item $\partial\Sigma\subseteq \iota(G)$.
\item If $l \geq 1,$ then $$\Sigma\setminus\iota(G)=\coprod_{i\in \I} D_i,$$ where each $D_i$ is a topological open disk, with $z_i\in D_i$. We call the disk $D_i$ the face marked $i.$
\item If $l=0$, then $\iota(G)=\partial\Sigma,$ and $k=3.$ Such a component is called \emph{trivalent ghost}.
\end{enumerate}
The genus $g(G)$ of the graph~$G$ is the genus of $\Sigma$. The number of the boundary components of~$G$ or~$\Sigma$ is denoted by $b(G)$ and $V^I(G)$ stands for the set of internal vertices. 
Denote by~$B(G)$ the set of boundary marked points~$\{x_i\}_{i\in\B},~I(G)\simeq\I$ is the set of faces.

%
%
%
%
%
\end{definition}
\begin{definition}
An \emph{odd critical nodal ribbon graph} is $G=\left(\coprod_i G_i\right)/N$, where
\begin{enumerate}
\item $\iota_i\colon G_i\to\Sigma_i$ are smooth trivalent ribbon graphs.
\item $N\subset (\cup_i V(G_i))\times(\cup_i V(G_i))$ is a set of \emph{ordered} pairs of boundary marked points $(v_1,v_2)$ of the $G_i$'s which we identify. After the identification of the vertices~$v_1$ and~$v_2$ the corresponding point in the graph is called a node. The vertex~$v_1$ is called the legal side of the node and the vertex~$v_2$ is called the illegal side of the node.
\item Ghost components do not contain the illegal sides of nodes.
\item For any component $G_i,$ any boundary component of it contains an odd number of points which are either marked points or legal sides of nodes.
%
\end{enumerate}
We require that elements of $N$ are disjoint as sets (without ordering).

The set of edges $E(G)$ is composed of the internal edges of the $G_i$'s and of the boundary edges.
The boundary edges are the boundary segments between successive vertices which are not the illegal sides of nodes. For any boundary edge $e$ we denote by $m(e)$ the number of the
illegal sides of nodes lying on it. The boundary marked points of~$G$ are the boundary marked points of~$G_i$'s, which are not nodes. The set of boundary marked points of~$G$ will
be denoted by~$B(G),$ the set of faces by $I(G).$
%

An odd critical nodal ribbon graph is naturally embedded into the nodal surface $\Sigma=\left(\coprod_i\Sigma_i\right)/N$. The genus of the graph is defined as the genus of $\Sigma.$
A $(g,k,l)-$odd critical nodal ribbon graph is a connected odd critical nodal ribbon graph, together with a pair of bijections, $\mmm^B:B(G)\to[k],\mmm^I:I(G)\to[l],$
called markings. 

Two marked odd critical nodal ribbon graphs $\iota\colon G\to\Sigma,~\iota'\colon G'\to\Sigma'$ are isomorphic, if there is an orientation preserving homeomorphism
$\Phi\colon(\Sigma,\{z_i\},\{x_i\})\to(\Sigma',\{z'_i\},\{x'_i\}),$ of marked surfaces, and an isomorphism of graphs $\phi\colon G\to G'$, such that
\begin{enumerate}
\item
$\iota'\circ\phi = \Phi\circ\iota.$
\item
The maps preserve the markings.
\end{enumerate}
\end{definition}

In Figure~\ref{fig:nodalgraph} there is a nodal graph of genus~$0$, with~$5$ boundary marked points, $6$ internal marked points, three components, one of them is a ghost, two nodes, where a
plus sign is drawn next to the legal side of a node and a minus sign is drawn next to the illegal side.

\begin{figure}
\centering
\includegraphics[scale=.4]{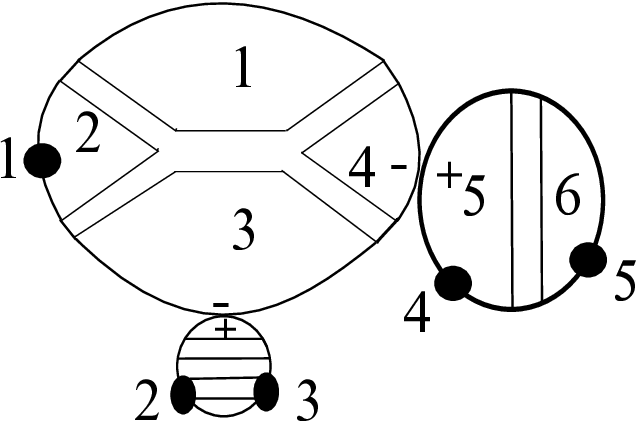}
\caption{A nodal ribbon graph.}
\label{fig:nodalgraph}
\end{figure}

\begin{nn}\label{nn:oR_intro}
Denote by $\oRc^m_{g,k,l}$ the set of isomorphism classes of odd $(g,k,l)-$critical nodal ribbon graphs with $m$ legal nodes.
\end{nn}
\begin{rmk}
In Section \ref{sec:strebel} we have to consider more general ribbon graphs, and the notions of this subsection are defined in a different but equivalent way.
\end{rmk}
The goal of this paper is to prove the following theorem

\begin{thm}\label{thm:comb_model}
Fix $g,k,l\ge 0$ which satisfy conditions \ref{eq:restrictions}. Let $\lambda_1,\ldots,\lambda_l$ be formal variables. Then we have
\begin{multline}\label{eq:combinatorial formula}
2^{\frac{g+k-1}{2}}\sum_{a_1,\ldots,a_l\ge 0}\blangle\tau_{a_1} \tau_{a_2} \cdots \tau_{a_l}\sigma^k\brangle_g^o\prod_{i=1}^l\frac{2^{a_i}(2a_i-1)!!}{\lambda_i^{2a_i+1}}\\
=\sum_{m\geq 0}\sum_{G=\left(\coprod_i G_i\right)/N\in \oRc^m_{g,k,l}}\frac{\prod_i 2^{|V^I(G_i)|+g(G_i)+b(G_i)-1}}{|\text{Aut}(G)|}\prod_{e\in E(G)}\lambda(e),
\end{multline}
where
\[
\lambda(e):=
\begin{cases}
\frac{1}{\lambda_i+\lambda_j},&\text{$e$ is an internal edge between faces $i$ and $j$};\\
\frac{1}{m+1}\binom{2m}{m}\lambda_i^{-2m-1},&\text{$e$ is a boundary edge of face $i$ and $m(e)=m$};\\
1,&\text{$e$ is a boundary edge of a ghost}.
\end{cases}
\]
\end{thm}
\begin{rmk}
The invariants of \cite{PST,ST} are defined as integrals of relative Euler classes, relative to canonical boundary conditions, over the moduli of graded surfaces, which are oriented orbifolds with corners.
Theorem \ref{thm:comb_model} is proven based on these definitions; more precisely, it assumes that the moduli spaces of graded surfaces are oriented orbifolds with corners, that the orientations satisfy some compatibility properties along nodal strata, and that (special) canonical multisections can be found.
Since \cite{ST}, which proves these assumptions in the positive genus case, has not appeared yet, in addition to defining everything we use, we also review the arguments.

First, the fact that the moduli of graded surfaces are smooth orbifolds with corners is a technical result, whose proof imitates of the proof of Theorem 2 of \cite{Zernik}, and is provided in Subsection \ref{subsub:moduli}. Second, the construction of special canonical boundary conditions is similar to the proof of Lemma 3.53~(a) in \cite{PST}, and appears in Subsection \ref{subsub:bc}.

On the other hand, proving that the high genus moduli is orientable, constructing the orientations and showing their properties is more involved, and is based on the discovery of the open Arf invariant in \cite{ST0}. However, in Section \ref{sec:critical} and Subsection \ref{sec:power_of_2}, we provide completely different proofs for the orientability and the orientation properties we need, using the stratification of the moduli and properties of Kasteleyn orientations.

It is also worth mentioning that one of the main results of \cite{PST,ST} is the independence of the open intersection numbers on choices. This fact is also a byproduct of the proof of Theorem \ref{thm:comb_model}, which uses just the defining properties of canonical boundary conditions and not a specific canonical multisection.
\end{rmk}

\subsubsection{Examples}
$\langle\tau_1\tau_0\sigma\rangle_0=1.$ Thus, for $g=0,k=1,l=2$ the left hand side of Equation \ref{eq:combinatorial formula} with $\lambda_1=\lambda,\lambda_2=\mu,$ is $\frac{2}{\lambda\mu^3}+\frac{2}{\mu\lambda^3}.$
The right hand side receives contributions from several graphs, see Figure \ref{fig:ex}, $(a).$
The two non nodal contributions in the first line are $\frac{1}{\lambda(\lambda+\mu)\mu^2}+\frac{1}{\mu(\lambda+\mu)\lambda^2}.$ The two non nodal contributions in the second line are
$\frac{2}{2\lambda^3(\lambda+\mu)}+\frac{2}{2\mu^3(\lambda+\mu)}.$ The nodal ones sum to $\frac{1}{\lambda\mu^3}+\frac{1}{\mu\lambda^3}.$ And the two sides agree.

The second example is of $\langle\tau_1\rangle_1=\frac{1}{2}.$ Consider case $(b)$ in Figure \ref{fig:ex}. The left hand side is $\frac{1}{\lambda^3}.$ Non nodal terms do not contribute, as the single relevant graph (leftmost graph of Example $b$) is not odd. The nodal contribution is exactly $\frac{1}{\lambda^3}.$

The last example $(c),$ is of $\langle\tau_2\sigma^5\rangle=8.$ The left hand side gives $\frac{384}{\lambda^5}.$
$24$ non nodal diagrams, one for each cyclic order of the boundary points, contribute $\frac{24}{\lambda^5}.$ There are $120$ diagrams with a single node, one for each order, each contributes $\frac{1}{\lambda^5}.$ There are $120$ diagrams with two nodes, each contribute $\frac{2}{\lambda^5},$ where $2$ comes from the Catalan term.
\begin{figure}
\centering
\includegraphics[scale=.4]{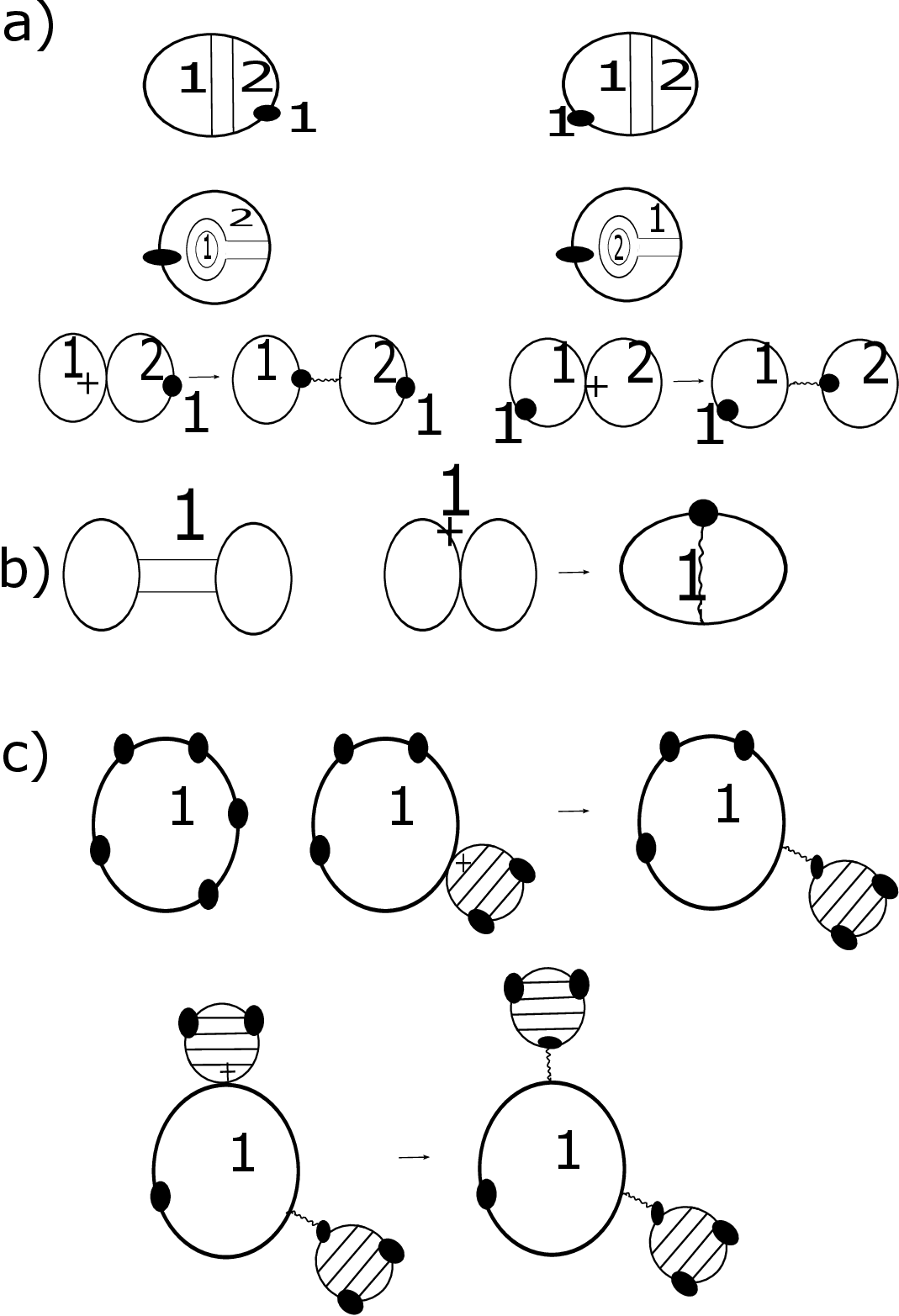}
\caption{Examples of contributing graphs.}
\label{fig:ex}
\end{figure}

\subsection{Proof of the conjecture and related works}\label{subsec:further}
Some recent developments, related works and open question are summarized below.
\begin{enumerate}
\item\textit{Proof of the open KdV conjecture.}
Based on the combinatorial formula presented here, the conjecture of \cite{PST} has been proven in \cite{BT}:
First, the combinatorial formula was transformed to a formula of matrix integrals, and then, by analytical tools and ideas from the theory of integrable hierarchies, the integral was shown to satisfy the open Virasoro constraints (which are equivalent to the open KdV equations by \cite{Bur}).
\item\textit{Boundary descendents.}
\cite{Bur2} showed that the string solution of the open KdV equation is closely related to the wave function of the KdV hierarchy.
In \cite{Bur} a more general generating function, which is a tau-function of the Burgers-KdV system, was introduced. It was conjectured that this function should correspond to an open intersection theory which includes \emph{descendents of boundary marked points}. Such a theory can be constructed, extending the construction of \cite{ST}, and, based on the combinatorial construction in this paper and on \cite{BT}, this theory can be shown to satisfy the Burgers-KdV hierarchy. The definition of the extended theory, its calculation and the proof of its relation with the Burgers-KdV hierarchy will appear soon.
\item\textit{Kontsevich-Penner matrix model, Refined open intersection numbers.}
An alternative description of the solution of the Burgers-KdV equations in terms of matrix integrals was found algebraically by A. Alexandrov in \cite{Alexandrov} in terms of the $N=1$ specification of the Kontsevich-Penner tau function.
\begin{open}\label{open2}
Is there a direct geometric way to derive Alexandrov's solution of the open KdV equations from the geometric construction of \cite{PST},\cite{ST}?
\end{open}
The combinatorial construction presented here was used in \cite{ABT} to write a formula for more refined open intersection numbers. The main conjecture of \cite{ABT}, which is a strengthening of a conjecture of \cite{Saf16a}, is that the generating series of the refined open numbers equals the Kontsevich-Penner tau function.
\item\textit{Open $r-$spin.}
In the recent work \cite{BCT1,BCT2}, a far reaching generalization of \cite{PST} to an intersection theory over the moduli of $r-$spin disks has appeared. The potential of the genus $0$ open $r-$spin integrals was shown to be closely related to the wave function of the $r$KdV hierarchy, and an all-genus generalization was conjectured. In a work in progress with Gross and Kelly this construction is being generalized to open FJRW theory, and the genus $0$ intersection numbers are explained using mirror symmetry.
\begin{open}\label{open3}
Generalize the formula presented in this work to the case of open $r-$spin intersection numbers.
\end{open}
\item\textit{Other interpretations of the theory.}
There were several related works in the physics literature; we mention two. In \cite{DijkWit} Dijkgraaf and Witten provide a physical interpretation to the open intersection theory of \cite{PST,ST}. In \cite{Rim}, Bawane, Muraki and Rim describe a solution for the open KdV equations in terms of minimal gravity on the disk.

In \cite{Saf16b}, Safnuk gives an interpretation of the $N=1$ specification of the Kontsevich-Penner tau function (which is, as explained above, a solution of the Burgers-KdV hierarchy) in terms of combinatorially defined volumes of moduli spaces.
\item\textit{Similar formulas for other OGW invariants.}
There are two newer works which present formulas for open GW invariants in terms of summation over graphs with boundary nodes and are of the same flavour as the formula given here, and the refined formula of \cite{ABT}.
\cite{ZernikPrep} presents an equivariant localization calculation of OGW disk invariants for the pair $(\C\mathbb{P}^{2n},\R\mathbb{P}^{2n}).$
\cite{BPTZ} constructs the stationary OGW theory of $(\C\mathbb{P}^1,\R\mathbb{P}^1)$, derives a localization formula for all intersection numbers, including descendents, and in \cite{BPTZ2} it is used it to prove a correspondence with open Hurwitz theory. Both formulas contain corners contributions, in addition to the naive contributions, in resemblance to \eqref{eq:combinatorial formula}. To the best knowledge of the author, such formulas have not appeared in literature before. Formulas for open GW invariants have appeared in the past, usually in the context of equivariant localization (see the calculations of \cite{KatzLiu} as a prototypic example). In the older formulas which involved graph summation, the graphs were dual to topological stable marked surfaces with boundaries (which parameterized fixed point loci). These surfaces included disk components which were connected by internal nodes to the closed part. There were no boundary nodes. The amplitudes of such graphs were usually similar to the analogous amplitudes in the closed case (and the disk contribution was usually more or less the square root of the sphere contribution). In the formulas of this work and of \cite{ABT,BPTZ,ZernikPrep} the boundary nodes contribute an additional factor to the amplitudes. 
It would be interesting to gain a general understanding of this new type of expressions, 
to understand when are they expected to appear, and to analyze them.
\end{enumerate}

\subsection{Plan of the paper}
In Section \ref{sec:geom} the constructions of \cite{PST},\cite{ST} are reviewed.
In particular, graded spin surfaces are defined, as well as their moduli space $\oCM_{g,k,l},$ tautological line bundles and special canonical boundary conditions. With these in hand, the open intersection numbers are then defined.

In section \ref{sec:sphere} the notions of sphere bundles and angular forms are recalled. It is explained how to calculate the integral of the relative Euler class, relative to nowhere vanishing boundary conditions. The main result of this section is an explicit formula for a representative of the angular form of a sphere bundle. This formula is the starting point of the paper.

Section \ref{sec:strebel} is devoted to constructing an open analog of Strebel's stratification. Symmetric stable Jenkins-Strebel differentials are defined, and used to stratify the moduli space of open surfaces, and then the moduli of graded surfaces.
In addition combinatorial sphere bundles are constructed.
It is then shown that special canonical multisections are pulled back from the combinatorial moduli. The main result of this section is that the open descendent integrals can be calculated as integrals over the combinatorial moduli.

Section \ref{sec:critical} describes in more details the cells in the stratification which will eventually contribute to the open descendents. Extended Kasteleyn orientation are defined, and equivalence classes of them are shown to be equivalent to the data of a graded spin structure. The Kasteleyn orientations are used to provide a more explicit description of the contributing cells, of the boundary conditions and of the orientation of the moduli. As a byproduct, an alternative proof that the moduli $\oCM_{g,k,l}$ is canonically oriented is given. The analysis of orientations is an important ingredient in the proof.

The last section, Section \ref{sec:last}, proves the combinatorial formula, Theorem \ref{thm:comb_model}.
With the aid of the explicit angular form constructed in Section \ref{sec:sphere} an integral representation of the open gravitations descendent is given. The integral depends explicitly on the boundary conditions. The properties of special canonical multisections are then used to iteratively integrate by parts, until an integrated form of the combinatorial formula, Theorem \ref{thm:int_form}, is obtained. 
Finally, by performing a detailed study of the Kasteleyn orientations and multiplicative constants they contribute\footnote{This study also applies to the closed case, and gives a conceptual calculation in terms of discrete spin structures of a constant appearing in Kontsevich's work \cite[Appendix C]{Kont}, which was the subject of several other works.}, we are able to apply Laplace transform to the integrated formula and obtain the main theorem, Theorem \ref{thm:comb_model}.

\subsection{Acknowledgments}
\[\text{The author dedicates this paper to his late grandfather, and first teacher, Yaki Shkolnik.}\]

The author would like to thank J. Solomon for the encouragement to work on this problem, and for many stimulating discussions.
In addition, the author would also like to thank A.~Alexandrov, A.~Buryak, K.~F.~Gutkin, D.~Kazhdan, A.~Lavon, A.~Okounkov, R.~Pandharipande, E.~Shustin, A.~Solomon, A.~Zernik and D.~Zvonkine for interesting discussions related to the work presented here. The author wishes to thank the anonymous reviewers for careful reading of this work, and for many valuable remarks that vastly improved the article. 

The author (incumbent of the Lillian and George Lyttle Career Development Chair)  was supported by the ISF grant No. 335/19 and by a research grant from the Center for New Scientists of Weizmann Institute.
In addition, he was supported by ISF Grant 1747/13 and ERC Starting Grant 337560 in the group of J.~P.~Solomon at the Hebrew university of Jerusalem.

\section{The moduli, bundles and intersection numbers}\label{sec:geom}
This section briefly summarizes the required definitions and results from \cite{PST,ST}.

\subsection{General conventions and notations}
For $l\in\mathbb{N}$ we write $[l] = \{1,2,\ldots,l\}.$ The set $[0]$ will denote the empty set.

Throughout this article a map $\mmm:A\to\Z,$ from an arbitrary set $A$ which is injective away from $\mmm^{-1}(0)$ will be called a \emph{marking} or a \emph{marking of $A$}. Given a marking, we shall identify elements of $\mmm^{-1}(\Z\setminus\{0\})$ with their images.

In what follows, the markings will be used to mark points in surfaces, half edges in dual graphs and vertices in ribbon graphs. The reason we allow non injective marking functions is that we will have to perform many graph or surface operations that will create new marked points. There will be no natural way to mark these new points, and therefore we will mark them all by $0.$

We will encounter many types of graphs in the next sections. Dual graphs, to be defined in Section \ref{sec:geom}, will be denoted by capital Greek letters. Ribbon graphs, to be defined in Sections \ref{sec:strebel},\ref{sec:critical}, will be denoted by capital English letters.

Many of the objects in this paper, such as surfaces or graphs, will have natural notions of genus, boundary labels and internal labels. A \emph{(g,B,I)-}object is an object whose genus is $g,$ the set of boundary labels is $B,$ and the set of internal labels is $I.$ Similarly, in the closed setting, a \emph{(g,I)-}object is an object whose genus is $g$ and the set of internal labels is $I.$

Given a permutation $\pi$ on a set $S,$ we write $s/\pi$ the $\pi-$cycle of $s\in S.$ For $a\in S/\pi,$ write $\pi^{-1}(a)$ for the elements which belong to the cycle $a.$

We shall sometimes use the shorthand notation $\mathbf y$ to denote a sequence $\{y_i\}_{i\in[r]},$ if the sequence we are referring to is understood from context.


\subsection{Open surfaces and their moduli space}
\subsubsection{Stable open surfaces}
We recall the notion of a \emph{stable marked open surface}.

\begin{definition}
We define a \emph{smooth pointed surface} to be a triple
\[
\left(\Sigma,\mathbf{x},\mathbf{z}\right)=\left(\Sigma, \{x_i\}_{i\in \B}, \{z_i\}_{i\in \I}\right)
\]
where
\begin{enumerate}
\item
$\Sigma$ is a Riemann surface, possibly with boundary.
\item
An injection $\B\to \partial \Sigma,~~i\mapsto x_i,$ where $\B$ is a finite set.
\item
An injection $\I\to \mathring{ \Sigma},~~i\mapsto z_i,$ where $\I$ is a finite set.
\end{enumerate}
%
In case $\partial\Sigma\neq\emptyset,$ we say that $\Sigma$ is an open surface. Otherwise it is closed.
We sometimes omit the marked points from our notations. Given a smooth pointed surface $\Sigma$, we write $B\left(\Sigma\right)$ for the set $\B,$ and sometimes also for the set $\{z_i\}_{i\in \B}$. We similarly define $I\left(\Sigma\right).$

A smooth closed pointed surface $\Sigma$ is called \emph{stable} if
\[
2g(\Sigma)+|I\left(\Sigma\right)|>2.
\]

A smooth open pointed surface $\Sigma$ is called \emph{stable} if
\[
2g(\Sigma)+|B\left(\Sigma\right)|+2|I\left(\Sigma\right)|> 2.
\]
\end{definition}

\begin{rmk}
$\Sigma$ is canonically oriented, as a Riemann surface. In case $\partial\Sigma\neq\emptyset,$ it is endowed with a canonical induced orientation.
\end{rmk}
\begin{definition}\label{def:dobuling}
For a pointed Riemann surface $\left(\Sigma, \{x_i\}_{i\in \B}, \{z_i\}_{i\in \I}\right),$ where in case $\Sigma$ is closed $\B=\emptyset,$ we denote by $\left(\overline{\Sigma}, \{x_i\}_{i\in \B}, \{\bar{z}_i\}_{i\in \I}\right)$ the same surface with opposite complex structure.
The \emph{doubling} of an open $\Sigma$ is
\[
\Sigma_\C = \Sigma\coprod_{\partial\Sigma}\overline\Sigma,
\]
the surface obtained by \emph{Schwarz reflection principle along the boundary $\partial\Sigma.$}
For an open connected $\Sigma$ we define the \emph{genus} $g(\Sigma)$ to be the genus of $\Sigma_\C.$ For $\Sigma$ closed and connected the genus is just the usual genus. In case $\Sigma$ is disconnected its genus is defined as the sum of the genera of its connected components.
\end{definition}
\begin{rmk}\label{rmk:top_type_open}
Note that for open surfaces the topological type is determined by two numbers, the doubled genus $g$ and the number of boundary components $h,$ and not only by the genus. $h$ is constrained by
\[h=g+1~(\text{mod }~2),~~0\leq h\leq g+1,\]
and for any $(g,h)$ satisfying these constraints there is a topological type of open surfaces.
\end{rmk}

\begin{definition}\label{def:prestable}
A \emph{pre-stable} surface is a tuple
\[
\Sigma = \left(\left\{\Sigma_\alpha\right\}_{\alpha\in\mathcal{O}\cup \mathcal{C}}, \sim = \sim_B\cup\sim_I,\text{CB}\right)
\]
where
\begin{enumerate}
\item
$\mathcal{O}$ and $\mathcal{C}$ are finite sets.
For $\alpha\in\mathcal{O},~\Sigma_\alpha$ is an open smooth pointed surface; for $\alpha\in\mathcal{S},~\Sigma_\alpha$ is a closed smooth pointed surface.
\item
$\sim = \sim_B\cup\sim_I,$ where $\sim_B$ is an equivalence relation on $\bigcup_\alpha B(\Sigma_\alpha),$ with equivalence classes of size at most $2,$ and $\sim_I$ is an equivalence relation equivalence relation on $\bigcup_\alpha I(\Sigma_\alpha),$ with equivalence classes of size at most $2$. We write $B(\Sigma), I(\Sigma)$ for the equivalence classes of size $1$ of $\sim_B,\sim_I$ respectively.
\item
$\text{CB}(\Sigma)$ is a subset of $ I(\Sigma).$
\end{enumerate}
Elements of $B(\Sigma)$ are called \emph{boundary marked points}. Elements of $I(\Sigma)\setminus \text{CB}(\Sigma)$ are called \emph{internal marked points}.
The $\sim_B$ (resp. $\sim_I$) equivalence classes of size $2$ are called boundary (resp. interior) nodes, elements which belong to these equivalence classes are called \emph{half nodes}.
Element of $\text{CB}$ are called contracted boundaries.
The equivalence classes of $\sim, (\sim_B,\sim_I)$ are collectively called \em{special (special boundary, special internal) points} of $\Sigma.$

We also write $\Sigma = \coprod_{\alpha\in\mathcal{O}\cup\mathcal{C}}\Sigma_\alpha/\!\!\sim.$ If $\mathcal{O}$ is empty and $\text{CB}$ is empty, $\Sigma$ is called a \emph{pre-stable closed surface}. Otherwise it is called a \emph{pre-stable open surface}.

A pre-stable surface is \emph{marked}, if in addition it is endowed with \emph{markings} $\mmm^B:B(\Sigma)\to \Z,~\mmm^I:I(\Sigma)\setminus{\text{CB}}\to \Z.$ Write $\mmm=\mmm^I\cup \mmm^B.$
Recall that a marking is injective outside of the preimage of $0.$

A pre-stable marked surface is called a \emph{stable marked surface} if each of its constituent smooth surfaces $\Sigma_\alpha$ is stable.

The \emph{doubled surface} $\Sigma_\C$ of a stable open surface is defined as
\[
\Sigma_\C = (\coprod_{\alpha\in\mathcal{O}}(\Sigma_\alpha)_\C \coprod_{\alpha\in\mathcal{C}}\Sigma_\alpha\coprod\overline{\Sigma}_\alpha  )/\sim_\C,
\]
where \[\sim_\C=\left(\sim_B\cup\sim_I\cup\sim_{\bar I}\cup\sim_{\text{CB}}\right)\] is defined as follows. $\sim_{\bar I}$ identifies internal marked points of $\{\overline{\Sigma}_\alpha\}_{\alpha\in\mathcal{C}}$ if and only if $\sim_I$ identifies the corresponding marked points in $\{\Sigma_\alpha\}_{\alpha\in\mathcal{C}}.~\sim_{\text{CB}}$ identifies $z_i\in\Sigma_\alpha,\bar{z}_i\in\overline{\Sigma}_{\alpha}$ whenever $i\in \text{CB}(\Sigma).$
$\Sigma_\C$ is endowed with an involution $\varrho,$ with $\bar{z}_i=\varrho(z_i),$ whose fixed point set is $\partial\Sigma\cup\text{CB}(\Sigma),$ and such that $\Sigma\simeq\Sigma_\C/\varrho.$
Write $D(\Sigma)=(\Sigma_\C,\varrho).$

$\Sigma$ is connected if the underlying space, $\coprod_{\alpha\in\mathcal{D}\cup\mathcal{S}}\Sigma_\alpha/\!\!\sim$ is.
$\Sigma$ is smooth if $\text{CB}(\Sigma)=\emptyset,$ and $\sim$ has only equivalence classes of size $1.$

The \emph{normalization} $\NNN(\Sigma)$ of the stable marked surface $\Sigma$ is the surface $\left(\left\{\Sigma_\alpha\right\}_{\alpha\in\mathcal{O}\cup \mathcal{C}}, \sim',
\text{CB}',\mmm'\right)$
where $\sim'$ has only size $1$ equivalence classes, $\text{CB}'$ is empty, and the marking $\mmm'$ agrees with $\mmm$ whenever is defined, and otherwise $\mmm^{'I}=\III,\mmm^{'B}=\BBB.$
For a marked point marked $i\neq0,$ write $\Sigma_i$ for the component of $\NNN(\Sigma)$ which contains marked point $z_i.$


A \emph{topological stable marked surface}, open or closed, is defined in the same way, only with $\Sigma_\alpha$ being topological surfaces rather than Riemann surfaces.
\end{definition}
In what follows, our defaultive choice of marking function $\mmm$ is a bijection $\mmm^I:I(\Sigma)\to[n]$ in case $\Sigma$ is closed, and if $\Sigma$ is open we usually take bijections $\mmm^I:I(\Sigma)\setminus\text{CB}(\Sigma)\to[l],~\mmm^B:B(\Sigma)\to[k]$. Therefore whenever a surface is written as $(\Sigma,z_1,\ldots,z_n)$ or $(\Sigma,x_1,\ldots,x_k,z_1,\ldots,z_l)$ we implicitly mean it is marked, and that the indices of the marked points represent the markings.

See figure \ref{fig:2_4} for examples of pre-stable surfaces and their normalizations.
\begin{figure}
\centering
\includegraphics[scale=.4]{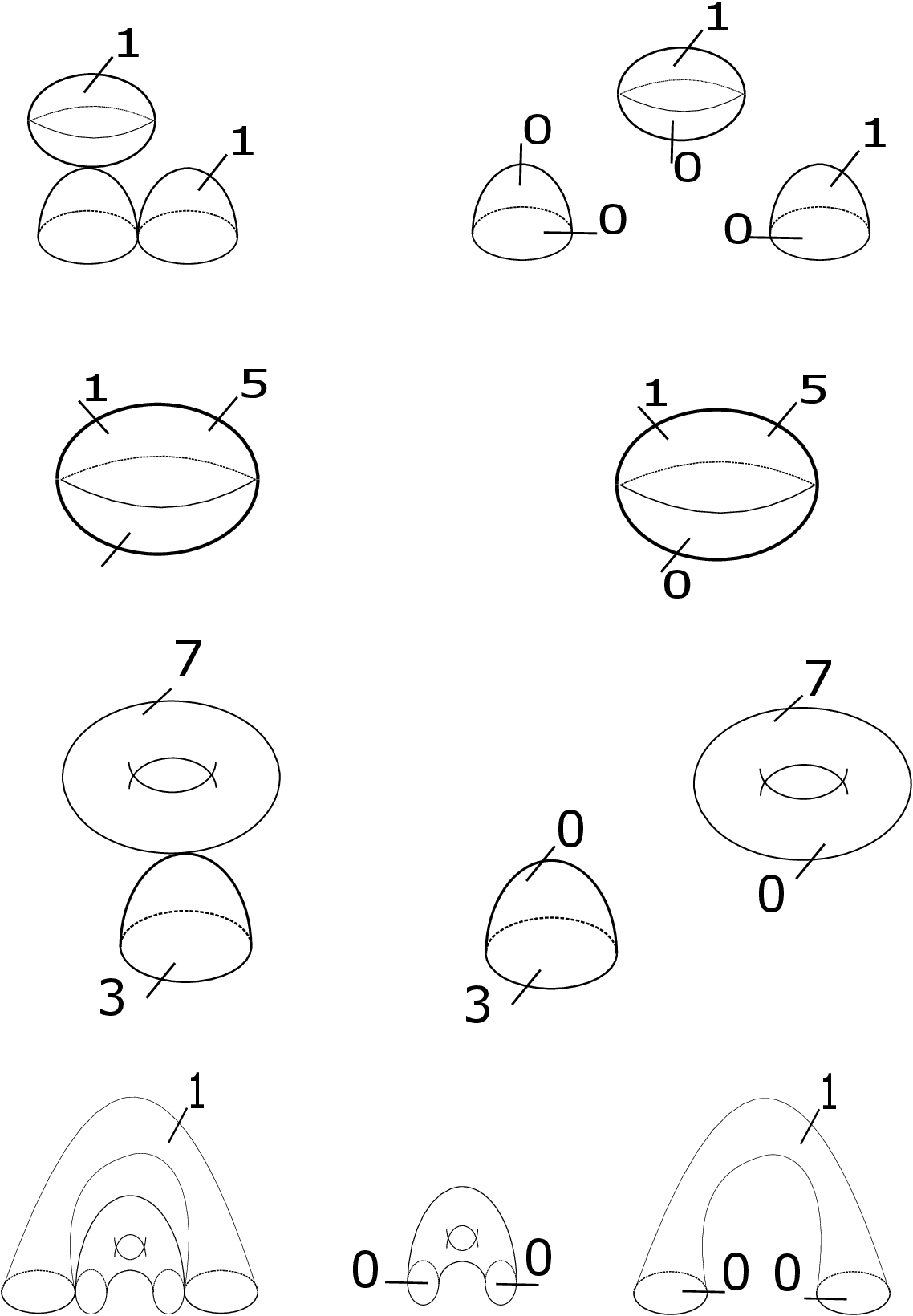}
\caption{In this diagram in every row the leftmost picture is a pre-stable surface, and on the right side of the same row is the normalization. In the top row there is a pre-stable marked surface with boundary, and its normalization into two stable marked disks and a prestable marked sphere. In the second row there is a stable sphere with an (unmarked) contracted boundary. Its normalization is a stable sphere with three markings. In the third row there is a stable surface with boundary which is normalized into a disk and a torus. The last row contains a stable surface whose normalization is the union of a cylinder and a genus $3$ surface with boundary.}
\label{fig:2_4}
\end{figure}

We sometimes identify $D(\Sigma),$ and $\Sigma_\C.$ 
%

\begin{definition}
An \emph{isomorphism} between two pre-stable marked surfaces, $\Sigma = \left(\left\{\Sigma_\alpha\right\}_{\alpha\in\mathcal{O}\cup \mathcal{C}}, \sim, \text{CB}, \mmm\right)$ and $\Sigma' = \left(\left\{\Sigma'_\alpha\right\}_{\alpha\in\mathcal{O}'\cup \mathcal{C}'}, \sim', \text{CB}', \mmm'\right)$
is a tuple
$
f = \left(f^{\mathcal{O}},f^{\mathcal{C}},\{f^\alpha\}_{\alpha\in\mathcal{O}\cup\mathcal{C}}\right)
$
such that
\begin{enumerate}
\item $f^\mathcal{O}:\mathcal{O}\to\mathcal{O}',~f^\mathcal{C}:\mathcal{C}\to\mathcal{C}'$ are bijections between the sets which index the components of the surfaces.
\item
For $\alpha\in\mathcal{O},$ 
$f^\alpha:\Sigma_\alpha\to\Sigma'_{f^{\mathcal{O}}(\alpha)}$ is a biholomorphism which induces a bijection on the sets of special points.
For $\alpha\in\mathcal{C},~f^\alpha:\Sigma_\alpha\to\Sigma'_{f^{\mathcal{C}}(\alpha)}$ is a biholomorphism which induces a bijection on the sets of special points.
\item $x\sim y$ for $x\in \Sigma_\alpha,~y\in \Sigma_\beta$ if and only if $f^\alpha(x)\sim'f^\beta(y).$
\item For any special point $x\in\Sigma_\alpha,$ $\mmm'(f^\alpha(x))=\mmm(x).$
\item
$\bigcup_\alpha f^\alpha(\text{CB})=\text{CB}'.$
\end{enumerate}
We denote by $\text{Aut}(\Sigma)$ the group of the automorphisms of $\Sigma.$

An isomorphism between stable topological surfaces is similarly defined, only that the maps $f^{\alpha}$ are required to be homeomorphisms rather than biholomorphisms.
\end{definition}

\subsubsection{Stable graphs}
It is useful to encode some of the combinatorial data of stable marked surfaces in graphs.
\begin{definition}
A (not necessarily connected) \emph{pre-stable dual graph $\Gamma$} is a tuple
\[\left(V = V^{O}\cup V^{C} ,H = H^{B}\cup H^{I}, \sig_0, \sim = \sim_B\cup\sim_I,g,H^{\text{CB}}, \mmm=\mmm^B\cup \mmm^I\right),\]
where
\begin{enumerate}
\item
$V^O , V^C$ are finite sets called \em{open and closed vertices}, respectively.
\item
$H^B,H^I$ are finite sets of \em{boundary and internal half edges}.
\item
$\sig_0:H\to V$ associates any half edge to its vertex.
\item
$\sim_B$ is an equivalence relation on $H^B$ with equivalence classes of sizes $1$ or $2.$
Denote by $T^B$ the equivalence classes of size $1$ of $\sim_B.$
$\sim_I$ is an equivalence relation on $H^I$ with equivalence classes of sizes $1$ or $2.$
Denote by $T^I$ the equivalence classes of size $1$ of $\sim_I.$
\item
$H^{\text{CB}}\subseteq T^I.$
\item
$g:V\to\Z_{\geq 0}$ is a genus assignment.
\item
$\mmm^B:T^B\to\Z,~\mmm^I:T^I\setminus H^{\text{CB}}\to\Z$ are markings.
\end{enumerate}
We call $T^B$ \emph{boundary tails}, $H^{\text{CB}}$ \emph{contracted boundaries}, and $T^I\setminus H^{\text{CB}}$ \emph{internal tails}. Set $T=T^I\cup T^B.$
$\sim_B$ induces a fixed point free involution on $H^B\setminus T^B.$
Similarly, $\sim_I$ induces a fixed point free involution on $H^I\setminus T^I.$
We denote this involution on $H\setminus T$ by $\sig_1.$
We set $E^B = (H^B\setminus T^B)/{\sim_B},$ the set of boundary \em{edges}. We define $E^I=(H^I\setminus T^I)/{\sim_I}\cup{H^{\text{CB}}}.$ We put $E=E^I\cup E^B,$ the set of edges.
We denote by $\sig_0^B$ the restriction of $\sig_0$ to $H^B,$ in a similar fashion we define $\sig_0^I.$

We require that for all $h\in H^B,~\sig_0(h)\in V^O.$


We say that $\Gamma$ is connected if its underlying graph, $\left(V,E\right)$ is connected.

For a vertex $v$ we set $k(v) = |(\sig_0^B)^{-1}(v)|.$ It is defined to be $0$ if $v$ is closed. We set $l(v) = |(\sig_0^I)^{-1}(v)|.$ Write $\text{CB}(v)$ for the number of contracted boundaries of $v.$ A dual graph is \emph{closed} if $V^O=H^{\text{CB}}=\emptyset,$ and otherwise it is \emph{open}.

The \em{genus of a stable connected closed dual graph} $\Gamma$ is defined by
\[
g(\Gamma) = \sum_{v\in V^C} g(v) +|E^I|-|V^C|+1.
\]
The \emph{genus of a stable connected open dual graph} $\Gamma$ is defined by
\[
g(\Gamma) = \sum_{v\in V^O} g(v) + 2\sum_{v\in V^C} g(v) + |E^B|+2|E^I|-|H^{\text{CB}}|-|V^O|-2|V^C|+1.
\]

A closed vertex $v\in V^C$ is \em{stable} if
\[2g(v)+l(v)>2.\]

An open vertex $v\in V^O$ is \em{stable} if
\[2g(v)+k(v)+2l(v)>2.\]

A dual graph $\Gamma$ is stable if all its vertices are.

$\NNN(\Gamma),$ the \emph{normalization} of the graph $\Gamma$ is the unique stable graph $\left(V' ,H' , \sig'_0, \sim' ,g',H^{'\text{CB}}, \mmm'\right)$
with $V'=V,H'=H,\sig'_0=\sig_0,g'=g,H^{'\text{CB}}=\emptyset,$ and $\sim'$ has only classes of size $1.$
$\mmm'$ agrees with $\mmm,$ whenever $\mmm$ is defined. 
Otherwise $\mmm'=0.$

For $i\in \text{Image}(\mmm^I)\setminus\{0\},$ we denote by $v_i(\Gamma)$ the connected component of $\NNN(\Gamma)$ which contains the tail marked $i.$
\end{definition}
It is easy to see that the genus is always non negative.
Figure \ref{fig:2_6} illustrates several dual graphs and their normalizations. Note that open vertices without boundary half edges are allowed.

\begin{definition}
An \emph{isomorphism} between graphs \[\Gamma = (V,H, \sig_0, \sim,g,H^{\text{CB}},\mmm),~\Gamma' = (V',H', \sig'_0, \sim',g',H^{'\text{CB}},\mmm')\] is a pair $f = \left(f^V,f^H\right)$
such that
\begin{enumerate}
\item
$f^V:V\to V'$ is a bijection; $f^H:H\to H'$ is a bijection.
\item
${g'}\circ f = g.$
\item
$h_1\sim h_2 \Leftrightarrow f(h_1)\sim' f(h_2).$
\item
$\sig'_0 = f\circ \sig_0.$
\item
$\mmm'\circ f = \mmm.$
\item
$f(H^{\text{CB}})=H^{'\text{CB}}.$
\end{enumerate}
We denote by $\text{Aut}(\Gamma)$ the group of the automorphisms of $\Gamma.$
\end{definition}

To each stable marked surface $\Sigma$ we associate an isomorphism class of connected stable graphs as follows. We set $V^O = \mathcal O$ and $V^C = \mathcal C.$
$H^B=\bigcup_\alpha B \left(\Sigma_\alpha\right),~H^I=\bigcup_\alpha I \left(\Sigma_\alpha\right).~H^{\text{CB}}=\text{CB}(\Sigma).$
The definitions of $g,\sim,\sig_0,\mmm$ are straightforward. 
In particular, a tail marked $a$ is associated to a marked point labeled $a.$ An edge between two vertices corresponds to a node between their corresponding components. See Figure \ref{fig:2_6} for the dual graphs which correspond to the surfaces of Figure \ref{fig:2_4}.
Note that this correspondence is at the level of isomorphism classes of topological stable surfaces, and that a surface is closed precisely if its corresponding graph is closed.
\begin{definition}\label{def:genus_stable_surf}
The graph associated to a stable surface $\Sigma$ is denoted by $\Gamma\left(\Sigma\right)$.
The \emph{genus} of a stable surface $\Sigma$ is defined as the genus of $\Gamma(\Sigma).$
\end{definition}
Observe that the genus of a stable closed surface agrees with the standard definition, while the genus of a stable open surface equals the standard genus of its doubled surface. The genus of a stable surface equals the genus of the surface obtained by smoothing its nodes, including the contracted boundaries which are smoothed to boundary components.
Observe also that $\NNN(\Gamma(\Sigma))=\Gamma(\NNN(\Sigma)),$ and that for any internal marked point which is marked $i\neq 0,$ $v_i(\Gamma(\Sigma))=\Gamma(\Sigma_i),$ where $\Sigma_i$ is the component of $\Sigma$ which contains marked point $z_i.$

\begin{figure}
\centering
\includegraphics[scale=.3]{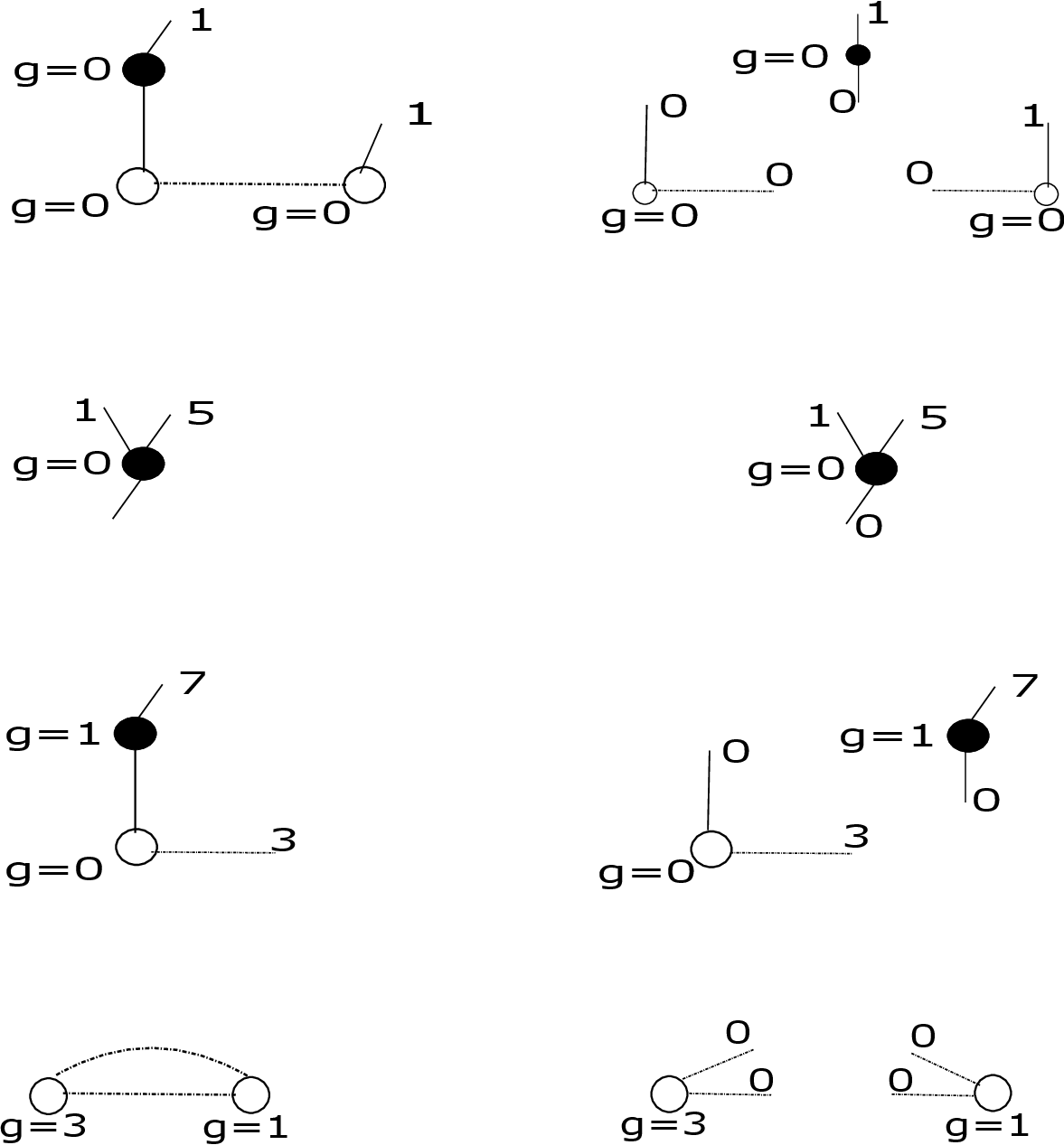}
\caption{This diagram presents the dual graphs which correspond to the surfaces from Figure \ref{fig:2_4}, under the correspondence of Definition \ref{def:genus_stable_surf}. Again the right hand side of each row is the normalization of the left hand side. Black vertices correspond to closed components, and empty vertices to open. The genus of the vertex is written next to it. Boundary edges or half edges are drawn as dashed lines, and the other edges or half edges are internal (the case of contracted boundary is included). The label of a tail is written next to it.
The genus of the dual graphs in the left hand side are, going from top to bottom, $0,0,2,5.$}
\label{fig:2_6}
\end{figure}

Throughout this paper we will sometimes write 'graph' instead of 'dual graph' when the meaning is clear from the context.
Dual graphs will be denoted by capital Greek letters, to help us distinguish them from another kind of graphs we shall meet below, ribbon graphs, which will be denoted by capital English letters.

We denote by $\RCG_{g,k,l}$ the set of isomorphism classes of all stable graphs of genus $g,$ with $k$ boundary tails, $l$ internal tails and
\[
\text{Image}(\mmm^B)=[k];\text{Image}(\mmm^I)=[l].
\]
We write $\RCG$ for the set of isomorphism classes of all stable graphs. Note that the cases $k=0$ or $l=0$ are not excluded, as surfaces without boundary or internal marked points will be considered in what follows.

\begin{nn}
Given nonnegative integers $k,l$ with $2g+k+2l>2,$ denote by $\RGamma_{g,k,l}$ the stable graph with $V^O = \left\{*\right\}, V^C = \emptyset,$ with
\[
g(*) = g,~T^B=H^B\simeq[k], T^I=H^I\simeq[l],
\]
where the equivalences with $[k],~[l]$ are obtained using $\mmm^B,\mmm^I,$ respectively. We similarly define $\Gamma_{g,n}$ as the closed graph with a single vertex of genus $g,$ and $T^I=H^I\simeq[n].$
\end{nn}

\begin{definition}
A stable dual graph is \emph{effective} if
\begin{enumerate}
\item Any internal half edge is a tail or a contracted boundary.
\item Any vertex without internal tails has exactly three boundary half edges and genus $0.$
\item Different vertices without internal half edges are not adjacent.
\end{enumerate}
A surface is called \emph{effective} if it is associated to an effective graph.
\end{definition}
The notion of effectiveness will be important later on, when we construct the combinatorial moduli space using Jenkins-Strebel differentials. On moduli strata which correspond to effective dual graphs, the map to the combinatorial moduli is a homeomorphism. This fact will turn out to be useful when we come to translate the geometric intersection numbers to combinatorial expressions.

In the leftmost column of Figure \ref{fig:2_4}, only the sphere from the second row is effective: The surface from the first row has an internal node, and in addition it is not stable, the surface from the third row also has an internal node as well, the surface from the lowest row has a component without internal markings, which is not a disk with three boundary markings. Equivalently, in the leftmost column of Figure \ref{fig:2_6} only the second graph is effective.
Additional examples of effective and non effective surfaces and graphs are illustrated in Figure \ref{fig:2_10}.
\begin{figure}
\centering
\includegraphics[scale=.4]{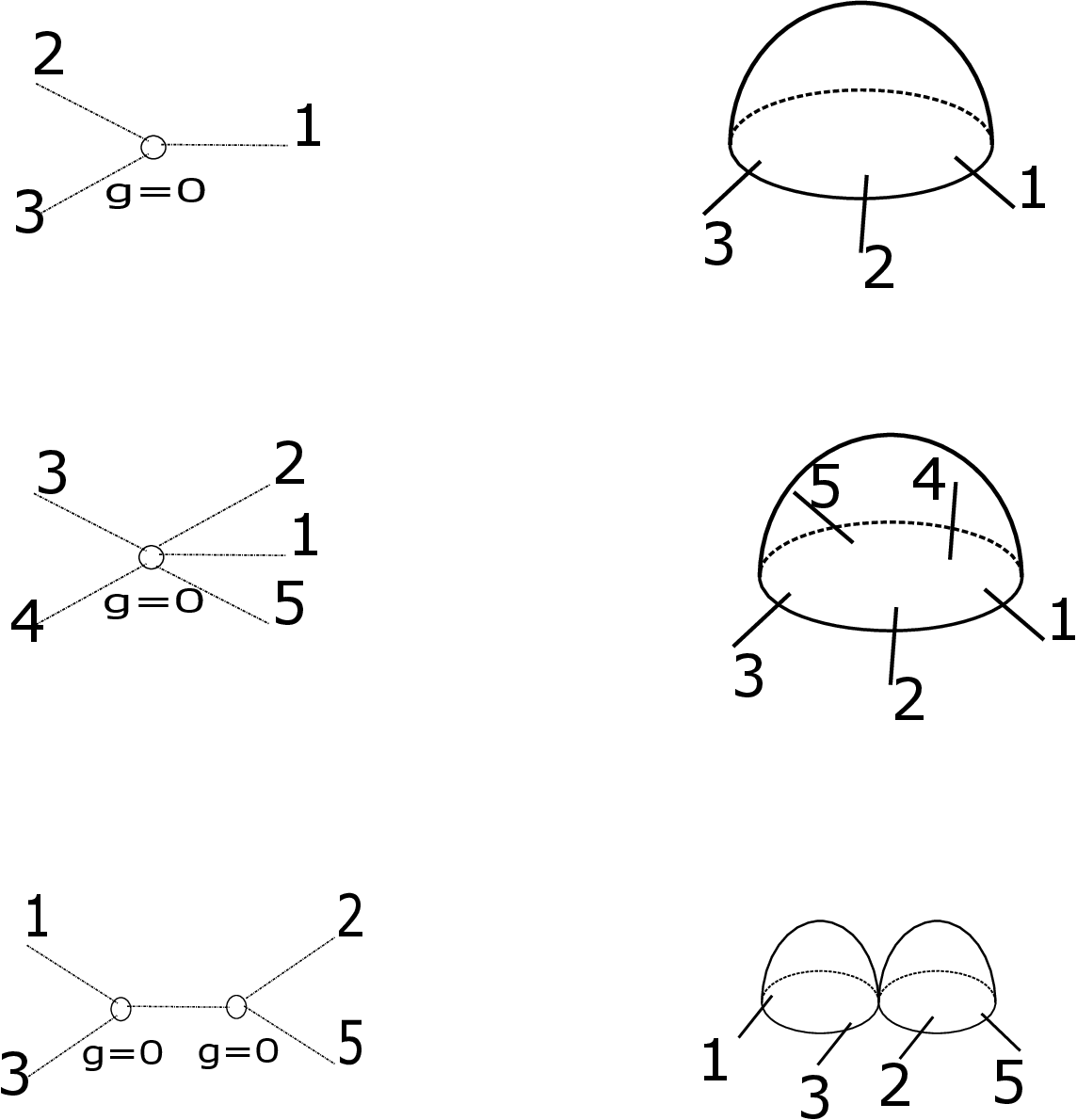}
\caption{Every row in this diagram illustrates a dual graph and the corresponding surface. Only the first row represent an effective graph/surface. Note that the cyclic order of boundary markings on boundaries cannot be read from the dual graph data.
}
\label{fig:2_10}
\end{figure}

\subsubsection{Some graph operations}
For the purpose of the next definition, write, for a vertex $v$ in a dual graph $\Gamma,$ $\varepsilon(v)=1$ if $v$ is open, and otherwise $\varepsilon(v)=2.$ For an edge $e$ set $\varepsilon(e)=0$ unless $e$ is an internal edge connecting two open vertices, and then put $\varepsilon(e)=1.$
\begin{definition}\label{def:smoothing}
Consider a stable graph $\Gamma.$
The \emph{smoothing} of $\Gamma$ at $f\in E$ is the stable graph
\[
d_f\Gamma = \Gamma' = \left(V', H', \sim',s'_0, g',\mmm'\right)
\]
defined as follows.
Suppose $f\notin H^{\text{CB}}(\Gamma)$ is the $\sim-$equivalence class $\{h_1,h_2\},$ write $\sig_0(h_1)=v_1,\sig_0(h_2)=v_2.$
The vertex set is given by
\[
V'= \left(V\setminus\left\{v_1,v_2\right\}\right)\cup\left\{v\right\}.
\]
The new vertex $v$ is closed if and only if both $v_1$ and $v_2$ are closed.
\[
H' = H\setminus\{h_1,h_2\}.
\]
and $\sim'$ is the restriction of $\sim$ to $H'.$
For $h\in \sig_0^{-1}(\{v_1,v_2\})$ we define $\sig'_0(h)=v.$ Otherwise $\sig'_0(h)=\sig_0(h).$
For any tail $t,~\mmm'(t)=\mmm(t).$
\[g'(v)=
\begin{cases}
g(v_1)+1+\varepsilon(f), &\mbox{if}~v_1=v_2 ,\\
g(v_1)+g(v_2)+\varepsilon(f), &\mbox{if}~v_1\neq v_2,~\varepsilon(v_1)=\varepsilon(v_2),\\
\varepsilon(v_1)g(v_1)+\varepsilon(v_2)g(v_2),&\mbox{otherwise.}
\end{cases}
\]
When $f\in H^{\text{CB}},$ a contracted boundary of vertex $v$, then $V'=V,~H'=H\setminus\{f\},H^{'\text{CB}}=H^{\text{CB}}\setminus\{f\}.$ We update $\sim',\sig'_0,\mmm'$ as above.
We put $g'(w)=g(w),$ for $w\neq v,$ and $g'(v)=g(v)+1$ if $v$ is open, otherwise we set $g'(v)=2g(v)$ and declare $v$ to be open. 

Observe that there is a natural proper injection $H'\hookrightarrow H$, so we may identify $H'$ with a subset of $H.$ This identification induces identifications of tails and of edges. Using the identifications, we extend the definition of smoothing in the following manner.
Given a set $S=\left\{f_1 , \ldots, f_n\right\}\subseteq E\left(\Gamma\right)$, define the smoothing at $S$ as
\[
d_S\Gamma = d_{f_n}\left(\ldots d_{f_2}\left(d_{f_1}\Gamma\right)\ldots\right).
\]
Observe that $d_S\Gamma$ does not depend on the order of smoothings performed.
\end{definition}

\begin{definition}\label{def:smoothing-top}
A stable topological surface $\Sigma'$
is a \emph{smoothing} of a topological stable marked surface $\Sigma$ at an internal node $z_\nu\sim z_\mu,$ if there exists a simple closed path $\gamma\hookrightarrow\Sigma',$ and a map $\varphi:\Sigma'\to\Sigma$ which takes $\gamma$ to the node, and restricts to an orientation preserving homeomorphism $\varphi:\Sigma'\setminus\gamma\simeq\Sigma\setminus\{z_\mu,z_\nu\}.$ In this case we say that $\gamma$ is contracted to the node.
~We say that $\gamma$ \emph{degenerates} to $z_\nu$ when this time $\gamma$ is an \emph{oriented} simple closed path in $\Sigma',$ if $\gamma$ is contracted to the node, and
the $\varphi-$preimage of a small enough neighborhood of $z_\nu$ lays in the left of $\gamma.$
~The definitions of smoothing in a boundary node, or degeneration to a boundary half node are analogous, only with a simple arc that connects two boundary points.

A topological stable surface $\Sigma'$ is the smoothing of a topological stable surface $\Sigma$ at a contracted boundary $z_\nu,$ if there exists a boundary component $\partial\Sigma'_\nu,$ and $\varphi:\Sigma'\to\Sigma,$ such that $\varphi(\partial\Sigma'_\nu)=z_\nu,$ and $\varphi:\Sigma'\setminus\partial\Sigma'_\nu\simeq\Sigma\setminus z_\nu.$
\end{definition}
Note that if $e$ is the edge of $\Gamma(\Sigma)$ which corresponds to the node $z_\nu\sim z_\mu$ in $\Sigma,$ then $\Gamma(\Sigma')=d_{e}\Gamma(\Sigma),$ where $\Sigma'$ is the smoothing of $\Sigma$ in that node. Similarly for smoothing in contracted boundaries.

Note that in case $\Gamma = d_S\Gamma'$, then $H'$ is canonically a subset of $H,$ and we have a natural identification between $E\left(\Gamma\right)$ and $E\left(\Gamma'\right)\setminus S$.

We can now define boundary maps
\begin{gather*}
\pu : \RCG \to 2^{\RCG},\qquad\partial : \RCG \to 2^{\RCG},
\end{gather*}
by putting
\[\pu \Gamma = \left\{\left.\Gamma' \right| \exists S\subseteq E\left(\Gamma'\right),\, \Gamma = d_S\Gamma'\right\},~~
\partial\Gamma =\pu\Gamma\setminus\{\Gamma\}.\]
These maps naturally extend to maps $2^{\RCG} \to 2^{\RCG}.$


\subsubsection{Moduli of open surfaces}
In this paper we consider orbifolds with corners; we follow the definitions of \cite[Section 3]{Zernik} (which build on the works \cite{Joyce,Joyce2}).

%
%
\begin{nn}
For $\Gamma \in \RCG$, denote by $\RCM_\Gamma$ the set of isomorphism classes of stable marked genus $g$ surfaces with associated graph $\Gamma.$

Define
\[
\oRCM_\Gamma = \coprod_{\Gamma'\in\pu\Gamma}\RCM_{\Gamma'}.
\]

We abbreviate $\oRCM_{g,k,l}=\oRCM_{\RGamma_{g,k,l}},\RCM_{g,k,l}=\RCM_{\RGamma_{g,k,l}}.$ We similarly define $\oCM_{g,n},~\CM_{g,n},$ which are just the usual Deligne-Mumford moduli spaces of stable and smooth curves respectively.

For $i\in \text{Image}(\mmm^I)\setminus\{0\},$ 
write $\CM_{v_i(\Gamma)}$ for the moduli of the graph $v_i(\Gamma),$ and denote by $v_i:\CM_\Gamma\to\CM_{v_i(\Gamma)}$ the natural map,
which on the level of objects sends $\Sigma\to\Sigma_i.$
\end{nn}

The space $\oRCM_{g,k,l}$ is a compact smooth orbifold with corners of real dimension
\[
\dim_\R \oRCM_{g,k,l} = k + 2l +3g - 3.
\]
We attribute this result to Amitai Netser Zernik \cite[Section 2]{Zernik}.
His setting is slightly different. He considers open stable genus $0$ maps to homogeneous varieties, and he proves that the moduli space of these maps is an orbifold with corners. In our case the target space is a point, but the genus is arbitrary. This change does not affect his results or techniques, since they only rely on convexity of the corresponding closed moduli problem, that is, on the fact the moduli space of (complex) stable maps is a smooth orbifold, which clearly holds for $\oCM_{g,n}.$
We review the argument.
Consider the following sequence:
\begin{equation}
\label{eq:moduli_no_spin_sequence}
\oCM_{g,k,l}^\R\stackrel{(4)}{\hookrightarrow}
\widetilde{\mathcal{M}}_{g,k,2l}^{\R}\stackrel{(3)}{\to}
\widetilde{\mathbb{R}\mathcal{M}}_{g,k,2l}\stackrel{(2)}{\to}
\overline{\mathbb{R}\mathcal{M}}_{g,k,2l}\stackrel{(1)}{\to}\overline{\mathcal{M}}_{g,k+2l}.
\end{equation}
We define the moduli spaces and maps appearing in \eqref{eq:moduli_no_spin_sequence} as follows.

\textbf{Step 1: }
First, $\overline{\mathbb{R}\mathcal{M}}_{g,k,2l}$ is the fixed locus of the involution on $\oCM_{g,k+2l}$ defined by
\[(C;z_1, \ldots, z_{k+2l}) \mapsto (\overline{C}; z_1, \ldots, z_k, z_{k+l+1},\ldots, z_{k+2l},z_{k+1},\ldots,z_{k+l}).\]
where $\overline{C}$ is the same smooth curve $C,$ but with the conjugate complex structure. This is a compact smooth real orbifold, as it is the fixed locus of an anti-holomorphic involution over a smooth complex orbifold. More details on the fixed point functor on stacks can be found in \cite[Section 2.5]{Zernik}. This orbifold parameterizes isomorphism types of stable marked curves with a conjugation.

\textbf{Step 2: }
The next step is to cut $\overline{\mathbb{R}\mathcal{M}}_{g,k,2l}$ along strata which parameterize surfaces with at least one real node. 
These strata form a real normal crossing divisor, as they are the fixed point loci of the previous involution, applied to the normal crossing divisor of nodal strata in $\oCM_{g,k+2l}.$ The cutting procedure is via the \emph{real hyperplane blow-up} of \cite[Section 3.3]{Zernik}, and it is proven there that the result of this blow up is an orbifold with corners which we denote by $\widetilde{\mathbb{\R}\mathcal{M}}_{g,k,2l}$.

\textbf{Step 3: }
$\widetilde{\mathbb{\R}\mathcal{M}}_{g,k,2l}$ is made of several connected component. Consider those components whose generic point is a real curve $C$ with a conjugation $\varrho$ such that $C\setminus C^\varrho$ is disconnected.  Then $\widetilde{\mathcal{M}}_{g,k,2l}^{\R}$ is the disconnected $2$-to-$1$ cover of the union of those connected components, given, at the level of objects $(C,\varrho),$ by the choice of a \emph{distinguished half}, a connected component of $C\setminus C^\varrho$.  Thanks to the real blow up procedure, this choice extends naturally to the boundary strata. The resulting space is still a compact orbifold with corners, as a degree $2$ cover of such a space.

\textbf{Step 4: }
$\oCM_{g,k,l}^\R$ is the submoduli of $\widetilde{\mathcal{M}}_{g,k,2l}^{\R}$ made of connected components such that the marked points $w_{k+1},\ldots, w_{k+l}$ lie in the distinguished half. This final space is a compact orbifold with corners, as it is the union of connected components of a compact orbifold with corners.

Set theoretically $\oCM_{g,k,l}^\R$ is naturally identified with the moduli space of stable marked open $(g,k,l)-$surfaces, and therefore we identify this moduli with $\oCM_{g,k,l}^\R.$  
The construction endows the moduli space $\oCM_{g,k,l}^\R$ with topology and an orbifold with corners structures. For the dimension see, for example, \cite[Theorem 1.2]{Liu}.

In general the space $\oRCM_{g,k,l}$ is non orientable and disconnected.
A stable marked surface with $b$ boundary nodes or contracted boundaries belongs to a corner of the moduli space $\oRCM_{g,k,l}$ of codimension $b.$ For further reading about the nodal strata of the real and open moduli spaces we refer the reader to \cite[Section~3]{Liu}.

\begin{nn}
Denote by $D:\oRCM_{g,k,l}\to\oCM_{g,k+2l}$ the moduli-level doubling map $\Sigma\to \Sigma_\C,$ which is the composition of the maps of \eqref{eq:moduli_no_spin_sequence}.
\end{nn}


\subsection{Graded surfaces and their moduli space}\label{subsec:graded}
We present here the extra structure needed for the definition of intersection theory for open Riemann surfaces, following \cite{ST,ST0}.
\subsubsection{Smooth graded surfaces}\label{subsub:smooth}
Let $\Sigma$ be a smooth closed genus $g$ surface. A \emph{spin structure} twisted in $\{z_i\}_{i\in \I_1},~\I_1\subseteq\I,$
is a complex line bundle $\CCCL\to \Sigma$ together with an isomorphism
\[
b:\CCCL^{\otimes 2}\simeq\omega_{\Sigma}\left(-\sum_{i\in \I_1} z_i\right),
\]
where $\omega_{\Sigma}(-\sum_{i\in \I_1} z_i)$ is the canonical bundle twisted in $\{z_i\}_{i\in {\I_1}}.$

Let $\Sigma$ be a smooth genus $g$ open surface. A \emph{real spin structure} twisted in $\{x_i\}_{i\in \B_1},\{z_i\}_{i\in {\I_1}},$ where $\B_1\subseteq\B,$ and $\I_1\subseteq\I,$
is a triple $(\CCCL,b,\widetilde\varrho),$ where $(\CCCL,b)$ is a spin structure on the doubled surface $D(\Sigma)=(\Sigma_\C,\varrho)$ twisted in $\{x_i\}_{i\in \B_1},\{z_i,\bar{z}_i\}_{i\in {\I_1}},$ i.e., $\CCCL\to\Sigma_\C$ is a line bundle, and
\[
b:\CCCL^{\otimes 2}\simeq\omega_{\Sigma_\C}\left(-\sum_{i\in \B_1}x_i-\sum_{i\in \I_1} (z_i+\bar{z}_i)\right)
\]
is an isomorphism,
where $\omega_{\Sigma_\C}(-\sum_{i\in \B_1}x_i-\sum_{i\in \I_1} (z_i+\bar{z}_i))$ is the canonical bundle twisted in $\{x_i\}_{i\in \B_1},\{z_i,\bar{z}_i\}_{i\in {\I_1}}$.
\\The map $\widetilde\varrho:\CCCL\to\CCCL,$ is an involution which lifts $d\varrho,$ the induced involution on $\omega_{\Sigma_\C}.$

$\widetilde\varrho,d\varrho$ restrict to conjugations on the fibers of $$\CCCL\to {\Sigma_\C^\varrho}\simeq\partial\Sigma,\quad\omega_{\Sigma_\C}(-\sum_{i\in \B_1}x_i)\to {\Sigma_\C^\varrho}\simeq\partial\Sigma.$$
These conjugations give rise to a $\varrho-$invariant real subbundle.
The real line bundle \[\omega_{\Sigma_\C}(-\sum_{i\in \B_1}x_i)^\varrho\to {\Sigma_\C^\varrho}\] is oriented: Take any nowhere vanishing section $\xi\in\Gamma(T {\Sigma_\C^\varrho}\to {\Sigma_\C^\varrho}),$ which points in the direction of the orientation on ${\Sigma_\C^\varrho},$ induced from its identification with $\partial\Sigma.$ The orientation of $\omega^\varrho_{\Sigma_\C}|_{\Sigma_\C^\varrho\setminus{i\in \B_1}},$ is defined by a section $\hat\xi$ which satisfies $\hat\xi (\xi)> 0.$ 
Such a section is said to be \emph{positive}. Thus,
using $b,$ it is seen that for any connected component of ${\Sigma_\C^\varrho}\setminus\{x_i\}_{i\in \B_1},$ either $\hat\xi$ or $-\hat\xi$ has a root in $\CCCL^{\widetilde\varrho}.$
In the case that for each connected component of ${\Sigma_\C^\varrho}\setminus A,$ where $A\supseteq\{x_i\}_{i\in \B_1}$ is a finite set of points, the positive sections have roots in $\CCCL^{\widetilde\varrho}$, we say that $(\CCCL,\widetilde\varrho)$ is \emph{compatible} away from $A.$ In case $A=\{x_i\}_{i\in \B_1}$ we say that the structure is compatible.

\begin{prop}\label{prop:arf_1_for_bdry_pts}
If $\B_1\neq\emptyset$ then there are no compatible real twisted spin structures.
\end{prop}
\begin{proof}
Suppose $i\in\B_1.$
Let $U$ be a contractible $\varrho-$invariant neighborhood of $x_i,$ which contains no other marked points.
One can find a $\varrho-$invariant section $s\in\Gamma(\CCCL\to U),$ which vanishes nowhere in $U,$ possibly after replacing $U$ by a smaller neighborhood.
In $\varrho-$anti-invariant local coordinates around $x_i,$ the real section $zdz$ generates $\omega_{\Sigma_\C}(U).$
Write $f(z)=zdz/b(s^{\otimes 2}),$ this is a nowhere vanishing holomorphic function in $U.$ Moreover, $f$ is conjugation invariant, and hence real on $U^\varrho.$
In particular, it does not change sign there. But this is impossible for a compatible structure since $zdz$ is positive on exactly one component of $U^\varrho\setminus\{x_i\}.$
%
%
%
\end{proof}
%

Given a compatible real spin structure,
a \emph{lifting} of the spin structure is a choice of a section in $$\Gamma(S^0(\CCCL^{\widetilde\varrho})\to{{\Sigma_\C^\varrho}\setminus\{x_i\}_{i\in\B}}),$$ where $S^0$ stands for the rank zero sphere bundle.
We say that the lifting \emph{alternates} in $x_j,$ and that $x_j$ is a \emph{legal} point, if this choice cannot be extended to $\Gamma(S^0(\CCCL^{\widetilde\varrho})\to{{\Sigma_\C^\varrho}\setminus\{x_i\}_{i\in\B\setminus\{j\}}}).$
Otherwise the lifting does not alternate in $x_j$ and $x_j$ is an \emph{illegal} point.

\begin{definition}\label{def:def2}
A \emph{twisted closed smooth spin surface} is a closed smooth surface $(\Sigma,\{z_i\}_{i\in\I}),$
together with a twisted spin structure twisted in $\{z_i\}_{i\in {\I_1}}.$
In case $\I_1=\emptyset,$ we call it a \emph{closed smooth spin surface}.

A \emph{twisted open smooth spin surface} is a smooth open surface $(\Sigma,\{x_i\}_{i\in\B},\{z_i\}_{i\in\I}),$
together with a compatible twisted real spin structure twisted in $\{z_i\}_{i\in {\I_1}}.$
In case $\I_1=\emptyset,$ we call it an \emph{open smooth spin surface}.
A \emph{(twisted) smooth spin surface with a lifting} is a (twisted) open spin surface, together with a lifting. A lifting with all boundary points being legal is called a \emph{grading}. A surface with a non twisted spin structure, that is, $\I_1=\B_1=\emptyset$, and a grading is called a \emph{graded surface}.
An isomorphism of twisted spin surfaces is an isomorphism of the underlying surfaces and of the line bundles which respects the twists, commutes with the maps between the canonical lines in the expected sense and, in the open case, also the with the involutions. An isomorphism of twisted spin surfaces with a lifting is an isomorphism of the twisted spin surfaces, which takes the lifting to the lifting in the target, and respects the alternations.
\end{definition}
We will see below in Proposition \ref{prop:parity_disks} that the only obstruction to the existence of a graded spin structure is the parity of $g+k:$ In a graded spin structure, $g+k$ must be odd.

\subsubsection{Stable graded surfaces}
We follow the terminology of \cite{Jarvispin}. 
Let $\Sigma=\{\Sigma_\alpha\}_{\alpha\in\mathcal{C}}$ be a stable  closed surface. A \emph{spin structure} twisted in $\{z_i\}_{i\in {\I_1}},~\I_1\subseteq\I,$
is a rank $1$ torsion free sheaf $\CCCL$ over $\Sigma$ together with a map
\[
b:\CCCL^{\otimes 2}\rightarrow\omega_{\Sigma}\left(-\sum_{i\in \I_1} z_i\right),
\]
where $\omega_{\Sigma}\left(-\sum_{i\in \I_1} z_i\right)$ is the dualizing sheaf, twisted in $\{z_i\}_{i\in {\I_1}}.$

We require
\begin{enumerate}
\item
\[\text{deg}(\CCCL)=\frac{\text{deg}(\omega_{\Sigma})-|\I_1|}{2}.\]
\item
$b$ is an isomorphism on the locus where $\CCCL$ is locally free.
\item
For any point $p$ where $\CCCL$ is not free the length of $\text{coker}(b_p)$ is $1.$
\end{enumerate}
In particular, $b$ is an isomorphism away from nodes.
Nodes where $b$ is not an isomorphism are called \emph{Neveu-Schwarz} (NS), at these nodes the last requirement says exactly that $b$ vanishes in order $2.$ The other nodes are called \emph{Ramond}.

Let $\Sigma=\{\Sigma_\alpha\}_{\alpha\in\mathcal{C}\cup\mathcal{O}}$ be a stable open $(g,k,l)-$surface. A \emph{real spin structure} twisted in $\{x_i\}_{i\in \B_1},\{z_i\}_{i\in {\I_1}},$ with $\I_1\subseteq\I,$ and $\B_1\subseteq\B,$
is a triple $(\CCCL,b,\widetilde\varrho),$ where $(\CCCL,b)$ is a spin structure over the doubled surface $D(\Sigma)=({\Sigma_\C},\varrho),$ and $\widetilde\varrho:\CCCL\to\CCCL,$ is an involution which lifts $d\varrho,$ the induced involution on $\omega_{\Sigma_\C}.$
Note that this means, in particular, that $b$ is a map
\[
b:\CCCL^{\otimes 2}\rightarrow\omega_{\Sigma_\C}\left(-\sum_{i\in \B_1}x_i-\sum_{i\in \I_1} (z_i+\bar{z}_i)\right),
\]
where $\omega_{\Sigma_\C}\left(-\sum_{i\in \B_1}x_i-\sum_{i\in \I_1} (z_i+\bar{z}_i)\right)$ is the dualizing sheaf, twisted in $\{x_i\}_{i\in \B_1},\{z_i,\bar{z}_i\}_{i\in {\I_1}},$
and that
\[\text{deg}(\CCCL)=\frac{\text{deg}(\omega_{\Sigma_\C})-2|\I_1|-|\B_1|}{2}.\]

\begin{rmk}\label{rmk:ramond}
Suppose $\Sigma$ is a nodal curve, open or closed, $z$ a node with preimages $z_\nu,z_\mu\in \Norm(\Sigma).$ Then there are natural residue maps $\text{res}_\eta:(\Norm^*\omega_\Sigma)_{z_\eta}\simeq\C.$ These induce an isomorphism $a:(\Norm^*\omega_\Sigma)_{z_\mu}\simeq(\Norm^*\omega_\Sigma)_{z_\nu},$ by $\text{res}(v)+\text{res}(a(v))=0.$
In the Ramond case, we also have an isomorphism $\widetilde{a}:(\Norm^*\CCCL)_{z_\mu}\to(\Norm^*\CCCL)_{z_\nu},$ and $\text{res}(b(v^{\otimes2}))+\text{res}(b(\widetilde{a}(v)^{\otimes2}))=0.$
For more details see \cite{Jarvispin}.
\end{rmk}

When $z\in\Sigma\subset\Sigma_\C$ is a contracted boundary which is Ramond, $d\varrho,\widetilde\varrho$ lift to complex anti-linear isomorphisms between the fibers of $\Norm^*\omega_{\Sigma_\C},$ $\Norm^*\CCCL$ in $z_\pm,$ where $z_+$ is the preimage of $z$ in $\Norm(\Sigma),$ and $z_-$ is the preimage of $z$ in $\Norm(\bar\Sigma).$
By composing with $a,\widetilde{a}$ we get anti-linear involutions on the fibers at $z_+.$
This defines real lines which we denote by $(\omega_\Sigma^\varrho)_{z_+},(\CCCL^{\widetilde\varrho})_{z_+},$ together with maps $\text{res}:(\omega_\Sigma^\varrho)_{z_+}\simeq\sqrt{-1}\R,$ where $\sqrt{-1}$ is the root of $-1$ in the upper half plane, and
$b^2:(\CCCL^{\widetilde\varrho})_{z_+}\to(\omega_\Sigma^\varrho)_{z_+},$ defined by $b^2(v)=b(v^{\otimes 2}).$

We say that the real spin structure is \emph{compatible in a contracted boundary $z$} if $z$ is a Ramond node of $\Sigma_\C$ and the image of $b^2$ is in the positive imaginary half line $\text{res}^{-1}(\sqrt{-1}\R_{\geq0}).$

The real spin structure is \emph{compatible} if it is compatible in contracted boundaries and away of special boundary points. Compatibility away from special points is defined as in the smooth case.

A \emph{lifting} of a compatible real spin structure is a choice of a section $$s\in\Gamma(S^0(\CCCL^{\widetilde\varrho})\to{{\Sigma_\C^\varrho}\setminus (\cup_{\alpha\in\mathcal{O}}B(\Sigma_\alpha))}),$$ where $S^0$ stands for the rank zero sphere bundle.
~The notions of alternations and of legal marked point or a legal half node are as in the smooth case.

Note that the definition of the lifting includes, for any contracted boundary node $z,$ a \emph{choice of a lifting for the contracted boundary}, i.e.,  with the above notations and identifications, a choice of direction in $(\CCCL^{\widetilde\varrho})_{z_+}$  which is mapped by $\text{res}\circ b^2$ to the ray $\sqrt{-1}\R_{\geq0}.$

\begin{prop}\label{prop:spin_on_comps}
\begin{enumerate}
\item
A real spin structure on a stable surface, twisted or not, induces a real spin structure, possibly twisted, on any open component of the normalization and a possibly twisted spin structure on any closed component of it. For any node of $\Sigma,$ the induced structure is either twisted in both of its preimages in the normalization, or not twisted in both.
The former case is the Ramond case, the latter is Neveu-Schwarz.
If there are no Ramond nodes then the spin structures on the closed components of the normalization, together with the real spin structures on its open components, determine the real spin structure on $\Sigma.$ 
\item
If the real spin structure is compatible, then so is the induced structure on any open component of the normalization. In this case, in particular, there are no twists in boundary marked points, and no boundary Ramond nodes. In case there are no Ramond internal nodes but there may be contracted boundaries, compatible spin structures on the normalization determine the compatible spin structure on $\Sigma.$ 
\item
A lifting on $\Sigma$ induces a lifting on the normalization. A lifting on the normalization, together with a choice of a direction in $(\text{res}\circ b^2)^{-1}(\sqrt{-1}\R_{\geq0})$ for the preimage $z_+$ of any contracted boundary, induces a lifting on $\Sigma.$ 
\end{enumerate}
\end{prop}
\begin{proof}
The fact that the twisted spin structure induces one on the normalization by pull back, and is induced by one, when there are no Ramond nodes is already true in the closed case, see for example \cite{Jarvispin}. Moreover, it is shown there that given the structures on the normalization and the identifications of the stalks in preimages of nodes, see Remark \ref{rmk:ramond}, the twisted spin structure on the surface is determined. The involution extends uniquely by continuity.

The second claim follows from the fact that one can examine compatibility away from special points. Ramond boundary nodes can not appear by Proposition \ref{prop:arf_1_for_bdry_pts}.
If $z$ is a contracted boundary, there is a single, up to sign, possible identification map $\widetilde{a},$ as in Remark \ref{rmk:ramond}.
Now, if $\widetilde{a}$ makes the contracted boundary compatible, with respect to the involution, $-\widetilde{a}$ will make it not compatible, and vice versa.
The last statement is evident.
\end{proof}

\begin{definition}
A \emph{twisted closed stable spin surface} is a closed stable surface $(\Sigma,\{z_i\}_{i\in\I}),$
together with a spin structure twisted in $\{z_i\}_{i\in {\I_1}}.$
In case $\I_1=\emptyset,$ we call it a \emph{stable closed spin surface}.
A \emph{twisted open stable spin surface} is a stable open surface $(\Sigma,\{x_i\}_{i\in\B},\{z_i\}_{i\in\I}),$
together with a compatible real spin structure twisted in $\{z_i\}_{i\in {\I_1}}.$
In case $\I_1=\emptyset,$ we call it a \emph{stable open spin surface}.
A \emph{(twisted) stable spin surface with a lifting} is a (twisted) open spin surface, together with a lifting such that for any boundary node, exactly one half node is legal. If all the boundary marked points are legal, the lifting is called a \emph{grading}. A (twisted) stable spin surface with a grading is \emph{effective} if the underlying surface is, and, for any component of the normalization, with genus $0,~3$ special boundary points and no special internal points, its special points are legal. A \emph{stable graded surface} is a (non-twisted) stable spin surface with a grading.
The isomorphism notions are as in the smooth case.
\end{definition}
The legality condition on the nodes may seem peculiar at first glance. However this is the condition which allows smoothing the stable graded surface at a boundary node. The closed analog of it is that the twists at the two half nodes of the same node must agree. In a nutshell, as we will see in the next subsection, in a twisted spin surface any closed path which does not pass through special points has a well defined notion of parity. By pinching the surface in that path, a node is formed, and this node is NS or Ramond according to the parity of the pinched path. Similarly, any \emph{oriented} arc between boundary points, which avoids special points, also has a well defined notion of parity. We will see in Proposition \ref{Q_alternates} that this parity changes if the orientation of the arc changes. By pinching the arc one obtains a surface with a new boundary node. The boundary node is NS, but the legality of its half nodes is determined by the parity of the corresponding oriented arcs. See Lemma \ref{lem:Q_q_and_bridges} for an exact statement. Interestingly, when the node is separating the legality can be determined from the parity considerations of Proposition \ref{prop:parity_disks}. Since in $g=0$ all nodes are separating, the genus $0$ theory could have been defined without referring to the graded spin structure.
These points will be discussed more in \cite{ST0}.
\begin{figure}
\centering
\includegraphics[scale=.4]{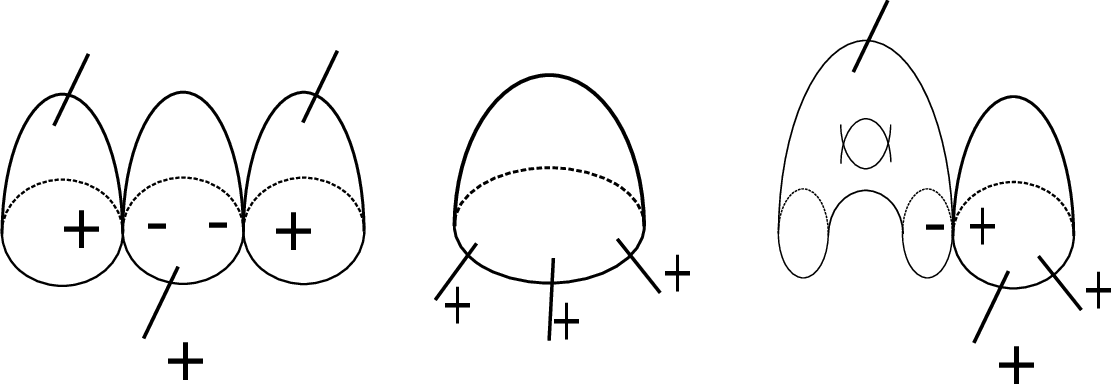}
\caption{In this figure we consider three stable graded spin surfaces. In these cases the underlying surfaces are effective. '+' near a boundary marking or a boundary half node indicates a legal point, while '-' indicates an illegal point. We omit the labels. The graded surface on the left is non effective, since in the normalization the middle components has one legal boundary marking, two illegal boundary markings and no internal markings. On the other hand, the remaining two graded surface are effective.
}
\label{fig:2_20}
\end{figure}
\begin{nn}
Denote by $\Spin(\Sigma)$ the set of isomorphism classes of graded spin structures on a stable open surface $\Sigma.$
\end{nn}
The definition of graded surfaces, together with Proposition \ref{prop:spin_on_comps}, yield the corollary
\begin{cor}\label{cor:spin_on_comps}
If $\Sigma$ has no internal nodes there is a bijection between $\Spin(\Sigma)$ and
\begin{enumerate}
\item isomorphism types of spin structures with a lifting on $\NNN(\Sigma),$ twisted precisely at preimages of contracted boundaries, such that any boundary marked point of $\Sigma$ is legal as a point of $\NNN(\Sigma),$ and for any node of $\Sigma$ exactly one half node in $\NNN(\Sigma)$ is legal.
\item A choice of a direction in $(\text{res}\circ b^2)^{-1}(\sqrt{-1}\R_{\geq0})$ for the preimage $z_+$ of any contracted boundary.
\end{enumerate} 
\end{cor}

\subsubsection{An alternative definition for the smooth case}\label{subsub:alter}
In this subsection we provide an alternative definition for smooth spin surfaces with a lifting.
This definition will be easier to work with.
Let $(\Sigma,\{x_i\}_{i\in \B},\{z_j\}_{j\in \I})$ be a smooth, open or closed, pointed Riemann surface. Choose any Riemannian metric on it.
\begin{nn}\label{nn:T1}
Denote by $T^1\Sigma$ the $S^1-$bundle of $T\Sigma.$
For a simple smooth arc or a simple smooth closed path $\gamma\subset\Sigma$ we denote the $S^0-$bundle of $T\gamma$ by $T^1\gamma.$

When the arc or path $\gamma$ is oriented, $T^1\gamma$ will stand for length $1$ oriented tangent vector field to $\gamma.$ In particular we shall use the notation $T^1\partial\Sigma$ for the branch of $T^1\partial\Sigma$ which covers the direction of the induced orientation on the boundary.
\end{nn}
We consider $T^1\Sigma$ as the $S^1-$subbundle of length $1$ vectors of $T\Sigma.$ Similarly for $T^1\gamma.$ We also consider $T^1\gamma$ as a $S^0-$subbundle of $T^1\Sigma|_\gamma.$ In what follows we use these identifications without mentioning a choice of metric. Different metrics will give rise to equivalent structures, and in fact, one can make these definitions metric-independent by considering the $S^0$ and $S^1$ bundles as sub quotients of the corresponding vector bundles.

For a point $p\in\Sigma,$ a vector $w\in T_p\Sigma,$ and an angle $\theta \in \R/2\pi\R,$ let $r_\theta w = r_\theta(p)w,$ be the operator of rotation by $\theta$ in the counterclockwise direction. We shall omit $p$ from the notation when it is clear from context. The operator $r_\theta(p)$ is induced on $T^1_p\Sigma,$ and we shall use the same notation.

If $u,w$ are two tangent vectors at $p$ denote the counter clockwise angle from $u$ to $w$ by $\measuredangle(u,w).$

For a smooth arc $\gamma:[0,1]\to\Sigma,$ there exists a canonical trivialization \[\varsigma:[0,1]\times S^1\to T^1\Sigma|_\gamma,\] defined by
$$\varsigma(t,\theta)=(\gamma(t),e^{i\theta}v_t),~v_t=(T^1)_{\gamma(t)}\gamma.$$
This trivialization defines a continuous family of maps
$$\{p(\gamma)_{s}^{t}:T^1_{\gamma(s)}\Sigma\to T^1_{\gamma(t)}\Sigma\}_{0\leq s,t\leq 1},$$
uniquely determined by the condition
\[
p_2(\varsigma^{-1}(\gamma(s),v)) = p_2(\varsigma^{-1}(\gamma(t),p(\gamma)_{s}^{t}v)),
\]
where $p_2$ is the projection on the second coordinate. One can extend the trivialization to the piecewise smooth context by approximation.
In case $s=0,t=1$ we omit them from the notation and write $p(\gamma).$ One can easily verify, in the piecewise smooth case, that
if $\gamma$ is composed of smooth sub arcs, $\gamma_i:[a_i, a_{i+1}]\to\Sigma,$ where $a_0=0<a_1<\ldots<a_n=1,$ and $\theta_{i+1}$ is $\measuredangle(\dot{\gamma}_i|_{\gamma_{i+1}(a_{i+1})},\dot{\gamma}_i|_{\gamma_i(a_{i+1})}),$ then
\[
p(\gamma) = p(\gamma_{n-1})r_{\theta_{n-1}}p(\gamma_{n-2})\dots r_{\theta_1}p(\gamma_0).
\]
We shall denote such $\gamma$ by $\gamma_1\to\gamma_2\to\ldots\gamma_n.$
For a \emph{closed} piecewise smooth path $\gamma,$ we slightly change the definition of $p$ to be
\[
p(\gamma) = r_{\theta_0}p(\gamma_{n-1})r_{\theta_{n-1}}p(\gamma_{n-2})\dots r_{\theta_1}p(\gamma_0),
\]
and note that this is in fact the identity map.
We shall denote such $\gamma$ by $\gamma_1\to\gamma_2\to\ldots\gamma_n\to\gamma_1.$
\begin{definition}
A \emph{twisted spin structure} $\SL\to\Sigma\setminus\{z_j\}_{j\in I}$ on a smooth marked $\Sigma$ is a $S^1-$bundle on $\Sigma\setminus\{z_j\}_{j\in I}$ together with a degree $2$ cover bundle map
\[
\pi=\pi^\SL:\SL\to T^1\Sigma|_{\Sigma\setminus\{z_j\}_{j\in I}}.
\]
\end{definition}
For a point $p\in\Sigma,$ a vector $w\in\SL_p,$ and an angle $\theta \in \R/4\pi\Z,$ let $R_\theta w = R_\theta(p)w,$ be the operator of rotation by $\theta$ in the counterclockwise direction. We shall omit $p$ from the notation when it is clear from context.

The \emph{parallel transport} along $\gamma:[0,1]\to\Sigma$ is the unique continuous family of maps
$$\{P(\gamma)_{s}^{t}:\SL_{\gamma(s)}\to\SL_{\gamma(t)}\}_{0\leq s,t\leq 1},$$
which covers $\{p(\gamma)_{s}^{t}\}.$
We shall sometimes call $P(\gamma)_{0}^{1}v$ the \emph{parallel transport of $v$ along $\gamma$}, and write it as $P(\gamma)v.$
\begin{rmk}\label{rmk:commutation_of_trans_and_rot}
Note that $R$ covers $r$ in the sense that
if $\pi(s)=v,$ for $s\in\SL_p,v\in T^1_p\Sigma$ then
\[
\pi(R_\theta(p)s)=r_\theta(p)v=r_{\theta(\text{mod } 2\pi)}(p)v.
\]

Observe that $R_\alpha R_\beta = R_{\alpha+\beta}.$ In addition, $P,R$ commute:
\[
R_\theta (\gamma(t))P(\gamma)_s^t v = P(\gamma)_s^t R_\theta (\gamma_s)v.
\]
\end{rmk}
\begin{definition}\label{def:spin_using_q}
A (twisted) spin structure $\SL$ is associated with a function
\[
\q = \q^\SL:H_1(\Sigma\setminus\{z_j\}_{j\in I},\mathbb{Z}_2)\to\Z_2,
\]
defined as follows. For $x\in H_1(\Sigma\setminus\{z_j\}_{j\in I},\mathbb{Z}_2),$ take a piecewise smooth connected representative $\gamma.$ Then $p(\gamma)$ is the identity. Hence $P(\gamma)$ is either the identity or minus the identity. We define $\q(x)=\q(\gamma)$ to be $1$ in the former case, otherwise it is $0.$

For any internal marked point $z_j,$ take a small disk $D_j$ which surrounds it and contains no other marked points in its closure. We define the \emph{twist} in $z_j$ to be $\q(\partial D_j).$
\end{definition}

The following well known theorem was proven by Johnson \cite{Johnson}, and states that $\q$ is a \emph{quadratic enhancement} of the Poincar\'e pairing $\langle\alpha,\beta\rangle.$
\begin{thm}\label{lem:arf_props}
$\q$ is a well defined function on $H_1(\Sigma\setminus\{z_j\}_{j\in I},\mathbb{Z}_2).$ For $\alpha,\beta\in H_1(\Sigma\setminus\{z_j\}_{j\in I},\mathbb{Z}_2),$
\[
\q(\alpha+\beta)=\q(\alpha)+\q(\beta)+\langle\alpha,\beta\rangle,
\]
\end{thm}

\begin{prop}\label{prop:arf1}
If $\gamma:[0,1]\to\Sigma\setminus\{z_j\}_{j\in I}$ is a piecewise smooth closed curve which bounds a contractible domain, then $P(\gamma)_0^1=R_{2\pi}.$
Moreover, suppose $\Sigma$ is a disk with a piecewise smooth boundary $\gamma.$ Let $\SL\to T^1\Sigma|_\gamma$ be a double cover by a $S^1$ bundle $\SL.$
Then $\SL$ can be extended to a non-twisted spin structure on $\Sigma$ if and only if $P(\gamma)_0^1=R_{2\pi}.$ In this case the extension is unique.
In particular, the spin structure can be extended to a marked point $z_i$ if and only if its twist is $0,$ in that case the extension is unique.
\end{prop}
The first part follows from Theorem \ref{lem:arf_props}, by taking $\alpha=\beta=[\gamma].$
The other parts are also simple and will be omitted.

\begin{definition}
Let $(\Sigma,\SL)$ be an open marked Riemann surface, together with a (twisted) spin structure. Suppose $\partial\Sigma\neq\emptyset.$
A \emph{lifting} is a choice of a section
\[
s: \partial\Sigma\setminus\{x_i\}_{i\in \B}\to\SL|_{\partial\Sigma\setminus\{x_i\}_{i\in \B}}
\]
which covers the \emph{oriented} $T^1(\partial\Sigma\setminus\{x_i\}_{i\in \B}).$

For $j\in \B,$ suppose $i:(-\frac{1}{2},\frac{1}{2})\to \partial\Sigma$ is a smooth orientation preserving embedding with $i(0)=x_j,$ and $x_b\notin i(-\frac{1}{2},\frac{1}{2}),~b\neq j.$
In case
\[
\lim_{x\to0^-}s(x)\neq\lim_{x\to0^+}s(x),
\]
we say that the structure \emph{alternates} in $x_j,$ and that $x_j$ is a \em{legal point}. Otherwise $x_j$ is \emph{illegal} and the structure does not alternate.
We extend the definition of $s$ to the boundary marked points by $s(x)=\lim_{x\to0^+}s(x).$

A \em{smooth spin surface with a lifting} $(\Sigma,\{x_i\}_{i\in \B},\{z_i\}_{i\in \I},\SL,s)$ is a smooth open Riemann surface together with a spin structure and a lifting.
A \em{smooth graded surface} is a smooth spin surface with a lifting, such that all boundary marked points are legal.

%
%
\end{definition}

The notion of alternation can be generalized in the following manner.
\begin{definition}\label{def:Q}
A \emph{bridge} is a piecewise smooth simple arc which meets the boundary only in its two distinct endpoints $x,y\in\partial\Sigma\setminus\{x_i\}_{i\in\B}.$
Suppose we orient the bridge and parameterize it as
\[\gamma:[0,1]\to\Sigma,~\gamma(0)=x,~\gamma(1)=y.\] 
Define $\Q(\gamma)\in\Z_2$ by the equation
\begin{equation}\label{eq:Q_def}
R_{2\pi-\alpha_y}(y)P(\gamma) R_{\alpha_x}(x) s(x)= R_{2\pi\Q(\gamma)}(y)s(y).
\end{equation}
where 
\[\alpha_x=\measuredangle((T^1)_x\partial\Sigma,(T^1)_x\gamma)\in[0,\pi],\quad
\alpha_y=\measuredangle((T^1)_y\partial\Sigma,(T^1)_y\gamma)\in[\pi,2\pi].\]
$\Q(\gamma)$ depends on the orientation but not on the parametrization. An oriented bridge with $\Q=1$ is called \emph{a legal side of the bridge}, otherwise it is called an \emph{illegal side}.
\end{definition}

\begin{prop}\label{Q_alternates}
Let $\Sigma$ be a smooth open spin surface with a lifting.
Let $\gamma$ be and denote by $\bar\gamma$ the same bridge with opposite orientation.
Then $\Q(\gamma)+\Q(\bar\gamma)=1.$ Thus, any bridge has exactly one legal side and exactly one illegal side.
\end{prop}
\begin{proof}
Work with the notations of Definition \ref{def:Q}. For $w\in\{x,y\},~\alpha'_w$ is defined by $\alpha'_w=\measuredangle((T^1)_w\partial\Sigma,(T^1)_w\bar{\gamma}).$
Observe that $\alpha'_x=\alpha_x+\pi,~\alpha'_y=\alpha_y-\pi.$
Apply $R_{2\pi\Q(\bar{\gamma})}(y)$ to the left hand side of \eqref{eq:Q_def}.
By Remark \ref{rmk:commutation_of_trans_and_rot} the left hand side becomes
\[
R_{2\pi\Q(\bar{\gamma})}(y)R_{2\pi-\alpha_y}(y)P(\gamma) R_{\alpha_x}(x) s(x)=
R_{2\pi-\alpha_y}(y)P(\gamma) R_{\alpha_x}(x) R_{2\pi\Q(\bar{\gamma})}(x)s(x).
\]
Using Equation \eqref{eq:Q_def} for $\bar{\gamma},$ Remark \ref{rmk:commutation_of_trans_and_rot} again, and the relations between $\alpha_x,\alpha'_x$ and $\alpha_y,\alpha'_y,$ the last expression simplifies to $R_\pi P(\bar{\gamma})R_\pi P(\gamma)s(y).$ By Proposition \ref{prop:arf1}, applied to the piecewise smooth closed curve $\gamma\to\bar \gamma\to \gamma,$ this is just $R_{2\pi}(y)s(y).$

Apply $R_{2\pi\Q(\bar{\gamma})}(y)$ to the right hand side of \eqref{eq:Q_def}, we obtain $R_{2\pi(\Q(\gamma)+\Q(\bar{\gamma}))}(y)s(y).$
Thus, \[R_{2\pi}(y)s(y)=R_{2\pi(\Q(\gamma)+\Q(\bar{\gamma}))}(y)s(y),\] and the claim follows.
%
%
%
%
%
\end{proof}

\begin{prop}\label{prop:parity_disks}
\begin{enumerate}
\item
Suppose $(\Sigma,\{z_i\}_{i\in \I},\SL)$ is a genus $g$ closed spin surface. Suppose that exactly $l_1$ marked points have twist $1.$
Then $l_1$ is even.
For any closed Riemann surface $(\Sigma,\{z_i\}_{i\in \I}),$ there exist $2^{2g}$ distinct non-twisted spin structures on $\Sigma.$
\item
Suppose $(\Sigma,\{x_i\}_{i\in \B},\{z_i\}_{i\in\I},\SL,s)$ is a genus $g$ open spin surface with a lifting. Suppose that exactly $k_+$ of the boundary marked points are legal, and $l_1$ internal marked points have twist $1.$
Then
\[
l_1 = g+1+k_+(\text{mod } 2).
\]
For any $(\Sigma,\{x_i\}_{i\in \B},\{z_i\}_{i\in\I})\in\RCM_{g,k,l}$ with $2|g+k+1,$ there exist exactly $2^g$ graded structures on $\Sigma.$
\end{enumerate}
\end{prop}
\begin{proof}
For the first claim, let $\{C_i\}$ be a family of non intersecting circles around each marked point. Then $\sum C_i$ is homologous to $0$. 
By Theorem \ref{lem:arf_props}, $\q(\sum C_i)=\sum \q(C_i)=0.$ For the number of spin structures see, for example, \cite{Jarvispin}.

Regarding the second claim, let $C_i$ be as above, and for any boundary component $\partial\Sigma_b,$ let $C_b$ be a curve surrounding this boundary, disjoint from it, but isotopic to it in $\Sigma\setminus\mathbf{z}.$ By the definitions of $\q,\Q$ one easily sees that $\q(C_b)$ is $1$ plus the number of legal marked points of $\partial\Sigma_b.$
Again \[\sum \q(C_i)+\sum \q(C_b)=0~(\text{mod }2),\] but this sum equals $l_1+k_++b~(\text{mod~}2),$ where $b$ is the number of boundaries. It is easy to see that $b=g+1~(\text{mod } 2).$
For the number of graded structures see \cite{ST0}.
We will also obtain it as a byproduct in Subsection \ref{sec:kasteleyn}, see the end of Example \ref{ex:macheta}.
\end{proof}

\begin{lemma}\label{lem:equiv_of_defs}
The definitions of smooth spin surfaces with a lifting, twisted or not, graded or not, given in this subsection are equivalent to the analogous ones given in Subsection \ref{subsub:smooth}.
\end{lemma}
Starting with a real spin structure $\CCCL$ in the sense of Subsection \ref{subsub:smooth}, $\SL$ is just the $S^1-$ bundle of $\CCCL^*,$ and the lifting is the reduction of the lifting to that bundle. See \cite{ST0} for more details, and for the rather straightforward proof of equivalence.
\subsubsection{A comment about the alternative definition in the stable case}
In the stable case, by Proposition \ref{prop:spin_on_comps}, the sheaf $\CCCL$ and the graded data determine the spin structures and liftings on the normalization, hence by Lemma \ref{lem:equiv_of_defs}, determines the data of $\SL,s$ for each component. However, it is determined by it, again, using the same lemma and proposition, only when there are no Ramond nodes.
Even when there are such nodes, the data of $\SL,s$ for each component determines $\CCCL$ and the graded data up to a finite choice of identification maps between stalks of half nodes and liftings at the preimages of the contracted boundaries, as explained in the proof of Proposition \ref{prop:spin_on_comps}.
Therefore, since working with the $S^1-$bundle and its lifting is more convenient, throughout this paper we shall usually write $(\Sigma,\SL,s)$ to indicate a spin structure with a lifting, and leave $\CCCL$ implicit.
We shall sometimes even leave $\SL,s$ implicit.

\subsubsection{Spin graphs}
It is useful to encode some of the combinatorial data of spin surfaces with a lifting in graphs.
\begin{definition}
A \emph{(pre-)stable spin graph $\Gamma$ with a lifting} is a (pre-)
stable graph
\[\Gamma = \left(V,H ,\sim = \sim_B\cup\sim_I\right),\]
together with a twist map $\tw:H^I\to\Z_2,$ and an alternation map $\Or:H^B\to\Z_2.$
we require
\begin{enumerate}
\item
$\tw(h)=\tw(\sig_1(h)),$ for any $h\in H^I\setminus T^I.$
\item
$\Or(h)+\Or(\sig_1(h))=1,$ for any $h\in H^B\setminus T^B.$
\item
$\forall h\in H^{\text{CB}},~\tw(h)=1.$
\item
For $v\in V^O,$ then
\[
\sum_{h\in (\sig_0^B)^{-1}(v)}\Or(h)+\sum_{h\in (\sig_0^I)^{-1}(v)}\tw(h)=g(v)+1 (mod~2).
\]
\item
For $v\in V^C$
\[
\sum_{h\in \sig_0^{-1}(v)}\tw(h)=0.
\]
\end{enumerate}
A boundary half edge $h,$ and in particular a tail with $\Or(h)=0$ is said to be \emph{illegal}, otherwise it is \emph{legal}.

We say that the graph is \em{stable} if $\Gamma$ is stable.
We call $\Gamma$ a \emph{graded graph} if $\Or(t)=1$ for all $t\in T^B,~\tw(t)=0$ for all $t\in T^I\setminus H^{\text{CB}}.$ 

$\Gamma$ is \emph{effective} if its underlying graph is effective, $\Or(t)=1$ for all $t\in T^B,$ and for any $v\in V^O$ without internal half edges its three boundary half edges have $\Or=1.$

\end{definition}
The normalization $\NNN(\Gamma)$ is just the normalization of the underlying graph $\Gamma,$ with the maps $\tw,\Or$ defined on the tails of $\NNN(\Gamma)$ by their values on the corresponding half edges of $\Gamma.$
As in the spinless case, whenever an internal tail of $\Gamma$ is marked $i\neq0,$ the graph $v_i(\Gamma)$ is the component of $\NNN(\Gamma)$ which contains tails $i,$ but with the additional data of $\tw,~\Or.$

When it is clear from the context that the dual graph under consideration is a spin graph with a lifting, we sometimes omit the maps $\tw,\Or$ from the notations.
\begin{figure}
\centering
\includegraphics[scale=.4]{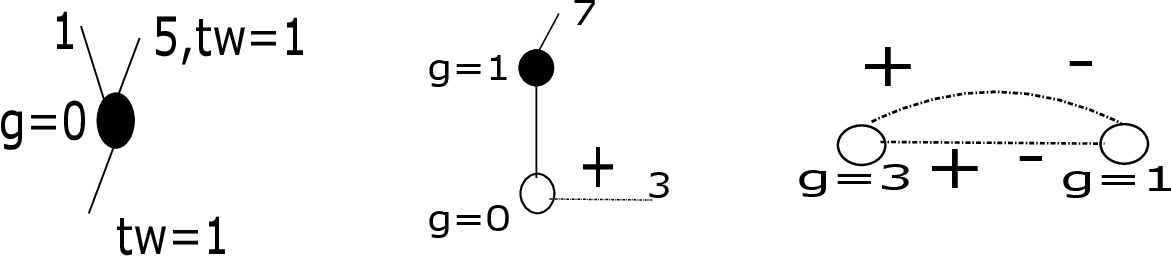}
\caption{Three examples of graded dual graphs. The numbers stands for the markings, all twists are $0$ unless '$\tw=1$' is written next to an element of $H^I.$ In order to avoid confusion, legal half edges, the elements $h\in H^B$ with $\Or(h)=1$ are decorated by $+$ signs.}
\label{fig:2_37}
\end{figure}

\begin{definition}
An \emph{isomorphism} between 
spin graphs with a lifting $(\Gamma,\tw,\Or)$ and $(\Gamma',\tw',\Or')$ is a tuple
\[
f = \left(f^V,f^H\right)
\]
such that
\begin{enumerate}
\item
$f:\Gamma\to\Gamma'$ is an isomorphism of stable graphs.
\item
$\tw' = \tw\circ f^H;~\Or' = \Or\circ f^H|_{H^B}.$
\end{enumerate}
We denote by $\text{Aut}(\Gamma)$ the group of the automorphisms of $\Gamma=(\Gamma,\tw,\Or).$
\end{definition}

We denote by $\CG$ the set of isomorphism classes of all spin graphs with a lifting.
We have a natural map
\[
\widetilde{\text{for}}_{\text{spin}}:\CG\to\RCG,~\widetilde{\text{for}}_{\text{spin}}(\Gamma,\tw,\Or)=\Gamma.
\]
Write ${\text{for}}_{\text{spin}}$ for its restriction to graded graphs.
We denote by $\CG_{g,k,l}$ the set of isomorphism classes of graded graphs with $\text{Image}(\mmm^B)=[k],\text{Image}(\mmm^I)=[l].$
Define $\Gamma_{g,k,l}$ as the unique connected graded dual graph with a single open vertex of genus $g,$ exactly $k$ boundary tails marked by $[k],$ exactly $l$ internal tails marked by $[l],~H^{\text{CB}}=\emptyset,$ and no further half edges.

To each graded stable marked surface $\Sigma$ we associate a graded stable graph $(\Gamma,\tw,\Or)$ as follows.
First, $\Gamma =\Gamma\left(\Sigma\right).$ Let $w\in\Sigma_\alpha$ be any special point of this component. It corresponds to some half edge $h.$
If $h\in H^I,$ then $\tw(h)$ is defined to be the twist in $w.$ If $h\in H^B,$ then $\Or(h)=1$ if and only if $h$ is legal.
For shortness we denote the graded stable graph corresponding to $\Sigma$ by $\Gamma\left(\Sigma\right),$ omitting $\tw,\Or$ from the notation.
Note that $\NNN(\Gamma(\Sigma))=\Gamma(\NNN(\Sigma)),$ and whenever a $i\neq 0$ marks an internal marked point, then $v_i(\Gamma(\Sigma))=\Gamma(\Sigma_i).$

We can also extend the graph operations to the graded case.
The \emph{smoothing} of a stable spin graph with a lifting $(\Gamma,\Or,\tw),$ at $f\in E$ is the stable graph
\[
d_f\Gamma = (\Gamma',\Or',\tw')
\]
such that $d_f(\Gamma) = \Gamma'.$ Recall we may identify $H'$ as a subset of $H.$ We define $\tw',\Or'$ as the restrictions of $\tw,\Or$ with respect to this identification. Given a set $S=\left\{f_1 , \ldots, f_n\right\}\subseteq E\left(\Gamma\right)$, define the smoothing at $S$ as
\[
d_S\Gamma = d_{f_n}\left(\ldots d_{f_2}\left(d_{f_1}\Gamma\right)\ldots\right).
\]
Note that again in case $\Gamma = d_S\Gamma'$, then $H'$ is canonically identified as a subset of $H,$ and $\Or,\tw$ respect this identification.

Again we define
$\pu : \CG \to 2^{\CG},$ and $\partial : \CG \to 2^{\CG},$
by
\[\pu \Gamma = \left\{\left.\Gamma' \right| \exists S\subseteq E\left(\Gamma'\right),\, \Gamma = d_S\Gamma'\right\},~~
\partial\Gamma =\pu\Gamma\setminus\{\Gamma\}.\]
And again these maps naturally extend to maps $2^{\CG} \to 2^{\CG}.$
%
%
%
%

\subsubsection{$\oCM_{g,k,l}$}\label{subsub:moduli}
\begin{nn}
For $\Gamma \in \CG$, denote by $\CM_\Gamma$ the set of isomorphism classes of marked stable spin surfaces with a lifting, associated to graph $\Gamma.$

Define
\[
\oCM_\Gamma = \coprod_{\Gamma'\in\pu\Gamma}\CM_{\Gamma'}.
\]
Define 
$\oCM_{g,k,l}=\oCM_{\Gamma_{g,k,l}}.$ Similarly define 
$\CM_{g,k,l}$ as the subspace parameterizing smooth surfaces.

For a marking $i,$ denote by $v_i:\CM_\Gamma\to\CM_{v_i(\Gamma)}$ the canonical map $[\Sigma]\to[\Sigma_i].$
Observe that in order to define this map we have used Proposition \ref{prop:spin_on_comps}. If $\Sigma$ has a contracted boundary, then $\Sigma_i$ has a marked Ramond point which corresponds to it. The passage from $\Sigma$ to $\Sigma_i$ \emph{forgets the lifting at contracted boundaries}.
%
\end{nn}
\begin{thm}[\cite{ST}]
The space $\oCM_{g,k,l}$ is a compact smooth orbifold with corners of real dimension $3g-3+k+2l.$ It is endowed with a canonical orientation.\end{thm}
We note that $\oCM_{g,k,l}$ is in general disconnected. Different connected components correspond to different topologies with the same doubled genus, to different partitions of the boundary points between boundary components, and sometimes also to different connected components of graded spin structures.

The main difficulty in this theorem is the proof of orientability. The properties of the canonical orientation will be detailed in Theorem \ref{lem:geometric_orientability} below. In Theorem \ref{thm:orientability}, Proposition \ref{prop: for_monster_calc} and Corollary \ref{cor:or_agree} below we will provide a different proof for the orientability and for the properties of the canonical orientations.
We now briefly review the proof that $\oCM_{g,k,l}$ is a compact smooth orbifold with corners. As in the spinless case, we rely on \cite{Zernik}. We also refer the reader to \cite[Lemma 3.5]{BCT1}, where a similar procedure, also based on \cite{Zernik}, is applied to the moduli of $r-$spin disks.

Our starting point is the fact that in the closed setting the moduli space $\oCM_{g,n}^{1/2}$ of twisted spin curves is a smooth orbifold, see, for example \cite{Jarvis2}.
Consider the following sequence:
\begin{equation}
\label{eq:spin_mod_Csequence}
\oCM_{g,k,l}\stackrel{(5)}{\to}\widehat{\mathcal{M}}_{g,k,l}\stackrel{(4)}
{\hookrightarrow}\widetilde{\mathcal{M}}_{g,k,l}\stackrel{(3)}{\to}
\widetilde{\R\mathcal{M}}_{g,k,l}\stackrel{(2)}{\to}
\overline{\R\mathcal{M}}_{g,k+2l}^{1/2}\stackrel{(1)}{\to}\overline{\mathcal{M}}_{g,k+2l}^{'1/2}.
\end{equation}
As in the spinless case we explain the notations throughout the steps below.

\textbf{Step 1: }
First, $\overline{\mathcal{M}}_{g,k+2l}^{'1/2}$ is the suborbifold of $\oCM_{g,k+2l}^{1/2},$ the moduli of stable marked $2-$spin curves, given by the condition that all the markings have twist $0.$ Inside this space, $\overline{\R\mathcal{M}}_{g,k+2l}^{1/2}$ is the fixed locus of the involution defined by
\[(C;w_1, \ldots, w_{k+2l}, \SL) \mapsto (\overline{C}; w_1, \ldots, w_k, w_{k+l+1},\ldots, w_{k+2l},w_{k+1},\ldots,w_{k+l}, \overline{\SL}),\]
where $\overline{C}$ and $\overline{\SL}$ are the same as $C$ and $\SL$ but with the conjugate complex structure. $k$ is required to satisfy $2\nmid g+k.$
As the fixed locus of an anti-holomorphic involution $\overline{\R\mathcal{M}}_{g,k+2l}^{1/2}$ is a smooth compact real orbifold.  It parameterizes isomorphism types of marked spin curves with an involution $\widetilde\varrho$ covering the conjugation $\varrho$ on $C$, and $0$ twists.

\textbf{Step 2: }
The next step is to cut $\overline{\R\mathcal{M}}_{g,k+2l}^{1/2}$ along the real simple normal crossings divisor consisting of curves with at least one real node, via the real hyperplane blow-up of \cite{Zernik}. As in the spinless case, this yields an orbifold with corners $\widetilde{\R\mathcal{M}}_{g,k,l}$.

\textbf{Step 3: }
Consider the subset of $\widetilde{\R\mathcal{M}}_{g,k+2l}^{1/2}$ whose generic point is a smooth marked real spin curves with nonempty real locus.  Then $\widetilde{\mathcal{M}}_{g,k,l}$ is the disconnected $2$-to-$1$ cover of this subset given, as in the spinless case, by the choice of a distinguished half $\Sigma,$ a connected component of $C\setminus C^\varrho$. Note that $C=D(\Sigma).$

\textbf{Step 4: }
Inside $\widetilde{\mathcal{M}}_{g,k,l}$, we denote by $\widehat{\mathcal{M}}_{g,k,l}$ the union of connected components such that the marked points $w_{k+1},\ldots, w_{k+l}$ lie in the distinguished half, and that the spin structure is compatible. The generic point in this orbifold has isotropy $\Z_2,$ coming, in the level of objects, from scaling the fibers of $\SL$ by $-1.$

\textbf{Step 5: }
Finally, $\oCM_{g,k,l}$ is the degree $2$ cover of $\widehat{\mathcal{M}}_{g,k,l}$ given by a choice of grading. The choice of the grading cancels the global $\Z_2$ isotropy, since the $-1$ map is no longer an automorphism, as it does not preserve the grading. As a cover, $\oCM_{g,k,l}$ is also endowed by an orbifold with corners structure.

$\oCM_\Gamma$ is a suborbifold with corners, which is the closure of $\CM_\Gamma,$ for any $\Gamma\in\CG_{g,k,l}.$ The map ${\text{For}}_{\text{spin}}$ is an orbifold branched cover.
A graded surface with $b$ boundary nodes and contracted boundaries belongs to a corner of the moduli space $\oCM_{g,k,l}$ of codimension $b.$ Thus $\partial \oCM_{g,k,l}$ consists of graded stable surfaces with at least one boundary node or contracted boundary.
For details see \cite{ST}. We should note that the same argument applies for the more general setting of the moduli space of twisted spin surfaces with a lifting. These more general moduli spaces are also smooth orbifolds with corners, but in general they are not orientable.
\begin{rmk}\label{rmk:spin_is_const}
By Proposition \ref{prop:parity_disks}, the degree of the map ${\text{For}}_{\text{spin}}$ is $2^g.$ The automorphism group of the underlying surface acts on the set of spin structures.
When the surface is smooth this group is generically trivial, but when it is not, it may happen that the fiber of ${\text{For}}_{\text{spin}}$ is of cardinality smaller than $2^g.$ Still, even in this case its weighted cardinality, which takes into account the isotropies, is $2^g,$ so that the orbifold degree in the smooth case is constant.
When the topology becomes nodal the number of graded spin structures on a given underlying surface may change. But still, for any graded dual graph $\Gamma$ the degree of ${\text{For}}_{\text{spin}}$ restricted to $\CM_\Gamma$ is generically constant, and, when isotropy groups are taken into account, it is always constant. This constant is a power of $2$ which can be calculated from the graph structure of $\Gamma$ using, for example, Proposition \ref{prop:spin_on_comps} and the first paragraph in its proof, which relate spin structures on a stable surface and twisted spin structures on its normalization.
\end{rmk}


The \emph{universal curve} $\oCC_{g,k,l}\to\oCM_{g,k,l}$ is the space whose fiber over $[\Sigma]\in\oCM_{g,k,l}$ is $\Sigma.$
Its topology can be defined as in the closed case.

The following simple lemma is useful for understanding the geometry of $\oCM_{g,k,l},$ see \cite{ST0,ST} for details.
\begin{lemma}\label{lem:Q_q_and_bridges}
\begin{enumerate}
\item
$\q,\Q$ are isotopy invariants, in the sense that if $(\Sigma_s)_{0\leq s \leq 1}$ is a path in $\oCM_{g,k,l},$ and $(\gamma_{t,s})_{0\leq s,t\leq 1}$ is a continuous family
of simple paths $\gamma_{\cdot,s}\subseteq\Sigma_s\hookrightarrow\oCC_{g,k,l},$ which miss the special points, and which are either all bridges or all closed.
Then in case they are all bridges then $\Q(\gamma_{\cdot,s})$ is fixed, for any continuous choice of orientations on $\gamma_{\cdot,s},$ if they are all closed, then $\q(\gamma_{\cdot,s})$ is fixed.
\item\label{it:legal_illegal}
Suppose now that $(\Sigma_s)_{0\leq s \leq 1}$ is a path in $\oCM_{g,k,l},$ and $(\gamma_{t,s})_{0\leq s,t\leq 1}$ is a continuous family
of paths $\gamma_{\cdot,s}\subseteq\Sigma_s\hookrightarrow\oCC_{g,k,l},$ which for $s<1$ are simple and miss the special points, and are either all bridges or all closed.
Assume $\gamma_{\cdot,1}$ is a constant path mapped to a node or a contracted boundary.
Then if $\gamma_{\cdot,s}$ are all closed, then the node is internal or a contracted boundary and its twist is $\q(\gamma_{\cdot,s}),$ for any $s<1.$
If  $\gamma_{\cdot,s}$ are all open, then the node is a boundary node. In this case, the illegal side of the bridges degenerate to the
illegal half node, in the sense of Definition \ref{def:smoothing-top}.

In particular, by Proposition \ref{Q_alternates}, exactly one of the half nodes of each boundary node is legal.
\item\label{it:qQ_determine_spin}
Two graded spin structures on $\Sigma,$ without a Ramond node which give rise to the same pair $(\q,\Q)$ are isomorphic.
\end{enumerate}
\end{lemma}
\begin{rmk}
A classification of all pairs $(\q,\Q)$ is given in \cite{ST0}.
\end{rmk}

\begin{nn}\label{nn:tilde_forget_spin}
We denote by $\widetilde{\text{For}}_{\text{spin}}$ the canonical map
\[\widetilde{\text{For}}_{\text{spin}}:\oCM_\Gamma\to\oRCM_{{\text{for}}_{\text{spin}}(\Gamma)}\]
defined by forgetting the twisted spin structure and the lifting. Write ${\text{For}}_{\text{spin}}$ for the restriction to graded moduli. The definitions of $\widetilde{\text{For}}_{\text{spin}}, ~{\text{For}}_{\text{spin}}$ make sense also when $\Gamma$ is closed (and then the lifting is trivial).
\end{nn}

We end this subsection with a brief illustration of the phenomenon underlying the branched cover property of the map $\text{For}_{\text{spin}}.$ The branching phenomenon occurs along strata which parameterize surfaces with \emph{internal} nodes, and therefore happens, and from the same geometric reasoning, also in the setting of the closed $2-$spin intersection theory. We shall explain it in this setting for the simplicity of notations. 
Let $\Sigma_0$ be a curve with a single non separating node, and let $\Sigma_1$ be its smoothing, so that $\Sigma_0$ is obtained from $\Sigma_1$ by pinching at some simple smooth closed path $\gamma.$
Let $(\Sigma_t)_{t\in[0,1]}$ be a path in the moduli of curves, interpolating between $\Sigma_1,\Sigma_0.$ This path induces an identification of $H_1(\Sigma_t,\Z_2)$ for $t>0,$ which in the limit $t\to 0$ corresponds to the surjection obtained by taking the quotient $[\gamma_t]=0,$ where $[\gamma_t]$ is the generator of $H_1(\Sigma_t,\Z_2)$ which corresponds to $[\gamma]\in H_1(\Sigma_1,\Z_2)$ under this isomorphism.
Let $\alpha$ be any element of $H^1(\Sigma_1,\Z_2)$ satisfying $\langle\alpha,\gamma\rangle=1.$ Denote by $\alpha'$ the element in $H^1(\Sigma_0,\Z_2)$ which corresponds to $\alpha$ after the pinching, via the aforementioned surjection.
Let $B_1$ be an ordered basis of $H_1(\Sigma_1,\Z_2)$ whose first two elements are $[\gamma]$ and $[\alpha],$ and the remaining basis elements do not intersect $[\gamma].$ Define, for $t>0$ $B_t$ as the image of $B_1$ under the isomorphism, and extend to $t=0$ via the mentioned surjection.
Now choose any spin structure of $\Sigma_0$ which gives all markings twist $0$ and makes the node NS. Recall that spin structures on smooth curves are determined by the map $\q$ of Definition \ref{def:spin_using_q}, using the rule of Theorem \ref{lem:arf_props}, and any map which satisfies this rule gives rise to such a spin structure. Recall also that spin structures on $\Sigma_0$ which give all markings twist $0$ and make the node NS are in bijection with spin structures on the normalization of $\Sigma_0,$ giving all of its special points twist $0.$
Assign a number $\q(\beta)$ to any element $\beta\in B_1\setminus\{\gamma\},$ and put $\q(\gamma)=0.$ Recall Lemma \ref{lem:Q_q_and_bridges}. The identifications between the different $B_t,~t>0$ define a spin structure $\SL_t$ on $\Sigma_t$ for any $t>0.$ It extends to a spin structure on $\Sigma_0$ with a NS node. We can also define spin structure $\SL'_t,~t>0$ whose restrictions to $B_t$ are the same, except for the elements which correspond to $\alpha,$ on which they are opposite. Both $(\Sigma_t,\SL_t)_{t\in[0,1]}$ and $(\Sigma_t,\SL'_t)_{t\in[0,1]}$ are paths in the moduli of $2-$spin curves which have the same limit point $(\Sigma_0,\SL_0)=(\Sigma_0,\SL'_0),$ and which cover the same path $(\Sigma_t)_{t\in[0,1]}$ in the moduli of curves. The existence of the paths is due to the fact that $\q(\alpha')$ is undefined, and this data loss is the reason for the appearance of the branched cover phenomenon.

In case of a separating NS node, this argument no longer works, however in this case the automorphism group of the spin structure becomes larger: Scaling the fibers of the spin bundle by $-1$ on each one of the two components is an automorphism. This growth of the automorphism group implies that the orbifold degree of the restriction of $\text{For}_{\text{spin}}$ to such strata decreases.
\subsection{The line bundles $\CL_i$}
\begin{definition}
Let $\Gamma$ be a stable graph with an internal tail marked $i\neq 0.$
The line bundle $\CL_i\to\oRCM_\Gamma$ is the line bundle whose fiber at $(\Sigma,\{x_j\}_{j\in\B},\{z_j\}_{j\in\I})\in\oRCM_\Gamma$ is $T^*_{z_i}\Sigma.$
This bundle can also be defined by pulling back the corresponding relative cotangent line over the closed moduli space, via the doubling map.

Let $\Gamma$ be a spin graph with a lifting and an internal tail marked $i\neq0.$
The line bundle $\CL_i\to\oCM_\Gamma$ is the line bundle whose fiber at $(\Sigma,\{x_j\}_{j\in\B},\{z_j\}_{j\in\I})\in\oCM_\Gamma$ is $T^*_{z_i}\Sigma.$
Equivalently, this bundle can be defined as the pullback of $\CL_i\to\oRCM_{{\widetilde{\text{for}}}_{\text{spin}}(\Gamma)}$ by the map ${\widetilde{\text{For}}}_{\text{spin}}.$
\end{definition}

\subsection{Boundary conditions and intersection numbers}\label{subsub:bc}
We begin with a simple observation
\begin{obs}\label{obs:forgetful_illegal}
Let $(\Sigma,\SL,s)$ be a smooth marked surface with a spin structure and a lifting, $\Sigma'$ the marked surface obtained by forgetting points $\{x_b\}_{b\in\B'}$ where $\B'$ is a 
subset of illegal boundary marked points. Then $\SL$ is canonically a (twisted) spin structure for $\Sigma',$ and $s$ canonically extends to a lifting on $\Sigma'.$
In particular, a marked point is legal for $(\Sigma',\SL,s)$ if and only if it is legal for $(\Sigma,\SL,s).$
\end{obs}
\begin{definition}\label{def:abs_vertex1}
Consider $\Gamma\in\CG_{g,k,l}$ and $i\in[l],$ and let $v=i/\sigma_0$ be the vertex of $\Gamma$ which contains the tail marked $i.$
Denote by $v_i^*(\Gamma)$ the following graph, which will be called the \emph{abstract vertex of $i$ in $\Gamma,$} or just the \emph{abstract vertex} for shortness.
\begin{enumerate}
\item
$V(v_i^*(\Gamma)) =\{*\},$ a singleton. It is open if and only if $v$ is.
\item
$T^I(v_i^*(\Gamma)) = (\sig_0^I)^{-1}(v).$ Any internal tail of $v_i^*(\Gamma)$ which corresponds to a tail marked by $j\in[l]$ is be marked $j,$ otherwise it is marked $\III.$ The twist of any tail of $v_i^*(\Gamma)$ is the same as the twist of the corresponding half edge of $v.~H^{\text{CB}}=\emptyset.$
\item
$T^B(v_i^*(\Gamma)) = \{h\in (\sig_0^B)^{-1}(v)|~\Or(h)=1\},$ and all of these boundary tails are marked $\BBB.$
\item
$g(v_i^*(\Gamma))=g(v),~E(v_i^*(\Gamma))=\emptyset.$
\end{enumerate}
Define the map ${\text{for}}_{illegal}:\CG\to\CG,$ which forgets all tails $t\in T^B$ with $\Or(t)=0.$ As a consequence of Observation \ref{obs:forgetful_illegal}, it induces a map at the level of moduli spaces, which will be denoted by ${\text{For}}_{illegal}.$

Write $\Phi_{\Gamma,i}={\text{For}}_{illegal}\circ v_i:\CM_\Gamma\to\CM_{v_i^*(\Gamma)}.$ This map extend to a map $\oCM_\Gamma\to\oCM_{v_i^*(\Gamma)},$ and we also denote the extension by $\Phi_{\Gamma,i}.$
\end{definition}
\begin{figure}
\centering
\includegraphics[scale=.4]{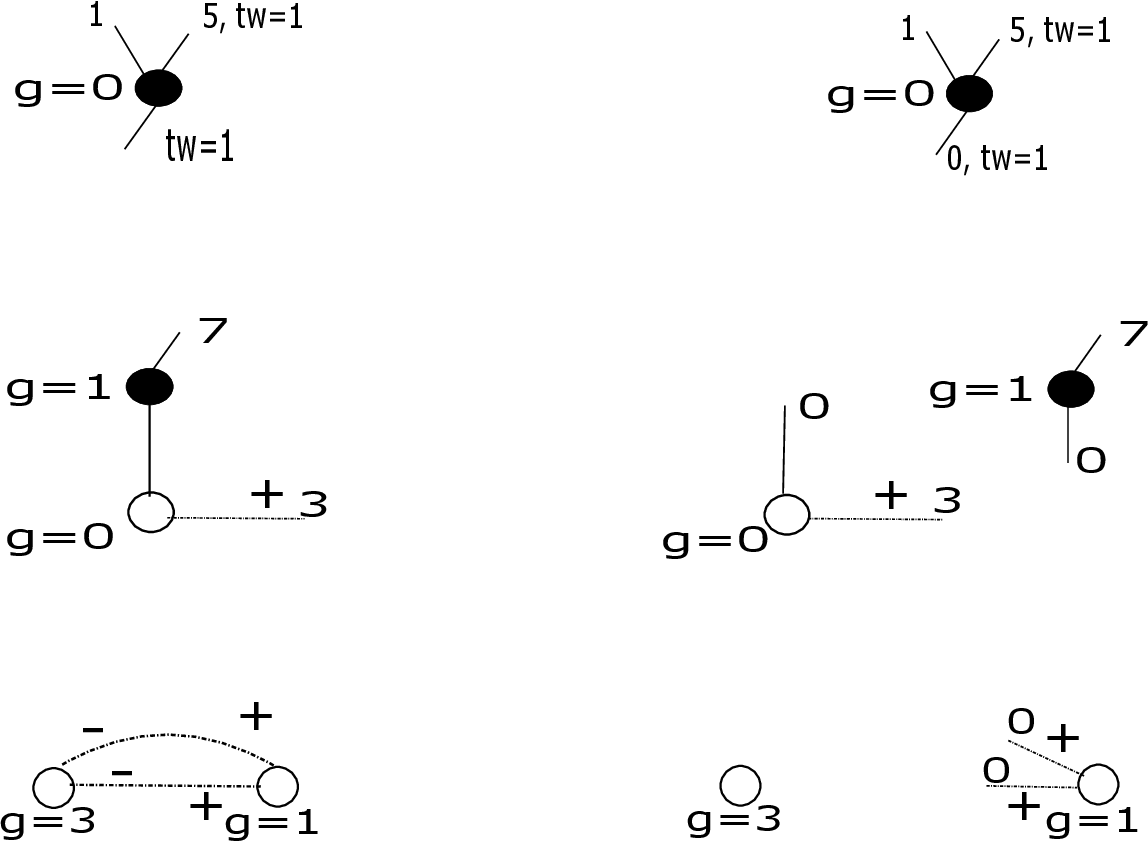}
\caption{In this figure to the right side of each dual graded spin graph its corresponding abstract vertices are shown. Again, half edges $h$ with $\Or(h)=1$ are decorated with '+'.}
\label{fig:2_49}
\end{figure}
At the level of surfaces, $\Phi_{\Gamma,i}(\Sigma),$ for $\Sigma\in\oCM_\Gamma,$ is the graded smooth surface obtained from $\Sigma$ by normalizing the nodes which correspond to the edges of $\Gamma,$ taking the component of $z_i,$ forgetting all illegal half nodes which were formed, renaming all remaining special points by $0,$ and forgetting the lifting at preimages of contracted boundaries, 
see Figure \ref{fig:2_50}
\begin{figure}
\centering
\includegraphics[scale=.4]{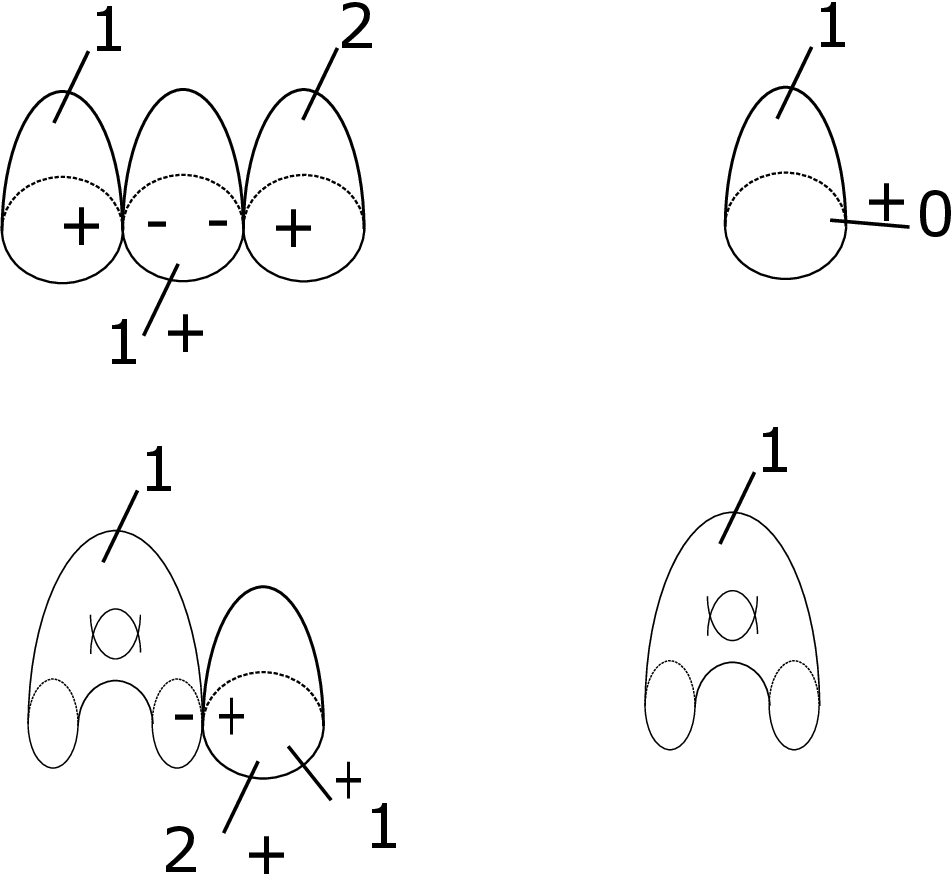}
\caption{In this figure the graded surface on the right of each row is the image (as a moduli point) of $\Phi_{\Gamma,1}(\Sigma)$ where $\Sigma$ is the corresponding surface to the left of the same row, and $\Gamma$ is the dual graph which corresponds to $\Sigma.$}
\label{fig:2_50}
\end{figure}

\begin{obs}\label{obs:tautological_under_base}
For $\Gamma$ as above,
the two orbifold line bundles, $\CL_i\to\CM_\Gamma,$ and $\Phi^*_{\Gamma,i}(\CL_i\to\CM_{v_i^*(\Gamma)})$ are canonically isomorphic.
\end{obs}
For a proof, see \cite{PST} (it is proven there for the $g=0$ case, but the same argument works in general).

In order to define the open intersection numbers we need to define, following \cite{PST,ST} special canonical multisections. We first recall what multisections are, and refer the reader to \cite[Appendix A]{BCT2} for more details and references.
\begin{definition}\label{def:multisec}
Let $E\to M$ be a orbibundle over an orbifold with corners, and we identify $E$ with its total space. A \emph{multisection} is a function $\kappa:E\to\mathbb{Q}_{\geq 0}$ which satisfies the following properties.
For any $p\in M,$ let $(F\to U)/G$ be a local model for $E\to M$ in a neighborhood of $p,$ where $U\simeq\R^m\times\R_{\geq 0}^{n-m},$ $p$ is identified with $0,$ $F\simeq U\times \R^h,$ the map $\pi:F\to U$ is the projection and $G$ is a finite group acting linearly on the pair, commuting with $\pi.$ Denote by $\hat\kappa$ the pullback of $\kappa$ to a $G-$invariant function on $F.$ Then,
\begin{enumerate}
\item for all $y\in U,$\[\sum_{v\in \pi^{-1}(y)}\hat\kappa(v)=1.\]
\item We can find sections $s_1,\ldots,s_N:U\to F,$ perhaps after replacing $U$ with a smaller neighborhood of $0,$ and non negative rational numbers $\mu_1,\ldots,\mu_N,$ such that for all $y\in U,$~$v\in \pi^{-1}(y)$
    \[\hat\kappa(v)=\sum_{i|s_i(y)=v}\mu_i.\]
\end{enumerate}
The sections $s_1,\ldots,s_N$ are called \emph{local branches} and the numbers $\mu_1,\ldots,\mu_N$ are their \emph{weights}. The locus where $\kappa\neq 0,$ which is locally the union of its local branches, is called the \emph{support} of the multisection. The elements in the support of $\kappa$ which lie in the fiber $E_p$ of $E$ over $p$ form set of \emph{values} of the multisection at $p$. 
\end{definition}
Although the support does not, in general, capture all the information of the multisection, we usually refer to the multisection $\kappa$ by its support $s,$ and write $s(x)$ for the values of the multisection at $x.$
If $N=1$ for all $p\in M$ then the multisection is just a usual section.
The multisection is smooth (piecewise smooth) if all its local branches are smooth (piecewise smooth). Many of the natural operations and properties of sections of vector bundles generalize to multisections of orbibundles in a natural way. These include addition of multisections, multiplications by functions $f: M\to \R$ and most transversality statements.
We say that the multisection is nowhere vanishing if none of its branches vanishes, or equivalently $\kappa(x,0)= 0$ for all $x\in M.$ The multisection is transverse to zero if all its branches are transverse to the zero section, and it has isolated zeroes, if all its local branches have isolated zeroes. $x$ is a zero of the multisection if $\kappa(x,0)\neq 0,$ that is, at least one of the local branches at $x$ vanishes at $x.$ The zero locus of a multisection is the set of its zeroes.
\begin{definition}\label{def:special_canonical}
Suppose $A\subseteq\CG_{g,k,l}$ is a collection of graphs with at least one boundary edge. A piecewise smooth multisection $s$ of $\CL_i\to\cup_{\Gamma\in A}\oCM_{\Gamma}$ is called
\emph{special canonical} on $\cup_{\Gamma\in A}\oCM_{\Gamma}$ if for all $\Lambda\in\pu\Gamma$
\[
s|_{\CM_\Lambda} = \Phi_{\Lambda,i}^*s^{v_i^*(\Lambda)},
\]
for some piecewise smooth multisection $s^{v_i^*(\Lambda)}$ of $\CL_i\to \CM_{v_i^*(\Lambda)}.$

In case $A\subseteq\CG_{g,k,l}$ is the collection of all graphs with at least one boundary edge, we say that $s$ as above is \emph{special canonical}.

A multisection $s=\bigoplus_{i\in[l],j\in[a_i]} s_{ij},$ of $\bigoplus_i\CL_i^{\oplus a_j}$ is special canonical if each component $s_{ij}$ is special canonical.
\end{definition}
Intuitively, being special canonical means that the multisection depends only on the irreducible component of $z_i$ in the normalization, after forgetting the locations of the illegal boundary half nodes and the liftings at contracted boundaries.

Still following \cite[Appendix A]{BCT2}, let $p\in M$ be an internal point, and let $s$ be a multisection with isolated zeroes. We assume that $E,M$ are oriented and than $\text{rk}(E)=\dim(M).$
Take a local model $(F\to U)/G$ for the neighborhood of $p$ as in Definition \ref{def:multisec}. Choose a metric on $U,$ a metric on the fibers $\R^h,$ and let $\pi':F\to\R^h$ be the projection on the $\R^h$ component. Let $B$ be a small ball around $0$ (which is identified with $p$), which contains no zero of $s$ except, possibly, $0.$ Denote by $S$ the unit sphere in $\R^h.$ We use the orientations of $M,~E$ to endow $S,\partial B$ with the induced orientations as the boundaries of oriented balls.
We define $\text{deg}_p(s_i)$, the local degree of $s_i$ at $p$ as the degree of the map $t:\partial B\to S,$ where \[t(x)=\frac{\pi'(s_i(x))}{|\pi'(s_i(x))|}.\]
This definition is independent of choices.
The \emph{weight} of $p$ in the zero locus of $s$ is defined as
\begin{equation}\label{eq:eps_p}\epsilon_p=\frac{1}{|G|}\sum_{i=1}^N\mu_i\text{deg}_p(s_i).\end{equation}
If $s$ has a finite zero locus $\{p_1,\ldots,p_t\}$, then the \emph{weighted signed zero count} of $s$ is $\sum_{i=1}^t\epsilon_{p_i}(s)\in\mathbb{Q}.$

Let $s$ be a piecewise smooth multisection of $E\to \partial M,$ where $E\to M$ is an oriented orbibundle over a compact oriented orbifold with corners. Suppose $s$ vanishes nowhere.
For \emph{any} piecewise smooth multisection $\widetilde{s}$ extending $s$ to the interior of $M$ with isolated zeroes, the weighted signed zero count of $\widetilde{s}$ is the same.
This follows from standard cobordism arguments (see, for example \cite[Section 3]{Fuk}, for the case $\partial M=\emptyset$ the addition of boundary does not complicate the argument\footnote{In \cite{Fuk} the definition of multisections is slightly different, as a section to the symmetric product of the orbifold vector bundle. However a multisection in our terminology induces in a natural way a multisection in the terminology of that paper, and the definitions of the zero counts agree.}), and it is also a consequence of Proposition \ref{prop:basic prop}, whose proof is sketched below. We denote this number by $\int_Me(E,s)$ and call it the \emph{integral of the relative Euler class} of $E$ relative to $s.$
\begin{rmk}
The \emph{relative Euler class} $e(E,s)\in H^{n}(M,\partial M,\mathbb{Q}),$ where $E\to M$ is an oriented orbibundle over a compact oriented orbifold with corners with $\text{rk}(E)=\dim (M)=n,$ is defined whenever $s$ is a nowhere vanishing boundary condition for $E\to M.$
Integrating, or capping with the fundamental class, gives by Poincar\'e-Lefschetz duality, an element of $H_0(M,\mathbb{Q})\simeq\mathbb{Q}.$ This element is precisely what we defined as the integral of the relative Euler class. For our needs the definition of the relative Euler class itself is not needed.
See the appendix in \cite{BCT2} for more details and references.

The integral relative Euler class can be defined for orbifold sphere bundles rather than orbifold vector bundles, for example by using an embedding of the sphere bundle into the vector bundle using a choice of a metric for the vector bundle, and inducing the boundary conditions by this embedding. The resulting integrals are the same when working with a vector bundle $E$ or with its associated sphere bundle. We shall use these two notions interchangeably throughout the paper.
\end{rmk}
\begin{obs}\label{obs:functoriality_Euler}
Suppose $E\to M$ is an oriented orbibundle over a compact oriented piecewise smooth orbifold with corners with $\text{rk}(E)=\dim (M)=n$ and $s$ a nowhere vanishing multisection of $E\to\partial M.$ Let $f:N\to M$ be a surjection between compact oriented piecewise smooth orbifolds with corners of dimension $n$ which maps $\partial N$ onto $\partial M.$ Suppose that $f$ is generically of degree $1,$ meaning that outside of a subspace $K\subset M$ which is a union of finitely many compact suborbifolds of $M$ of real codimension $1$ $f$ is injective. Then
\[\int_Ne(f^*E,f^*s)=\int_Ne(E,s).\]
\end{obs}
Indeed, standard transversality arguments show that a generic piecewise extension of $s$ to $M$ will have no zeroes in $K.$ Using the pull back to $N$ of such a generic extension proves the claim.

The following theorem has appeared in \cite{PST} in the genus $0$ case, and will appear in \cite{ST} for all genera.
\begin{thm}\label{thm:PST}
Suppose $a_1,\ldots,a_l\geq 0$ are integers which sum to $\frac{k+2l+3g-3}{2}.$ Then one can choose multisections
$\{s_{ij}\}_{i\in[l],j\in[a_i]}$ such that
\begin{enumerate}
\item
For all $i,j~s_{ij}$ is a special canonical multisection of $\CL_i\to\partial\oCM_{g,k,l}.$
\item
The multisection $s = \bigoplus_{i,j}s_{ij}$ vanishes nowhere.
\end{enumerate}
Moreover, for any two choices $\{s_{ij}\},\{s'_{ij}\}$ which satisfy the above requirements we have
\[
\int_{\oCM_{g,k,l}}e(\bigoplus_i\CL_i^{\oplus a_j},s) = \int_{\oCM_{g,k,l}}e(\bigoplus_i\CL_i^{\oplus a_j},s')
\]
where 
$s' = \bigoplus_{i,j}s'_{ij}.$
\end{thm}
For completeness, and since \cite{ST} is yet to appear, we will shortly review the proof of first claim in the theorem. We will not review the 'Moreover' part, since it will be a consequence of our main theorem, Theorem \ref{thm:comb_model}, which calculates the integral of the relative Euler class, and obtains an answer which does not involve the special canonical multisection, without relying on the assumption that the integral is independent of the multisection.

The proof that nowhere vanishing special canonical boundary conditions exist has two steps.
The first step shows that for any boundary point $p\in \partial\oCM_{g,k,l}$ there exists a special canonical multisection such that none of its branches vanishes at $p.$ This step is the heart of the argument, it is similar but not identical to \cite[Proposition 3.49(a)]{PST} and we will review it in the next paragraph.
The second step uses the multisections constructed in the first step to construct nowhere vanishing boundary conditions: Using the first step, and compactness one can find finitely many canonical multisections $s_1,\ldots,s_N$ of $E=\bigoplus_{i\in[l],j\in[a_i]}\CL_i^{\oplus a_i},$ such that for any boundary point $p\in \partial\oCM_{g,k,l},$ and any choice of local branches $s'_i$ of $s_i$ at $p,$ the vectors $(s'_i)_p,~i\in [N]$ span the fiber $E_N.$ Then, by a standard transversality argument, a generic linear combination of $s_1,\ldots,s_N$ will be a nowhere vanishing canonical multisection. By generic we mean that the subset of linear combinations of $s_1,\ldots,s_N$ with this property is residual in set of all possible linear combinations. The proof of this step is identical to \cite[Lemma 3.53(a)]{PST}, and we refer the interested reader there.

We turn to explain the first step. Fix $p\in \partial\oCM_{g,k,l}$ and $i\in [l].$ Suppose $p$ belongs to the stratum $\CM_\Gamma,$ for some graded spin dual graph $\Gamma.$ corresponds to the graded surface $\Sigma.$ Let $u\in(\CL_i)_p$ be an arbitrary non zero vector. Finally, let $[\Sigma']$ be the image of $p=[\Sigma]$ under the map $\Phi_{\Gamma,i},$ and write $G=\text{Aut}(\Sigma).$ The action of $G$ lifts to an action on the cotangent of the $i^{th}$ marking, that is, on $(\CL_i)_{[\Sigma']},$ the fiber of $\CL_i$ at $[\Sigma'].$ By Observation \ref{obs:tautological_under_base}, the fibers of $\CL_i$ at $[\Sigma'],[\Sigma]$ are isomorphic, canonically up to the action of $G$ on $(\CL_i)_{[\Sigma']}.$ Thus, the $G-$action lifts also to $(\CL_i)_{[\Sigma]}.$ Write \[\mathbf{u}=\{u_1,\ldots,u_m\}=\{g\cdot u|g\in G\}.\]

We will construct a special canonical multisection of $\CL_i$ whose branches at $p$ have values $\mathbf{u},$ with equal weights. Set \[V_{g,k,l}=\{v_i^*(\Lambda)|\Lambda\in\pu\Gamma_{g,k,l}\},\]i.e. $V_{g,k,l}$ is the collection of abstract vertices $v_i^*(\Lambda)$ for any graded spin graph $\Lambda$ that corresponds to a stratum of $\oCM_{g,k,l}.$ We will construct for any $v\in V_{g,k,l}$ a special canonical multisection $s^v$ for $\CL_i\to\oCM_{v}.$
These multisections are required to be compatible in the following sense. Let $v\in V_{g,k,l},$ let $\Lambda\in\partial v$ be a graph which corresponds to a boundary stratum of $\oCM_v,$ and let $v'=v_i^*(\Lambda).$ It is easy to see that $v'\in V_{g,k,l}.$ We require, for all such $v,\Lambda$ that
\begin{equation}\label{eq:compatibility}
s^v|_{\CM_\Lambda}=\Phi_{\Lambda,i}^*s^{v'}.
\end{equation}
These constraints, for different $\Lambda$ are compatible. See the explanation in the beginning of the proof of \cite[Proposition 3.49]{PST}, which extends to our setting.
This construction will provide, in particular, a construction of a special canonical multisection for $v_i^*(\Gamma_{g,k,l}),$ which is the same graded dual graph as $\Gamma_{g,k,l}$ except that the boundary tails which are marked $\BBB.$ The pull-back of this section by the canonical map
\[\oCM_{\Gamma_{g,k,l}}\to\oCM_{v_i^*(\Gamma_{g,k,l})},\]
which changes the boundary markings to $\BBB$ will be the required multisection.

Write $v^*=v^*_i(\Gamma)$ and $a=\dim (\CM_{v^*}),$ where $\Gamma$ is the dual graph which corresponds to $\Sigma.$
The construction of multisections $s^v,~v\in V_{g,k,l},$ will be by induction on $d=\dim \CM_{v}.$
The basis is $d=-1,$ which holds trivially since there are no such vertices.
Suppose we have constructed multisections with the required properties for all $v'$ with $\dim\CM_{v'}<d.$
Consider $v\in V_{g,k,l}$ with $d=\dim\CM_{v}.$ Note that $v$ needs not to be an open vertex, and may even have internal tails with $\tw=1.$ Write $\Upsilon=\coprod_{\Lambda\in\partial v}\CM_\Lambda.$ Define first $s^v|_{\Upsilon}$ according to \eqref{eq:compatibility}, where the right hand side of the compatibility equations is already defined by induction. We now extend $s^v$ to the whole moduli space $\oCM_v.$ Here we separate into cases. If $v\neq v^*,$ we extend arbitrarily. If $v=v^*$ we extend arbitrarily, but under the requirement that $s^{v^*}_{[\Sigma']}=\mathbf{u},$ meaning that each $u_i$ appears in some branches of $s^{v^*},$ and with the same total weight. This can be done for example in the following way. Let $\rho:\oCM_{v^*}\to [0,1]$ be a smooth function which is $1$ near $[\Sigma']$ and $0$ near $\Upsilon.$ Let $s'$ be an arbitrary extension of the already defined $s^{v^*}|_{\Upsilon}$ to $\oCM_{v^*},$ and $s''$ an arbitrary multisection of $\CL_i\to\oCM_{v^*}$ which has the required values $[\Sigma'].$ Then one can take \[s^{v^*}=\rho s''+(1-\rho)s'.\] The induction follows\footnote{In the proof of the corresponding claim in \cite{PST} the multisection were also required to satisfy some invariance under symmetry groups. In our case, since we work with orbifolds and orbibundles, this invariance is part of the definition of being a multisection, see the appendix of \cite{BCT2}. In \cite{PST} the orbifoldness was implicit, and was a result of forgetting the boundary markings. In higher genus even the moduli with injective markings is an orbifold.}, and thus also the proof.

For the benefit of the reader we now explain the difference between this proof and the proof of \cite[Proposition 3.49(a)]{PST}, and the intuitive reason for why canonical boundary conditions should give rise to well defined intersection numbers. In \cite{PST} there were no contracted boundaries and all boundary nodes where separating. In this case the definition of canonical boundary conditions can be given without spin structure, only by using parity considerations: For each node precisely one half node is forgotten, and the forgotten half nodes are chosen in the unique way which leaves on each connected component of genus $s$ of the normalized surface a total number of unforgotten special boundary points whose parity is $s+1~(\text{mod}~2).$ This numerical reasoning cannot work when there are non separating nodes. However, as it turns out, this parity notion neatly generalizes to the notion of a graded spin structure, and the forgotten half nodes are precisely the illegal ones. The importance of this scheme is that it forces the boundary conditions to be pulled back from a real codimension $2$ space rather than from a codimension $1$ space (the codimension is with respect to the dimension of the whole moduli space).

This idea cannot work in case of moduli strata which parameterize surfaces with a contracted boundary component. However, for such surfaces, for any contracted boundary component, there are two possible choices of liftings. Moreover, by the 'Moreover' part of Theorem \ref{lem:geometric_orientability} below, the boundary strata of the moduli which correspond to the different choices of liftings come with opposite orientations. Since, in the definition of the base, the lifting in such points is forgotten, the boundary conditions should be the same, for these two boundary strata\footnote{Essentially this discussion says that such codimension $1$ boundaries can be glued, and that the integrals can be calculated with respect to the glued moduli space. In an earlier version of this manuscript we have chosen this path, but we believe that this gluing is less elegant than the equivalent choice of unglued boundaries we make here. The cost of this choice is that there are now additional boundary conditions to impose and to analyze.}.

These two properties are strong enough to guarantee that the integrals are well defined:
The dimension reduction, together with a standard transversality argument, enables one to construct a homotopy between any two choices of canonical boundary conditions $s,s'$, which does not vanish on boundary strata which correspond to surfaces with a boundary node. It may vanish on boundary strata which correspond to surfaces with contracted boundaries, but these vanishings cancel in pairs, which differ in the liftings of these contracted boundaries. This homotopy argument thus shows that $s,s'$ determine the same integral. In the course of proof of Theorem \ref{thm:comb_model} this independence will become manifest.

Based on Theorem \ref{thm:PST} can now define open intersection numbers.
\begin{definition}\label{def:int_nums}
With the notations of Theorem \ref{thm:PST},
define the open intersection number
\[
\langle\tau_{a_1}\ldots\tau_{a_l}\sigma^k\rangle_g :=  2^{-\frac{g+k-1}{2}}\int_{\oCM_{g,k,l}}e(\bigoplus_i\CL_i^{\oplus a_j},s),
\]
where $s$ is a nowhere vanishing special canonical multisection.
\end{definition}
The power of $2$ is a normalization factor chosen in \cite{PST} which makes some initial conditions nicer, but has no geometric or algebraic importance.

Since we define the intersection numbers to be $0$ unless the numerical condition of Theorem \ref{thm:PST} holds, the genus is determined from knowing $k,l,a_1,\ldots, a_l$ and for this reason we will usually omit it from the notation and simple write $\langle\tau_{a_1}\ldots\tau_{a_l}\sigma^k\rangle$ for $\langle\tau_{a_1}\ldots\tau_{a_l}\sigma^k\rangle_g.$

\subsection{The orientation of $\oCM_{g,k,l}$}
As mentioned above, the spaces $\oCM_{g,k,l}$ were proved to be orientable, and moreover were given canonical orientations.
In order to state properties of these orientations that will be required for later, we need the following definition.
\begin{definition}\label{def:induced_or}
Let $M$ be an oriented orbifold with corners. Then $\partial M$ is also orientable. The \emph{induced orientation} on $\partial M,$ is defined by the exact sequence
\[
0\to N\to TM|_{\partial M}\to T\partial M\to 0,
\]
where $N,$ the dimension $1$ normal bundle of $\partial M$ in $M,$ is oriented by taking the outward normal as a positive direction and the orientation on $TM$ is the given one. 
\end{definition}
{For the benefit of the reader, we recall the construction of the induced orientation also in terms of local coordinates. Let $p$ be a boundary point which is not a corner, a local neighborhood of $p$ is diffeomorphic to $(\R_{\geq 0}\times \R^{n-1})/G,$ for some finite group $G$ which acts on $\R^n,$ and under the diffeomeorphism $p$ is mapped to the origin. By the orientability assumption $G$ acts in an orientation preserving manner, and we may assume that the orientation induced on $\R^n$ by the diffeomorphism is the standard one. Since $p$ is a boundary point, $\{0\}\times\R^{n-1}$ is preserved by $G,$ and since $G$ acts on $\R_{\geq 0}\times \R^{n-1},$ by definition $\R_{\geq 0}\times \R^{n-1}$ is preserved. Take an oriented frame $(v_1, v_2,\ldots, v_n)$ for $\R^n,$ which is in the class of the standard orientation, such that $v_1$ has negative first coordinate and the remaining vectors of the frame have the first coordinate $0.$ Then $(v_2,\ldots, v_n)$ is a frame for $\{0\}\times \R^{n-1}.$ For $g\in G,$ $(gv_1,\ldots,gv_n)$ is in the same orientation class as the original frame. Since $\R_{\geq 0}\times \R^{n-1},$ is preserved under $G,$ $gv_1$ has a negative first coordinate. Since $\{0\}\times\R^{n-1}$ is preserved under $G,$ the first coordinate of each $gv_i,~i\geq 2$ is $0,$ we obtain that \[(gv_2,\ldots,gv_n),~(v_2,\ldots,v_n)\]are in the same orientation class. This class is defined as the orientation frame, which defines the local orientation at $p.$
We extend the orientation to the whole boundary by continuity.}

The next theorem, proven in \cite{ST}, describes some useful properties of the canonical orientations of $\oCM_{g,k,l},$ properties that characterize these orientations uniquely.
\begin{thm}\label{lem:geometric_orientability}
There is a unique choice of orientations $\mathfrak{o}_\Gamma,$ for any graded graph $\Gamma$ all of whose connected components contain a single vertex,
which satisfy the following requirements:
\begin{enumerate}
\item
The $0-$ dimensional spaces $\oCM_\Gamma,$ for $\Gamma\in\{\Gamma_{0,1,1},\Gamma_{0,3,0}\},$ are oriented positively.
\item
If $\Gamma=\{\Gamma_1,\ldots,\Gamma_r\},$ the connected components, then $\mathfrak{o}_\Gamma=\boxtimes_{i=1}^r\mathfrak{o}_{\Gamma_i}.$
\item
Let $\Gamma$ be a graded stable graph with a single boundary edge, $e,$ and put $\Lambda = d_e\Gamma$. Denote by $\Gamma'$ the graph obtained by detaching that edge into two tails $t,~t'$ with $\Or(t)=1,~\Or(t')=1,$ and forgetting the tail $t.$
Note that we have a fibration $\CM_\Gamma\to\CM_{\Gamma'},$ whose fiber over the graded surface $\Sigma\in\CM_{\Gamma'},$ is naturally identified with $\partial\Sigma\setminus{\{x_i\}}_{i\in B(\Gamma')}.$
Then the induced orientation on $\CM_{\Gamma}$ as a codimension $1$ boundary of $\oCM_{\Lambda},$
agrees with the orientation on $\CM_{\Gamma}$ induced by the fibration $\CM_\Gamma\to\CM_{\Gamma'},$ where the base is given the orientation $\mathfrak{o}_{\Gamma'},$ and the fiber over $\Sigma$ gets the orientation of $\partial\Sigma.$
\end{enumerate}
Moreover, these orientations have the following additional property.
For $\Gamma$ as above, let $C$ be a connected component of $\oCM_\Gamma$ which parameterizes surfaces with at least one boundary component which contains no boundary markings. Let $C'$ be another connected component which parameterizes surfaces which differ from those of $C$ only at the grading in that boundary component, which is opposite. There is a natural map $\FFF:C\to C'$ which maps a stable graded marked surface to the same surface only with the opposite grading at this boundary component.
Let $C_\times,~C'_\times$ be the boundary strata of $C,~C'$ respectively which parameterize surfaces in which this boundary component is contracted,
and let $\mathfrak{o}_{C_\times},~\mathfrak{o}_{C'_\times}$ be the orientations induced on these subspaces by $C,~C'$ respectively.
Then $\FFF$ maps $C_\times$ bijectively on $C'_\times,$ \[\mathfrak{o}_{C'_\times}=-\FFF_*\mathfrak{o}_{C_\times}.\]
\end{thm}
The difficulty in this theorem lies in the existence and the 'Moreover' parts, which will be proven by other means below. Given the existence, the uniqueness follows easily using induction on dimension. In \cite{ST} also the behavior of the orientations with respect to strata with internal nodes is explained, but it is not required for our needs.

\section{Sphere bundles and relative Euler class}\label{sec:sphere}
Given a rank $n$ complex vector bundle $\pi:E\to M,$ and a metric on it, one can define the sphere bundle $\pi:S=S(E)=S^{2n-1}(E)\to M$ whose fiber $S_p$ at $p\in M$ is the set of length $1$ vectors in $E_p,$ the fiber of $E$ at $p,$ with the induced orientation.
Given a sphere bundle $S\to M,$ its \emph{linearization} is the space
\[
S\times\R_{\geq 0}/\sim,
\]
where $(v,r)\sim(v',r')$ if either $r=r'=0,$ or $v=v',r=r'.$ This space can be endowed with a natural linear structure, a metric and a projection to $M.$ When $S=S(E)$ the linearization of $S$ recovers $E.$
The sphere bundle of $E$ can be defined also without referring to a metric, by removing the $0-$section and taking the quotient by the $\R_+$ action. Different metrics give rise to isomorphic sphere bundles

\begin{definition}\label{def:angular}
An \emph{angular form} for $E$ (or for $S$) is a $(2n-1)-$form $\Phi$ on $S$ which satisfies the following two requirements:
\begin{enumerate}
\item
$\int_{S_p}\Phi = 1,$ for all $p\in M.$
\item
$d\Phi = -\pi^*\Omega$ where $\Omega$ is some $2n-$form on $M.$
\end{enumerate}
Then, the form $\Omega$ is a local representative of the top Chern form of $E\to M,$ and will be called the \emph{Euler form} which corresponds to $\Phi.$
Denote by $\Phi$ also the form on $E\setminus M,$ where we identify $E$ and its total space, defined by $P^*\Phi,$ where $P:E\setminus M\to S(E),$ is the map
\[
(p,v)\to (p,v/|v|),~p\in M,v\in E\setminus M.
\]
\end{definition}
It is straightforward that
\begin{obs}\label{obs:extending_angular}
The form $|v|\Phi$ extends to a form on all the total space of $E.$
\end{obs}

We will use the following claim.
\begin{pr}\label{prop:basic prop}
Let $E\to M$ be a real oriented rank $2n$ vector bundle on a smooth oriented manifold with boundary $M$ of real dimension $2n$. Let $\Phi$ be an angular form, and $\Omega$ its corresponding Euler form. Given a nowhere vanishing section $s\in\Gamma(E\to\partial M),$ one can define the integral of the relative Euler class, 
and it holds that
\[
\int_Me(E,s) = \int_M \Omega +\int_{\partial M}s^*\Phi.
\]
Moreover, the statement also holds if $E\to M$ is an orbifold vector bundle over an orbifold with corners, and $s$ is a nowhere vanishing multisection over the boundary.
\end{pr}
This claim is well known, in the case of manifolds, and the extension to orbifolds is straight forward.
We briefly recall the proof of the claim for manifolds, referring the reader to \cite[Ch. 11]{BottTu} for further details, we then explain the changes required for handling the orbifold case. As usual we are interested in the integral of the relative Euler class, rather than the class itself.

We wish to calculate $\int_Me(E,s),$ the weighted number of zeroes of an extension of $s$ to $M$ to a section with isolated zeroes.
Let $\bar{s}$ be such an extension, and let $p_1,\ldots, p_m$ be its zeroes. By choosing diffeomorphisms from neighborhood of $p_1,\ldots, p_m$ to open sets in $\R^n,$ we can define, for small enough $r,$ $M_r=M\setminus \bigcup_{i=1}^m B_r(p_i),$ where $B_r(p)$ is the ball around $p,$ and sections $s_r$ which are the restrictions of $\bar{s}$ to $\partial M_r.$ By taking $r$ to be even smaller we may assume that the balls are disjoint.
By Stokes's theorem, $\bar{s}$ being a global section over $M_r,$ and the definition of the angular form,
\begin{align*}\int_M\Omega&=\lim_{r\to 0}\int_{{M_r}}\Omega=\lim_{r\to 0}\int_{M_r}\bar{s}^*\pi^*\Omega\\
&=-\lim_{r\to 0}\int_{M_r}\bar{s}^*d\Phi=-\lim_{r\to 0}\int_{\partial M_r}s_r^*\Phi=
-\int_{\partial M}s^*\Phi+\sum_{i=1}^m\lim_{r\to 0}\int_{\partial B_r(p_i)}s_r^*\Phi.\end{align*}
For each $i=1,\ldots, m,$ and small enough $r,$ $\int_{\partial B_r(p_i)}s_r^*\Phi$ is the order of vanishing of $\bar{s}$ at $p_i$ (see \cite[Theorem 11.16]{BottTu}). Thus, the RHS of the previous equation equals $\int_M e(E,s)-\int_{\partial M}s^*\Phi,$ as needed.

The argument works also in the orbifold case. One first shows that Stokes's theorem generalizes to the case of orbifolds with corners and multisections of the vector bundle $\Lambda^\bullet(T^*M)$ instead of sections of this bundle (differential forms). For differential forms over orbifolds with corners this is shown for example in \cite{SaraJake}. The extension to multisections is proven similarly. Then, the integral around $p_i$ becomes, in the local model and notations of Definition \ref{def:multisec},
\[\sum_{i=1}^N\mu_i\int_{\partial B(0)}\bar{s}_i^*\Phi,\] $B\subset U$ a small ball around $0,$ and $\bar{s}_i,~\mu_i$ the local branches and weights.
But this is precisely the weight \eqref{eq:eps_p} in the definition of $\int_Me(E,s),$ so again the result follows. 

Suppose now that $E = \bigoplus_{i=1}^n L_i,$ is the sum of $n$ complex line bundles $L_i.$ Choose a metric for $E$ for which the line bundles $L_i$ are pairwise orthogonal.
Write $\alpha_i$ for an angular form for $S_i=S(L_i),$ and $\omega_i$ for the corresponding Euler form, i.e. the curvature of $L_i.$
Define the functions
$$r_i:E\to \R,$$ to be the length of the projection of $(p,v)\in E$ to $L_i.$ The sphere bundle can be described as the set of vectors which satisfy $\sum r_i^2 = 1.$
For convenience, denote by $\omega_i, r_i\alpha_i$ the pull-backs of $\omega_i,r_i\alpha_i$ to the total space of $E$ and of $S(E),$ where for the latter form we use Observation \ref{obs:extending_angular}
%

As far as we know, the following theorem has not appeared in the literature before.
\begin{thm}\label{thm:angular}
The following form,
\begin{equation}\label{eq:angular}
\Phi = \sum_{k=0}^{n-1} 2^k k!\sum_{i\in[n]}r_i^2\alpha_i\wedge
\sum_{\substack{I\subseteq[n]\setminus\{i\}\\ |I|=k}}\bigwedge_{j\in I}(r_jdr_j\wedge\alpha_j)\bigwedge_{h\notin I\cup\{i\}}\omega_h.
\end{equation}
is an angular form for $E$ whose corresponding Euler form is $\bigwedge_{i=1}^n\omega_i.$
\end{thm}

\begin{proof}
We first need to show that the integration on a fiber gives $1.$ Since $\omega_i$ are pulled back from the base, for all $i,$ the only term in $\Phi$ that may have a non zero integral over a fiber is the term
$$ \Phi^{\text{top}}=2^{n-1} (n-1)!\sum_{i\in[n]}r_i^2\alpha_i
\bigwedge_{j\neq i}(r_jdr_j\wedge\alpha_j).$$
We wish to show that $\int_{S(E_p)}\Phi^{\text{top}}=1,$ for an arbitrary $p\in M.$ We first integrate all the $\alpha_i$ terms. By using that $\alpha_i$ is an angular form for $L_i$ the integral of $\alpha_i$ is $1,$
and we are left with calculating
\[\int_{\sum r_i^2=1}2^{n-1} (n-1)!\sum_{i\in[n]}r_i^2
\bigwedge_{j\neq i}r_jdr_j.\]
By changing the variables to $t_i=r_i^2,~dt_i=2r_idr_i$ the integral becomes
\[(n-1)!\int_{\substack{\sum t_i=1,\\t_1,\ldots,t_n\geq 0}}\sum_{i=1}^nt_i\bigwedge_{j\neq i}dt_j=
n!\int_{\substack{\sum_{i=1}^{n-1} t_i\leq1,\\t_1,\ldots,t_{n-1}\geq 0}}(1-\sum_{i=1}^{n-1}t_i)\bigwedge_{1\leq j\leq n-1}dt_j=
n!\int_{\substack{\sum_{i=1}^{n} t_i\leq1,\\t_1,\ldots,t_{n}\geq 0}}\bigwedge_{1\leq j\leq n}dt_j,\]
where in the first equality we have used the symmetric role of the variables $t_i$ and then eliminated $t_n,$ and in the second equality we have used that \[1-\sum_{i\leq n-1}t_i=\int_{0\leq s\leq 1-\sum_{i\leq n-1}t_i}ds.\]
The left hand side is just $n!$ times the Euclidean volume of the $n-$simplex \[\{t_1+\ldots+t_n\leq 1|t_1,\ldots,t_n\geq 0\}.\] It is well known that this volume is $\frac{1}{n!},$ and the first property of the angular form follows.

For the second property, we will now show that when calculating $d\Phi,$ one gets a telescopic sum which turns out to be equal $\bigwedge\omega_i.$ Write,
$$S_{I,i} :=  2^k k! r_i^2\alpha_i
\bigwedge_{j\in I}(r_jdr_j\wedge\alpha_j)\bigwedge_{h\notin I\cup\{i\}}\omega_h,$$
the contribution for given $I,i\notin I,$ where $k=|I|.$ Taking the derivative, as $\omega_i$ and $r_i dr_i$ are closed, only $r_i^2$ or $\alpha_j$ may contribute.
We obtain
$$dS_{I,i}=d_1S_{I,i}+d_2S_{I,i}+\sum_{l\in I}d_{3,l}S_{I,i},$$
where
$$d_1 S_{I,i}:=  2^{k+1} k! r_i dr_i\wedge\alpha_i
\bigwedge_{j\in I}(r_j dr_j\wedge\alpha_j)\bigwedge_{h\notin I\cup\{i\}}\omega_h,$$
$$d_2 S_{I,i}:= - 2^{k} k! r_i^2\omega_i
\bigwedge_{j\in I}(r_j dr_j\wedge\alpha_j)\bigwedge_{h\notin I\cup\{i\}}\omega_h,$$
$$d_{3,l} S_{I,i}:= - 2^{k} k! r_i^2\alpha_i \wedge r_l dr_l\wedge \omega_l
\bigwedge_{j\in I\setminus{\{l\}}}(r_j dr_j\wedge\alpha_j)\bigwedge_{h\notin I\cup\{i\}}\omega_h,$$
for $l\in I.$ The third contribution appears only when $I\neq\emptyset.$

Now, fixing $I,$ one has
\begin{equation}\label{eq:d1}
\sum_{i\in I} d_1 S_{I\setminus{\{i\}},i}=  k 2^{k} (k-1)!
\bigwedge_{j\in I}(r_j dr_j\wedge\alpha_j)\bigwedge_{h\notin I}\omega_h,
\end{equation}

\begin{align}
\label{eq:d2}
\notag\sum_{i\notin I} d_2 S_{I,i} &= -\sum_{i\notin I} 2^{k} k! r_i^2
\bigwedge_{j\in I}(r_j dr_j\wedge\alpha_j)\bigwedge_{h\notin I}\omega_h  \\&=-(1-\sum_{i\in I}r_i^2) 2^{k} k!
\bigwedge_{j\in I}(r_j dr_j\wedge\alpha_j)\bigwedge_{h\notin I}\omega_h\\\notag
&=-2^{k} k! \left(\bigwedge_{j\in I}(r_j dr_j\wedge\alpha_j)\bigwedge_{h\notin I}\omega_h\right. \\ \nonumber
    &\qquad\qquad\quad -\left.\sum_{i\in I} r_i^3 dr_i\wedge\alpha_i\bigwedge_{j\in I\setminus{\{i\}}}(r_j dr_j\wedge\alpha_j)\bigwedge_{h\notin I}\omega_h\right),
\end{align}
where we have used $\sum r_i^2=1$ in the second equality.
And, fixing $I,i\notin I,$
\begin{align}\label{eq:d3}
\notag\sum_{l\notin I\cup{\{i\}}}d_{3,l} S_{I\cup\{l\},i} &= -\sum_{l\notin I\cup{\{i\}}} 2^{k+1} (k+1)! r_i^2\alpha_i \wedge r_l dr_l\wedge \omega_l
\bigwedge_{j\in I}(r_j dr_j\wedge\alpha_j)\bigwedge_{h\notin I\cup\{i,l\}}\omega_h
\\&=\notag- 2^{k+1} (k+1)!\sum_{l\in I\cup\{i\}} r_l dr_l\wedge r_i^2\alpha_i
\bigwedge_{j\in I}(r_j dr_j\wedge\alpha_j)\bigwedge_{h\notin I\cup\{i\}}\omega_h
\\&=-2^{k+1} (k+1)!r_i^3 dr_i\wedge\alpha_i
\bigwedge_{j\in I}(r_j dr_j\wedge\alpha_j)\bigwedge_{h\notin I\cup\{i\}}\omega_h,
\end{align}
where the identity $\sum r_i dr_i=0$ was used for the second equality. The last passage follows from noting that except for the $l=i$ term, for all other $l\in I$ we will get a monomial with two $dr_l$ terms.

Summing equations \eqref{eq:d1},\eqref{eq:d2},\eqref{eq:d3}, over all possibilities for $I,$ and in \eqref{eq:d3} also for $i\notin I,$ we see that:
\begin{itemize}
\item \eqref{eq:d1} vanishes if $I=\emptyset.$ For $I\neq\emptyset$ the contribution of \eqref{eq:d1} cancels with the first term in the right hand side of \eqref{eq:d2} for the same $I.$
\item For a given $J\neq \emptyset,$ the sum of \eqref{eq:d3} over all pairs $(I,i)$ with $i\in J,~I=I\setminus\{i\}$ cancels with the second term of \eqref{eq:d2} with $I=J.$
\item For $I=\emptyset,$ the second term of \eqref{eq:d2} vanishes.
\end{itemize}
Thus, the only term which is left uncancelled is $\bigwedge\omega_i,$ coming from the first term of \eqref{eq:d2} with $I=\emptyset.$
Hence,
$$d\Phi=\sum_{I,i}dS_{I,i}= -\bigwedge\omega_i.$$
As needed.
\end{proof}
\begin{rmk}\label{rmk:phi_as_poly}
In what follows we will sometimes use forms on $S(E)$ which are defined similarly to $\Phi,$ but depend a subset of its arguments. For this reason it will be useful to extend $\Phi$ and similar expressions to multilinear functions in the variables $r_i,dr_i,\alpha_i,\omega_i,~i=1,\ldots,n,$ without imposing $\sum r_i^2=1,~\sum r_idr_i=0.$

Without putting these constraints the right hand side of \eqref{eq:d2} gets a correction of
\[2^k k!\left(1-\sum_{h\in [n]}r_h^2\right)
\left(\bigwedge_{j\in I}r_jdr_j\wedge{\alpha}_{j}\right)\wedge\bigwedge_{h\notin I}{\omega}_{j},\]
while the right hand side of \eqref{eq:d3} gets a correction of
\[2^{k+1}(k+1)!\left(\sum_{l\in [n]}r_l dr_l\right)\wedge
r_i^2{\alpha}_{i}\wedge
\left(\bigwedge_{j\in I} r_j dr_j \wedge{\alpha}_{j}\right)\wedge\bigwedge_{h\in [n]\setminus( I\cup\{i\})}{\omega}_h
\]
Summing the first correction over all $I,$ and adding the sum of the second correction over all $I,i\notin I,$ we obtain
\begin{align*}
&Z=\left(1-\sum_{h\in [n]}r_h^2\right) \\
&\qquad\wedge\sum_{m\geq 0} 2^m m!\sum_{|I|=m,I\subseteq [n]}
\left(\bigwedge_{j\in I}r_jdr_j\wedge{\alpha}_{j}\right)\wedge\bigwedge_{j\in [n]\setminus I}{\omega}_{j}\\
&\qquad+\left(\sum_{h\in [n]}r_h dr_h\right)\wedge
\sum_{i\in [n]\setminus\{h\}}r_i^2{\alpha}_{i}\wedge
\sum_{m\geq 0}
2^{(m+1)}(m+1)!\\
&\qquad\qquad\cdot\sum_{|I|=m, I\subseteq [n]\setminus\{i,h\}}\wedge
\left(\bigwedge_{j\in I} r_j dr_j \wedge{\alpha}_{j}\right)\wedge\bigwedge_{j\in [n]\setminus( I\cup\{i\})}{\omega}_j
\end{align*}
Therefore, without imposing $\sum r_i^2=1,~\sum r_idr_i=0$ we have
\[d\Phi= Z-\bigwedge_{i=1}^n\omega_i.\] Clearly $Z$ vanishes if we do make these assumptions.
\end{rmk}

\begin{construction}\label{nn:spherization}
Suppose $S_1,\ldots,S_l\to M$ are piecewise smooth $S^1$ bundles over a piecewise smooth orbifold with corners.
Denote by $S(S_1,\ldots,S_l)\to M$ the $2l-1-$sphere bundle on $M,$ whose fibers are \[S(S_1,\ldots,S_l)_x
 =\{(r_1,P_1,r_2,P_2,\ldots,r_l,P_l)|P_i\in (S_i)_x,~r_i\geq 0,~\sum r_i^2=1\}/\sim,
\]
where $\sim$ is the equivalence relation generated by \[(r_1,P_1,\ldots,0,P_i,\ldots,r_l,P_l)\sim(r_1,P_1,\ldots,0,P'_i,\ldots,r_l,P_l),\]
and with the natural topology.
\end{construction}

\section{Symmetric Jenkins-Strebel stratification}\label{sec:strebel}
In the remainder of the article all open spin surfaces we will encounter, twisted or not, will have a lifting. Similarly, we will encounter several types of graphs, the dual graphs we have defined above, ribbon graphs and nodal graphs. These graphs will also be classified as open or closed and will sometimes carry spin structures, twisted or not. All the open spin graphs we shall meet will have a lifting. For this reason we will sometimes slightly abuse notations and omit the suffix 'with a lifting' from the terminology. We will also usually omit the addition 'twisted'. It will be clear from the context if we mean a closed or open object, twisted or not etc.
\subsection{JS stratification for the closed moduli}
\subsubsection{JS differential and the induced graph}
In this subsection we briefly describe the stratification of moduli of closed stable curves following \cite{Kont,Zvonk,Looij}.

%
Let $\Sigma$ be a nodal Riemann surface with $2g-2+n\geq0.$ A meromorphic section $\gamma$ of the tensor square of the cotangent bundle defined on each component of the normalization of $\Sigma$ can be written in a local coordinate $z$ as $f(z)dz^2.$ If $\gamma$ has a double pole at $w\in\Sigma,$ the \emph{residue} of $\gamma$ at $w$ is the coefficient of $\frac{dz^2}{(z-w)^2},$ in the expansion of $\gamma$ around $w.$ The residue is independent of the choice of the local coordinate.
A \emph{quadratic differential} $\gamma$ is
such a section which has at most double poles, all the poles are located either at the marked points or at the nodes, and for any node, the residues of $\gamma$ at its two branches are the same.

Let $\gamma$ be a quadratic differential, and $w\in\Sigma$ a point which is neither a zero nor a pole. In a neighborhood $U$ we can take its unique, up to sign, square root $\sqrt{\gamma}.$ This is a $1-$form, hence can be integrated along a path. This defines a map
\[
g:U\to\C,~~g(z)=\int_w^z\sqrt\gamma,
\]
where the integral is taken along any path in $U.$

A \emph{horizontal trajectory} is the preimage of $\R\subset\C,$ and it is a smooth path containing $w$ in its interior.
It turns out that the notion of horizontal trajectories can be defined also in the case where $w$ is a zero of order $d\geq -1,$ where as usual a zero of order $-m$ is a pole of order $m.$ In this case there are exactly $d+2$ horizontal rays leaving $w.$
When $w$ is a pole of order $2,$
if its residue is $-\left(\frac{p}{2\pi}\right)^2,$ for some $p\in \R_+,$ there is a family of nonintersecting horizontal trajectories surrounding it, whose union is a topological open disk, punctured at $w.$
Moreover, with respect to the metric defined by $|\sqrt{\gamma}|,$ the perimeter of each of these trajectories is $p.$
\begin{ex}\label{ex:2_pts}
Let $\Sigma$ be the Riemann sphere. For all $p>0,$
\[\gamma_p = -\left(\frac{p}{2\pi}\right)^2\left(\frac{dz}{z}\right)^2,\]
is a quadratic differential, whose only poles are in $0,\infty$ and whose horizontal lines are the sets $|z|=r,$ for $r>0,$ whose lengths are indeed $p.$
Their union is an open punctured disk.
It should be noted that actually this is the only quadratic differential on the sphere, up to scaling, which invariant under the reflection in the equator whose only poles are at $0$ and $\infty.$
\end{ex}

\begin{definition}\label{def:stable_JS}
Let $(\Sigma,z_1,\ldots,z_n,z_{n+1},\ldots,z_{n+n_0})$ be a marked genus $g$ nodal Riemann surface with $2g-2+n\geq0,$ where the subscript of $z_i$ indicates its marking. 
Let $p_1,\ldots,p_n$ be positive reals, and $p_i=0,$ for $i>n.$
A \emph{marked component} is a smooth component of the curve with at least one marked point $z_i,i\in[n].$ The other components are called \emph{unmarked}.
A \emph{Jenkins-Strebel differential}, or a \emph{JS-}differential for shortness, is a quadratic differential $\gamma$ such that
\begin{enumerate}
\item
$\gamma$ is holomorphic outside of special points. In nodes it has at most simple poles and in the $i^{th}$ marked point it has a double pole with residue $-(p_i/2\pi)^2.$ In particular, if $p_i=0$ there is at most a simple pole at that point.
\item
$\gamma$ vanishes identically on unmarked components.
\item
Let $\Sigma'$ be any marked component of $\Sigma.$ When $p_i\neq 0,$ if $D_i$ is the punctured disk which is the union of horizontal trajectories surrounding $z_i\in\Sigma',$ then
\[
\bigcup_{i\in[n]} \overline{D}_i = \Sigma'.
\]
\end{enumerate}
\end{definition}
The following theorem was proved in \cite{Streb} for the smooth case, the nodal case was treated in \cite{Looij,Zvonk}.
\begin{thm}\label{thm:JS}
Given a stable marked surface $(\Sigma,z_1,\ldots,z_{n+n_0}),$ with $n>0$ and $\pp=(p_1,\ldots,p_{n+n_0})\in\R_+^n\times(0,\ldots,0)$ as above, JS differential exists and is unique.
\end{thm}

Given $(\Sigma,\mathbf{z}),\pp$ as above, define the decorated surface $\widetilde\Sigma,$ and the map $K_{n_0}:\Sigma\to\widetilde\Sigma$ as follows.
$\widetilde\Sigma$ is obtained from $\Sigma$ by contracting any unmarked component to a point, and decorating any such point by its \emph{genus defect} and \emph{marking defect}. The genus defect is the genus of the preimage of the point in $\Sigma,$ and if that preimage is a single point it is defined to be $0.$ The marking defect is the set of marked points in this preimage, which is labeled by a subset of $[n+n_0]\setminus[n].$
We should stress that $\gamma$ needs not to vanish on a preimage of a node in the normalization, but it can have at most a simple pole there. Thus, from the discussion about horizontal trajectories, each node or unmarked component and in particular any point $z_i$ for $i>n$ must be mapped to a point which touches at least one horizontal trajectory. Note also that an unmarked component always touches a node (unless $n=0$ and then the whole surface is unmarked).

The JS differential $\gamma$ induces a metric graph on $\widetilde\Sigma$ whose vertices are zeroes of order $d\geq-1$\footnote{We consider a simple pole as a zero of order $-1,$ and a point which is neither a zero nor a pole to be a zero of order $0.$} of $\gamma,$ including the images of unmarked components, and whose  edges are the horizontal trajectories, with their intrinsic length. These embedded graphs can be fully described.
\begin{definition}\label{def:stable_closed_ribbon}
A $(g,(n,n_0))-$\emph{stable closed ribbon graph} is a graph $G=(V,H,s_0,s_1,g,f),$ where
\begin{enumerate}
\item
$V$ is the set of vertices, $H$ is the set of half edges.
\item
$s_0$ is a permutation of the half edges emanating from each vertex.
\item
$s_1$ is a fixed point free involution of $H.$
\item
A map $g:V\to\Z_{\geq 0},$ called the \emph{genus defect}.
\item
A map $f:[n+n_0]\setminus[n]\to V.$ 
\end{enumerate}
The \emph{faces} of the graph are $s_2-$equivalence class of half edges, where $s_2=s_0^{-1}s_1.$ We write $F=H/s_2.$ The edges are $E=H/s_1.$
The \emph{genus} of $G$ can be defined as follows. Glue disks along the faces to obtain a surface $\widetilde\Sigma.$ The genus of $G$ is the (arithmetic) genus of $\widetilde\Sigma$ plus the sum of genus defects in vertices. The \emph{marking defect} of a vertex $v$ is defined as $f^{-1}(v).$
We require
\begin{enumerate}
\item
For a vertex $v$ of degree $1$ or of degree $2,$ but such that the assigned permutation is a transposition,
\[
g(v)+|f^{-1}(v)|\geq 1.
\]
\item
The genus of the graph is $g.$
\item
The number of faces is $n.$
\end{enumerate}
A \emph{stable metric ribbon graph} is a stable ribbon graph together with a metric
\[
\ell:E\to\R_+.
\]
We usually write $\ell_e$ instead of $\ell(e).$

A graph is \emph{smooth} if all the vertices' permutations $s_0$ are cyclic, all genus defects are $0$ and all marking defects are of size at most $1.$ The ribbon graph is connected if the underlying graph is.
We define isomorphisms and automorphisms in the expected way. Write $\text{Aut}(G)$ for the automorphism group of $G.$
\end{definition}
Note that case $(a)$ above occurs when $v$ is either the image of a contracted unmarked component, or the image of one of the points $p_i,~i>n.$
\begin{rmk}\label{rmk:measure_dist}
To a stable metric ribbon graph one can associate in a natural way a decorated metric space made of a disjoint union of closed intervals, one for each $e\in E,$ modulo the identification of endpoints dictated by the graph structure. The vertices, which are the equivalence classes of endpoints of intervals are endowed with genus and marking defects, and the closed interval which corresponds to the edge $e$ is associated a metric structure which makes it isometric to the interval $[0,\ell_e]\subset\R.$
The associated decorated metric space is unique up to the expected notion of isomorphism.
Stable metric ribbon graph which arise from a JS differential (we will see in Theorem \ref{thm:minimal_compacs_closed} below that all stable metric ribbon graphs arise this way) are endowed with this additional structure of isometries between the embedded edges and intervals of $\R.$ This will be used below, when we give coordinates to the combinatorial $S^1$-bundles. For more details we refer the reader to \cite{Zvonk}.
\end{rmk}

\begin{nn}\label{nn:Aut}
Throughout this article, given a ribbon graph, possibly with extra structure such as a graded ribbon graph, or a nodal graph, which will be defined later, we shall write $[h]$ for the class of the half edge or the edge $h$ under the action of the automorphism group. We similarly define $[A]$ for a subset of edges or half edges.
\end{nn}
\begin{rmk}\label{rmk:normalization-closed.differential}
If $\NNN:\NNN(\Sigma)\to\Sigma$ is the normalization of $\Sigma,$ and $\gamma$ is the JS differential on $\Sigma$ with prescribed perimeters,
then $\NNN^*\gamma$ is a JS differential, hence the unique JS differential, on $\NNN(\Sigma),$ with the same perimeters, and such that marked points which are preimages of nodes have $0-$ perimeter.
\end{rmk}

\subsubsection{Combinatorial moduli}
For a closed stable ribbon graph $G,$ write $\CM_G$ for the set of all metrics on $G,$ write $\CM_G(\pp)$ for the set of all such metrics where the $i^{th}$
face has perimeter $p_i.$ 
$\CM_G\simeq\R_+^{E(G)}/\text{Aut}(G)$ canonically, and this identification endows it with a smooth structure.

For $e\in E(G),$ the edge between vertices $v_1,v_2,$ define the graph $\partial_eG,$ the edge contraction, as follows. Write $h_1,h_2$ for the two half edges of $e.~V(\partial_eG)=V(G)\setminus\{v_1,v_2\}\cup\{v_1v_2\},H(\partial_e G)=H(G)\setminus \{h_1,h_2\}.~s'_1,g',f'$ are just $s_1,g,f$ when restricted to vertices and half edges of $G.$
For the new vertex $v=v_1v_2,~f'(v)=f(v_1)\cup f(v_2),g'(v)=g(v_1)+g(v_2)$ whenever $v_1\neq v_2,$ otherwise it is $g(v_1)+\delta,$ where $\delta=1$ if $h_1,h_2$ belong to different $s_0-$cycles, or else $0.$
For any half edge $h,~h/s_1\neq e,$ define $s_2'(h)$ to be the first half edge among $s_2(h),s_2^2(h),\ldots,$ which is not a half edge of $e.$ We then put $s'_0=s_1' (s'_2)^{-1}.$

Edge contractions commute with each other, and allow us to define a cell complex $\oCM_G=\coprod_{G'}\CM_{G'},$ where the union is over all graphs obtained from $G$ by edge contractions, and we glue the cell $\CM_{G'}$ of $G'=\partial_{e_1,\ldots,e_r}G$ to the cell $\CM_G$ along $\ell_{e_1}=\ldots=\ell_{e_r}=0.$ 
We similarly define $\oCM_G(\pp).$

Write $\CM^{\text{comb}}_{g,(n,n_0)}=\coprod\CM_{G},$ where the union is taken over smooth closed $(g,(n,n_0))$ ribbon graphs.
Write $\oCM^{\text{comb}}_{g,(n,n_0)}=\coprod\oCM_{G}/\sim=\coprod\CM_{G},$ where the union is taken over all closed stable $(g,(n,n_0))$ ribbon graphs, and $\sim$ is induced by edge contractions.
Define $\oCM^{\text{comb}}_{g,(n,n_0)}(\pp),\CM^{\text{comb}}_{g,(n,n_0)}(\pp)$ by constraining the perimeters to be $p_i.$
In all cases we define the cell attachment using edge contractions, and the resulting spaces are piecewise smooth Hausdorff orbifolds, see \cite{Looij,Zvonk} for details.



Set $\comb=\comb_{n_0}$ as the canonical maps
\[
\comb:\oCM_{g,n+n_0}\times\R_+^n\to\oCM^{\text{comb}}_{g,(n,n_0)},~\comb_{\pp}:\oCM_{g,n+n_0}\to\oCM^{\text{comb}}_{g,(n,n_0)}(\pp),
\]
which sends a stable curve and a set of perimeters to the corresponding graph.

We have, \cite{Kont,Looij,Zvonk},
\begin{thm}\label{thm:minimal_compacs_closed}
Suppose $n>0.$ The maps $\comb,\comb_{\pp}$ are continuous surjections of topological orbifolds. $\comb_{\pp}$ takes the fundamental class to a fundamental class.
Moreover, the cell complex topology described above is the finest topology with respect to which $\comb$ is continuous.
The maps are isomorphisms onto their images when restricted to $\CM_{g,n+n_0}\times\R_+^n,\CM_{g,n+n_0}.$

More generally, suppose $\Gamma$ is a closed dual graph with the property that any vertex without a tail marked by $[n]$ is of genus $0,$ and has exactly $3$ half edges, and any two such vertices are not adjacent. Then, with the same proofs, $\comb,\comb_{\pp}$ restricted to $\CM_\Gamma\times\R_+^n,\CM_\Gamma$ are isomorphisms onto their image.
\end{thm}

\subsubsection{Tautological line bundles and associated forms}
\begin{definition}\label{def:CF_i}
Suppose $p_i>0.$ Define the space $$\CF_i(\pp)\to\oCM_{g,n}^{\text{comb}}(\pp)$$ as the collection of pairs $(G,\ell,q),$ where $(G,\ell)\in\oCM^{\text{comb}}_{g,n}(\pp)$ and $q$ is a boundary point of the $i^{th}$ face. These spaces, glue together to the bundle $\CF_i\to\oCM^{\text{comb}}_{g,n}.$
Define $\phi_j$ to be the distance from $q$ to the $j^{th}$ vertex, taken along the arc from $q$ in the counterclockwise direction, so that $0<\phi_1<\phi_2<\ldots<\phi_N<p_i,$ where $N$ is the number of edges in the $i^{th}$ face, counted with multiplicities, and the distances are measured using the identifications of the edges with subintervals of $\R,$ see Remark \ref{rmk:measure_dist}.
Write $\ell_j=\phi_{j+1}-\phi_j.$
Orient the fibers with the clockwise orientation.

Define the following $1-$form and $2-$form on each cell of $\oCM_{g,(n,n_0)}^{\text{comb}}(\pp)$
\begin{equation}\label{eq:alpha_i_omega_i}
\alpha_i = \sum_{j=1}^N\frac{\ell_j}{p_i}d\left(\frac{\phi_j}{p_i}\right),\quad
\omega_i = -d\alpha_i=\sum_{1\leq a<b\leq N}d\left(\frac{\ell_a}{p_i}\right)\wedge d\left(\frac{\ell_b}{p_i}\right).
\end{equation}
\end{definition}
Later we will integrate forms which are made out of $\alpha_i,\omega_i,$ and we will perform Laplace transform over $\pp.$ For this reason it will be convenient to define the scaled versions of $\alpha_i,\omega_i$ which do not contain $p_i$ in their denominators. We thus put\[ \bar{\alpha}_i=p_i^2\alpha_i,\quad\bar{\omega}_i=p_i^2\omega_i,\quad\bar{\omega}=\sum_i\bar{\omega}_i.\]

The bundles $\CF_i$ carry natural piecewise smooth structures. Moreover,
\cite{Kont} says (see also Theorem $5$ in \cite{Zvonk})
\begin{thm}\label{thm:zvonk}
\begin{enumerate}
\item For $i\in[n],~\comb^*\CF_i\simeq S^1(\CL_i)$ canonically.
\item $\alpha_i,\omega_i$ are piecewise smooth angular $1-$form and Euler $2-$form for $\CF_i.$
\end{enumerate}
%
%
\end{thm}
\begin{rmk}
In \cite{Kont} $\CF_i$ was given the opposite orientation and the equivalence was hence to the bundle $S^1(\CL_i^*),$ which is canonically $S^1(\CL_i)$ with the opposite orientation.
\end{rmk}
Thus, combined with Theorem \ref{thm:minimal_compacs_closed} we see that all descendents may be calculated combinatorially on $\oCM^{\text{comb}}_{g,n}.$
In fact, all descendents can be calculated as integrals over the highest dimensional cells of $\oCM_{g,n}^{\text{comb}}.$ These are parameterized by trivalent ribbon graphs.

\subsection{JS Stratification for the open moduli}
\subsubsection{Symmetric JS differentials}
Motivated by Definition \ref{def:stable_JS} and Example \ref{ex:2_pts} we define
\begin{definition}\label{def:sym_JS}
Let $(\Sigma,\{z_i\}_{i\in\I\cup\PP},\{x_i\}_{i\in\B})$ be a stable open marked Riemann surface, $\pp=(p_i)_{i\in\I\cup\PP}\in\R_+^{\I}\times(0,\ldots,0).$ 
A \emph{symmetric JS differential} on $\Sigma$ is the restriction to $\Sigma$ of the unique JS differential of $D(\Sigma)$ whose residues at $z_i,\bar{z}_i$ are $-(p_i/2\pi)^2,$ which are $0$ for $i\in \PP$.
We extend the definition to the case $g=0,\I=[1],\PP=\B=\emptyset,$ where the differential is defined to be the restriction of the section $\gamma_{p_1}$ of Example \ref{ex:2_pts}.
\end{definition}
The existence and uniqueness follow from Theorem \ref{thm:JS} and the discussion in Example \ref{ex:2_pts}.

As before, the symmetric JS differential defines a cell decomposition of $D(\Sigma),$ in the smooth case, and in general a metric graph embedded in $\widetilde{D(\Sigma)},$ the surface obtained from $D(\Sigma)$ by contracting components with no $z_i,\bar{z}_i,~i\in\I$ ,whose complement is a disjoint union of disks. Note that $\widetilde{D(\Sigma)}$ inherits the conjugation from $D(\Sigma),$ which we also denote by $\varrho.$ The uniqueness forces the decomposition to be $\varrho-$invariant.
\begin{lemma}\label{lem:sym_JS}
The $\varrho-$fixed locus of $\widetilde{D(\Sigma)}$ is a union of (possibly closed) horizontal trajectories and isolated vertices.
Any $\varrho-$fixed point is a zero the differential of an even order, possibly $0.$
\end{lemma}
\begin{proof}
The case $g=0,\I=[1],\PP=\B=\emptyset$ follows from the discussion in Example \ref{ex:2_pts}. In other cases, take an arbitrary point in $\widetilde{D(\Sigma)}^\varrho.$
It cannot belong to the disk cell of any $z_j,$ since otherwise it would have belonged to the cell of $\bar{z}_j$ as well. Thus, $\widetilde{D(\Sigma)}^\varrho$ is contained in the one-skeleton of the decomposition.
Consider $p\in \widetilde{D(\Sigma)}^\varrho.$ If $p$ is an isolated vertex in the $\varrho-$fixed locus, then by connectivity it must be incidenet to some non $\varrho-$fixed horizontal trajectory which, without loss of generality, lies in the image of $\Sigma^o$ in $\widetilde{D(\Sigma)}.$ Suppose it touches $r$ such trajectories. Then it also touches their $\varrho-$conjugate trajectories, which lie in the image of $\overline{\Sigma}^o$ in $\widetilde{D(\Sigma)}.$ Thus, $2r$ horizontal trajectories emanate from $p,$ for $r\geq 1,$ hence $p$ is a a zero of order $2r-2\geq 0.$ The second case is that $p$ is not isolated, so it lies on the image of $\partial\Sigma$ in $\widetilde{D(\Sigma)},$ which, as explained, is contained in the $1-$skeleton. In this case, at least two horizontal trajectories which are contained in the image of $\partial\Sigma$ emanate out of $p,$ one to its left and one to its right. In addition, there are also $r\geq 0$ such trajectories in the image of $\Sigma^o,$ and because of symmetry there are also $r$ such trajectories in the image of $\overline{\Sigma}^o.$ In total, there are $2r+2$ horizontal trajectories emanating from $p,$ which means that it is a zero of order $2r\geq 0.$
\end{proof}

Lemma \ref{lem:sym_JS} has the following corollary
\begin{cor}\label{cor:forgetting_a_bdry_point_differential}
Suppose $\Sigma,\pp$ are as above, and $\gamma$ is the associated symmetric JS differential. Assume that for some $i\in\B,$ forgetting $x_i$ makes no component of $\Sigma$ unstable. Denote by $\Sigma'$ the resulting surface, and let $\iota:\Sigma'\to\Sigma$ be the natural map between the surfaces.
Then if $\gamma,\gamma'$ are the unique JS differentials for $\Sigma,\Sigma'$ respectively, with the prescribed perimeters, then
\[
\gamma'=\iota^*\gamma.
\]
\end{cor}
Indeed, both $\gamma,\gamma'$ are JS differentials on $\Sigma',$ since there is no pole in $x_i.$ Hence they must be equal.

Remark \ref{rmk:normalization-closed.differential} has the following consequence
\begin{cor}\label{cor:normalization-open.differential}
If $\NNN:\NNN(\Sigma)\to\Sigma$ is the normalization of $\Sigma,$ and $\gamma$ is the JS differential on $\Sigma$ with prescribed perimeters,
then $\NNN^*\gamma$ is the unique JS differential, on $\NNN(\Sigma),$ with the same perimeters, and such that marked points which are preimages of nodes have $0-$ perimeter.
\end{cor}

\begin{rmk}
Although throughout the article we will be mainly interested in internal markings with positive perimeters, markings of perimeter zero occur naturally when one considers normalizations, see Proposition \ref{prop:for_base_JS}. In the open intersection theory the normalizations are crucial for the definition of intersection numbers, Definition \ref{def:special_canonical}, and therefore considering markings with zero perimeters is unavoidable. In addition, since boundary markings carry no descendents, we to not lose from fixing their perimeters to be zero, and it simplifies calculations. For these reasons throughout this section we shall allow marked points to have perimeter $0,$ at the cost of making the notations somehow more cumbersome.
\end{rmk}

\subsubsection{Open Ribbon graphs}
\begin{nn}\label{nn:isotopy_classes}
Let $I,B$ be finite sets.
Denote the set of isotopy types of open connected genus $g$ smooth oriented marked surfaces, with $I$ being the set of internal marked points and $B$ being the set of boundary marked points by $\D(g,I,B)$.
Write $\D(g,I)$ for the set of isotopy types of closed connected genus $g$ smooth oriented marked surfaces, which is just a singleton.
\end{nn}
\begin{definition}\label{def:open_stable_ribbon}
An \emph{open ribbon graph} is a tuple
\[G=(V=V^I\cup V^B,H=H^I\cup H^B,s_0,s_1,f=f^{I}\cup f^{B}\cup f^{\PPP}, g,d)\]
where
\begin{enumerate}
\item
$V^I$ is the set of internal vertices, $V^B$ the set of boundary vertices.
\item
$H^B$ is the set of boundary half edges, $H^I$ is the set of internal half edges;
$s_1$ is a fixed point free involution on $H$ whose equivalence classes are the edges, $E.~E^B$ is the set of edges which contain a boundary half edge.
\item
A permutation $s_0$ assigned to each vertex, should be thought of as a cyclic order of the half edges issuing each vertex.
We write $s_0$ also for the product of all these permutations.

We denote by $\widetilde V$ the set of cycles of $s_0.$
Write $\widetilde V^I$ for cycles which do not contain boundary half edges. Set $\widetilde V^B=\widetilde V\setminus \widetilde V^I.$
Put by $N:\widetilde V\to V$ the map which takes a cycle to the vertex which contains its half edges, and let $N^{\PPP},N^{B}$ be the restrictions to $\widetilde V^I,\widetilde V^B,$ respectively. 
\item
A map $f^{B}:\B\to V^B,$ where $\B$ is a finite set.
\item
A map $f^{\PPP}:\PP\to V,$ where $\PP$ is a finite set.
\item
An injection $f^{I}:\I\hookrightarrow H/s_2,$ where $s_2:=s_0^{-1}s_1.$
\item
A map $g:V\to\Z_{\geq 0},$ called the \emph{genus defect}.
\item
For any $v\in V^B,$ an element \[d(v)\in \D(g(v),(f^{\PPP})^{-1}(v)\cup(N^{\PPP})^{-1}(v),(f^{B})^{-1}(v)\cup(N^{B})^{-1}(v)).\]
For any $v\in V^I,$ the unique element $d(v)\in \D(g(v),(f^{\PPP})^{-1}(v)\cup(N^{\PPP})^{-1}(v)).~d$ is called the \emph{topological defect} of $v.$
\end{enumerate}
Write $\text{deg}(v)$ for the degree of the vertex $v.$
A \emph{closed contracted component} is a vertex $v\in V^I$ with $$2g(v)+|(f^{\PPP})^{-1}(v)|+|N^{-1}(v)|>2.$$ Denote their collection by $\SC^C(G).$
An \emph{open contracted component} is a vertex $v\in V^B$ with $$2\left(g(v)+|(f^{\PPP})^{-1}(v)|+|(N^{\PPP})^{-1}(v)|\right)+|(f^{B})^{-1}(v)|+|(N^{B})^{-1}(v)|>2.$$ Denote their collection by $\SC^O(G).$

We have the following requirements.
\begin{enumerate}
\item
Any half edge appears in the permutation $s_0$ of exactly one vertex.

We define a graph whose vertices are the elements of $V$ and whose half edges are the elements of $H.$ A half edge is connected to a vertex if and only if it appears in the vertex's permutation $s_0.$
\item
$N(\widetilde V^B)\subseteq V^B.$
\item
If $h\in H^B,$ then $s_1h\notin H^B.$
\item
$s_2$ preserves the partition $H=H^I\cup H^B.$
The image of $f^{I}$ is exactly $H^I/s_2.$
\item
For $v\in V^I,$ if the degree of $\text{deg}(v)=1,$ or $\text{deg}(v)=2$ but $|N^{-1}(v)|=1,$ then $|(f^{\PPP})^{-1}(v)|+g(v)\geq 1.$
%
\item
For $v\in V^B,$
if $v$ has at least one boundary edge and $\text{deg}(v)=2$ then $|(f^{\PPP})^{-1}(v)|+|(f^{B})^{-1}(v)|+g(v)\geq 1.$
\item
Any vertex of degree $0$ is a \emph{contracted component}.
%
\end{enumerate}
We call the elements of $H^B/s_2$ \emph{boundary components}, and the elements of $F=H^I/s_2$ are called \emph{faces}. $b(G)=|H^B/s_2|$ is the number of boundary components. The \emph{marking defect} of $v\in V$ is defined as $(f^\PPP)^{-1}(v)\cup(f^\B)^{-1}(v).$
The sets $\I,\PP,\B$ are called the sets of internal markings, internal markings of perimeter $0,$ and boundary markings respectively. $\B$ is also denoted by $B(G),$ define $I(G),\PPP(G)$ similarly.
An \emph{internal node} is either a contracted component with at least one edge and no boundary edges, or an internal vertex whose assigned permutation is not transitive.
%
A boundary vertex $v$ without boundary half edges, with an empty marking defect and such that $g(v)=0,|N^{-1}(v)|=1$ is called a \emph{contracted boundary}. We denote their collection by $\text{CB}(G).$
A boundary vertex $v$ which is either a contracted component with at least one boundary edge, 
or that whose assigned permutation is not transitive is called a \emph{boundary node}.
A \emph{boundary marked point} is an image of $f^{B}$ which is not a node. An \emph{internal marked point of perimeter $0$} is an image of $f^{\PPP}$ which is not a node.
A \emph{boundary half node} is a $(N^{B})^{-1}-$preimage of a node. Denote their collection by $\HN(G).$
A vertex which is either a node or a contracted component, or the $f-$image of a unique element in $\PP\cup\B$ is called a \emph{special point}.


%

We write $i(h)=h/s_2,$ and $H_i =\{h\in H | i(h)=i\}.$
%

An \emph{open metric ribbon graph} is an open ribbon graph together with a positive metric $\ell:E\to\R_+.$ We sometimes write $\ell_h,h\in H$ instead of $\ell_{h/s_1}.$

\emph{Markings} of an open ribbon graph are markings,
\[
\mmm^I: \I\cup\PP\to \Z,~\mmm^B:\B\to\Z,
\]
such that $\mmm^I(\PP)=0,\mmm^I(\I)\subset\Z_{\neq0}.$
A graph together with a marking is called a \emph{marked graph}.

An isomorphism of marked graphs, and an automorphism of a marked graph are the expected notions. $\text{Aut}(G)$ denotes the group of automorphisms of $G.$ A metric is \emph{generic} if $(G,\ell)$ has no automorphisms.

A ribbon graph is said to be closed if $V^B=0,$ it is said to be connected if the underlying graph is connected.
\end{definition}
The maps $f^B, ~f^\PPP$ should be thought of as the associations of the boundary marked points and the internal marked points of perimeter $0$ respectively, to the vertices of the graph formed by the symmetric JS differential. Requirements $(e),~(f)$ in this definition are the open counterparts of Requirement $(a)$ of Definition \ref{def:stable_closed_ribbon}.
Note that a half edge $h$ is canonically oriented away from its base point $h/s_0.$
Throughout the paper we identify boundary marked points, which are vertices, with their (unique) preimages in $B(G)=\B.$

\begin{rmk}
Here, unlike the closed case, the genus defect is not enough to classify surfaces with contracted components. In particular, there are several topologies for a given genus, as mentioned in Remark \ref{rmk:top_type_open}, and the set of topologies grows as we add boundary marked points, which may be divided between different boundary components. 
\end{rmk}
Figure \ref{fig:ribbon4_17} shows some examples of ribbon graphs.
\begin{figure}
\centering
\includegraphics[scale=.4]{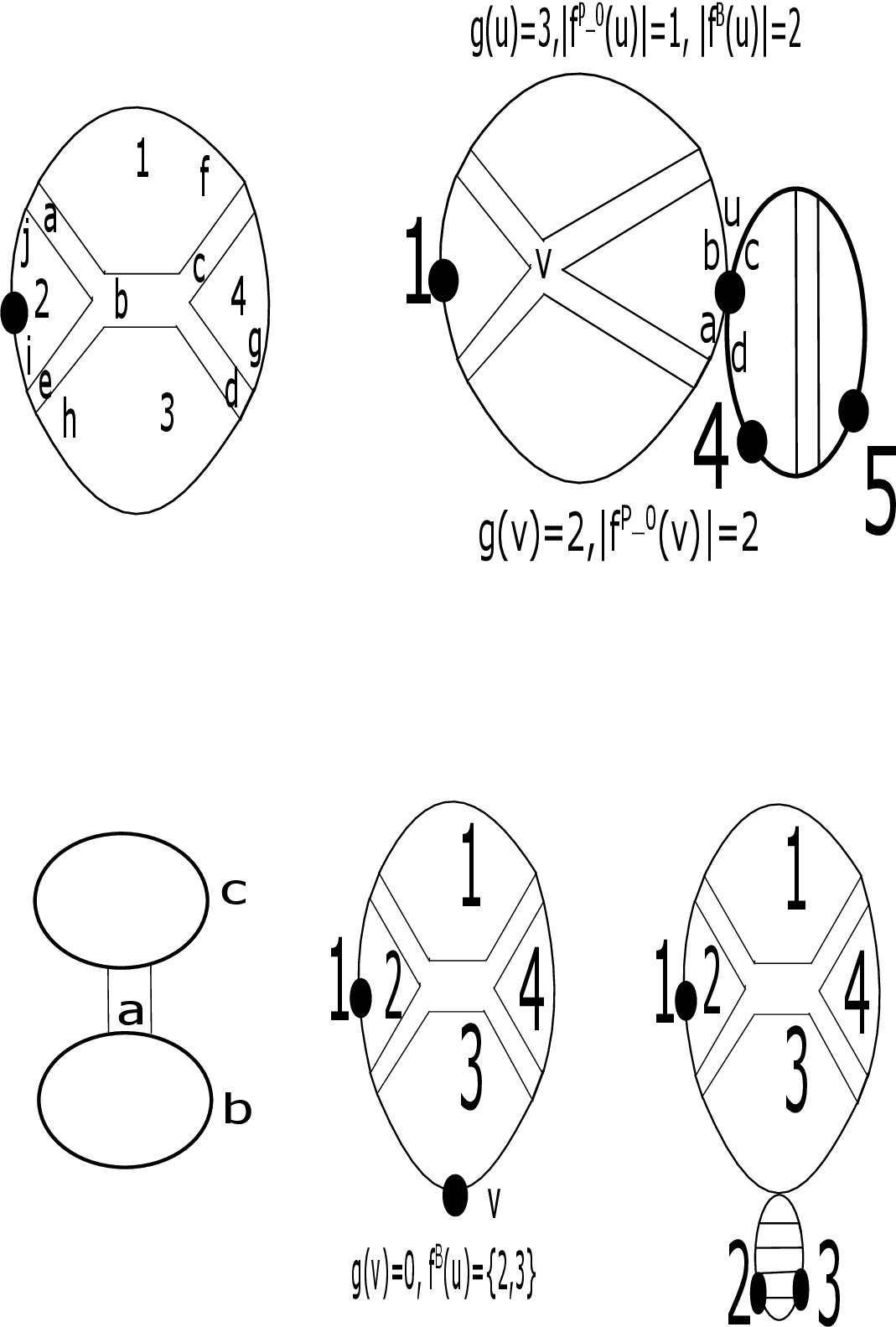}
\caption{Examples of ribbon graphs. Internal edges are drawn as strips. The top left corner shows a ribbon graph with one boundary markings and four internal markings (the name of a half edge appears next to the vertex from which it emanates). Its underlying surface is a disk, and the boundary edges are $s_1f,s_1g,s_1h,s_1i,s_1j.$ The cyclic orders in the internal vertices are $s_1a,s_1b,e$ and $s_1b,s_1d,c.$ Face $1,$ for example, is the $s_2-$cyclic order $a,b,c,f.$
In the low right corner a ribbon graph on a cylinder is drawn. It has one face, the $s_2-$cycle $a,c,s_1a,b$ and two boundary components, each made of a single boundary edge $s_1b,$ and $s_1c.$
The ribbon graph in the top right corner has one boundary node $u$, which is also an open contracted component and an internal node $v$ which is also a contracted component. The permutation of half edges at $u$ is $(ab)(cd).$ The contracted component is open, of genus defect $3,$ has an internal marking of perimeter $0,$ and four special boundary points, the markings $1,2$ and the half nodes $(ab),(cd).$ The topological defect can be any topology which corresponds to doubled genus $3,$ one internal marking and four boundary markings. $v$ has genus defect $2$ and two perimeter $0$ internal markings.
The middle low picture has an open contracted component at $v,$ it is a contracted disk with two boundary markings $2,3$ and a boundary half node and no special internal points.
Contracted component which are disks with three boundary markings and no internal markings will play an important role in what follows. We shall therefore draw such components as disks cut by parallel lines, as in the bottom right picture.
}
\label{fig:ribbon4_17}
\end{figure}

\begin{nn}\label{nn:Sigma_G}
By gluing disks along the faces, any open ribbon graph gives rise to a topological open oriented surface $\Sigma_G.$ This surface is a union of smooth surfaces, identified in a finite number of points. One can easily define its double, $D(\Sigma_G)=(\Sigma_G)_\C,$ as in the non topological case.
\end{nn}
\begin{definition}\label{def:genus_of_ribbon}
The genus of the open graph $G$ is defined by
\[
g(G):=g((\Sigma_G)_\C)+\sum_{v\in V^B}g(v)+2\sum_{v\in V^I}g(v).
\]
The graph is \emph{stable} if $2g-2+|\B|+2(|\I|+|\PP|)>0.$
\end{definition}

For a stable open surface $(\Sigma,\{z_i\}_{i\in\I\cup\PP},\{x_i\}_{i\in\B}),$ define the marked components to be components with at least one $z_i,~i\in\I.$
The other components are unmarked. Define the decorated surface $\widetilde\Sigma=K_{\B,\PP}(\Sigma),$ and the map $K_{\B,\PP}:\Sigma\to\widetilde\Sigma$ to be the surface obtained by contracting unmarked components to points, and $K_{\B,\PP}$ is the quotient map. We decorate any point $p$ in $\widetilde\Sigma$ by its genus defect, marking defect, and the topological defect which can be defined by the genus, boundary markings and topological type of the surface obtained by smoothing the nodes in $K_{\B,\PP}^{-1}(p).$
\begin{rmk}
This definition agrees with the one given for closed surfaces, in the sense that one can also define the doubling $D$ of $\widetilde\Sigma$ in a natural way, and then
$D(\widetilde\Sigma)\simeq\widetilde{D(\Sigma)}.$
\end{rmk}

\begin{definition}\label{def:ghost_effective_trivalent}
A \emph{ghost} is a ribbon graph without half edges. 
A \emph{smooth open ribbon graph} is a stable open ribbon graph such that none of its connected components contains a node or a contracted boundary.

A stable ribbon graph, open or closed, is \emph{effective} if
\begin{enumerate}
\item Any genus defect is $0.$
\item There are no internal nodes
\item Contracted components or ghost components $v$ must have \[(N^{\PPP})^{-1}(v)=\emptyset,~|(N^{B})^{-1}(v)|+|(f^{B})^{-1}(v)|=3.\]
\end{enumerate}
The graph is \emph{trivalent} if
\begin{enumerate}
\item it is effective,
\item $\PP=\emptyset,$
\item it has no contracted boundaries,
\item all vertices which are not special boundary points are trivalent,
\item and for every special boundary point all the $s_0-$cycles are of length $2.$
 \end{enumerate}

A boundary marked point or a boundary half node in a trivalent graph $G$ which is not a ghost is said to \emph{belong to} a face $i$ if its unique internal half edge belongs to that face.

\end{definition}
In Figure \ref{fig:ribbon4_17} the diagrams on the left represent smooth graphs, and all but the top right are effective.
\begin{rmk}\label{rmk:ghost}
The only non zero open intersection number which does not involve internal markings is the genus $0$ intersection number with three boundary markings $\langle\sigma^3\rangle^o_0.$ The graph which corresponds to this picture is precisely the trivalent ghost.
\end{rmk}

The following proposition is a consequence of Lemma \ref{lem:sym_JS}, and the closed theory, the proof is in the appendix.
\begin{prop}\label{prop:streb_gives_graph}
Let $\Sigma$ be a stable open marked Riemann surface. The unique symmetric JS differential of $\Sigma$ defines a unique metric graph $(G,\ell)$ embedded in $K_{\B,\PP}(\Sigma).$
This graph is an open ribbon graph, whose vertices are $K_{\B,\PP}-$images of zeroes of the differential, its edges are $K_{\B,\PP}-$images of horizontal trajectories. The boundary edges, if there are any, are embedded in the boundary and cover it, and the defects of vertices agree with the defects of their image in $K_{\B,\PP}(\Sigma),$ in particular boundary nodes go to boundary nodes. Under this embedding the orientation of any half edge $h\in s_1H^B$ agrees with the orientation induced on $\partial K_{\B,\PP}(\Sigma).$
Topologically $K_{\B,\PP}(\Sigma)\simeq\Sigma_G.$

Moreover, for any stable $(g,\B,\I\cup\PP)-$metric graph is the graph associated to some stable open $(g,\B,\I\cup\PP)-$surface and a set of perimeters $\pp.$
This surface is unique if the graph is smooth or effective.
\end{prop}
We sometimes identify the graph with its image under the embedding. In particular, throughout this article we shall consider an edge as a trajectory in the surface, and a half edge $h$ as trajectory oriented outward from $h/s_0.$
\begin{nn}
With the notations of the above observation, denote by $\Rcomb_\pp$ the map between surfaces and open metric ribbon graphs, defined by $(G,\ell)=\Rcomb_\pp(\Sigma).$
Write also $(G,\ell)=\Rcomb(\Sigma,\pp).$
\end{nn}

\begin{definition}\label{def:norm}
The \emph{normalization} $\NNN(G)$ of a stable connected open ribbon graph $G$ is the unique smooth, not necessarily connected, open ribbon graph, defined in the following way. If $G$ is smooth, $\NNN(G)=G.$
Otherwise
the vertex set is $\widetilde V^I\cup\widetilde V^B\cup \SC^C(G)\cup \SC^O(G),$ contracted components are isolated vertices in the graph, and the half edges are $H^I\cup H^B.$
The genus and topological defects of vertices in ${\widetilde V^I\cup\widetilde V^B}$ are $0.$

For a contracted component $v,$ the genus and topological defects are given by $g^{\NNN(G)}(v)=g(v),~d^{\NNN(G)}(v)=d(v).$ The marking defect and the maps $f^{\PPP,v},~f^{B,v}$ are derived from $d^{\NNN(G)}(v).$ In particular $B(v)=(N^{B})^{-1}(v)\cup(f^{B})^{-1}(v).$ The permutations $s_0^v,s_1^v$ are the trivial permutations, and the set $I(v)=\emptyset.$

For any connected component $C$ of $\NNN(G),$ not in $\SC^C(G)\cup \SC^O(G),$ define $s_0=s_0^C,~s_1=s_1^C,f^{I}=f^{I,C}$ as those induced from $G.$ 
Let $\PPP(C)$ be the union of the set of elements of $\PP$ which map to vertices whose unique $N-$preimage is in $C,$ and the set of preimages of internal nodes of $C,$ i.e., internal vertices $v$ of $C$ such that $|N^{-1}(N(v))|>1.$
In other words, we can write $\PPP(C)=(\PPP(C)\cap\PP)\cup (\PPP(C)\setminus\PP).$ 
We define $f^{\PPP}=f^{\PPP,C}:\PPP(C)\to V^I(C)$ as follows.
On $\PPP(C)\cap\PP$ we put $f^{\PPP,C}(p_i)=N^{-1}(f^{\PPP}(p_i)),$ where $f^{\PPP}$ of the right hand side is the function from the definition of $G,$ while on $\PPP(C)\setminus\PP,$ 
the preimages of a nodes, we set $f^{\PPP,C}(v)=v.$ 
Define $B(C),~f^{B}=f^{B,C}:B(C)\to V^B(C)$ similarly.

The normalization $\NNN(G)$ of a marked graph is the marked graph whose underlying graph is the normalization of the underlying graph of $G,$ new marked points are marked $\BBB.$

Write $\Norm:\NNN(G)\to G$ to be the evident normalization map.
\end{definition}
Observe that the normalization of a trivalent graph is trivalent, and that if $v$ is a contracted component which touches at least one edge in $G,$ then $|\Norm^{-1}(v)|=|N^{-1}(v)|+1.$

Figure \ref{fig:norm_ribbon4_25} shows the normalizations of the graphs in the right column of Figure \ref{fig:ribbon4_17}.
\begin{figure}
\centering
\includegraphics[scale=.4]{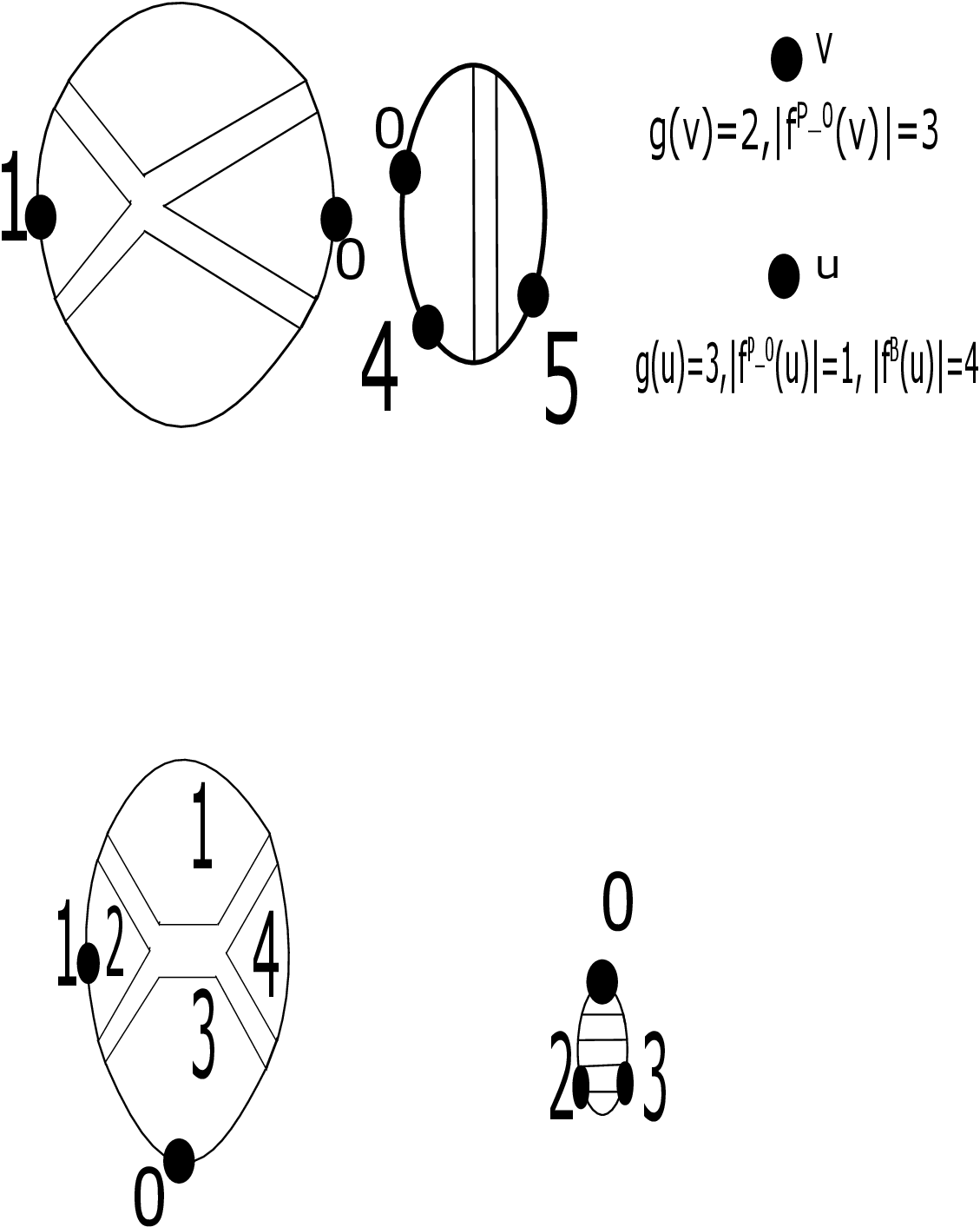}
\caption{The normalizations of the graphs in the right column of Figure \ref{fig:ribbon4_17}. The upper normalization has four components, two are contracted components. The one which corresponds to $v$ has three internal points of perimeter $0,$ the original two and the node. The one which corresponds to $u$ has $4$ boundary markings, the original two and two that corresponds to half nodes. The lower normalization is made of two components. New special points in both normalizations are labeled $0.$}
\label{fig:norm_ribbon4_25}
\end{figure}
\begin{nn}\label{nn:extended_s_1}
There is a canonical injection $B(G)\hookrightarrow B(\NNN(G)).$ On $B(\NNN(G))\setminus B(G)$ there is a fixed point free involution which we also denote by $s_1,$ which on preimages of a node which is not a contracted component just interchanges its two preimages. If $v$ is a contracted component, its new boundary markings correspond to elements $u\in(N^{B})^{-1}(v).$ Any such $u$ corresponds also to a unique marking $w$ in another \emph{non contracted} component. Write $s_1u=w,~s_1w=u.$
\end{nn}

\subsubsection{Moduli of open metric graphs}
For a stable open ribbon graph $G,$ denote by $\RCM_G$ the set of all metrics on $G,$ write $\RCM_G(\pp)$ for the set of all such metrics were the $i^{th}$
face has perimeter $p_i.$
Note that $\RCM_G=\R_+^{E(G)}/\text{Aut}(G)$ canonically. 

\begin{construction}
For $e\in E(G),$ the edge between vertices $v_1,v_2,$ one can define the graph $\partial_eG,$
as the graph obtained by contracting $e$ to a point, identifying its vertices to give a new vertex $v_1v_2$ and updating the permutations and marking defects as in the closed case.
When $v_1,v_2$ are internal, then so is $v_1v_2.$ The genus defect is updated as in the closed case, and this determines the whole defect.
Suppose $v_1$ is a boundary vertex. Then so is $v_1v_2.$
If $v_2\neq v_1,$ then $g(v_1v_2)=g(v_1)+g(v_2),$ if $v_2\in V^B,$ and otherwise $g(v_1v_2)=g(v_1)+2g(v_2).$
When $v_1=v_2,$ let $h_1,h_2$ be the half edges of $e.$ Let $\widetilde{h}_i\in N^{-1}(v_1)$ be the $s_0-$cycle of $h_i.$
Then $g(v_1v_2)=g(v_1)+\delta,$ where
\[
\delta=\begin{cases}
0,&\mbox{if }\widetilde{h}_1=\widetilde{h}_2,\\
1,&\mbox{if }\widetilde{h}_1\neq \widetilde{h}_2,~\widetilde{h}_1,\widetilde{h}_2\in\widetilde{V}^B,\\
2,&\mbox{otherwise. }
\end{cases}
\]
$d(v_1v_2)\in \D=\D(g(v_1v_2),I_{v_1v_2},B_{v_1v_2}),$ or $d(v_1v_2)\in \D=\D(g(v_1v_2),I_{v_1v_2}),$
where $B_{v_1v_2}=(f^{B})^{-1}(v_1v_2)\cup(N^{B})^{-1}(v_1v_2),~I_{v_1v_2}=(f^{\PPP})^{-1}(v_1v_2)\cup(N^{\PPP})^{-1}(v_1v_2).$
These two sets are already known from what we have constructed so far.
In particular, whenever $\D$ is trivial, which is always the case for internal vertices, and for boundary vertices it happens when $2g(v_1v_2)+2|I_{v_1v_2}|+|B_{v_1v_2}|\leq 2,$ we know $d(v_1v_2).$
For shortness we will not describe the general update of the topological defect. We do describe a special case of particular importance. Suppose $e\in E^B,$ and $v_1\neq v_2$ are boundary vertices with $d(v_i)\in \D(0,\emptyset,B_i)$ where $|B_i|=2.$ This is the case when each $v_i$ is a marked point or a boundary node which is not a contracted component.
Write $B_i=\{\widetilde{h}_i,a_i\},$ where $\widetilde{h}_i$ is as above. Suppose $h_2\in H^B,$ that is, its orientation disagrees with the orientation of the boundary.
Then $d(v_1v_2)\in \D(0,\emptyset,\{a,a_1,a_2\}),$ where $a$ is the new cycle of $s_0h_2,$ obtained from concatenating $\widetilde{h}_1,\widetilde{h}_2$ after erasing $h_1,h_2,~d(v_1v_2)$ is the element which corresponds to cyclic order $a\to a_1\to a_2.$

Suppose $E'=\{e_1,\ldots,e_r\}\subseteq E,$ then there is an identification between $E(G)\setminus E'$ and $E(\partial_{e_1,\ldots,e_r}G).$ Throughout this paper we shall use this identification without further comment.
\end{construction}
Figure \ref{fig:nn_cons_2} illustrates several examples of edge contractions.
\begin{figure}
\centering
\includegraphics[scale=.4]{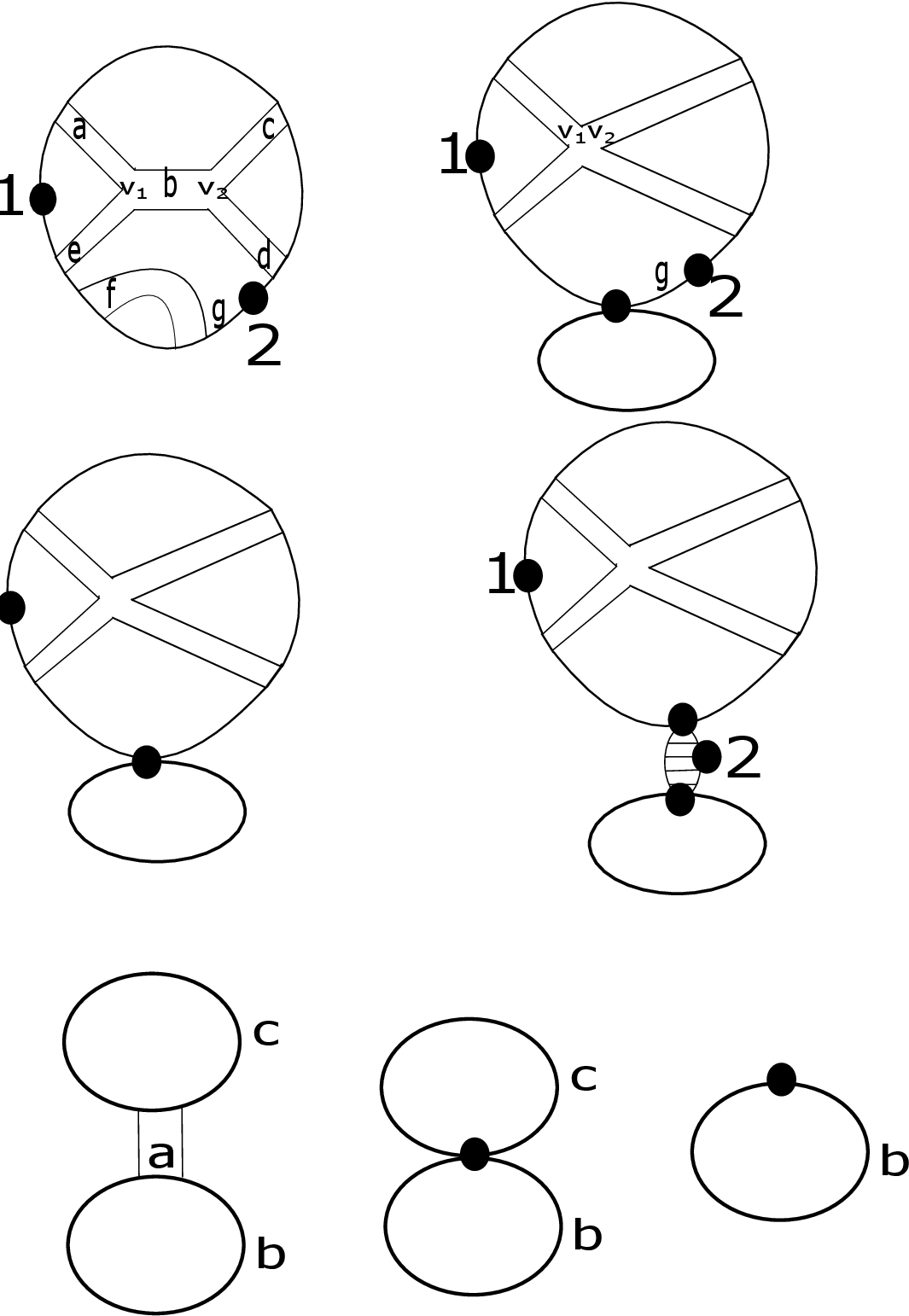}
\caption{Examples of edge contractions. Contracting the internal edges $b,f$ of the smooth graph on the top left corners gives rise to the nodal graph on the top right corner. The vertex $v_1v_2$ corresponds to the permutation $(ae)(cd).$ By further contracting the boundary edge $g$ between the boundary node and the marked point $2$ we obtain the graph in the left hand of the middle row. The boundary node there corresponds to a contracted component which contains two nodes and the marking $2.$ The graph in the right hand side of the same row is equivalent to the left hand one, only that the ghost is illustrated and there the cyclic order of half nodes is seen.
In the bottom left corner a genus $1$ ribbon graph is drawn. After contracting the edge $a$ we obtain a nodal graph. Further contracting $c$ we obtain the graph on the right corner, which contains an open contracted component. The genus defect of the contracted component is $1$ and its topological defect is that of a cylinder with one special boundary point - the node.
}
\label{fig:nn_cons_2}
\end{figure}

For a stable open ribbon graph $G,$ we define the orbifold cell complex $\oRCM_{G}$ as the cell complex whose cells are $\RCM_{G'},$ for all graphs $G'$ obtained from $G$ by edge contractions. The cell $\RCM_G$ which corresponds to contracting the empty subset of $E(G)$ is included. If $G',G''$ are two such cells, and $G''$ is obtained from $G'$ by contracting the edges $\{e_1,\ldots,e_r\},$ then the corresponding cell $\RCM_{G''}$ is the boundary of the cell $\RCM_{G'}$ glued to it along $\ell_{e_1}=\ldots=\ell_{e_r}=0.$ In this case we say that $\RCM_{G''}$ is a face of $\RCM_{G'}.$
Write ${\oRCM}^{\text{comb}}_{g,k,l}=\coprod\oRCM_{G}/\sim=\coprod\RCM_{G},$ where the union is over all open $(g,k,l)-$ribbon graphs, and $\sim$ is induced by the canonical injections $\oRCM_{G'}\hookrightarrow\oRCM_{G},$ over pairs $(G,G')$ where $G'$ is obtained from $G$ by edge contractions.
Write ${\RCM}^{\text{comb}}_{g,k,l}$ for the locus which is the union over smooth graphs.
Define $\oRCM_G(\pp),~{\oRCM}^{\text{comb}}_{g,k,l}(\pp),~{\RCM}^{\text{comb}}_{g,k,l}(\pp)$ by restricting perimeters to be $p_i.$
In these cases we also define the cell attachments using edge contractions.


The pointwise maps $\Rcomb$ induce moduli maps
\[
\Rcomb:{\oRCM}_{g,k,l}\times\R_{+}^l\to{\oRCM}^{\text{comb}}_{g,k,l},~\Rcomb_{\pp}:{\oRCM}_{g,k,l}\to{\oRCM}^{\text{comb}}_{g,k,l}(\pp),
\]
which send a stable open surface and a set of perimeters to the corresponding graph.

\begin{lemma}\label{lem:minimal_compacs_open}
${\oRCM}^{\text{comb}}_{g,k,l}$ with the cell structure defined above is a piecewise smooth Hausdorff orbifold with corners.
This is the finest topology on the moduli of $(g,k,l)-$graphs such that the map $\Rcomb$ is continuous. ${\oRCM}^{\text{comb}}_{g,k,l}(\pp)$ is compact for any $\pp.$
$\Rcomb:\RCM_{g,k,l}\times\R_{+}^l\simeq{\RCM}^{\text{comb}}.$
Moreover, the analogous claims remain true if we declare some, but not all, of the internal marked points to have perimeters $0.$
In fact, for any effective dual graph $\Gamma$ 
the map $\Rcomb$ restricted to $\RCM_\Gamma\times\R_{+}^l$ is an isomorphism onto its image.
\end{lemma}
The proof is similar to the closed case, see \cite{Zvonk,Looij} for a proof of the analogous theorem.

\subsection{JS Stratification for the graded moduli}
\subsubsection{Graded ribbon graphs}
For a metric, open or closed ribbon graph, $(G,\ell),$ write
\[\widetilde{Z}_{G,\ell}=\pi_0(\widetilde{\text{For}}_{\text{spin}}^{-1}((\Rcomb)^{-1}(G,\ell)), {Z}_{G,\ell}=\pi_0({\text{For}}_{\text{spin}}^{-1}((\Rcomb)^{-1}(G,\ell)),\]
where the maps $\widetilde{\text{For}}_{\text{spin}},~{\text{For}}_{\text{spin}}$ are defined in Notation \ref{nn:tilde_forget_spin}. 
For any two generic metrics $\ell,\ell'$ the sets $Z_{G,\ell},Z_{G,\ell'}$ are isomorphic, see Remark \ref{rmk:spin_is_const}. When $G$ has non trivial automorphisms the sets are non canonically isomorphic.
For any $G$, let $Z_G$ be the set $Z_{G,\ell}$ for a fixed generic $\ell.$
Define $\widetilde Z_G$ similarly.

\begin{definition}\label{def:graded_ribbon}
A \emph{metric spin ribbon graph with a lifting} $(G,z,\ell)$ is a metric ribbon graph together with an element $z\in \widetilde Z_{G,\ell}.$ The graph is called \emph{graded} when $z\in Z_{G,\ell}.$ A graded graph is a pair $(G,z),z\in Z_G.$ Similarly, in the closed setting, a \emph{metric spin ribbon graph} $(G,z,\ell)$ is a metric ribbon graph together with $z\in \widetilde Z_{G,\ell}.$

The \emph{normalization} $\NNN(G,z,\ell)$ of $(G,z,\ell)$ is the smooth, not necessarily connected graph $\coprod(G_i,\ell_i,z_i),$ where $(G_i,\ell_i)$ are the components of $\NNN(G,\ell),$ and $z_i\in \widetilde{Z}_{G_i,\ell_i}$ are induced from $z$ by Proposition \ref{prop:spin_on_comps}. A half node is \emph{legal} if it is legal as a marked point in the graded structure of $\NNN(G,z).$
\end{definition}
It follows from Proposition \ref{prop:streb_gives_graph} that
a graded surface, together with perimeters $\{p_i\}_{i\in\I},$ defines a unique graded metric graph $(G,z,\ell),$ where $(G,\ell)$ is embedded in $K_{\B,\PP}({\text{For}}_{\text{spin}}(\Sigma)),$
as in Proposition \ref{prop:streb_gives_graph} and $z$ is the class of graded spin structures which contains the graded structure of $\Sigma.$ When $(G,\ell)$ generic and effective, all possible automorphisms of $(G,\ell)$ leave all half edges in place, and may only act non trivially only on isolated contracted components, which are of genus $0.$ Thus, the action of this automorphism group on $Z_G=Z_{G,\ell}$ is trivial, and hence, in this case, $Z_G$ isomorphic to $\Spin(\Sigma),$ and any element $z$ of it corresponds to a unique graded structure.

Moreover, by Corollary \ref{cor:spin_on_comps}, if in addition $G$ has no contracted boundaries, then $Z_{G}$ is in one to one correspondence with isomorphism classes of tuples $(\SL_1,\ldots,\SL_r)$ where each $\SL_i$ is a spin structure with a lifting on the $i^{th}$ component of $\NNN(\Sigma)$ such that all original boundary marked points are legal and for any boundary node of $\Sigma$ exactly one half is legal.
\begin{definition}\label{def:oR}
A spin ribbon graph with a lifting $(G,z),$ with or without a metric $\ell,$ is called \emph{effective} if $G$ is effective, and $z$ is a spin structure with a lifting in which for every contracted component $v\in V(G),$ all boundary marked points of the isolated component in $\Norm^{-1}(v)$ are legal. In case $v$ is not isolated, it is equivalent to all half nodes in $(N^{B})^{-1}(v)$ being illegal. An effective graded graph $(G,z)$ is \emph{trivalent} if $G$ is trivalent.
The graph is \emph{smooth} if its underlying graph is. These definitions extend to the closed case, without the assumptions on boundary nodes.

Denote by $\oSR^0$ the set of isomorphism classes of graded smooth trivalent ribbon graphs, 
write $\oR^0$ for the set of their underlying open ribbon graphs. Denote by $\oSR^0_{g,k,l}\subseteq\oSR^0$ the subset whose faces are marked $[l]$ and the boundary points are marked by $[k].$ Define $\oR^0_{g,k,l}$ similarly. 

Let $\oSRc^0_{g,k,l}$ be the collection of all graphs in $\oSR^0_{g,k,l}$ with an odd number of boundary marked points on each boundary component. Define $\oRc^0_{g,k,l}$ similarly.
\end{definition}
Note that in a trivalent graph, by definition if $v$ is a contracted component, the unique ghost component in $\Norm^{-1}(v)$ has all its marked points \emph{legal}.

Recall that smooth graded surfaces have no internal markings of twist $1$ or illegal boundary markings. Therefore an immediate corollary of Proposition \ref{prop:parity_disks}, which can be taken as an alternative definition of $\oR^0_{g,k,l},$ is
\begin{cor}\label{cor:parity_oR}
$\oR^0_{g,k,l}\neq\emptyset$ precisely when $2|g+k-1.$ Every trivalent smooth graph satisfying this constraint belongs to $\oR^0_{g,k,l}.$
\end{cor}

\begin{nn}
We define the map $\comb$ between graded surfaces and graded metric ribbon graphs by $\comb(\Sigma,\SL,s,\pp)=(G,z,\ell)$ where $(G,\ell)=\Rcomb(\Sigma,\pp)$ and $z\in Z_{G,\ell}$ is the corresponding class.
Write $\comb_\pp=\comb(-,-,-,\pp).$ Write ${\text{For}}_{\text{spin}}^{\text{comb}}(G,z,\ell)=(G,\ell).$
\end{nn}

\begin{prop}\label{prop:for_base_JS}
Suppose $\comb(\Sigma,\pp)=(G,z,\ell).$
\begin{enumerate}
\item
Then $\comb(\NNN(\Sigma),\pp)=\NNN(G,z,\ell),$ where preimages of nodes in $\Sigma$ will be internal markings of perimeter $0.$
\item
Suppose $\Sigma'$ is obtained from $\Sigma$ by forgetting an illegal marked point $x_\nu$ whose removal makes no component unstable. Suppose that $x_\nu$ is mapped to vertex $v$ of $G.$ Write $(G',z',\ell')=\comb(\Sigma',\pp).$ Then $(G',\ell')$ is obtained from
$(G,z,\ell)$ by the following procedure. If $\text{deg}(v)=2,$ and $v$ has a zero genus defect and marking defect $\{\nu\},$ remove $v$ from the graph, unite its two edges $e_1,e_2$ to one edge $e,$ define $\ell'(e)=\ell(e_1)+\ell(e_2)$ and for the other edges put $\ell'=\ell.$ Otherwise the graph and metric do not change, but the marking $\nu$ is removed from the marking defect of $v$.
$z'$ is the image of $z$ under the natural map $Z_{G,\ell}\to Z_{G',\ell'}$ obtained from Observation \ref{obs:forgetful_illegal} with $\B'=\{\nu\}.$
\end{enumerate}
\end{prop}
\begin{proof}
The first item is a consequence of Corollary \ref{cor:normalization-open.differential}. The second follows from Corollary \ref{cor:forgetting_a_bdry_point_differential} and Observation \ref{obs:forgetful_illegal}.
\end{proof}

\subsubsection{Combinatorial moduli for graded surfaces, bundles and forms}
Denote by $\oCM_{g,k,l}^{\text{comb}}$ the set of metric graded $(g,k,l)-$ribbon graphs. Write $\oCM_{g,k,l}^{\text{comb}}(\pp)$ for the subspace of graphs with fixed perimeters $\pp.$ Define $\CM_{g,k,l}^{\text{comb}}$ as the subspace of smooth graphs. Define similarly $\CM_{g,k,l}^{\text{comb}}(\pp).$
The pointwise maps $\comb$ induce moduli maps
\[
\comb:\oCM_{g,k,l}\times\R_{+}^l\to\oCM^{\text{comb}}_{g,k,l},~\comb=\comb_{\pp}:\oCM_{g,k,l}\to\oCM^{\text{comb}}_{g,k,l}(\pp),
\]
which send a stable graded surface and a set of perimeters to the corresponding graph.
Endow these spaces with the finest topology so that $\comb$ is continuous.

We now study the cell structure of $\oCM_{g,k,l}^{\text{comb}}.$ Recall that a metric $\ell$ is generic if the metric graph has no automorphisms. In particular, in the open and connected setting, metrics which give all edges different lengths are generic.
For a generic $\ell\in\RCM_G,$ choose $z\in Z_G=Z_{G,\ell},$ define $\CM_{(G,z)}$ to be the connected component of $({\text{For}}_{\text{spin}}^{\text{comb}})^{-1}(\RCM_G)$ which contains $(G,z,\ell).$

The map ${\text{For}}_{\text{spin}}^{\text{comb}}$ is continuous. Moreover, by the same reasoning as in the non combinatorial case, see the discussion in the end of Subsection \ref{subsub:moduli}, it is an orbifold branched cover, and over any $\RCM_G$ it is an orbifold cover.

Thus, $({\text{For}}_{\text{spin}}^{\text{comb}})^{-1}(\RCM_G)$ is an orbibundle over $\RCM_G,$ with a generic fiber $Z_G.$
Since $\RCM_G =\R_+^{E(G)}/\text{Aut}(G),$ such a bundle must be of the form
\[({\text{For}}_{\text{spin}}^{\text{comb}})^{-1}(\RCM_G) \simeq (\R_+^{E(G)}\times Z_G)/\text{Aut}(G),\]
for some action of $\text{Aut}(G)$ we now explain.

Let $C\subseteq\RCM_G$ be the locus of generic metrics, and $\overline{C}\subseteq\R^{E(G)}$ its preimage under the quotient by $\text{Aut}(G).$ Except from some borderline cases which can be treated separately its complement is of real codimension at least $3.$ Over $C$ the fiber of the bundle is always of size $|Z_G|.$ Denote this fiber bundle by $E,$ and let $\overline{E}\to\overline{C}$ be its pullback to $\overline{C}.$ $\pi_1(\overline C)$ is trivial, as $\R^{E(G)}\setminus\overline{C}$ is of codimension at least $3.$ Thus $\overline{E}$ must be trivial, and is hence isomorphic to $\overline{C}\times Z_G.$

Let $\overline{\ell}\in\overline{C}$ be any point, let $\ell$ be its image in $C.$
Recall that as an orbispace, $\text{Aut}(G)\simeq \pi_1(\overline{C}/\text{Aut}(G),\ell),$ and this isomorphism can be made explicit as follows:
For $g\in \text{Aut}(G),$ choose any path $\bar{\gamma}^g:[0,1]\to\overline{C},$ with $\bar{\gamma}_0^g=\overline{\ell}\in\R_+^{E(G)}, \bar{\gamma}_1^g=g\cdot\overline{\ell},$
and set $\gamma_g$ to be its $\bar{\gamma}^g$ to $C.$

Parallel transport $z=z_0$ along $\gamma^g$ to get $z_1.$ This can be done as the fiber is $0-$dimensional.
Define $g\cdot(\overline{\ell},z)=(g\cdot\overline{\ell},z_1).$
This action is independent of choices, and can be defined continuously over all $\overline{E}.$ This gives us the orbibundle structure over $C.$ Again by continuity, it can be uniquely extended to an action on $\R^{E(G)}_+\times Z_G.$

Note that in particular, we have defined an action of $\text{Aut}(G)$ on $Z_G.$
Define the group $\text{Aut}(G,z)$ as the subgroup of $\text{Aut}(G)$ which leaves $z$ invariant. Then $\CM_{(G,z)}\simeq\R^{E(G)}_+/\text{Aut}(G,z).$
Define $\CM_{(G,z)}(\pp)$ as the subspace where the perimeters are $\pp.$ 

For $e\in E(G),$ define the \emph{edge contraction} to be $\partial_e(G,z)=(\partial_eG,\partial_ez),$
where $\partial_ez\in Z_{\partial_e G}$ using the cell structure of $\oRCM_G$ and the topology of $\oCM_{g,k,l}^{\text{comb}}.$ Explicitly, fix $\pp$ and take an arbitrary continuous path $([\Sigma_t])_{t\in[0,1]}\subset\oCM_{g,k,l},$ such that $\text{comb}([\Sigma_t])\in \CM_{(G,z)}$ for $t>0$ and ${\text{For}}_{\text{spin}}(\text{comb}([\Sigma_0]))\in\RCM_{\partial_eG}.$
Suppose that $\text{comb}([\Sigma_0])\in\CM_{(\partial_eG,z')}.$ Then $z'=\partial_ez,$ and this definition is easily seen to be independent of choices.

An explicit combinatorial description for the special case of trivalent graphs appears in Subsection \ref{subsec:codim_geq_1}.

As in the spinless case $\oCM_{(G,z)},$ the closure of $\CM_{(G,z)}$ in $\oCM_{g,k,l}^{\text{comb}},$ is the union of cells $\CM_{(G',z')},$ where $(G',z')$ is obtained from $(G,z)$ by edge contractions, and the attachment of the cells is also defined via the edge contractions, i.e. $\CM_{(G',z')}$ is glued to $\CM_{(G,z)}$ along $\ell_{e_1}=\ldots=\ell_{e_r}=0,$ if $e_1,\ldots,e_r$ are the edges of $G$ which are contracted to obtain $G'.$
%
In this case we say that $\CM_{(G',z')}$ is a face of $\CM_{(G,z)}.$ We similarly define $\oCM_{(G,z)}(\pp).$ 
Again as in the spinless case we can now define the orbifold cell complex structure on $\oCM^{\text{comb}}_{g,k,l}.$ \[\oCM^{\text{comb}}_{g,k,l}=\coprod\oCM_{(G,z)}/\sim=\coprod\CM_{(G,z)},\] where the union is over all connected components which correspond to graded $(g,k,l)-$ribbon graphs, and $\sim$ is induced by edge contractions. We similarly define the orbifold cell complex structure on $\oCM^{\text{comb}}_{g,k,l}(\pp).$ In both cases the cell structure agrees with the topology.
Denote the quotient-by-$\sim$-map by $\Xi.$

A graph $(G,z)$ corresponds to a boundary stratum of $\oCM_{g,k,l}^{\text{comb}},$ that is $\CM_{(G,z)}\subseteq\comb(\partial\oCM_{g,k,l}\times\R_+^l)$ if and only if it has at least one boundary node or contracted boundary. In this case we call it a \emph{boundary graph}. All of the above constructions extend to the setting of spin ribbon graphs with a lifting, and to (closed) spin ribbon graphs.

\begin{lemma}\label{lem:minimal_compacs_spin}
Suppose $2|g+k-1.$ Then
$\oCM_{g,k,l}^{\text{comb}},~\oCM_{g,k,l}^{\text{comb}}(\pp)$ are piecewise smooth Hausdorff orbifolds with corners, the latter is compact.

The maps $\comb,\comb_{\pp}$ are isomorphisms onto their images when restricted to the open dense subsets $\CM_{g,k,l}\times\R_{+}^l,\CM_{g,k,l}.$

$\comb_{\pp}$ induces an orientation on $\oCM_{g,k,l}^{\text{comb}},$ with this orientation $\text{deg}(\comb_{\pp})=1.$

Analogous claims are true if we declare some, but not all, of the internal marked points to have perimeters $0.$ Analogous claims are also true if we allow some internal markings to be Ramond or if we consider the case of closed (twisted) spin surfaces.
In addition, for an effective dual spin graph with a lifting $\Gamma,$ the maps $\comb,\comb_{\pp}$ restricted to $\CM_\Gamma\times\R_{+}^l,\CM_\Gamma$ are isomorphisms onto their images.
\end{lemma}
The proof is similar to the closed case and will be omitted. The orientation on $\oCM_{g,k,l}^{\text{comb}}$ will be constructed explicitly later.

The combinatorial $S^1-$bundles $\CF_i,~i\in[l],$ are defined as in Definition \ref{def:CF_i}. Again these carry a natural piecewise smooth structure, compatible with the natural piecewise smooth structures on $\oCM_{g,k,l}^{\text{comb}}.$
The forms $\alpha_i,\omega_i,\bar{\alpha}_i,\bar{\omega}_i,\bar{\omega}$ defined as in Definition \ref{def:CF_i} and Equation \eqref{eq:alpha_i_omega_i}.
\begin{definition}\label{def:l-sets}
Let $S\subseteq\mathbb{N}$ be a finite set. A $(S,l)-$set $L$ is a function $L : S\rightarrow [l].$ We write $S=\text{Dom}(L).$
In case $S=[d],$ we simply write a $(d,l)-$set. We say that $L$ is a $l-$set if the set $S$ is understood from the context. 

Given two $l-$sets, $L,L',$ we write
\[
L'\subseteq L,
\]
and say that $L'$ is a subset of $L,$ and write $L'\subseteq L,$ if
\[
\text{Dom}(L')\subseteq \text{Dom}(L)~\text{and}~ L|_{\text{Dom}(S')}=L'.
\]
In this case we define the $l-$set $L\setminus L'$ by
\[
L\setminus L' :\text{Dom}(L)\setminus \text{Dom}(L')\to[l], ~(L\setminus L')(s)=L(s).
\]
In case $j\in \text{Dom}(L)$ we write $j\in L.$ For $i\in[l]$ we put
\[
L_i = L^{-1}(i).
\]

\end{definition}
The $(S,l)-$sets will be used to encode direct sums of tautological lines in the following way.
\begin{nn}\label{nn:comb_sphere}
Recall Construction \ref{nn:spherization}. To any $(S,l)-$set $L$ we associate a vector bundle $E_L$ and a sphere bundle $S_L$ given by
\[E_L = \sum_{i\in S} \CL_{L(i)}\to\oCM_{g,k,l},~~S_L=S((\CF_{L(i)})_{i\in S}).\]
We will also consider the sphere bundle $S(E_L)$ associated to $E_L.$

Define an angular form $\Phi_L$ for $S_L$ by Formula \eqref{eq:angular}, and using Kontsevich's forms for the copy $\CF_{L(i)},$ of the $L(i)^{th}~S^1-$bundle.
Explicitly,
\begin{align*}
\Phi_{L}(\{{r}_i\}_{i\in S},& \{\hat{\alpha}_i\}_{i\in S},\{\hat{\omega}_i\}_{i\in S}) =\\\notag&
\sum_{k=0}^{|S|-1} 2^k k!\sum_{i\in S}{r}_i^2\hat{\alpha}_{i}
\sum_{I\subseteq S\setminus\{i\}, |I|=k}\bigwedge_{j\in I}(r_jdr_j\wedge\hat{\alpha}_{j})\bigwedge_{h\notin I\cup\{i\}}\hat{\omega}_{h},
\end{align*}
where $\hat\omega_{i}$ is Kontsevich's 2-form $\omega_{L(i)},$ and $\hat{\alpha}_i$ is a \emph{copy} of Kontsevich's 1-form $\alpha_{L(i)}.$ We refer to it as a copy, since for $i_1,i_2\in L_j,$ $\hat\alpha_{i_1},\hat\alpha_{i_2}$ are given by the same formula of the angular $1-$form of $\CF_{j},$ but with different $\phi$ variables.
Write $$\omega_L  = -d\Phi_L=\bigwedge_{i\in S}\omega_{L(i)},~p^{2L}=\prod_{i\in S} p^2_{L(i)},~\bar{\omega}_L=\pp^{2L}\omega_L,~\bar{\Phi}_L=\pp^{2L}\Phi_L.$$
When $S\neq [d]$ we will sometimes omit the assumption that $\sum_{i\in S} r_i^2=1,$ and then $-d\Phi_L$ gets a correction, see Remark \ref{rmk:phi_as_poly}.

When it is not clear from context, we write $\alpha^G_j$ to indicate the specific graph $G.$ The same remark goes to the other forms.
\end{nn}

Exactly as in the closed case, we have
\begin{lemma}\label{lem:zvonk_open}
\begin{enumerate}
\item For $i\in[l],~\comb^*\CF_i\simeq S^1(\CL_i)$ canonically. As a result, $\comb^*S_L\simeq S(E_L)$ canonically.
\item $\alpha_i,\omega_i$ are piecewise smooth angular $1-$form and Euler $2-$form for $S^1(\CL_i).~\Phi_L$ is an angular form of $S_L,~\omega_L$ its Euler form.
\item For $(G,z)\in\oSR_{g,k,l}^0,$ there is a canonical identification \[(\CF_i\to\oCM_{(G,z)})\simeq\Xi^*(\CF_i\to\oCM_{g,k,l}^{\text{comb}}).\]  Similarly for the bundles $S_L.$
\end{enumerate}
\end{lemma}

\begin{nn}
Recall Proposition \ref{prop:for_base_JS}.
Let $(G,z,\ell)$ be a metric spin ribbon graph with a lifting. Define the graph $\CBB(G,z,\ell)=(\CBB G,\CBB z,\CBB \ell)$ by first taking the normalization of $(G,z,\ell),$ and then forgetting isolated components, the lifting data in contracted boundaries, and the new illegal marked points. 
Let $\CBB:\CM_{(G,z)}\to\CM_{(\CBB G,\CBB z)}$ the induced map on the moduli.
\end{nn}
Observe that
\begin{obs}\label{obs:identificationCBB}
For any spin ribbon graph with a lifting $(G,z),$ and a face marked $i,
~\CF_i\to\CM_{(G,z)}\simeq\CBB^*\left(\CF_i\to\CM_{\CBB(G,z)}\right)$ canonically. A similar claim holds for $S_L.$
\end{obs}
The observation follows from the natural identification of the boundary of the $i^{th}$ faces in $G,\CBB G.$

\begin{prop}
A special canonical multisection $s$ of $S(E_L)$ is a pull back of a multisection $s'$ of $S_L.$
\end{prop}
\begin{proof}
Take $\CM_\Gamma\subseteq\partial\oCM_{g,k,l},$ let $i_1,\ldots,i_r$ be labels of internal tails, one for each vertex of $\Gamma.$
$\comb(\CM_\Gamma\times\R_+^l)=\coprod_{(G,z)}\CM_{(G,z)},$ where the union is taken over some graded graphs $(G,z).$ Consider one of them, denote it by $(G,z).$
Write \[\Phi_\Gamma=\prod_{j=1}^r \Phi_{\Gamma,i}.\]
Consider the diagram
\begin{equation}\label{eq:comb_canonical_seq}
\xymatrix{
\comb^{-1}\CM_{(G,z)} \ar[r]^{\Phi_\Gamma} \ar[d]^{\comb} &\comb^{-1}\CM_{(\CBB G,\CBB z)} \ar[d]^{\comb}\\
\CM_{(G,z)} \ar[r]^{\CBB} &\CM_{(\CBB G,\CBB z)} .
}
\end{equation}
This diagram commutes, by Proposition \ref{prop:for_base_JS}. 
Now $(\CBB G,\CBB z)$ is smooth, hence the right vertical arrow is an isomorphism, by Lemma \ref{lem:minimal_compacs_spin}. A special canonical multisection over $\CM_\Gamma\times\R_+^l$ is pulled back via $\Phi_\Gamma,$ from $S(E_L)\to\prod_{j=1}^r \CM_{v_i^*(\Gamma)}\times\R_+^l.$ Let $s$ be special canonical, we now construct $s'$ with $s=\comb^*s'$. Write $s|_{\comb^{-1}\CM_{(G,z)}}=\Phi_\Gamma^*(\comb^*(s''))$ where $s''$ is a multisection of $S_L\to \CM_{(\CBB G,\CBB z)}.$ Define $s'|_{\CM_{(G,z)}}=\CBB^*s''.$
These multisections for different strata evidently glue.
\end{proof}
\begin{definition}
A \emph{special canonical multisection} of $S_L\to\oCM_{g,k,l}^{\text{comb}}$ is a multisection $s$ with $\comb^*s$ being special canonical.
~A \emph{special canonical multisection} of $S_L\to\oCM_{(G,z)}$ is a $\Xi-$pull back of a special canonical multisection on $\oCM_{g,k,l}^{\text{comb}}.$
~Write $s^{(G,z)}$ for the restriction of $s$ to $\oCM_{(G,z)}.$
\end{definition}
The proof of proposition yields the following immediate corollary.
\begin{cor}\label{cor:main_prop_of_comb_section}
Suppose $(G,z)$ is a boundary $(g,k,l)-$graded ribbon graph, $s$ is a special canonical multisection of $S_L,$ where $L$ is a $(d,l)-$set, restricted to the boundary cell $\CM_{(G,z)}$
then $s=\CBB^*s'$ where $s'$ is a multisection of $S_L\to\CM_{\CBB(G,z)}.$
\end{cor}

The main result of this section is that the descendents can be calculated over the combinatorial moduli.
\begin{lemma}\label{lem:combinatorial_translation0}
Let $s$ be a special canonical multisection for $S(E_L).$ Denote by $s'$ the multisection on $S_L$ with $s=\comb^*s'.$
Then
\begin{align}\label{eq:comb_translate0}
\int_{\oCM_{g,k,l}}e(S(E_L),s) = \int_{\oCM_{g,k,l}^{\text{comb}}(\pp)}e(S_L,s').
\end{align}
The orientations are the ones induced on the combinatorial moduli by $\comb_*.$
\end{lemma}
The proof is an immediate consequence of Lemmas \ref{lem:minimal_compacs_spin}, \ref{lem:zvonk_open} and Observation \ref{obs:functoriality_Euler}.

\subsubsection{Intersection numbers as integrals over the combinatorial moduli}
We can now use the natural piecewise linear structure on $\oCM_{g,k,l}^{\text{comb}}$ and the associated bundles
to write an explicit integral formula for them.
\begin{definition}\label{def:bridges_and_cont_half_edges}
A \emph{boundary loop} in a graded graph $(G,z)$ is a boundary edge which is a loop. We denote the collection of these elements by $\text{Loop}(G).$ 
A \emph{bridge} in a graded graph $(G,z)$ is either a boundary edge between two distinct special legal boundary points
or an internal edge between two boundary vertices, see Figure \ref{fig:for4_43} and the left hand sides of Figure \ref{fig:Feynman}, rows (d),(e). Denote by $\text{Br}(G,z)$ the set of bridges of $(G,z).$ Usually we shall omit $z$ from the notation and write $\text{Br}(G)$ instead.
A \emph{compatible sequence of bridges} $\{e_1,\ldots,e_r\}$ is a sequence of bridges such that $e_{i+1}$ is a bridge in $\partial_{e_1,\ldots,e_i}G$ for all $i.$

Suppose $e$ is a bridge and $h\in H^I$ satisfies $h/s_1=e.$
Set $h'=s_2h.$ We define $\partial_eh\in \HN(\partial_eG)$ ($\HN$ is defined in Definition \ref{def:open_stable_ribbon}) to be the unique vertex $v\in V(\NNN(\partial_eG))$ with $h'/s_0=v,$ where we consider $h'$ as an edge of $\NNN(\partial_eG),$ using the canonical identification, see Figure \ref{fig:for4_43} and the right hand sides of rows (d),(e) of Figure \ref{fig:Feynman}. When there is $h\in H^B$ with $h/s_1=e,$ contracting $e$ creates a contracted component $v,$ which is identified with a ghost component of $\NNN(G),$ see Figures \ref{fig:for4_43} and \ref{fig:Feynman} (d) again.
We denote by $\partial_e h\in B(v)$ the marking which is the $s_0-$cycle of $s_2(s_1h)$ in $(N^{B})^{-1}(v).$ This is equivalent to writing $\partial_e h = s_1\partial_e(s_1h),$ recalling Notation \ref{nn:extended_s_1}.
\end{definition}
\begin{figure}
\centering
\includegraphics[scale=.3]{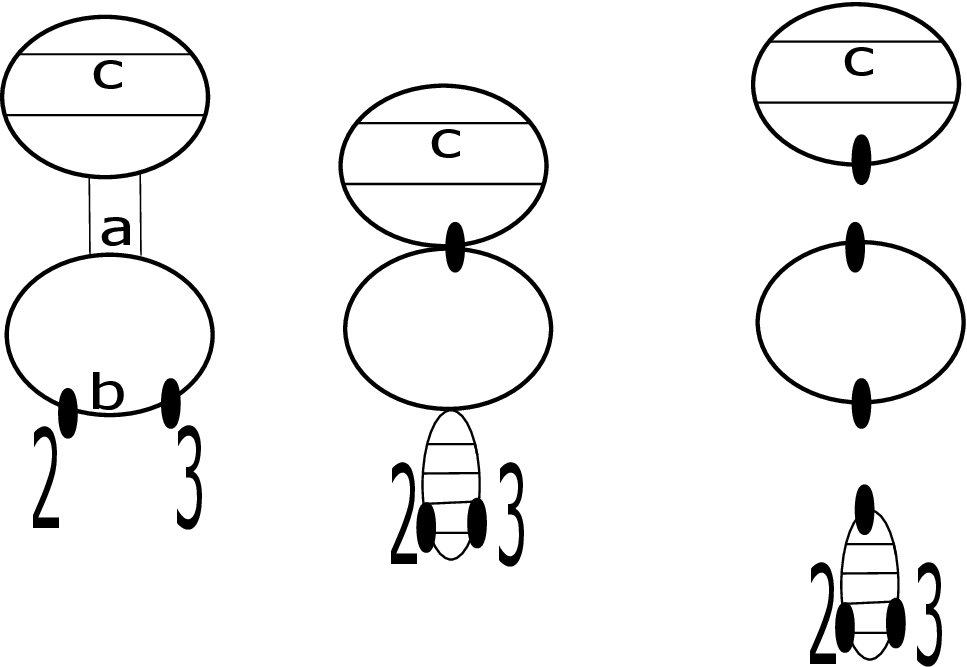}
\caption{Bridges and their contractions. On the left hand side three compatible bridges are drawn, $a,b,c.$ In the middle picture $b,c$ are contracted, and on the right hand the normalization is presented. If $h_b$ is the boundary half edge which corresponds to $b,$ then $\partial_b h$ corresponds to the half node in the ghost component of the normalization. If $h_1,h_2$ are the half edges of $c$ then $\partial_c h_1,\partial_c h_2$ are the two half nodes in the normalization of the node which corresponds to $c.$}
\label{fig:for4_43}
\end{figure}
\begin{figure}
\centering
\includegraphics[scale=.3]{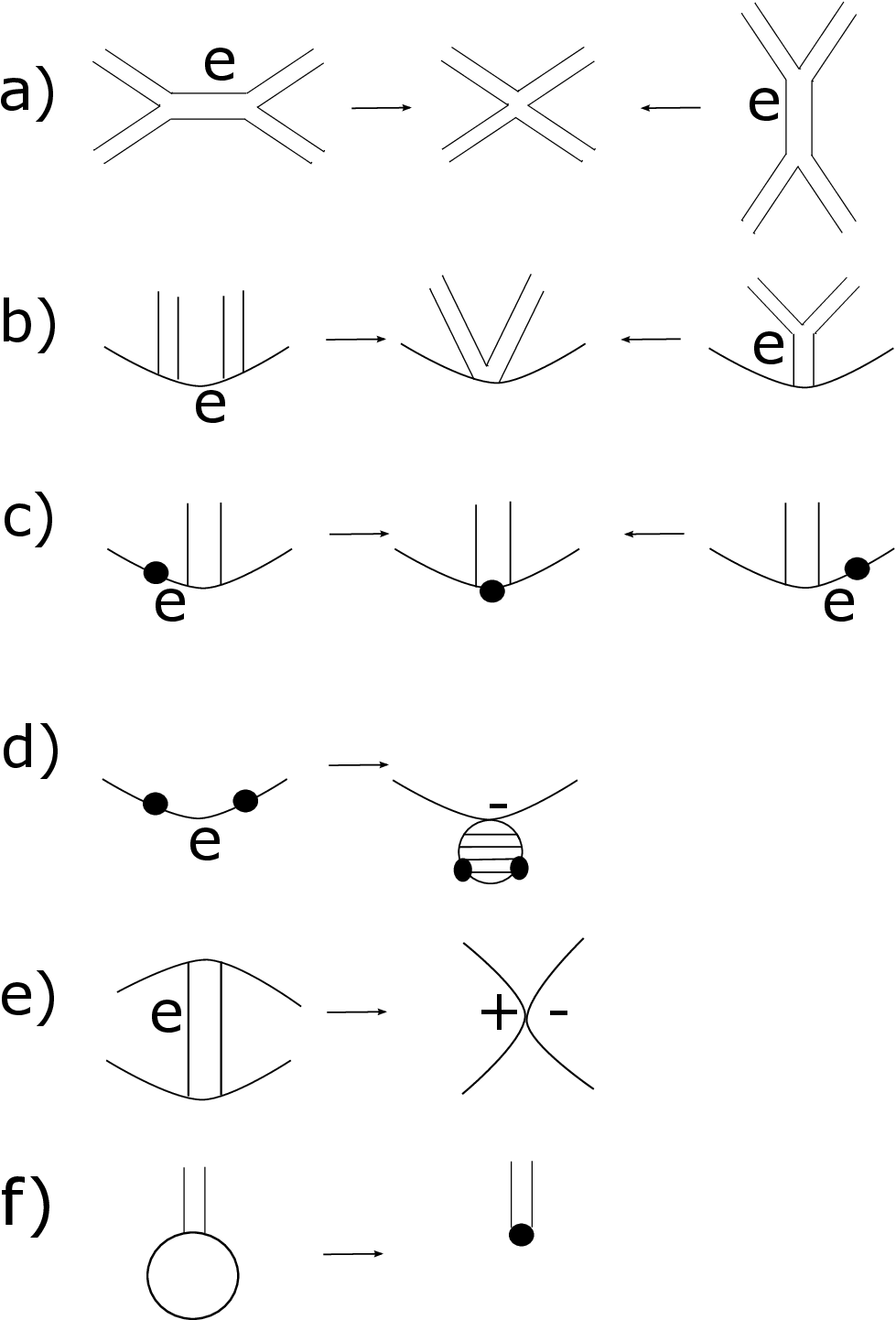}
\caption{Edge contractions and Feynman moves. In rows $d,e$ bridge contractions are presented. In row f, a boundary contraction is shown.
In the remaining rows it is shown how the other types of contracted edges can be obtained as the results of two different contractions.}
\label{fig:Feynman}
\end{figure}

The following observation is immediate
\begin{obs}\label{obs:highest_bdry_cels}
\begin{enumerate}
\item\label{it:full}
$\dim\CM_{(G,z)}(\pp)=\dim\oCM_{g,k,l}$ if and only if $(G,z)\in\oSR^0_{g,k,l}.$
\item\label{it:bdry}
In addition, $(G,z)$ is a boundary graph if and only if it can be represented as $\partial_{e_1,\ldots,e_r}(G',z'),$ where $(G',z')\in\oSR^0_{g,k,l},$ and at least one $e_i$ is a bridge or a loop.
The only boundary graphs whose $(G,z)$ whose moduli is of full dimension $\dim\oCM_{g,k,l}-1,$ are those which can be written as $\partial_e(G',z'),$ for $(G',z')\in\oSR^0_{g,k,l},~e\in \text{Br}(G')\cup\text{Loop}(G').$
\item\label{it:trivalent}
If $\{e_1,\ldots,e_r\}$ is a compatible sequence of bridges in a trivalent graph $(G,z)$ then $\partial_{e_1,\ldots,e_r}(G,z)$ is trivalent.
Any trivalent graph can be written as $\partial_{e_1,\ldots,e_r}(G,z),$ where $(G,z)$ is smooth trivalent and $\{e_1,\ldots,e_r\}$ is a compatible sequence of bridges. This representation is unique up to reordering the bridges in the sequence.
\end{enumerate}
\end{obs}
See Figure \ref{fig:Feynman} (d),(e) for examples.

Recall Definition \ref{def:int_nums}.
Using Observation \ref{obs:highest_bdry_cels}, Lemma \ref{lem:combinatorial_translation0} and Proposition \ref{prop:basic prop},
we immediately get
\begin{lemma}\label{lem:combinatorial_translation}
Let $L$ be a $(d,l)$-set where $d=\frac{3g-3+k+2l}{2},$ and let $s$ be a special canonical multisection for $S_L.$
Then
\begin{align*}
&2^{\frac{g+k-1}{2}}
\langle \tau_{a_1}\ldots\tau_{a_l}\sigma^k \rangle = \\
&\sum_{(G,z)\in\oSR^0_{g,k,l}}\int_{\CM_{(G,z)}(\pp)}{\omega}_L+\sum_{\substack{(G,z)\in\oSR^0_{g,k,l},\\ [e]\in [\text{Br}(G)\cup\text{Loop}(G)]}}
\int_{\oCM_{\partial_e(G,z)}(\pp)}s^*{\Phi}_L.
\end{align*}
equivalently,
\begin{align*}
&\pp^{2L}2^{\frac{g+k-1}{2}}
\langle \tau_{a_1}\ldots\tau_{a_l}\sigma^k \rangle = \\
&\sum_{(G,z)\in\oSR^0_{g,k,l}}\int_{\CM_{(G,z)}(\pp)}\bar{\omega}_L+
\sum_{\substack{(G,z)\in\oSR^0_{g,k,l},\\ [e]\in [\text{Br}(G)\cup\text{Loop}(G)]}}
\int_{\oCM_{\partial_e(G,z)}(\pp)}s^*\bar{\Phi}_L.
\end{align*}
The orientations are the ones induced on the combinatorial moduli by $\comb_*.$
\end{lemma}
\begin{rmk}
The formalism of piecewise linear forms and their integration is treated, for instance, in \cite{Zvonk}.
\end{rmk}

\begin{construction}\label{cons:z_e}
For later purposes we now define \emph{Feynman moves} in edges.
Suppose that $G$ is a trivalent graph, $e\in E\setminus \text{Br}(G).$ If $e$ is a boundary edge, we require that least one of its vertices is not a special point.
Define the graph $G_e=G,$ in case $e$ is a boundary loop. Otherwise, define $G_e$ as the graph obtained from $G$ by first contracting $e$ and then reopening it in the unique different possible way, see Figure \ref{fig:Feynman} (a),(b),(c). 

Let $(G,z)$ be a graded trivalent graph. For a boundary loop $e$ define the graded structure $z_e\in Z_{G}$ as the graded structure which is identical to $z$ except that the lifting on the boundary component $e$ is opposite. For an edge $e\notin \text{Br}(G)\cup\text{Loop}(G)$ write $z_e\in Z_{G_e}$ for the graded structure on $G_e,$ defined by the following proposition.
\begin{prop}\label{prop:G_and_G_e}
For $(G,z)$ and $e$ as above there is a unique graded structure $z_e$ such that
~if $G$ is smooth, $\CM_{(G_e,z_e)}$ is the unique codimension $0$ cell of $\oCM_{g,k,l}^{\text{comb}}$ adjacent to $\CM_{(G,z)}$ along $\CM_{\partial_e(G,z)}.$
For non smooth $G,$ write $(G,z)=\partial_{e_1,\ldots,e_r}(H,w),~e_1,\ldots,e_r\in E(H),$ with $(H,w)$ trivalent and smooth. Then
\[(G_e,z_e)=\partial_{e_1,\ldots,e_r}(H_e,w_e).\]
\end{prop}
\begin{proof}
For a smooth trivalent $G,$ and an edge $e,$ $\partial_e\CM_{(G,z)}$ is a codimension $1$ face, hence, since $\oCM_{g,k,l}^{\text{comb}}$ is an orbifold with corners, this face must be adjacent to at most one additional codimension $0$ cell. Since $e$ is neither a boundary loop nor a bridge, this face is not contained in the boundary of the moduli, hence has it is adjacent to two codimension $0$ cells. Since ${\text{For}}_{\text{spin}}^{\text{comb}}$ is continuous, this cell must be of the form $\CM_{(G,z_e)}$ for some graded structure $z_e\in Z_G\setminus{z}$ or $\CM_{(G_e,z_e)},$ for $z_e\in Z(G_e).$
~The map $\CM_{g,k,l}^{\text{comb}}\simeq\CM_{g,k,l}\to\RCM_{g,k,l}\simeq{\RCM_{g,k,l}}^{\text{comb}},$ when restricted to the open dense set of generic metrics is a (non branched) covering map, as there are no automorphisms to the objects, and since the neighboring cell in ${\RCM_{g,k,l}}^{\text{comb}}$ to $\RCM_G$ along $\partial_e\RCM_G$ is $\RCM_{G_e},$ the neighboring cell of $\CM_{(G,z)}$ along the boundary $\partial_e\CM_{(G,z)}$ must be $\CM_{(G_e,z_e)}.$
The rest of the claim follows from the cell structure and Observation \ref{obs:highest_bdry_cels}, part \ref{it:trivalent}.
\end{proof}

The operations $G\to G_e,(G,z)\to(G_e,z_e)$ are called \emph{Feynman moves}.
\end{construction}

\section{Trivalent and critical nodal graphs}\label{sec:critical}
It follows from Lemma \ref{lem:combinatorial_translation} that all intersection numbers can be calculated
as integrals over the highest dimensional cells of $\oCM_{g,k,l}^{\text{comb}},$ and of $\partial\oCM_{g,k,l}^{\text{comb}}.$
In the next section we will describe an iterative integration formula for the integrals. We will see that the cells that contribute to this iterative process are those which correspond to trivalent graded ribbon graphs. Analyzing their contribution is done by using a new type of graphs which we define below and name \emph{critical nodal graphs}.
It turns out that both for trivalent graded graphs, and for critical nodal graphs, the extra data of the graded spin structure can be described in an explicit combinatorial manner. In this section we shall provide this combinatorial interpretation, use it to describe the boundary conditions
and to write an explicit expression for the canonical orientations.

\subsection{Kasteleyn orientations}\label{sec:kasteleyn}
From here until the end of this subsection fix a graph $G\in\oR^0_{g,k,l},$ where $\oR^0_{g,k,l}$ was defined in Definition \ref{def:oR}.
\begin{definition}
Consider the set $\mathcal{A}$ of all assignments $H^I\to\Z_2.$
A \emph{vertex flip} is the involution $f_v:\mathcal{A}\to\mathcal{A}$ defined as follows. For $A\in\mathcal{A},~f_vA$ is the assignment which satisfies the following condition.
$f_vA(h) \neq A(h)$ if and only if exactly one of the vertices of $h,~h/s_0,s_1(h)/s_0$ is $v.$

A \emph{Kasteleyn orientation on $G$} is an assignment $\K\in \mathcal{A}$ which satisfies the following conditions.
\begin{enumerate}
\item\label{it:K1}
If $h$ belongs to a boundary edge, that is $s_1h\in H^B,$ then
$$\K(h)=1.$$
\item\label{it:K2}
For other half edge $h$
\[\K(h)+\K(s_1(h))=1.\]
\item\label{it:K3}
For every face $i,$
\[
\sum_{h\in H_i}\K(h)=1.
\]
\end{enumerate}
For convenience extend $\K$ to $H^B$ by $0,$ so that Property \ref{it:K2} holds for any half edge.
$\K(G)$ will stand for the set of all Kasteleyn orientations of $G.$
Vertex flips act on the set $\K(G).$
Two Kasteleyn orientations are \emph{equivalent} if they differ by vertex flips.
Write $[\K(G)]$ for the set of equivalence classes of Kasteleyn orientations, and $[\K]$ for the equivalence class of $\K.$
\end{definition}
\begin{obs}\label{obs:bridges_under_equivalent_kasteleyn}
Equivalent assignments give the same value to any half edge of a bridge.
\end{obs}
\begin{definition}
The \emph{legal side} of a bridge $e$ is the half edge $h\in s_1^{-1}(e)$ with $\K(h)=0.$ The other side is \emph{illegal}.
\end{definition}
The main goal of this subsection is to show that there is a natural bijection between $\oSR^0_{g,k,l}$ and $\{(G,[\K])|G\in\oR_{g,k,l}^0,[\K]\in[\K(G)]/\text{Aut}(G)\}.$

We first show how a graded structure induces an element in $[\K(G)].$
Take a graded surface $(\Sigma,\SL,s)$ whose corresponding embedded ribbon graph, defined by the JS differential, is $G.$
\begin{definition}\label{def:lifting_ribbon}
Let $v\in V^I,$ and $\{h_i\}_{i=1,2,3},$ are its three half edges, ordered so that $s_0 h_i = h_{i+1}.$
A \emph{choice of lifting for $v$} is a choice of lifts, $l_{h_i}\in\SL_v$ for the \emph{oriented} $T^1_v h_i$ (see Notation \ref{nn:T1}), for which
\begin{equation}\label{eq:choice_lift}
l_{h_{i+1}} = R_{\theta_i+2\pi}l_{h_i},~i=1,2,
\end{equation}
where $\theta_i = \measuredangle(T_v h_i,T_v h_{i+1}).$

Let $\partial\Sigma_b$ be a boundary component. Write $H_b = \{h_i\}_{i=1}^m,$ where $h_i\in H^I,$ are the half edges which are embedded in $\partial\Sigma_b,$ ordered so that $h_{i+1}=s_1(s_2^{-1}(s_1(h_i))).$ Put $v_i = h_i/s_0.$
A \emph{lifting for $\partial\Sigma_b$} is the unique choice of lifts $l_h\in\SL_{v_i}$ of $T_{v_i}^1 h,$ for any $i$ and any $h\in H_{v_i},$ which satisfies the following requirements.
\begin{enumerate}
\item
For $h=h_i\in s_1H_b, l_h=s(v_i).$
\item\label{it:lift3}
If $v_i$ is not a marked point, let $f=s_0 h_i,$ and put $\theta = \measuredangle(h_i,f).$ Then $l_f = R_{\theta+2\pi} l_{h_i},$ and $l_{s_0^{-1}h_i}=R_\pi l_{h_i}.$
\item
If $v_i$ is a marked point $l_{s_0^{-1}h_i}=R_{3\pi} l_{h_i}.$
\end{enumerate}
A \emph{choice of a lifting} is a choice of lifting for any vertex, and a lifting for any boundary component of the graph.
\end{definition}
Note that given a choice of a lifting in a vertex $v,$ \eqref{eq:choice_lift} holds also for $i=3,$ since composing \eqref{eq:choice_lift} for $i=1,2,3,$ yields \[R_{6\pi+\sum_{i=1}^3\theta_i}l_{h_1}=R_{8\pi}l_{h_1}=l_{h_1},\]
where the second passage follows from $\sum\theta_i=2\pi,$ and the last passage uses that $R_{4\pi}$ is the identity map. This also shows that a choice of a lifting for an internal vertex does not depend on the choice of which half edge is taken to be $h_1.$
In addition, note that a lifting of a boundary does not depend on choices.

Figure \ref{fig:lifting} illustrates the three types of liftings described above.
\begin{figure}
\centering
\includegraphics[scale=.4]{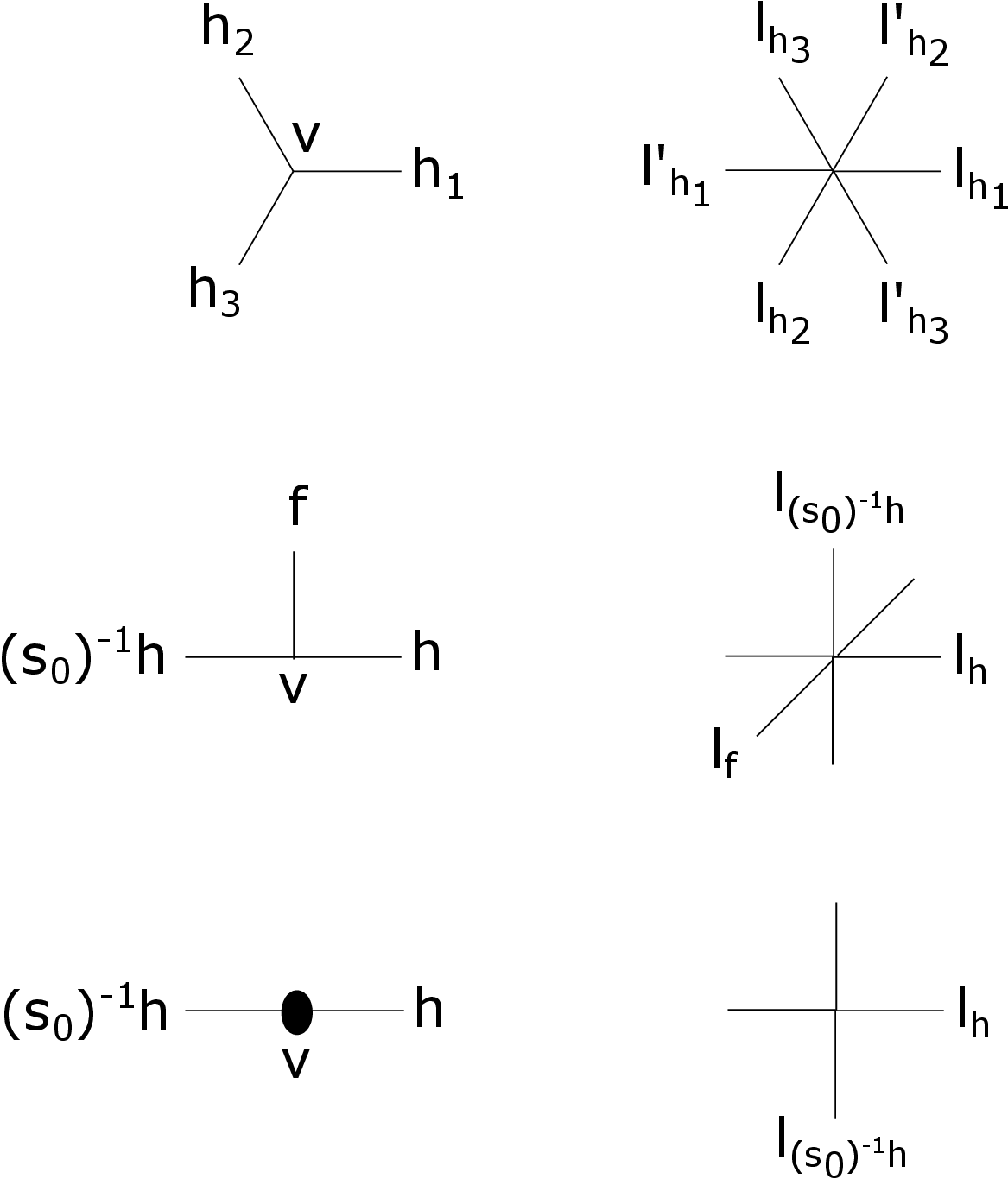}
\caption{In this figure the three types of liftings from Definition \ref{def:lifting_ribbon} are illustrated. The left column represent the local picture at the surface, while the right column represents the corresponding picture at the level of the spin fiber. Each vector on the left side has two preimages on the right side (where the angles between consecutive vectors on the right are half of those from the left). On the top row an internal trivalent vertex $v$ is drawn. For $v$ there are two possible lifts, $\{l_{h_1},l_{h_2},l_{h_3}\}$ and $\{l'_{h_1},l'_{h_2},l'_{h_3}\}.$
In the middle row $v$ is a trivalent boundary vertex and in the lowest row $v$ is a boundary marked point. In both of these cases the horizontal line in the left column represents the boundary, and in both cases $l_h$ is determined from the data of the grading, so there is no choice in the liftings, and they are as in the figure.}
\label{fig:lifting}
\end{figure}


The next observation is a consequence of the definition of the graded boundary conditions.
\begin{obs}\label{obs:bdry_edges}
Consider a lifting for the boundary $\partial\Sigma_b.$ With the above notations,
if $v_i$ is a marked point, then $l_{h_i}=R_{2\pi}P(h_{i-1})l_{h_{i-1}}.$
If $v_i$ is a boundary vertex which is not a marked point, then $l_{h_i}=P(h_{i-1})l_{h_{i-1}}.$
In both cases $R_\pi P(h_{i-1})l_{h_{i-1}}=l_{s_1(h_{i-1})}=l_{s_0^{-1}h_i}.$
\end{obs}
\begin{rmk}\label{rmk:boundary_and_arf}
Iterating Observation \ref{obs:bdry_edges} over all boundary vertices, we are led to the single constraint $l_{h_i} = R_{2 k_b \pi}l_{h_i},$ where $k_b$ is the number of boundary marked points of the boundary component $\partial\Sigma_b.$ By unwinding the alternations in boundary marked points, we see that $\q(\gamma)=k_b+1,$ for $\gamma$ a simple closed path isotopic to $\partial\Sigma_b.$
\end{rmk}

A choice of a lifting induces an assignment $\K\in\mathcal{A}$ as follows. $\K(h)=1,$ if $s_1h\in H^B.$
For an internal half edge $h,$ considered as an arc from $u$ to $v,$ we have lifts $l_h,l_{s_1(h)}$ of $T^1_u h,T^1_v s_1h$ respectively.
Now, $R_{\pi}P(h)l_h$ also covers $T^1_v s_1h,$ hence it equals either $l_{s_1(h)}$ or $R_{2\pi}l_{s_1(h)}.$ In the first case we define $\K(h) = 1,$ otherwise $\K(h)=0.$
Write $\K(\Sigma,\SL,s)$ for the set of all assignments of $G$ induced by choices of liftings.

\begin{definition}
A \emph{vertex lift flip} in a vertex $v\in V^I$ is the involution of the set of choices of lifts, which takes one choice to the choice which differs exactly in the lift at $v.$
\end{definition}

We have the following lemma
\begin{lemma}\label{lem:vertex_flip}
If $C,C'$ are two choices of lifts which differ by a vertex lift flip in $v,$ the corresponding assignments $\K,\K'$ differ by a vertex flip $f_v.$
The vertex flips act commutatively freely transitively on $\K(\Sigma,\SL,s).$ The correspondence between choices of lifts and $\K(\Sigma,\SL,s)$ is a bijection.
As a conclusion $|\K(\Sigma,\SL,s)|=2^{|V^I(G)|}.$
\end{lemma}
\begin{proof}
The first assertion, the commutativity and transitivity of the action are straightforward. The rest will from proving that the action is free.
In order to show this, note that we can think of $\K(\Sigma,\SL,s)$ as subset of $\Z_2^{H^I}.$ This is a vector space and a vertex flip $f_v$ is just an addition of an element $\widetilde f_v\in \Z_2^{H^I},$ which is $s_1-$invariant, and zero everywhere except for edges with exactly one of their ends is $v.$ Thus, we can also think of $\widetilde f_v$ as a function from $E$ to $\Z_2,$ which vanishes identically on boundary edges. In other words, $\widetilde f_v$ is canonically a $1-$cochain with coefficients in $\Z_2$ relative to boundary. In fact, if $\delta$ is the coboundary operator on the relative cochain complex defined on $\Sigma$ by the $1-$skeleton $G,$ then $\widetilde f_v = \delta e_v,$ where $e_v$ is the $0-$cochain which is $1$ only at $v.$ If the action of vertex flips were not free, there would have been a subset $A\subseteq V^I$ such that
$$\sum_{v\in A} \widetilde f_v = 0,$$
or equivalently
$$\delta\sum_{v\in A} e_v = 0,$$
so $\sum_{v\in A} e_v $ is $\delta-$closed in $H^0(\Sigma,\partial\Sigma)\simeq H_2(\Sigma)^*,$ by Poincar\'e-Lefschetz duality. But $H_2(\Sigma)=0,$ which means $A=\emptyset.$
\end{proof}

We now study $\K(\Sigma,\SL,s)$ more carefully.
\begin{prop}\label{prop:kasteleyn_in_vertices}
Fix $\K\in\K(\Sigma,\SL,s).$

For $h\in H^I.$ Put $v=h/s_0,u=(s_1 h)/s_0,~f = s_0^{-1}s_1 h,$ and, in case $u$ is not a marked point, $f' = s_0s_1 h.$ Write
$\theta = \measuredangle(P(h)T^1_{v}h,T^1_{u}f)\in(-\pi,\pi)$ and $\alpha = \measuredangle(f',f)\in(0,2\pi),$ if $u$ is not a marked point.
Let $l_h,~l_f$  denote the lifts of $T^1_{v}h,~T^1_{u}f,$ respectively, induced by $\K,$ and when $u$ is not a marked point, let $l_{f'},$ be the lift of $T^1_{u}f'.$ Finally, let $\varepsilon = \K(h).$ Then we have the following equalities,
\begin{enumerate}
\item\label{it:for_kasteleyn_lemma}
$l_f = R_{2\pi\varepsilon+\theta}P(h)l_h.$
\item\label{it:for_arf}
If $u$ is not a marked point, $l_{f'} = R_{2\pi(1+\varepsilon)+\theta-\alpha}P(h)l_h,$ and $\theta-\alpha\in(-\pi,\pi).$
\end{enumerate}

For $h\in H^B,$ from $v$ to $u,$ write $f=s_2h.$
If $u$ is a marked point then $R_{2\pi}P(h)l_h=l_f.$
If $u$ is not a marked point, write $f'=s_0s_1h$ and $\theta = \measuredangle(P(h)T^1_{v}h,T^1_{u}f')\in(-\pi,0).$
Then $P(h)l_h=l_f,~R_{\theta+2\pi}l_h=l_{f'}.$
\end{prop}
\begin{proof}
We prove for $h\in H^I.$ The proof for boundary half edges is similar and follows from Observation \ref{obs:bdry_edges}.
\begin{align*}
\K(h) = \varepsilon &\Leftrightarrow  R_{\pi}P(h)l_{h} = R_{(1+\varepsilon)2\pi}l_{s_1(h)}\\
&\Leftrightarrow  R_{\pi}P(h)l_{h} = R_{(1+\varepsilon)2\pi}(R_{2\pi + \pi-\theta} l_{f}) \\
&\Leftrightarrow  R_{\theta}P(h)l_{h} = R_{\varepsilon 2\pi}l_{f},
\end{align*}
where the equivalence in the second line follows from the definition of a choice of lift in a vertex, while the equivalence with the last line is a consequence of Remark \ref{rmk:commutation_of_trans_and_rot}. The second claim follows from $l_{f'}=R_{-2\pi-\alpha}l_f$ and the cyclic order of the half-edges.
\end{proof}

We now prove
\begin{lemma}\label{lem:kasteleyn_props}
If $\K\in\K(\Sigma,\SL,s),$ then $\K$ is a Kasteleyn orientation.
\end{lemma}
\begin{proof}
Property \ref{it:K1} of Kasteleyn orientations is just Observation \ref{obs:bdry_edges}.
Property \ref{it:K2} is reduced, thanks to Remark \ref{rmk:commutation_of_trans_and_rot} and the construction of $\K,$ to
\[
R_\pi P(s_1(h))R_\pi P(h) = R_{2\pi},
\]
but this follows from Proposition \ref{prop:arf1} applied to the piecewise smooth curve closed $h\to\bar h\to h,$ where $\bar h$ is $h$ with the opposite orientation.

For property \ref{it:K3}, let $h_1,\ldots h_m$ be an ordering of $H_i$ such that $s_2(h_j)=h_{j+1}.$ Set $v_j=h_j/s_0.$ Let $l_{h_j}$ be the lift of $T^1_{v_j}h_j$ determined by $\K,$ using Lemma \ref{lem:vertex_flip}.
Proposition \ref{prop:arf1} applied to the piecewise smooth curve $\gamma_i=h_1\to h_2\to\ldots h_m\to h_1,$ is equivalent to $P(\gamma_i)l_{h_1} = R_{2\pi}l_{h_1}.$ Put $\theta_{j+1} = \measuredangle(P(h_j)T^1_{v_j}h_j, T^1_{v_{j+1}}h_{j+1})\in(-\pi,\pi).$ Now, by Proposition \ref{prop:kasteleyn_in_vertices},
\[
R_{\theta_{j+1}}P(h_j)l_{h_j} = R_{\varepsilon_j 2\pi}l_{h_{j+1}},~~\varepsilon\in\Z_2,
\]
where $\varepsilon_j= \K(h_j).$
Iterating this equation, for $j=1,\ldots, m,$ we get
\begin{align*}
l_{h_1}=R&_{2\pi\varepsilon_m+\theta_{1}}P(h_m)R_{2\pi\varepsilon_{m-1}+\theta_{m}}P(h_{m-1})\dots R_{2\pi\varepsilon_{1}+\theta_{2}}P(h_{1})l_{h_1}=\\
&\,\,\,\,\,\,=R_{2\pi\sum_{i=1}^m\varepsilon_i}R_{\theta_{1}}P(h_m)R_{\theta_{m}}P(h_{m-1})\dots R_{\theta_{2}}P(h_{1})l_{h_1}.
\end{align*}
On the other hand, $R_{\theta_{1}}P(h_m)R_{\theta_{m}}P(h_{m-1})\dots R_{\theta_{2}}P(h_{1})=R_{2\pi(1+\q(\gamma_i))},$ by the definition of $\q.$ But $\q(\gamma_i)=0,$ since $\gamma_i$ is trivial in the homology of $\Sigma,$
so $\sum_{i=1}^m\varepsilon_i=\sum_{i=1}^m\K(h_i)$ must be odd.

\end{proof}

\begin{thm}\label{thm:Kesteleyn_is_spin}
Let $G,\Sigma$ be as above. There is a bijection between $\Spin(\Sigma),$ the set of isomorphism classes of graded spin structures on $\Sigma,$ and $[K(G)].$
\end{thm}
\begin{proof}
Given a graded spin structure $(\SL,s)$ on $\Sigma,$ we have constructed an equivalence class of Kasteleyn orientations, and this equivalence class depends only on the isomorphism type of $(\SL,s),$ so that we get a map
\[
[\K]: \Spin(\Sigma)\to [\K(G)].
\]
We shall construct a map $\Spin$ in the other direction.

Fix $\K\in\K(G).$
~We first construct the restriction of the spin bundle to $G,$ the $1-$skeleton of $\Sigma.$
For any vertex $v,$ write
\[
N_v = \cup_i\{h'_i\},
\]
where $h'_i$ are the open half edges emanating from $v,$ after removing their second endpoint. We define $\Spin(\K)|_{N_v}$ as the trivial spin cover of $T^1\Sigma|_{N_v}.$ On any fiber of $\Spin(\K)$ there is an action of $\R/4\pi\Z,$ denote it by $R_\theta.$

For a vertex $v,$ choose sections $l_{h_i}: h'_i\to \Spin(\K)|_{h'_i},$ which cover $T^1_vh_i,$ such that
for any $h_i\notin H^B,$
\[
R_{2\pi+\theta_i}(v)l_{h_i}(v) = l_{s_0(h_{i})}(v),
\]
where $\theta_i = \measuredangle(T^1_v h_i,T^1_v s_0(h_{i})).$

The transition map $g_{e',s_1(e)'}:\Spin(\K)|_{e'}\to \Spin(\K)|_{s_1(e)'}$
is given by identifying
$R_{2\K(e)\pi-\pi}l_h$ and $l_{s_1h},$ and extending using the $\R/4\pi\Z$ action.

It follows from construction, and from Property \ref{it:K3} of Kasteleyn orientations that for each $i\in[l],$ the spin structure on the boundary of face $i$ of $G,$ which is a topological disk, satisfies Proposition \ref{prop:arf1}, and hence can be extended uniquely to the face. Thus, we have constructed a spin structure on $\Sigma.$ The section $\{l_h\}_{h\in s_1 H^B},$ is evidently a grading. Call this graded spin structure $\Spin(\K).$
It can be verified easily that equivalent Kasteleyn orientations give rise to isomorphic graded spin structure, and that the maps $[\K],\Spin$ are inverse to each other.
\end{proof}

Now that we know that the data of an equivalence class of Kasteleyn orientations is equivalent to the data of a graded spin structure,
we may try to calculate $\q,\Q$ using $\K.$
\begin{definition}
Let $\gamma=(h_1\to\ldots\to h_m(\to h_1))$ be an open (closed) directed path in $G\in\oR^0_{g,k,l}$ without backtracking, that is, the directed edge $s_1 h$ cannot follow $h$ in the path. Put $v_i=h_i/s_0.$
We say that $\gamma$ makes a \emph{bad turn} at $v_i$ if either \begin{enumerate}
\item $h_{i-1}\in H^I$ and $h_i\neq s_2 h_{i-1},$ or
\item $h_{i-1}\in H^B$ and $h_i= s_0s_1h_{i-1},$
\end{enumerate}
where $i-1$ is taken modulo $m$ in the closed case. Otherwise it makes a \emph{good turn}. $BT(\gamma)$ is the number of bad turns.
\end{definition}
See figure \ref{fig:badturns} for illustrations of good and bad turns.
\begin{figure}
\centering
\includegraphics[scale=.4]{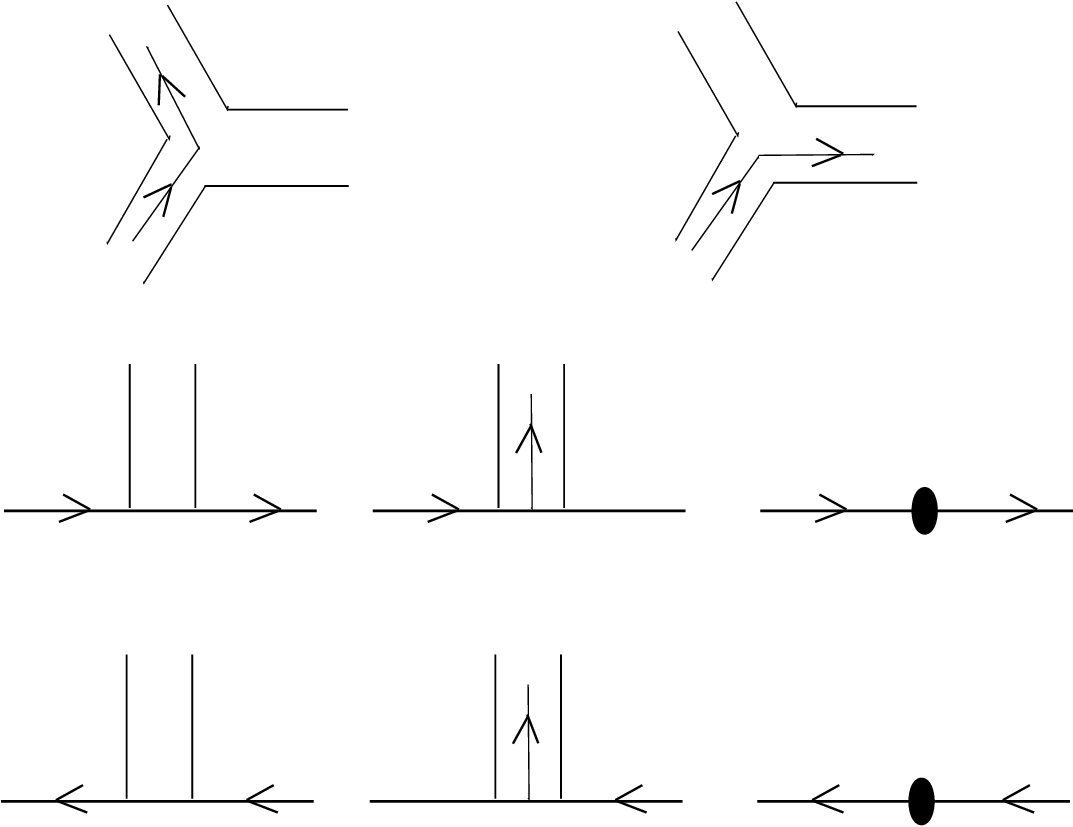}
\caption{Good and bad turns. In this figure a line with an arrow represents a half edge in a directed path, and the orientation is always counterclockwise. In the upper row an internal vertex is drawn. In the left image a good turn is drawn, while the right is a bad turn.
In the middle row the horizontal line is the boundary, and the surface lies above it. The oriented half edges in the boundary belong to $s_1 H^B.$ Only the leftmost image represents a bad turn.
In the lower row the oriented half edges in the boundary component are boundary half edges. The image on the left is a good turn, while the other two are bad.}
\label{fig:badturns}
\end{figure}
\begin{prop}\label{prop:arf}
Fix $[\K].$ With the conventions of the previous definition,
\begin{enumerate}
\item
For $\gamma$ closed,
$\q(\gamma) = \q_{\K}(\gamma):= 1+ \sum_i \K(h_i) + BT(\gamma),$ for any $\K\in[\K].$
\item
For $\gamma$ open, with $h_1,h_m\in s_1H^B,$ let $\widetilde\gamma$ be the sub arc obtained from $\gamma$ after removing small neighborhoods of its endpoints, then
$\Q(\widetilde\gamma) = \Q_{\K}(\gamma) :=1+ \sum_i \K(h_i) + BT(\gamma),$ for any $\K\in[\K].$ \end{enumerate}
\end{prop}
We defined $\widetilde\gamma$ in order to avoid marked points as endpoints.
\begin{proof}
Fix $\K\in[\K].$ Recall the correspondence between Kasteleyn orientations and lifts, \ref{lem:vertex_flip}, and take the corresponding lift $l.$
Put $\theta_{j+1} = \measuredangle(P(h_j)T^1h_j, T^1h_{j+1})\in(-\pi,\pi),$ $\varepsilon_j= \K(h_j),$ and $bt_{j+1}\in\Z_2$ is $1$ if and only if $\gamma$ makes a bad turn in $v_{j+1}.$
Proposition \ref{prop:kasteleyn_in_vertices} is equivalent, in these notations, to 
\begin{equation}\label{eq:BT}
R_{\theta_{j+1}}P(h_j)l_{h_j} = R_{(\varepsilon_j +bt_{j+1})2\pi}l_{h_{j+1}}.
\end{equation}
When $\gamma$ is closed, iterating \eqref{eq:BT}, for $j=1,\ldots, m,$ we get that
\begin{align*}
	l_{h_1} &= R_{2\pi(\varepsilon_m+bt_1)+\theta_{1}}P(h_m)R_{2\pi(\varepsilon_{m-1}+bt_m)+\theta_{m}}P(h_{m-1})\cdots \\
    &\;\quad\cdots R_{2\pi(\varepsilon_{1}+bt_2)+\theta_{2}}P(h_{1})l_{h_1}\\
    &=R_{2\pi\sum_{i=1}^m\varepsilon_i+bt_i}R_{\theta_{1}}P(h_m)R_{\theta_{m}}P(h_{m-1})\dots R_{\theta_{2}}P(h_{1})l_{h_1}\\
    &=R_{2\pi (BT(\gamma)+\sum_{i=1}^m\varepsilon_i)}R_{(1+\q(\gamma))2\pi}l_{h_1}=R_{2\pi(\q(\gamma)+1+BT(\gamma)+\sum_{i=1}^m\varepsilon_i)}l_{h_1},
\end{align*}
where 
the one before last passage uses the definition of $\q,$ Definition \ref{def:spin_using_q}.

Similarly, when $\gamma$ is open, iterating \eqref{eq:BT} over $j=1,\ldots, m-1,$ and applying the same reasoning, this time using Definition \ref{def:Q}, we obtain
\[l_{h_m}=R_{2\pi(BT(\gamma)+\sum_{i=1}^{m-1}\varepsilon_i+\Q(\gamma))}l_{h_1}=R_{2\pi(1+BT(\gamma)+\sum_{i=1}^{m}\varepsilon_i+\Q(\gamma))}l_{h_1},\]
where we used $\varepsilon_m=\K(h_m)=1.$ As needed.
\end{proof}
\begin{rmk}
The first case of the proposition appeared before in \cite{Cimasoni}.
Note that although the formula depends on the orientation of $\gamma,$ the result is orientation independent in the closed case. Indeed, flipping the orientation changes each $\K(h)$ to $\K(s_1 h) = \K(h)+1,$ and interchanges the                                 
sets of good turns and of bad turns. Thus, the total change is the number of edges plus the number of vertices of $\gamma,$ that is, a change by $2m=0.$                                                                                                              
A similar argument shows that in the open case the result changes by $1$ when the orientation is flipped.                                                                                                                                                              
\end{rmk}                                                                                                                                                                                                                                                             

\begin{definition}
An automorphism $\phi:G\to G$ defines an action, $\phi_*$ on $\K(G),[\K(G)]$ by
\[
(\phi_*\K)(h)=\K(\phi^{-1}(h)).
\]
An automorphism $\phi$ of $(G,[\K])$ is an automorphism $\phi$ of $G$ for which $\phi_*[\K]=[\K].$ We write $\text{Aut}(G,[\K])$ for the group of these automorphisms.
\end{definition}
\begin{prop}\label{prop:cells_and_kasteleyn}
For any $G\in\oSR_{g,k,l}^0,$ the map \[\coprod_{z\in Z_{G}/\text{Aut}(G)}\CM_{(G,z)}\to\coprod_{[\K]\in[\K(G)]/\text{Aut}(G)}\R_+^{E(G)}/\text{Aut}(G,[\K]),\]
which takes a metric graded graph $(G,z,\ell)$ to $([\K],\ell),$ where $[\K]$ is the Kasteleyn orientation associated to the graded spin structure of $\comb^{-1}(G,z,\ell)$ is a homeomorphism.
\end{prop}
\begin{proof}
It is enough to show that along a path $(\Sigma_t)_{0\leq t \leq 1}$ in $\comb^{-1}(\CM_{(G,z)}),$ the equivalence classes $[\K_t]=[\K_t(\Sigma_t,\SL_t,s_t)]\in[\K(G)]$ are the same.
Take $\K_0\in[\K(\Sigma_0,\SL_0,s_0)].$ This determines the maps $\Q_0,\q_0$ by Proposition \ref{prop:arf}, and the fact that any piecewise smooth path may be isotoped to a non backtracking one on the $1-$skeleton $G\hookrightarrow\Sigma_0.$ Now, varying $(\Sigma_t,\SL_t,s_t)$ is equivalent to varying the metric $\ell_t$ on $G,$ in the component $\CM_{(G,z)}$ continuously. But then it is evident that the maps $\Q_t,\q_t$ determined by $\K_0$ on the paths in resulting embedded graph do not change.
By Lemma \ref{lem:Q_q_and_bridges} we see that $[\K_t]=[\K_0].$
\end{proof}
In light of Proposition \ref{prop:cells_and_kasteleyn}, we can redefine $\oSR^0$ and the related combinatorial moduli spaces.
\begin{nn}
From now on we write \[\oSR_{g,k,l}^0=\{(G,[\K])|G\in\oR^0_{g,k,l},[\K]\in[\K(G)]/\text{Aut}(G)\}.\]
Define $\CM_{(G,[\K])}=\R_+^{E(G)}/\text{Aut}(G,[\K]),$ the moduli of metrics on $G,$ together with a fixed equivalence class of Kasteleyn orientations. $\CM_{(G,[\K])}\hookrightarrow\oCM_{(G,z)},$ for a unique $z\in Z_G,$ as in Proposition \ref{prop:cells_and_kasteleyn}. We therefore set $\oCM_{(G,[\K])}=\oCM_{(G,z)}.$
Define analogously $\CM_{(G,[\K])}(\pp),$ and $\oCM_{(G,[\K])}(\pp).$
%
\end{nn}

\begin{ex}\label{ex:macheta}
Fix a connected component $C$ of $\oRCM_{g,k,l}.$ 
Suppose that smooth surfaces in $C$ have $b$ boundary components and write $g_s = \frac{g-b+1}{2}.$ Let $k_j,~j=1,\ldots,b,$ be the number of boundary marked points on boundary component $j,$ for some locally defined numbering of the boundary components.
One ribbon graph which corresponds to surfaces in $C$ is the following graph $G\in\oR^0_{g,k,l},$ see also Figure \ref{fig:macheta}
\begin{align*}
V =& \{v_{j,j+1}^-\}_{j=2,\ldots,b}\cup\{v_{j,j+1}^+\}_{j\in[b-1]}\cup\{p_{j,i}\}_{j\in[b],i\in[k_j]} \\
&\cup\{v_i^\pm\}_{i=2,\ldots,l}\cup\{u^\pm_i,w^\pm_i\}_{i\in[g_s]}.
\end{align*}
Only $v_i^-$ are internal vertices, the vertices $p_{j,i},v_{j,j+1}^+,v_{j-1,j}^-$ belong to the $j^{th}$ boundary component. The other boundary vertices belong to the first boundary.
\[
H^I = \bigcup_{i\in[b]}H_{\text{bdry},i}\cup H_{\text{bridges}}\cup H_{\text{genus}}\cup H_{\text{internal marked}},
\]
where
\begin{enumerate}
\item
For $j\neq 1,~H_{\text{bdry},j} = \{e_{j,i}\}_{0\leq i\leq k_j+(1-\delta_{jb})}$ are the boundary edges of boundary component $j$ and of face $1.~e_{j,i}/s_0=p_{j,i}$
for $1\leq i\leq k_j.$
In addition, $e_{j,0}/s_0=v^+_{j,j+1},~(s_1e_{j,0})/s_0=p_{j,1}.$
For $j\neq b,1~e_{j, k_j}$ connects $p_{j, k_j}$ to $v_{j-1,j}^-,$ and
$e_{j, k_j+1}/s_0=v_{j-1,j}^-,$\\$s_1(e_{j, k_j+1})/s_0=v_{j,j+1}^+.$
For $j=b,~e_{j, k_j}/s_0=v_{b-1, b}^-.$
They are ordered so that $e_{j,i+1}=s_2' e_{j,i},$
where $s'_2(e) := s_1(s_2^{-1}(s_1(e))),$ for $e\in s_1 H^B.$
\item
\[H_{\text{bdry},1}=a_1,b_1,c_1,d_1,a_2,\ldots,d_{g_s},h_2,\ldots,h_{l},e_{10},e_{1,1},\ldots,e_{1,k_1}\] are the boundary edges of the first boundary, which all belong to face $1,$ ordered by $s'_2$ order.
The boundary vertices, in counterclockwise order starting from $v_{1,2}^+,$ the vertex of the bridge, are
\[v_{1,2}^+,u_1^+,w_1^+,u_1^-,w_1^-,u_2^+,\ldots,w_{g_s}^-,v_2^+,\ldots,v_l^+,p_{1,1},\ldots,p_{1,k_1}.\]
The adjacency relation is therefore
$a_1/s_0 = v_{1,2}^+.~a_i/s_0 = w_{i-1}^-,$ for $i>1,$ while $b_i/s_0 = u_i^+, c_1/s_0 = w_i^+,d_1/s_0 = u_i^-.$
Next, $h_2/s_0 = w_{g_s}^-, h_i/s_0 = v_{i-1}^+,$ for $i>1.$
Finally $e_{1,0}/s_0 = v_l^+,$ and for $i>0,~e_{1,i}/s_0 = p_{1,i}.$
\item
$H_{\text{bridges}} = \{b_{j,j+1}, \bar b_{j,j+1}\}_{j\in[b-1]}$ is the set of bridges between consecutive boundaries.
$$b_{j,j+1}/s_0 = v_{j,j+1}^+, ~\bar b_{j,j+1} = s_1 b_{j,j+1}, ~\bar b_{j,j+1}/s_0 = v_{j,j+1}^-.$$
\item
$H_{\text{genus}} = \{f_i,\bar f_i,g_i,\bar g_i\}_{i\in g_s}$ is a set of internal half edges of face $1,$ such that $f_i$ goes from $u_i^+$ to $u_i^-,\bar f_i = s_1 f_i,$ and $g_i$ goes from $w_i^+$ to $w_i^-,\bar g_i = s_1 g_i.$
\item
$H_{\text{internal marked}} = \{x_i,\bar x_i,y_i, \bar y_i\}_{i=2,\ldots,l},$ is the following set. $y_i$ is the unique edge of face $i,~y_i/s_0=v_i^-,$ and $\bar y_i = s_1y_i.$
The third half edge of $v_i^-$ is $x_i,$ and $\bar x_i = s_1 x_i,~\bar x_i/s_0 = v_i^+.$
\end{enumerate}
\begin{figure}
\centering
\includegraphics[scale=.4]{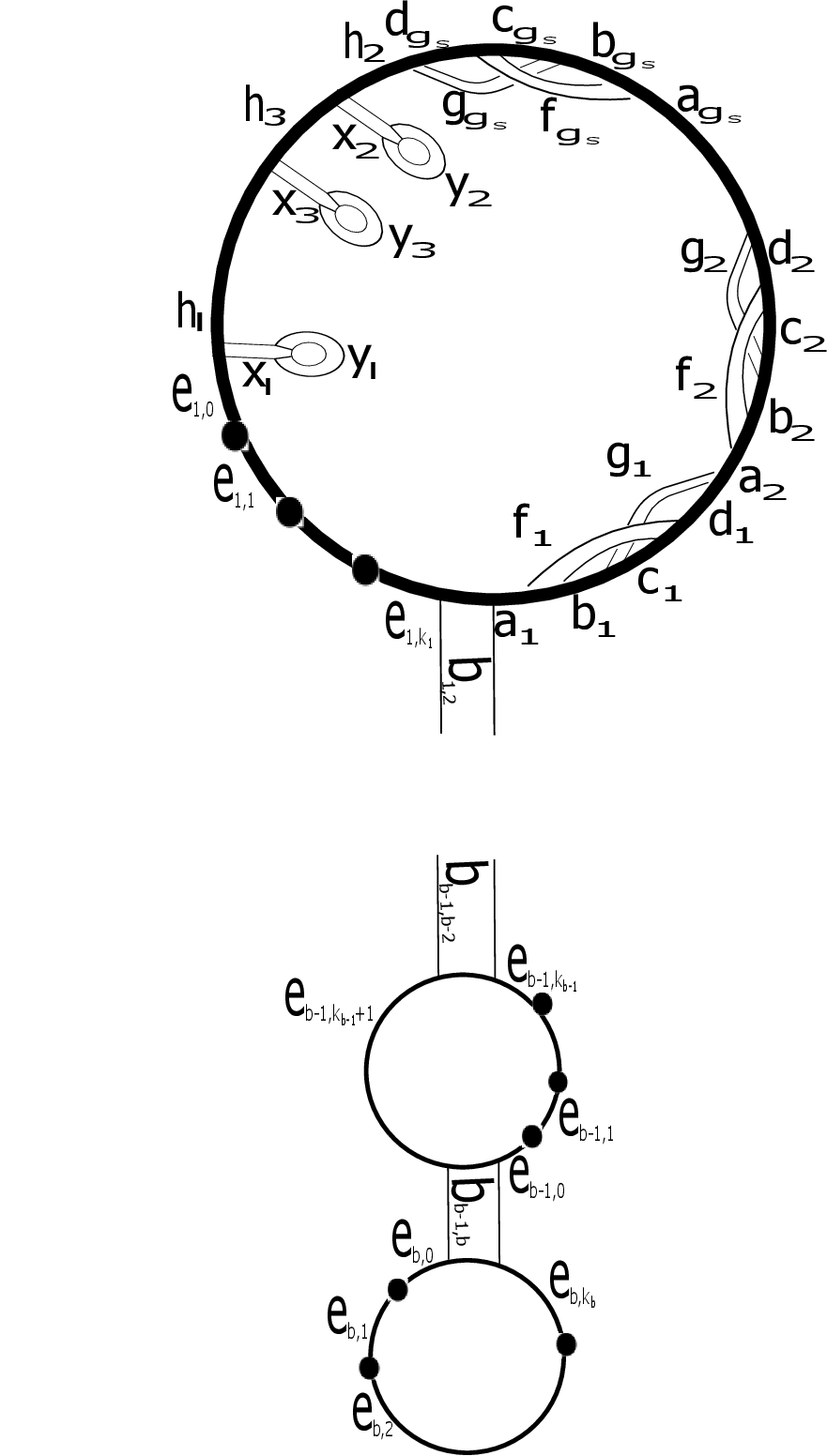}
\caption{}
\label{fig:macheta}
\end{figure}
We now describe $\K(G).$ First of all, $\K(h) = 1$ if $s_1h\in H^B$ or $h=y_i.$ There is no constraint on $\K(x_i),$ but different values are equivalent by flips in $v_i^-.$ Since there are no more internal vertices, for all other edges there are no constraints and no relations. Thus there is a total number of $2^{2g_s+b-1}=2^g$ different graded spin structures in this case. Since this is a topological invariant, for any generic open genus $g$ surface in $C$ there are $2^g$ graded structures. Thus, for any generic open genus $g$ surface which satisfies \eqref{eq:restrictions} there are $2^g$ graded structures.
\end{ex}
\begin{rmk}\label{rmk:machetta}
In \cite{ST0} a notion of parity is defined for smooth graded surfaces with an odd number of boundary point for each component.
It is defined as follows. Given such a graded surface $(\Sigma,\SL,s),$
choose a symplectic basis $\{\alpha_i,\beta_i\}_{i\in[g_s]}$ to $H_1(\Sigma,\Z_2)/H_0(\partial\Sigma,\Z_2).$ The quadratic form $\q$ factors through this quotient. Define $\text{Arf}(\Sigma) = \sum q(\alpha_i)q(\beta_i)(\text{mod } 2).$ This is an isotopy invariant. A spin structure is said to be even if the Arf is $0,$ otherwise it is odd. This notion is generalized, also in \cite{ST0}, to give the \emph{open Arf invariant}, which is defined for any graded surface, and specializes to the parity if there is an odd number of marking on each boundary.

For example, with the notation of Example \ref{ex:macheta}, suppose that each $k_j$ is odd. A possible choice for the symplectic basis is
\[
\alpha_i = b_i\to c_i\to\bar f_i\to b_i,~~\beta_i = c_i\to d_i\to \bar g_i\to c_i.
\]
Now, by Proposition \ref{prop:arf}, \[q(\alpha_i) = 1+ \K(b_i)+\K(c_i)+\K(\bar f_i)+BT(\alpha_i) = \K(\bar f_i),\]
since there is one bad turn. Similarly, $q(\beta_i) = \K(\bar g_i).$
Therefore,
\[
\text{Arf}(\Sigma) = \sum_{i\in [g_s]}\K(\bar f_i)\K(\bar g_i).
\]
A simple calculation now shows that the difference between even and odd spin gradings in this case is $2^{g_s+b-1}=2^{\frac{g+b-1}{2}}.$ 
\end{rmk}

\begin{rmk}
Kasteleyn orientation are named after W. Kasteleyn, who used them to analyze dimer statistics, see for example \cite{Kastel}.
The connection between Kasteleyn orientations and spin structures on closed surfaces is obtained in \cite{Kuper,Cimasoni}.
\end{rmk}
\subsubsection{Adjacent Kasteleyn orientations}
Recall Construction \ref{cons:z_e}. By Proposition \ref{prop:G_and_G_e}, in the cell structure of $\oCM_{g,k,l}^{\text{comb}},$
the cell $(G,[\K])$ is adjacent to cells of the form
$(G_e,[\K_e]),$ for some edge $e\notin\text{Br}(G)\cup\text{Loop}(G),$ $[\K_e]\in[\K(G_e)].$
We now describe $[\K_e]$ explicitly in terms of $[\K].$

Fix a Kasteleyn orientation $\K\in[\K].$
Write $h$ for the unique half edge with $\K(h)=1,h/s_1=e.$
\begin{figure}
\centering
\includegraphics[scale=.4]{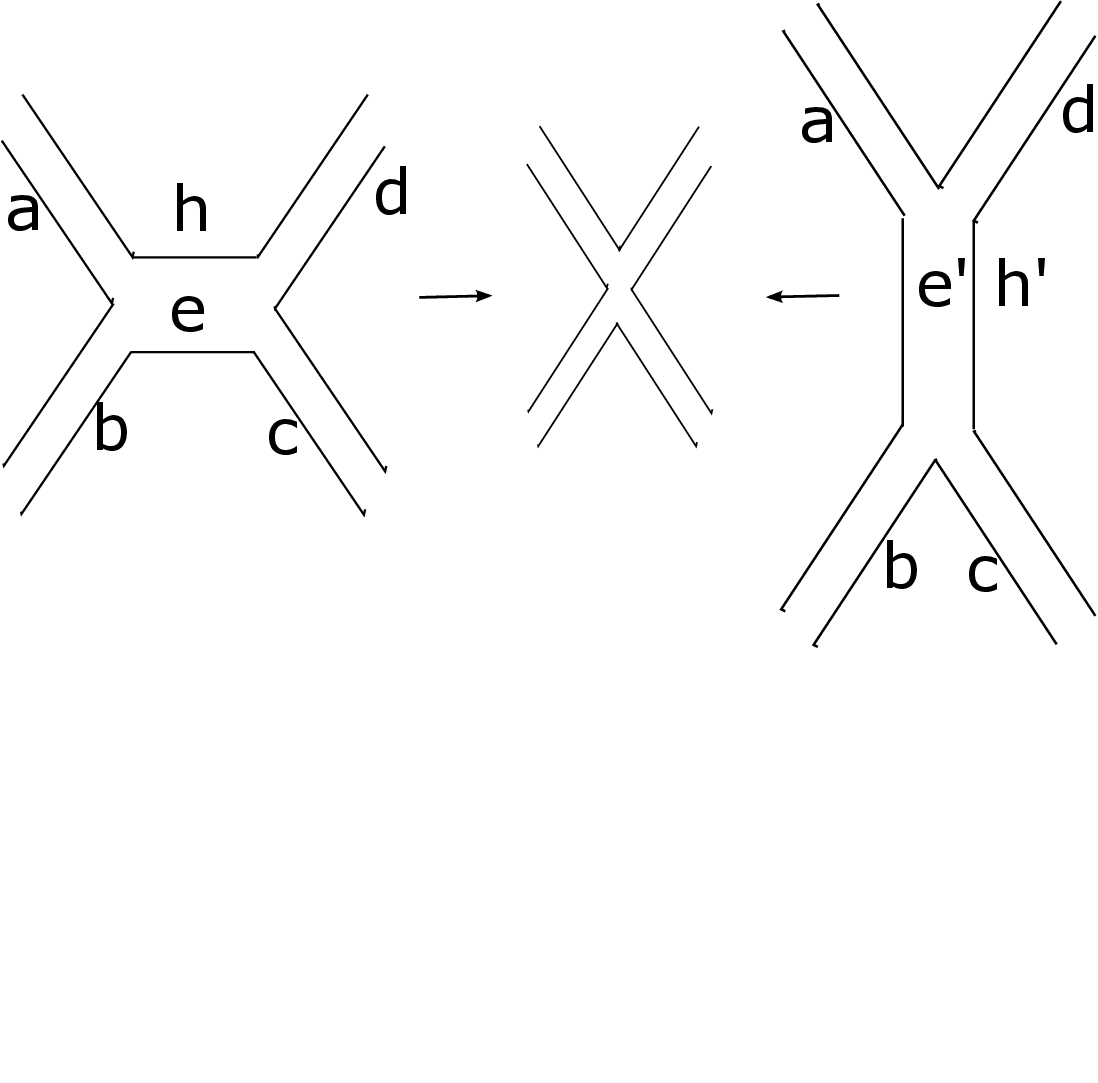}
\caption{$G,\partial_eG,$ and $G_e.$ The middle graph is $\partial_e G.$ We draw an half edge inside the face which contains it.}
\label{fig:neighboring}
\end{figure}
Write $a = s_0(h),b=s_0^2(h),c=s_1(s_0(s_1(h))),d=s_1(s_0^2(s_1(h))),$ see Figure \ref{fig:neighboring}. For shortness write $\bar{x}$ for $s_1(x).$ Apart from some borderline cases which may be treated separately, we may assume all these vertices and half edges are distinct, and then, using vertex flips if needed, we may also restrict ourself to the case where $\K(\bar{d})=1.$
Note that $E(G)\setminus e=E(G_e)\setminus e'$ canonically, for some $e'\in E(G_e).$
We therefore identify these sets, and also identify $H(G)\setminus\{h,s_1h\}$ and $H(G_e)\setminus s_1^{-1}e'.$
In $G_e,$ let $v'_1$ be the vertex from which $a,\bar d$ issue, and $v'_2$ be the vertex from which $b,\bar c$ issue. We may take the half edge $h'$ to be the third half edge from $v'_1.$
Define the assignment $\K':H^I(G)\to\Z_2$ by
\[
\K'(h')=1, ~\K'(\bar h') = 0,~\K'(d) = \K(d)+1=1,~\K'(\bar d) = \K(\bar d)+1=0,
\]
and $\K'(f)=\K(f)$ for any other half edge $f.$

For later purposes, define, for a boundary loop $e,$ and a Kasteleyn orientation $\K\in[\K],$ an assignment $\K',$ by $\K'(h)=\K(h)$ for any $h$ with $h/s_1\neq f,$ where $f$ is the unique edge which shares a vertex with $e,$ and otherwise $\K'(h)=\K(h)+1.$
\begin{lemma}\label{lem:adjacent_spin_cells}
In both cases, $\K'\in[\K(G_e)],$ and moreover $\K'\in[\K_e].$
\end{lemma}
\begin{proof}
The claim is straight forward when $e$ is a boundary loop. 
Suppose $e\notin\text{Br}(G)\cup\text{Loop}(G).$ The first assertion is simple, we focus the second one.
Write $C(G),C(G')$ for the set of closed paths without backtracking in $G,G'$ respectively.
Write $O(G),O(G')$ for the set of open directed paths without backtracking in $G,G'$ respectively, which connect boundary vertices which are not marked points.
We have bijections $f_C:C(G)\to C(G'), f_O:O(G)\to O(G')$ defined as follows.
For a path $(e_1\to e_2\to\ldots e_m)\in C(G),$ the path $f_C(e_1\to e_2\to\ldots e_m)\in C(G')$ is defined by erasing any appearance of $e$ in the sequence and adding $e'$ any time we have a move $f\to f'$ where the third edge of the vertex between $f$ and $f'$ is $e.$
The inverse map is defined similarly, but with changing the roles of $e,e'.$ The map $f_O$ is defined in the same way.

Using Proposition \ref{prop:arf} it is straightforward to verify that for any $\gamma\in C(G),~\q_{\K}(\gamma)=\q_{\K'}(f_C(\gamma)),$ and for any $\gamma\in O(G),~\Q_{\K}(\gamma)=\Q_{\K'}(f_C(\gamma)).$

Now, let $(\Sigma_t,\SL_t,s_t)_{t\in[0,1]}$ be a continuous path in $\oCM_{g,k,l}^{\text{comb}},$ with
$(\Sigma_t,\SL_t,s_t)\in\comb^{-1}(\CM_{(G_t,z_t)}),$ where
\[
G_t=\begin{cases}
G,&\mbox{if }t<\frac{1}{2},\\
\partial_eG,&\mbox{if }t=\frac{1}{2},\\
G',&\mbox{if }t>\frac{1}{2},
\end{cases}
\]
and the graded structure $z_0\in Z_{G}$ corresponds to the Kasteleyn orientation $[\K].$
In light of Lemma \ref{lem:Q_q_and_bridges}, Proposition \ref{prop:cells_and_kasteleyn} and isotopy arguments, the Kasteleyn orientation on $G'$ defined by $(\Sigma_t,\SL_t,s_t)_{t\in(\frac{1}{2},1)}$ is the unique class of Kasteleyn orientation for which for any continuous family $(\gamma_t\subseteq\Sigma_t)$ of closed paths or bridges,
$\q(\gamma_t),$ or $\Q(\gamma_t)$ is constant.
By performing an isotopy, we may assume that $\gamma_t$ is in fact a path in the graph $G_t.$
It is easy to see that for $\varepsilon$ small enough, $f_C(\gamma_{\frac{1}{2}-\varepsilon})=\gamma_{\frac{1}{2}+\varepsilon},$ in case $\gamma_t$ are closed, or
$f_O(\gamma_{\frac{1}{2}-\varepsilon})=\gamma_{\frac{1}{2}+\varepsilon},$ in case they are open.
In the first case, $\q_{[\K]}(\gamma_{\frac{1}{2}-\varepsilon})=\q_{[\K']}(\gamma_{\frac{1}{2}+\varepsilon}),$ while in the second the same equation holds for $\Q.$ By Lemma \ref{lem:Q_q_and_bridges}, part \ref{it:qQ_determine_spin}, and Theorem \ref{thm:Kesteleyn_is_spin},
the graded structure $z_{t},t>\frac{1}{2}$ must correspond to $[\K'].$
\end{proof}

\subsubsection{Trivalent graphs}\label{subsec:codim_geq_1}
\begin{definition}
Recall Definition \ref{def:ghost_effective_trivalent}.
Let $G$ be a trivalent graph. Recall that a half node is a $(N^{B})^{-1}-$preimage of a node, and that their collection is denoted $\HN(G).$
An \emph{extended Kasteleyn orientation} on $G$ is a map $\K:H(G)\cup \HN(G)\to\Z_2,$ which satisfies
\begin{enumerate}
\item
For any $h\in H^B,~\K(h)=0.$
\item
For any $h\in H,~\K(h)+\K(s_1h)=1.$
\item
For any node $v,$ if $|N^{-1}(v)|=3,$ then $\K|_{N^{-1}(v)}=1.$ Otherwise $\K(v_{i,1})+\K(v_{i,2})=1,$ where $N^{-1}(v)=\{v_{i,1},v_{i,2}\}.$
\item
For any face $f,~\sum\K(x)=1,$ where the variable $x$ is taken from the set of half edges with $x/s_2=f,$ together with the set of half nodes which belong to $f.$
\end{enumerate}
Two extended Kasteleyn orientations are equivalent if they differ by the action of internal vertex flips. Write $[\K]$ for the equivalence class of $\K.$
Define $\K(G),[\K(G)]$ as the sets of extended Kasteleyn orientations and the set of equivalence classes of extended Kasteleyn orientations. Write $\text{Aut}(G,[\K])$ as the automorphism subgroup of $G$ which preserves $[\K].$
\end{definition}
Item (c) above deals with the case $v$ is a contracted component, whose normalization contains at least three half nodes. In the trivalent case this can only happen if the unique contracted component in $\Norm^{-1}(v)$ is a ghost, and its three marked points are legal. Therefore there are exactly three corresponding half nodes in the non contracted parts, and they are illegal.

With the exact same techniques of Subsection \ref{sec:kasteleyn}, together with Corollary \ref{cor:spin_on_comps}, we obtain
\begin{lemma}\label{lem:equivalence_kasteleyn_critical}
For a trivalent $G$ and a metric $\ell,$ there is a natural bijection between $[\K(G)]$ and $\Spin((\Rcomb)^{-1}(G,\ell)).$
The induced map \[\coprod_{z\in Z_{G}/\text{Aut}(G)}\CM_{(G,z)}\to\coprod_{[\K]\in[\K(G)]/\text{Aut}(G)}\R_+^{E(G)}/\text{Aut}(G,[\K]),\] is a homeomorphism.
In particular, $Z_{G}\simeq[\K(G)]$ canonically.
A half node $v$ in $(G,z)$ is \emph{illegal} if and only if $\K(v)=1$ for any $\K\in[\K]$ which corresponds to $z.$
\end{lemma}
From now on we denote trivalent graphs $(G,z)$ by $(G,[\K]),$ for the corresponding $[\K]\in[\K(G)].$
\begin{definition}
Define $\CM_{(G,[\K])}=\R_+^{E(G)}/\text{Aut}(G,[\K]),$ the moduli of metrics on $\CM_{G},$ together with a fixed equivalence class of Kasteleyn orientations.
Define $\oCM_{(G,[\K])}=\oCM_{(G,z)},$ for the unique $z$ which corresponds to $[\K]$ by the above lemma.
For $f_1,\ldots,f_s\in E(G),$ set $\partial_{f_1,\ldots,f_s}\oCM_{(G,[\K])}$ as the face of $\oCM_{(G,[\K])}$ defined by setting the coordinates $\ell_{f_1},\ldots,\ell_{f_s}$ to $0.$ For $p_1,\ldots,p_l>0$ define $\CM_{(G,[\K])}(\pp),~\oCM_{(G,[\K])}(\pp)$ by setting the perimeters to these values.
\end{definition}

Suppose $G$ is a trivalent graph $\K\in\K(G),$ and $e\in \text{Br}(G).$ In case $e$ is a boundary edge, let $h_1$ be its internal half edge, $h_1/s_1=e,~h_1\in H^I.$
In case $e$ is an internal edge, write $s_1^{-1}(e)=\{h_1,h_2\},$ where $\K(h_i)=i(\text{mod } 2).$
Define $\partial_e\K$ to be the unique map $\partial_e\K:H(\partial_eG)\cup \HN(\partial_eG)\to\Z_2,$ which agrees with $\K$ on any half edge $h'\notin s_1^{-1}e,$
and $\partial_e\K(\partial_e h_i)=i(\text{mod } 2).$
In a similar way, one can define $\partial_{e_1,\ldots,e_r}\K$ for a compatible sequence of bridges.
%
%

It is straightforward that
\begin{obs}\label{obs:degenerating_kasteleyn_comb}
For any trivalent $(G,[\K]),$ and a bridge $e,$ the graph $(\partial_eG,[\partial_e\K])$ is a well defined trivalent graph, in particular $\partial_e\K\in[\K(\partial_eG)].$
Moreover, $\partial_e:[\K(G)]\to[\K(\partial_e G)]$ is a bijection.

In addition, for any trivalent connected graph $(G,[\K]),$ there is a unique smooth trivalent $(G',[\K']),$ and a unique, up to order, compatible sequence of bridges $e_1,\ldots,e_r$
with $(G,[\K])=\partial_{e_1,\ldots,e_r}(G',[\K']).$
\end{obs}
With the same techniques of the proof of Lemma \ref{lem:adjacent_spin_cells} one obtains
\begin{lemma}\label{lem:degenerating_kasteleyn}
Let $G$ be a trivalent graph, $e_1,\ldots,e_r$ a compatible sequence of bridges. Under the identification of Lemma \ref{lem:equivalence_kasteleyn_critical} between $Z_H,[\K(H)],$ for $H=G,\partial_{e_r}G,\ldots,\partial_{e_1,\ldots,e_r}G,$
\[
\overline{\CM}_{\partial_{e_1,\ldots,e_r}(G,[\K])}\simeq
\partial_{e_1,\ldots,e_s}\oCM_{\partial_{e_{s+1},\ldots,e_r}(G,[\K])},
\]
canonically.
\end{lemma}
In what follows we shall identify $\oCM_{(G,z)}$ and the corresponding $\oCM_{(G,[\K])}$ without further notice.
%
%
%
%

\subsection{Orientation}\label{sec:or}
In this subsection we construct an orientation to $\oCM_{g,k,l}^{\text{comb}}.$
We do it by writing an explicit formula for the orientation of each highest dimensional cell of $\oCM_{g,k,l}^{\text{comb}}(\mathbf{p}),$ that is, for cells $\CM_{(G,[\K])}(\pp)$ where $G\in\oR^0,[\K]\in[\K(G)],$ and then showing that on codimension $1$ faces between two such cells, the induced orientations disagree. We also discuss the induced orientation on the boundary, and prove that these orientations are the ones induced from $\oCM_{g,k,l}$ by $\comb_*.$

For $G\in\oR^{0}_{g,k,l},$ we have a map
\begin{equation}\label{eq:or_exact_seq}
A_G:\R_+^{E(G)}\to\R^{F(G)}=\R^{[l]},
\end{equation}
which takes as input a collection of edge length and outputs the face perimeters. \[\CM_{(G,[\K])}(\mathbf{p}) = A_G^{-1}(\mathbf{p})/{\text{Aut}(G,[\K])}.\]
In particular, orienting $\CM_{(G,[\K])}$ is equivalent to orienting $\text{ker}(A_G)/{\text{Aut}(G,[\K])}.$
Using the exact sequence
\begin{equation}\label{eq:induced_or}
0\to \text{ker}(A_G)\to\R^{E(G)}\to\R^{F(G)}=\R^{[l]}\to 0,
\end{equation}
we see that orienting $\R^{E(G)},\R^{[l]},$ or equivalently, ordering $E(G),[l],$ up to even permutations, gives an orientation to $\CM_{(G,[\K])}(\mathbf{p}),$
as long as the action of $\text{Aut}(G,[\K])$ preserves the orientation.

Fix any order for $[l],$ for example $1,2,\ldots, l.$ Choose any Kasteleyn orientation $\K\in [\K].$ Define $\mathfrak{o}_i = \mathfrak{o}_{(G,\K,i)}$ by
\[
\bigwedge_{\K(h)=1, h/s_2=i} d\ell_h,
\]
that is, we take the wedge of $d\ell_h$ over half edges $h$ of face $i,$ with $\K(h)=1.$ The wedge is taken counterclockwise. Because there is an odd number of half edges of the $i^{th}$ face with $\K=1,$ the element $\mathfrak{o}_i$ is well defined, and independent on which half edge appears first. In addition, $\mathfrak{o}_i$ is an odd degree form.

\begin{definition}
Choose any Kasteleyn orientation $\K.$ Put \[\mathfrak{o}_{(G,\K)}= \bigwedge_{i=1}^l\mathfrak{o}_i.\]
Define $\bar{\mathfrak{o}}_{(G,\K)}$ as the orientation on $\text{ker}(A_G)$ induced from exact sequence \eqref{eq:induced_or}, when $\R^{E(G)}$ is oriented by $\mathfrak{o}_{(G,\K)},$ and $\R^{[l]}$ by $\bigwedge_{i=1}^l dp_i.$
\end{definition}
\begin{rmk}
Because both $dp_i$ and $\mathfrak{o}_i$ are odd variables, choosing another order on $[l]$ does not change $\bar{\mathfrak{o}}_G.$
\end{rmk}
\begin{lemma}\label{lem:or_indep}
$\bar{\mathfrak{o}}_{(G,\K)}$ depends only on $[\K].$
\end{lemma}
Before we get to the proof, we add a few auxiliary definitions.
\begin{definition}\label{def:good_order}
Let $G$ be any open ribbon graph.
A \emph{good ordering} is a bijection $n:H^I\to|H^I|,$ which satisfies the following properties. First, if $i(h) < i(h'),$ that is $h$ belongs to face marked $i,$ and $h'$ to face marked $i'>i,$ then $n(h)<n(h').$ Thus, half edges of the same face are clustered together. Second, the ordering $n,$ when restricted to half edges of a single face, agrees with the counterclockwise ordering.

Let $n$ be a good ordering, as in Definition \ref{def:good_order}, and $\K\in\K(G)$ a Kasteleyn orientation. Define $H_\K = \{h\in H^I| \K(h)=1\}.$
We also define $n_\K:|H^I|\to \Z$ by
\[
n_\K(i) = |\{h\in H_\K|n(h)<i\}|.
\]

\end{definition}
\begin{figure}
\centering
\includegraphics[scale=.4]{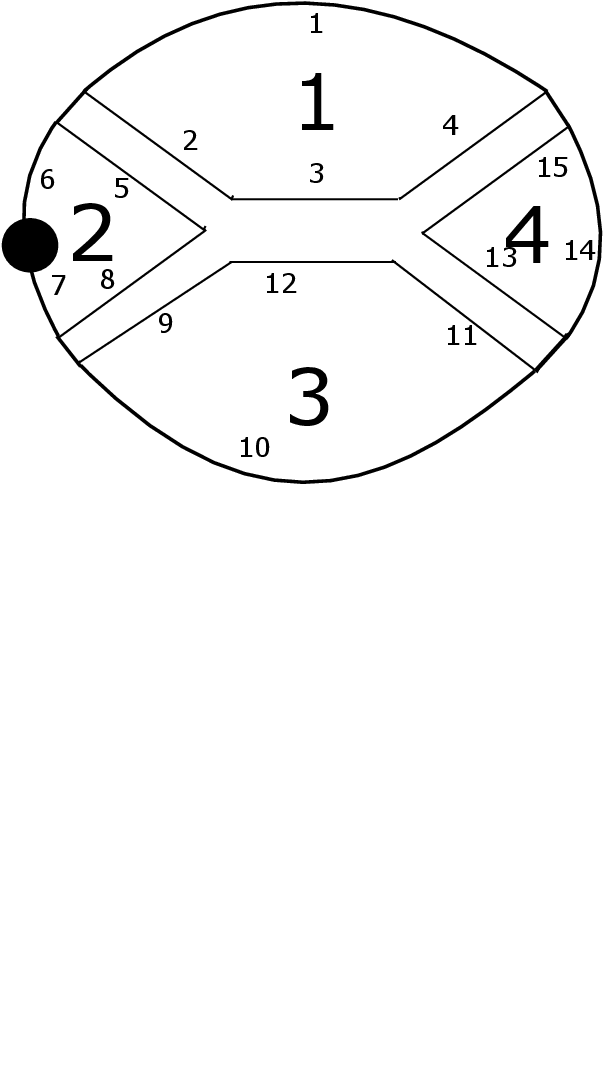}
\caption{A good ordering of internal half edges: the large numbers in the middles of the faces are the labels of the faces, the smaller ones next to the half edges are the half edges numbers in the ordering. The numbers of half edges in face $i$ are smaller than those of face $j$ if $i<j.$ In each face the numbers of half edges agree with the cyclic order induced by the face's orientation.}
\label{fig:good_order_fig}
\end{figure}
Figure \ref{fig:good_order_fig} illustrates a good ordering.
Note that the restriction of a good ordering to a subset of $H^I$ induces an order on its elements.
\begin{proof}[Proof of Lemma \ref{lem:or_indep}]
Take any $\K\in[\K].$ We recall from Lemma \ref{lem:vertex_flip} that any other element of $[\K]$ can be obtained from $\K$ by successive flips in vertices. It will thus be suffices to prove that the orientations induced by $\K,\K'$ are the same when $\K,\K'$ differ by a single flip in vertex $v.$ It will be enough to prove that $\mathfrak{o}_{(G,\K)}=\mathfrak{o}_{(G,\K')}$

Fix a good ordering $n.$ By definition $$\mathfrak{o}_{(G,\K)} = \bigwedge_{e\in H_\K} d\ell_e,$$ where the order of the wedging is the order $n$ restricted to $H_\K.$
The sign difference between $\mathfrak{o}_{(G,\K)},\mathfrak{o}_{(G,\K')}$ can be found geometrically by the following procedure, also illustrated in Figure \ref{fig:good_order_pf}.
Define
\[
L_\K = \{(n(h),0)|h\in H_\K\}, L_{\K'}=\{(n(h),1)|h\in H_{\K'}\}\subseteq \mathbb{R}^2.
\]
For any $e\in E$ draw the chord $c(e)$ between $(n(h_0),0)\in L_\K,(n(h_1),1)\in L_{\K'}$ where $h_0/s_1=h_1/s_1.$
By definition the change of signs between $\mathfrak{o}_{G,K},\mathfrak{o}_{G,K'}$ is just the parity of the number of intersections of these chords (slightly perturbed, if necessary).
We shall prove that this number is always even.
\begin{figure}
\centering
\includegraphics[scale=.4]{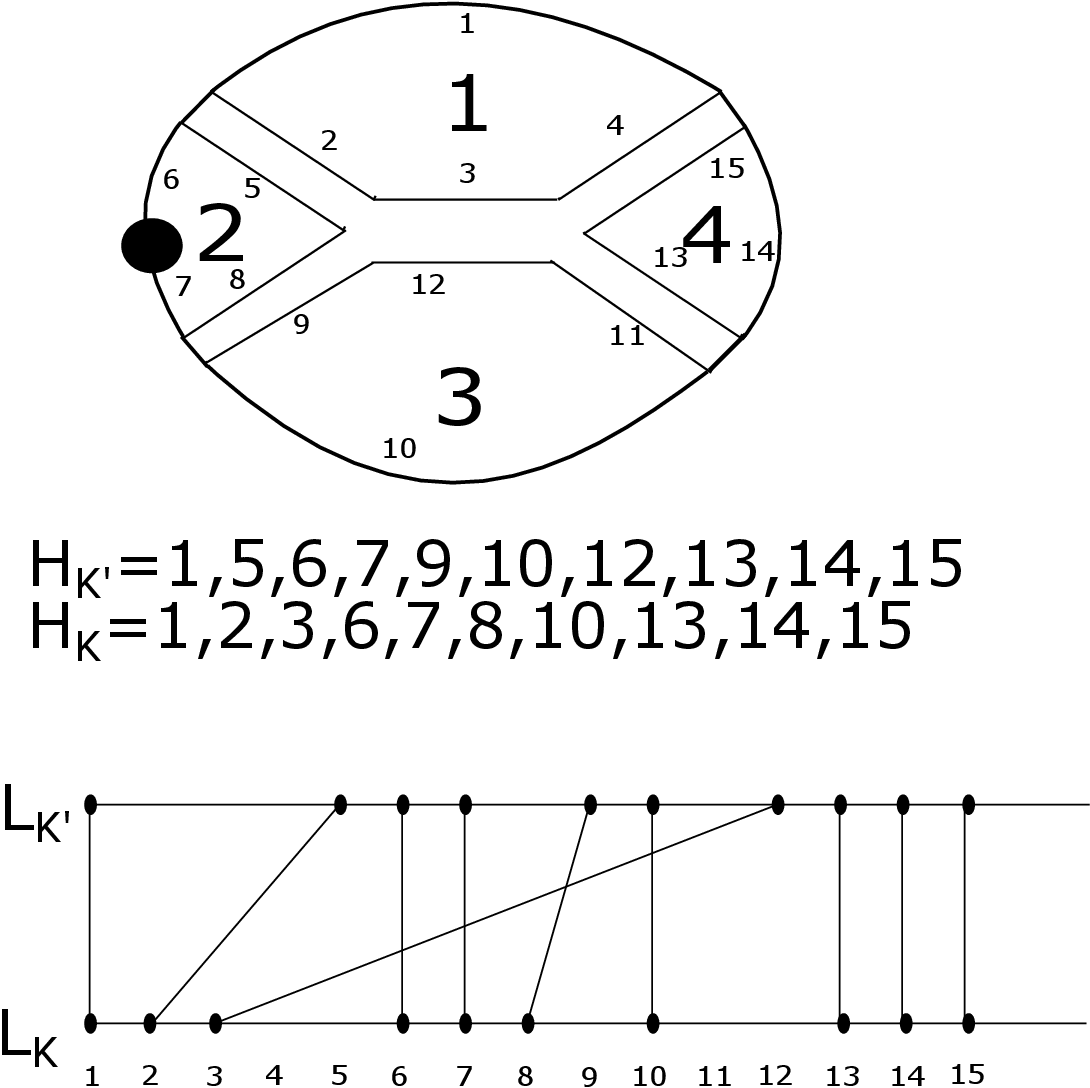}
\caption{In this diagram $H_\K,H_{\K'}$ are listed, for two Kasteleyn orientations $\K,\K'$ for the graph $G$ in the picture, whose half edges are identified with their $n-$value for some good order $n.$ The Kasteleyn orientations $\K,\K'$ can be read from $H_\K,H_{\K'},$ and they differ by a flip in the left internal vertex. Below the chord diagram of $L_\K,L_{\K'}$ is drawn, and the number of intersections is indeed even.}
\label{fig:good_order_pf}
\end{figure}
Note that for all edges except for those issuing from $v,$ the chords are parallel and vertical.

Let $h_1$ be a half edge of $v.$ Put $h_2=s_0(h_1),h_3=s_0^2(h_1),$ and $\bar{h}_j=s_1(h_j).$
Apart from some borderline cases which can be treated separately, we may assume that we are in the following scenario,
\begin{align*}
n(\bar{h}_2)=i_1, n(h_1)=i_1+1,n(\bar{h}_3)=i_2,
\\ n(h_2)=i_2+1,n(\bar{h}_1)=i_3, n(h_3)=i_3+1.
\end{align*}
Thus, the chord $c_{h_j}$ is either the chord between $(i_j+1,0)$ and $(i_{j-1},1),$ or the chord between $(i_j+1,1)$ and $(i_{j-1},0).$
It is easy to see that the number of vertical chords it intersects is the size of
\[
I_j = \{h\in H_K\setminus\{h_i,\bar{h_i}\}_{i=1,2,3}|n(h)\in(a_j,b_j)\},
\]
where $a_j = min(n_\K(i_j+1),n_\K(i_{j-1})), b_j = max(n_\K(i_j+1),n_\K(i_{j-1})).$
For exactly one $j\in \{1,2,3\}$ we have $I_j=I_{j+1}\cup I_{j+2},$ where addition is modulo $3,$ and the union is disjoint. Thus, any vertical chord either misses the chords $c_{h_j}$ or meets exactly two of them. In addition, it can be checked directly that the chords $c_{h_j}$ intersect each other an even number of times. And the lemma follows.
\end{proof}
\begin{cor}
For any $G\in\oR^0_{g,k,l},[\K]\in[\K(G)],$ the group $\text{Aut}(G,[\K])$ acts in an orientation preserving manner.
In particular, the orientation $\bar{\mathfrak{o}}_{(G,\K)}$ induces, for any $\pp$ an orientation on $\CM_{(G,[\K])}.$
\end{cor}
Denote this orientation by $\bar{\mathfrak{o}}_{(G,[\K])}.$
The main theorem of this subsection is
\begin{thm}\label{thm:orientability}
The orientations $\bar{\mathfrak{o}}_{(G,[\K])}$ induce a canonical orientation on the space $\oCM_{g,k,l}^{\text{comb}}(\mathbf{p}).$
\end{thm}
\begin{proof}
We shall show that the orientations $\mathfrak{o}_G$ for $G\in \oSR^0_{g,k,l}$ are compatible on codimension $1$ faces. This will show that a suborbifold of $\CM_{g,k,l}^{\text{comb}},$ which differs from $\CM_{g,k,l}^{\text{comb}}$ in codimension $2$ cells is oriented, hence also $\CM_{g,k,l}^{\text{comb}}$ is.
Since $\CM^{\text{comb}}_{g,k,l}$ itself differs from $\oCM^{\text{comb}}_{g,k,l}$ by codimension $2$ strata in the interior, and in codimension $1$ boundary, this argument will show that $\oCM_{g,k,l}^{\text{comb}}$ is also endowed with a canonical orientation.

We therefore have to show that for any $(G,[\K])\in \oSR^0_{g,k,l},e\notin \text{Br}(G)\cup\text{Loop}(G),$ $(G',[\K'])=(G_e,[\K_e])$ the orientation induced on $\partial_e\CM_{(G,[\K])}$ by $\CM_{(G,[\K])}$ and by $\CM_{(G',[\K'])}$ \emph{disagree}.

Put $H^I=H^I(G), H^{'I}=H^I(G').$ Note that we have a natural identification of $E(G)\setminus e$ and $E(G')\setminus e',$ for some edge $e',$ so from now on we treat them as the same set.
Choose an good ordering $n$ for $H^I.$ There exists a good ordering $n'$ of $H^{'I},$ which, when restricted to $H^{'I}\setminus s_1^{-1}(e'),$ defines the same order as the restriction of $n$ to $H^{'I}\setminus s_1^{-1}(e') \simeq H^{I}\setminus s_1^{-1}(e).$
Fix a Kasteleyn orientation $\K\in\K(G),$ set $h\in s_1^{-1}(e)$ with $\K(h)=1.$
Write $a = s_0(h),b=s_0^2(h),c=s_1(s_0(s_1(h))),d=s_1(s_0^2(s_1(h))),$ see Figure \ref{fig:neighboring2}. For shortness write $\bar{x}$ for $s_1(x).$ Apart from some borderline cases which may be treated separately, we may assume all these vertices and half edges are distinct, and then, using vertex flips if needed, we may also restrict ourself to the case where $\K(\bar{d})=1.$
In this case we can assume $n$ was chosen in such a way that
\begin{align*}
n(\bar a)&=i,~n(h)=i+1,~n(\bar d) = i+2,\\
n(d) &= m,~n(\bar c)= m+1,\\
n(c) &= p,~n(\bar h) = p+1,~n(b) = p+2,\\
n(\bar{b}) &= j,~ n(a) = j+1,
\end{align*}
for some $i,m,p,j,$ as in Figure \ref{fig:neighboring2}.
\begin{figure}
\centering
\includegraphics[scale=.4]{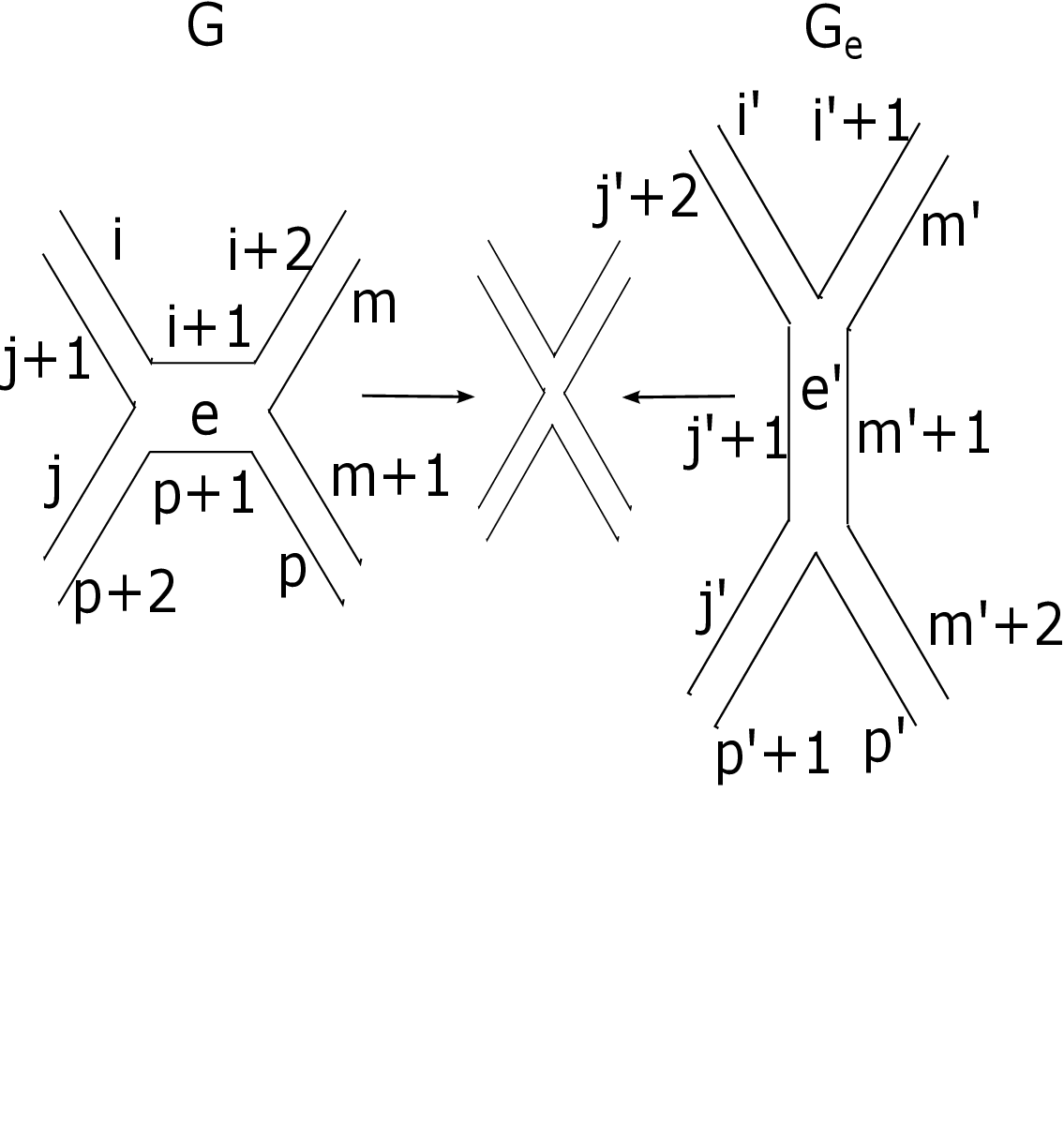}
\caption{The restrictions of the good orderings $n,n'$ to the half edges of $G,G_e.$}
\label{fig:neighboring2}
\end{figure}

A canonical outward normal for $\CM_{\partial_e G}\hookrightarrow\oCM_G$ is just $-d\ell_e.$ We see that the induced orientation on $\CM_{\partial_e G}$ is just
\begin{equation}\label{eq:or_1}
(-1)^{n_\K (n(h))+1}\bigwedge_{f\in H_\K\setminus\{h\}} d\ell_f = (-1)^{n_\K (i+1)+1}\bigwedge_{f\in H_\K\setminus\{h\}} d\ell_f,
\end{equation}
where as usual the wedge is taken in the order $n_\K$ induced by $n.$

In $G',$ let $v'_1$ be the vertex from which $a,\bar d$ issue, and $v'_2$ be the vertex from which $b,\bar c$ issue. We may take the half edge $h'$ to be the third half edge from $v'_1.$
Then, for some $i',m',p',j'$ we have
\begin{align*}
n'(\bar a)&=i',~n'(\bar d) = i'+1,\\
n'(d) &= m',~n'(h') = m'+1, n'(\bar c)= m'+2,\\
n'(c) &= p',~n'(b) = p'+1,\\
n'(\bar{b}) &= j', ~n'(\bar h') = j'+1,~n'(a) = j'+2.
\end{align*}
By Lemma \ref{lem:adjacent_spin_cells} we have a representative $\K'$ of $[\K_e],$ described by
\[
K'(h')=1, ~K'(\bar h') = 0,~K'(d) = K(d)+1=1,~K'(\bar d) = K(\bar d)+1=0,
\]
and $K'(f)=K(f)$ for any other half edge $f.$
As above, a canonical outward normal for $\CM_{\partial_{e'} G'}\hookrightarrow\oCM_{G'}$ is just $-d\ell_{e'}.$ We see that the induced orientation on $\CM_{\partial_{e'} G'}$ is just
\begin{equation}\label{eq:or_2}
(-1)^{n_{\K'} (n'(h'))+1}\bigwedge_{f\in H_{\K'}\setminus\{h'\}} d\ell_f = (-1)^{n'_{\K'} (m'+1)+1}\bigwedge_{f\in H_{\K'}\setminus\{h'\}} d\ell_f.
\end{equation}
The choice of $n,n',\K',$ makes the terms $\bigwedge_{f\in H_\K\setminus\{h\}} d\ell_f,\bigwedge_{f\in H_{\K'}\setminus\{h'\}} d\ell_f$ differ only in the relative location of $d\ell_{d}.$
By our assumptions on $\K(\bar d), \K'(\bar d)$ the difference is just the difference between $n_\K(\bar d)-1 = n_\K(i+2)-1$ and $n'_{\K'}(d) = n'_{K'}(m').$ We subtracted $1$ from $n_\K(\bar d)$ because we did not want to count $h$ which occurs before $\bar d$ in the order $n.$
Now, $n_\K(i+2)-1 = n_\K(i+1),$ as $n(h)=i,\K(h)=1.$ Similarly, $n'_{K'}(m') = n'_{K'}(m'+1)-1,$ since $n'(d)=m',\K'(d)=1.$

The total difference between the two orientations is thus $$(-1)^{n'_{\K'} (m'+1)+1 +n'_{K'}(m'+1)-1 + n_\K (i+1)+1 +n_\K(i+1)}=-1,$$
as claimed.
\end{proof}
\begin{rmk}
The spaces $\CM_{g,k,l},\CM_{g,k,l}^{\text{comb}}(\pp)$ are homeomorphic, therefore the last theorem gives, in fact, another proof that $\oCM_{g,k,l}$ is oriented.
Later we shall see that the orientation constructed here agrees with the orientation of \cite{ST}.
\end{rmk}

\begin{cor}\label{cor:G_e orientation1}
For $G\in\oSR^0_{g,k,l}$ and an internal edge $e$ which is not a bridge, the orientations on $\partial_e\oCM_{(G,[\K])}(\mathbf{p})\simeq\partial_e\oCM_{(G_e,[\K_e])}(\mathbf{p}),$ induced as boundaries of $\CM_{(G,[\K])}(\pp),\CM_{(G_e,[\K_e])}(\pp)$ are opposite.
\end{cor}

\subsection{Critical nodal graphs and their moduli}
\subsubsection{Critical nodal ribbon graphs}
In this subsection we describe effective and critical nodal graphs. They will parameterize strata which will participate in the analysis of the intersection numbers and that will contribute to the combinatorial formula. For completeness we first describe slightly more general graphs.
\begin{definition}\label{def:nodal1}
A \emph{nodal spin ribbon graph with a lifting (graded nodal ribbon graph)}, or a \emph{nodal graph} for shortness, is a spin ribbon graph with a lifting (graded ribbon graph) $(G,z),$ together with a subset $\N$ of legal points in $B(\Norm(G))\setminus B(G).$ We call $\N$ the set of \emph{legal nodes} of the nodal graph and $s_1\N$ the illegal nodes, where $s_1$ was defined in Notation \ref{nn:extended_s_1}.
The vertices and edges of the nodal graph are the vertices and edges of $\Norm(G,z)$ after forgetting the illegal nodes $s_1\N.$ A metric is a metric on these edges. If $e$ is an edge in the nodal graph $(G,z,\N)$, contracting the edge $e$ yields the nodal graph $\partial_e(G,z,\N),$ whose underlying graph is $\partial_e(G,z),$ and its legal nodes are those legal nodes in $\partial_e(G,z)$ which remain special points in $\Norm(\partial_e(G,z))$ after the contraction, where we use the natural correspondence between special points in $\Norm(G,z)$ and in $\Norm(\partial_e(G,z)).$

The \emph{components} of the nodal graph are the connected components created after removing $s_1\N.$ More precisely, define an equivalence relation $\sim_N$ on the components of $\Norm(G,z)$ as follows. Components $C_1,C_2\in \pi_0(\Norm(G,z))$ are neighbours if one of them contains a legal point $u\notin \N$ such that $s_1u$ belong to the other component. We write $C_1\sim_N C_2,$ for $C_1,C_2\in \pi_0(\Norm(G,z)),$ if they can be connected in a path of neighboring components. The components of the nodal graph are defined to be the $\Norm-$image of $\sim_N-$equivalence classes.
\end{definition}

In case the underlying graph is effective we have a more convenient definition.
\begin{definition}\label{def:nodal2}
An \emph{effective nodal spin ribbon graph with a lifting (effective graded nodal ribbon graph)}, or an \emph{effective nodal graph} for shortness, is a tuple $(G_i,z_i,m,\N=\{\N_e\}),$ or $(G,z)$ for shortness, where
\begin{enumerate}
\item
$(G_i,z_i)$ is an effective spin ribbon graph with a lifting (effective graded ribbon graph).
\item
$m:\bigcup_i s_1H^B(G_i)
\to\Z_{\geq 0}.$
\item
$\N_e:[m(e)]\to \bigcup_i B(G_i),~e\in\bigcup_i s_1H^B(G_i)
$ are injections.
\end{enumerate}
We require the sets $\N_e=\N_e([m(e)])$ to be disjoint. Denote by $C(G_i,z_i,m,\{\N_e\})$ the different graded components of the graph, that is the collection of $(G_i,z_i).$

Let $G$ be the graph obtained by choosing $m(e)$ points $p_{e,1},\ldots,p_{e,m(e)}$ on $e,$ ordered according to the orientation of the boundary and identifying $p_{e,i}$ with $\N_e(i).$ The effective nodal graph is said to be connected if $G$ is connected.
%

Write $E(G)=\cup_i E(G_i),$ similarly define $H^I(G),H^B(G),V(G),F(G).$
For a boundary edge $e=h/s_1,$ where $h\in s_1H^B$ we sometimes write $m(e)=m(h).$
Vertices in the image of $\N_e$ are called legal nodes and their set is denoted by $\N(G).$
The boundary marked points of~$G$ are boundary marked points of~$G_i$'s which are not legal nodes. Denote them by $B(G).$ Define $I(G)=\cup I(G_i).$

An effective nodal ribbon graph is naturally embedded into the (topological) nodal surface $\Sigma=\left(\coprod_i\Sigma_i\right)/\sim,$ defined as follows. $\Sigma_i$ is the topological open marked surface to which $G_i$ embeds, and in case $G_i$ is a ghost it is a point. We identify $G_i$ with its image in $\Sigma_i.$ We add $m(e)$ points on the edge $e,~p_{e,1},\ldots,p_{e,m(e)},$ and quotient by $p_{e,i}\sim\N_e(i).$ The genus of the graph is defined to be the (doubled) genus of $\Sigma.$

A \emph{marked effective nodal graph} is an effective nodal graph together with markings \\$\mmm^B:B(G)\to\Z,~\mmm^I:I(G)\to\Z.$

A \emph{graded critical nodal ribbon graph} is an effective nodal graph such that each $(G_i,z_i)\in\oSR^0.$ In this case we use the Kasteleyn notation for components, $(G_i,[\K_i])$ rather than $(G_i,z_i),$ and we denote the whole graph by $(G,[\K])$ for shortness.

A graded critical nodal graph~$G$ is \emph{odd}, if each $G_i\in\oSRc^0.$

The notion of an isomorphism is the expected one.
Write $\oSR^{m}_{g,k,l}$ for the collection of isomorphism classes of marked critical nodal graded ribbon graphs $G$ with $m$ nodes, genus $g,$ such that \\$\mmm^B:B(G)\simeq[k],~\mmm^I:I(G)\simeq[l].$ Let $\oSRc^m_{g,k,l}$ be the subset of such graphs which are odd. Write $\text{Aut}(G,[\K])$ for the group of automorphisms of $(G,[\K])\in\oSR^m_{g,k,l}.$

Define \emph{non graded critical nodal ribbon graphs} $G=(G_i,m,\N),$ in the same way, only without the data of Kasteleyn orientations, so that each $G_i$ belongs to $\oR^0,$ rather than to $\oSR^0.$
Denote by $\oR^m_{g,k,l}$ the collection of isomorphism classes non graded critical nodal ribbon graphs $G$ with $m$ nodes, genus $g,$
such that $\mmm^B:B(G)\simeq[k],~\mmm^I:I(G)\simeq[l].$
Let $\oRc^m_{g,k,l}$ be the subset of such graphs which are odd. Write $\text{Aut}(G)$ for the group of automorphisms of $G\in\oR^m_{g,k,l}.$

A metric on a nodal ribbon graph is an assignment of positive length to its edges.

A bridge $e\in E(G)$ is an edge which is a bridge in one component $G_i$ of $G.$ An \emph{effective bridge} is a bridge with $m(e)=0,$ if $m$ is defined.
Let $\text{Br}(G,[\K])$ to be the collection of bridges, and $\text{Br}^{\text{eff}}(G,[\K])$ the collection of effective bridges. As in the non nodal case, for shortness we shall usually omit $[\K]$ from the notations of $\text{Br},\text{Br}^{\text{eff}}.$ We similarly define boundary loops as boundary loops in one component $G_i$ of $G,$ and effective loops are boundary loops $e$ with $m(e)=0.$ Write $\text{Loop}(G),~\text{Loop}^{\text{eff}}(G)$ for the collection of boundary loops and effective loops respectively. 
%
\end{definition}

When it is understood from context whether or not the critical nodal graph is graded or non graded, we omit the words graded$\backslash$non graded, and just say critical nodal.
 %
\begin{rmk}
It is simple to verify that when $(G,z,m,\N)$ is effective the two definitions \ref{def:nodal1},~\ref{def:nodal2} are equivalent. We shall therefore use Definition \ref{def:nodal2} which is more explicit whenever possible.
It is also straightforward to verify that the definition of $\oRc_{g,k,l}^m$ agrees with the one given in \ref{nn:oR_intro}.

Note that in a metric effective nodal ribbon graph the data of distances between illegal nodes to other vertices is absent. On the other hand, the discrete data of which illegal node lays on which edge, and the relative order of illegal nodes on a given edge, are included. See the example in the bottom of Figure \ref{fig:nodal}.
\end{rmk}
\begin{obs}\label{obs:oR_and_oSR}
Under ${\text{for}}_{\text{spin}}:\oSR^m_{g,k,l}\to\oR^m_{g,k,l},$ which forgets the Kasteleyn orientation, odd graphs go to odd graphs, and the preimage of $G$ is canonically $[\K(G)]/\text{Aut}(G).$
\end{obs}

\subsubsection{Trivalent graphs versus graded critical nodal graphs}
In the analysis required for proving Theorem \ref{thm:comb_model} we will mainly need to analyze critical graded nodal graphs, and effective graphs which are obtained from them by contracting a single edge and possibly forgetting some data. We will now describe operations between nodal and non nodal ribbon graphs. Although these operations can be defined in full generality, we are interested only in cases where their output is trivalent or effective. We will therefore restrict our definitions to this setting, leaving the relatively straight forward details of the more general setting to the interested reader.

Given a connected effective spin ribbon graph with a lifting $(G,z),$ we define an effective nodal graph $\Y(G,z).$ Its components are the components of $\NNN(G,z),$
after erasing every illegal boundary point and concatenating its two edges to one edge.
Note that under this map a contracted boundary becomes a Ramond marking of perimeter $0.$ Suppose $e$ is an edge obtained by concatenating $e_1,\ldots,e_{m+1}$ in the described process, and in this order. Define $m(e)=m.$ Suppose $v_i$ is the vertex between $e_i,e_{i+1},$ then $\N_e(i)=s_1v_i,$ where we use Notation \ref{nn:extended_s_1}.
When $(G,z)=(G,[\K])$ critical trivalent, we denote $\Y(G,z)$ by $\Y(G,[\K]).$
It is easy to verify that
\begin{obs}\label{obs:correspondence1}
The map $\Y$ is a surjection from the collection of connected effective spin ribbon graphs 
to the collection of nodal connected effective spin ribbon graphs 
all of whose components are smooth.
It restricts to a bijection between connected trivalent graphs and connected graded critical nodal ribbon graphs.
For any connected effective spin ribbon graph 
$(G,z),$ there is a bijection between bridges (boundary loops) in $(G,z)$ and effective bridges (effective loops) in $\Y(G,z).$
\end{obs}
We now extend the definition of $\Y$ to metric effective spin ribbon graphs. 
For such a graph $(G,z,\ell)$ define the effective nodal metric graph $\Y(G,z,\ell)=(\Y(G,z),\Y\ell),$ by $\Y\ell_e=\ell_e,$ if the edge $e$ is an edge of $\NNN(G,z).$ Otherwise, if $e$ is the union of $e_1,\ldots,e_{m+1},$ define $\Y\ell_e=\sum_{i=1}^{m+1}\ell_{e_i}.$ Note that the perimeters are left unchanged.

We also define a map from effective nodal graphs to effective spin ribbon graphs: 
given an effective nodal graph $(G,z,m,\N),$ define the spin ribbon graph 
$\CBB(G,z)$ as the graph obtained by forgetting the data of $m,~\N,$ and applying $\CBB$ to each component $(G_i,z_i).$ The analogous definition holds for metric effective nodal graphs.

If $(G,z,m,\N)$ is an effective nodal graph and $e$ is either en internal edge or a boundary edge with $m(e)=0,$
$\partial_e(G,z,m,\N)$ is the nodal graph whose underlying ribbon graph is the graph obtained by contracting $e,$ and the data of $m,~\N$ is induced from $G$ by the usual identification of edges of $\partial_eG$ as a subset of edges of $G.$ Similarly, when $(G,[\K],m,\N)$ is critical trivalent and $e$ either an internal edge or an effective loop, we define $(G_e,[\K_e],m',\N')$ as the critical trivalent graph whose underlying graph is $(G_e,[\K_e])$ and $m'=m,~\N'=\N,$ where we again use the identification between edges of $G$ and $G_e.$

\begin{nn}
Suppose $(G,[\K])\in\oSR^m_{g,k,l}(\pp),~e=\{h_1,h_2=s_1h_1\}\in \text{Br}^{\text{eff}}(G)\cup\text{Loop}^{\text{eff}}(G),$ with $\K(h_1)=0.$
Define the nodal ribbon graph $\CB_e(G,[\K])$ 
as follows.
Suppose $G$ is made of the components $G_1,\ldots, G_n.$ Without loss of generality assume $e$ is an edge of component $G_n.$
Write $v_i=\partial_e(h_i),$ the vertex obtained by contracting $h_i$ in $\partial_e G_n.$
Write $x=s_2h_1, y=s_1(s_2^{-1}h_1)\in H^I(\partial_e G_n).$

The first $n-1$ components of the graph $\CB_e (G,[\K])$, are $G_i'=G_i,~i\leq n-1,$ and $\K'_i=\K_i,m'=m,\{\N'_f\}=\{\N_f\}$ for these components.

When $e$ is a boundary loop, $(G'_n,z_n)=\CBB\partial_e(G_n,[\K_n]),$ and also in this component $m'=m,~\{\N'_f\}=\{\N_f\},$ where we use the natural identifications between edges of $G_n$ other than $e$ and edges of $G_n.$

If $e$ is an effective bridge, then in case the normalization $\NNN(\partial_e G_n)$ is disconnected, let $G'_n$ be the component which does not contain $v_2,~\K',m',\N'$ will be the induced maps. Note that $G'_n$ may be a ghost.
Define the component $G'_{n+1}$ as the graph obtained by the component of $v_2$ in $\NNN(\partial_eG_n),$ after gluing the half edges $x/s_1,y/s_1$ to a new edge $xy,$ and removing the vertex $v_2.$
The updated Kasteleyn orientation is the unique Kasteleyn orientation which gives any internal half edge its value under $\K_n.$ For any half edge $e'\neq xy,~m'(e')=m(e'),~m(xy)=m(x)+m(y)+1.$ Similarly $\N'(e')=\N(e')$ for $e'\neq xy,$ while
\begin{equation}
\N'_{xy}(a)=
\begin{cases}
\N_y(a),~~&a\leq m(y)\\
v_1,& a=m(y)+1\\
\N_x(a-m(y)-1), &a>m(y)+1.
\end{cases}
\end{equation}

If $\partial_e G_n\setminus\{v_e\}$ is connected, set $G'_n$ to be the component of $v_1$ in the normalization, where again edges $x,y$ are glued and $v_2$ is removed, and $\K',m',\N'$ are defined in the same way as above.

There is a canonical surjection, which we shall also denote by $\CB_e,$
\[E(G)\cup\N(G)\to E(\CB_e G)\cup \N(\CB_e G).\] It takes $e$ to $v_1,$ and all other edges to the corresponding edges, so that it is one to one, except on the edges $x,y$ which go to $xy.$

Given a metric $\ell$ on the graph, with $\ell_e=0,$ the graph $\CB_e(G,[\K],\ell)$ is the graded nodal ribbon graph with underlying graph $\CB_e(G,[\K]),$ and the metric is induced from $\ell$ if $e$ is a boundary loop, while if $e$ is a bridge, then with the same notations as above, $(\CB_e\ell)_{e'}=\ell_{e'}$ for $e'\neq x,y,$ and $\CB_e\ell_{xy}=\ell_x+\ell_y.$ For convenience we usually denote $\CB_e\ell$ by $\ell$ as well.

A \emph{compatible sequence of effective bridges}, $e_1,\ldots,e_r$ is a sequence of bridges such that $e_{i+1}$ is an effective bridge in $\CB_{e_i}\ldots\CB_{e_1}G$ for all $i.$ For such a sequence define $\CB_{e_1,\ldots,e_r}(G,[\K],\ell)=\CB_{e_r}\ldots\CB_{e_1}(G,[\K],\ell),$ and the map $\CB_{e_1,\ldots,e_r}=\CB_{e_r}\circ\cdots\circ\CB_{e_1}.$
\end{nn}
\begin{figure}
\centering
\includegraphics[scale=.4]{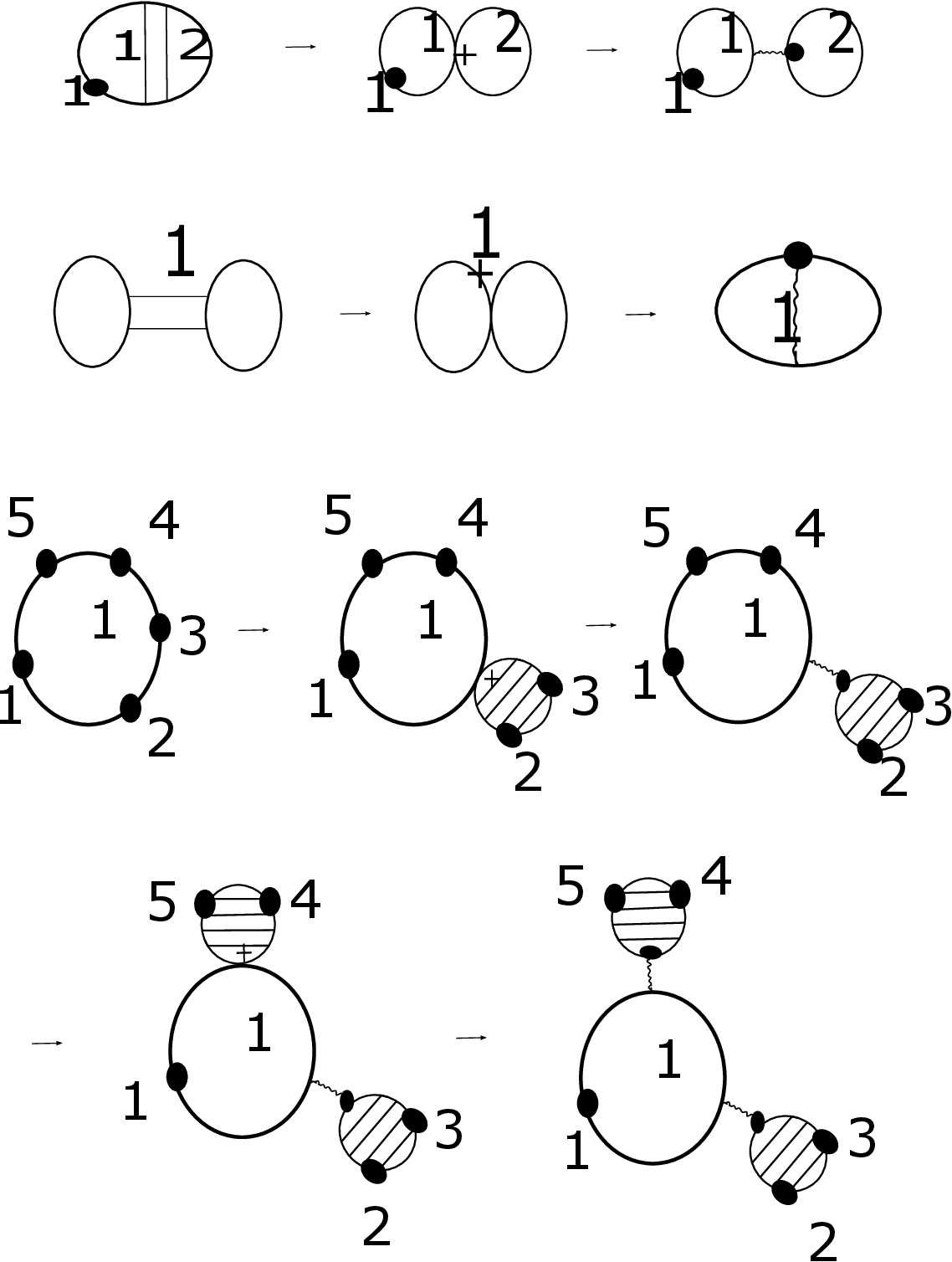}
\caption{This diagram presents trivalent graphs, their effective bridge contractions and the operation $\CB.$ $+$ represents a legal side of node, and, after performing $\CB,$ the wigly lines contain the data of $\N,$ namely, which edges contain which legal nodes, and at what order.
In the upper left corner a an effective trivalent smooth graph $(G,[\K])$ on a disk is shown, to the right of it its bridge $e$ is contracted, and then in the rightmost corner $\CB_e(G,[\K])$ is drawn. The second row descibes a similar scenario, but for a graph on a cylinder. The third row presents a graph on a disk. First the bridge between boundary markings $2,3$ is contracted, and then the bridge between $4,5$ is contracted. These bridges are compatible. The bridges between $2,3$ and $3,4,$ on the other hand, are not compatible with each other.}
\label{fig:nodal}
\end{figure}
The next observation follows easily from Observations \ref{obs:correspondence1} and \ref{obs:degenerating_kasteleyn_comb}.
\begin{obs}\label{obs:base_of_nodal_graphs}
If $(G,[\K])\in\oSR^m_{g,k,l}$ and $e\in \text{Loop}^{\text{eff}}(G),$ then $\CB_e G$ is an effective nodal ribbon graph.

If $(G,[\K])\in\oSR^m_{g,k,l}$ and $e\in \text{Br}^{\text{eff}}(G),$ then $\CB_e G\in \oSR^{m+1}_{g,k,l}.$
Moreover, for any $(G,[\K])\in \oSR^{m+1}_{g,k,l},$ and any legal node $v,$ there exists a unique graph $(H,[\K'])\in\oSR^{m}_{g,k,l}$ and an edge $e\in \text{Br}^{\text{eff}}(H)$ with $\CB_e (H,[\K']) = (G,[\K]),$ and $\CB_e e=v.$
In addition, if $(G,[\K])$ is connected trivalent, $e\in \text{Br}(G,[\K])$
\[\Y(\partial_e(G,[\K]))=\CB_e(\Y(G,[\K])),\] where we use the identification of bridges of Observation \ref{obs:correspondence1}.
\end{obs}
\begin{nn}\label{nn:CB_notations}
Recall notation \ref{nn:Aut}. For $(G,[\K])\in\oSR^{m+1}_{g,k,l},$ denote by $\CB^{-1}_{h,a}(G,[\K])=\CB^{-1}_{[h],a}(G,[\K])$ the isomorphism class of triples $(H,[\K'],e),$ where $H\in\oSR^m_{g,k,l},~\CB_e(H,[\K'])=(G,[\K]),\CB_ee=\N_h(a),$ for $h\in s_1(H^B(G)),$ and $a\in[m(h)].$ Let $$\CB^{-1}G= \{\CB^{-1}_{[h],a}(G,[\K])|[h]\in [s_1(H^B(G))], a\in[m(h)]\}.$$
\end{nn}
In other words, $(H,[\K'],e)=\CB^{-1}_{h,a}(G,[\K])$ should be thought as the graph $(H,[\K'])$ obtained by cancelling the $\CB$ operation, i.e., by returning the $a^{th}$ forgotten illegal node of $h,$ gluing it with its legal side, and then un-contracting the resulting node to obtain the bridge $e.$

\subsubsection{The moduli space of critical nodal graphs, the line bundles and the boundary conditions}
\begin{definition}\label{def:oCM_G_and_partial_e}
For an effective nodal ribbon graph $(G,z,m,\N)$ define $\CM_{(G,z,m,\N)}\simeq \R_{+}^{E(G)}/\text{Aut}(G,z,m,\N)$ to be the moduli of positive metrics on $G,$ and $\CM_{(G,z,m,\N)}$ as the subspace in which the $i^{th}$ perimeter equals $p_i>0,~i\in[l].$
In particular, given $(G,[\K])\in\oSR^m_{g,k,l},$ we have $\CM_{(G,[\K])}\simeq\R_{+}^{E(G)}/\text{Aut}(G,[\K]).$ Define $\oCM_{(G,z,m,\N)}$ and $\oCM_{(G,z,m,\N)}(\pp)$ as the cell complexes whose cells correspond to nodal ribbon graphs obtained from $(G,z,m,\N)$ by edge contractions, and the gluing maps are induced by these edge contractions.

For $e\in E(G),$ write $\partial_e\oCM_{(G,z,m,\N)}$ to be the face of $\oCM_{(G,z,m,\N)}$ where $e$ is contracted, i.e., the length of the edge $e$ is set to be $0.$
The boundary of $\oCM_{(G,z,m,\N)}$ can be written as
\[\partial\oCM_{(G,z,m,\N)}=\bigcup_{[e]\in{[E(G)]}}\partial_e\oCM_{(G,z,m,\N)},\]
where as in Notation \ref{nn:Aut}, $[E(G)]=E(G)/\text{Aut}(G,z,m,\N).$  
~We similarly define $\partial_{e_1,\ldots,e_r}\oCM_{(G,z,m,\N)}.$

The maps $\CBB,~\Y$ and $\CB_{e_1,\ldots,e_r}$ on metric graphs induce moduli level maps. We denote these maps by the same letters. When $e_1,\ldots,e_r$ are understood from the context we denote the former map by $\CB.$
%
\end{definition}
Note that $\oCM_{\partial_e(G,z,m,\N)}\simeq\partial_e\oCM_{(G,z,m,\N)},$ and that $\CB_{e_1,\ldots,e_r}$ factors $\CBB.$
The maps $\CB,\CBB,\Y$ are easily seen to be piecewise linear submersions.

\begin{definition}
For an effective nodal $(G,z,m,\N),~i\in[l]$ define the $S^1-$orbibundle $\CF_i\to\oCM_{(G,z,m,\N)}$ to be the set of pairs $(\ell,x)$ where $\ell\in\oCM_{(G,z,m,\N)}, x$ is a point on the $i^{th}$ face, with the natural topology.
For a $(d,l)-$set $L,$ write $S_L\to\oCM_{(G,z,m,\N)}$ to be the sphere bundle associated to $\{S_{L(i)}|i\in [d]\},$ as in Notation \ref{nn:spherization}.
We define the forms $\alpha_i,\omega_i,\bar{\alpha}_i,\bar{\omega}_i$ as the pull-backs of the corresponding forms defined on the component which contains face $i.$
\end{definition}
If ${(G',z',m',\N')}$ is obtained from ${(G,z,m,\N)}$ by edge contractions, we have the usual natural identification between $\CF_i\to\oCM_{(G',z',m',\N')}$ and the restriction of $\CF_i\to\oCM_{(G,z,m,\N)}$ to the corresponding cell.


By the constructions we immediately get
\begin{obs}\label{obs:identification_of_bundles}
For any effective spin ribbon graph 
$(G',z')$ and $i\in[l]$ we have a natural identification
\[(\CF_i\to\oCM_{(G',z')})\simeq \Y^*(\CF_i\to\oCM_{\Y(G',z')}).\]
while for an effective nodal spin ribbon graph 
$(G,z)$ and $i\in[l]$ we have a natural identification
\[(\CF_i\to\oCM_{(G,z)})\simeq \CBB^*(\CF_i\to\oCM_{\CBB(G',z')}).\]
As a consequence,
\begin{enumerate}
\item
for $(G,[\K])\in\oSR_{g,k,l}^m,e\notin \text{Br}(G)\cup\text{Loop}(G),$
there is a canonical identification \[(\CF_i\to\oCM_{\partial_e(G,[\K])})\simeq(\CF_i\to\partial_e\oCM_{(G,[\K])})\simeq(\CF_i\to\partial_e\oCM_{(G_e,[\K_e])}),\] and similarly for the bundles $S_L.$
\item
For $(G,[\K])\in\oSR_{g,k,l}^m,e\in \text{Br}^{\text{eff}}(G),$ then
there is a canonical identification
\[(\CF_i\to\oCM_{\partial_e(G,[\K])})\simeq(\CF_i\to\partial_e\oCM_{(G,[\K])})\simeq\CB^*_e(\CF_i\to\oCM_{\CB_e(G,[\K])}),\]  and similarly for the bundles $S_L.$
\item
For $(G,[\K])\in\oSR_{g,k,l}^m,e\in \text{Loop}(G),$ then
there is a canonical identification
\[(\CF_i\to\oCM_{\partial_e(G,[\K])})\simeq(\CF_i\to\partial_e\oCM_{(G,[\K])})\simeq (\FFF^{\text{comb}})^*(\CF_i\to\partial_e\oCM_{(G_e,[\K_e])}),\] and similarly for the bundles $S_L.$
\end{enumerate}
\end{obs}

\begin{prop}\label{prop:section_on_oSR}
Let $s$ be a special canonical multisection of $S_L\to\oCM_{g,k,l}^{\text{comb}}.$ Let $A$ be the collection of effective graded $(g,k,l)-$boundary ribbon graphs, so that $s$ restricts, in particular, to multisections $s^{(G,z)}$ for all $(G,z)\in A.$  
~Then $s$ induces multisections $s^{(G,z,m,\N)}$ of $S_L\to\oCM_{(G,z,m,\N)}$ for all effective nodal ribbon graphs $(G,z,m,\N)\in\Y(A),$
which satisfy the following relations:

For any effective graded $(G',z'),$ \[s^{(G',z')}=\Y^*s^{\Y(G',z')},\]
and for any effective nodal $(G,z,m,\N)$
\[s^{(G,z,m,\N)}=\CBB^*s',\]where $s'$ is a multisection of $S_L\to\oCM_{\CBB(G,z)}.$
In particular,
\begin{enumerate}
\item For any $(G,[\K])\in\oSR^m_{g,k,l},~e\notin \text{Br}(G)\cup\text{Loop}(G),$
\[
s^{(G,[\K])}|_{\partial_{e}\oCM_{(G,[\K])}}=s^{(G,[\K])}|_{\partial_{e}\oCM_{(G_e,[\K_e])}}.
\]
\item For any $(G,[\K])\in\oSR^m_{g,k,l},~e\in \text{Br}^{\text{eff}}(G),$
\[
s^{(G,[\K])}|_{\partial_e\oCM_{(G,[\K])}}=\CB_e^*s^{\CB_e(G,[\K])}.
\]
\item  For any $(G,[\K])\in\oSR^m_{g,k,l},~e\in \text{Loop}^{\text{eff}}(G),$
\[
s^{(G,[\K])}|_{\partial_e\oCM_{(G,[\K])}}=(\FFF^{\text{comb}})^*s^{(G_e,[\K_e])}.
\]
\end{enumerate}
where we compare multisections using the identifications of Observation \ref{obs:identification_of_bundles}.
\end{prop}
\begin{proof}
Let $s$ be a special canonical multisection as above. Consider 
an effective nodal $(G,z,m,\N)\in \Y(A).$ 
$(G,z,m,\N)$ can be written as
$\Y(G',z'),$ for some effective boundary graph. Now $s^{\Y(G',z')}=\CBB^*s^{\CBB \Y(G',z')}.$
We have a factorization
\begin{displaymath}
\xymatrix{
\CM_{(G',z')} \ar[r]^{\Y} \ar[rd]^{\CBB} &\CM_{(G,z,m,\N)} \ar[d]^{\CBB}\\
&\CM_{\CBB(G',z')},
}
\end{displaymath}

The identifications of bundles $S_L,$ see Observations \ref{obs:identificationCBB}, \ref{obs:identification_of_bundles}, is also compatible with this diagram. Since $s$ is canonical, by Corollary \ref{cor:main_prop_of_comb_section},
\[s^{(G',z')}=\CBB^*s^{\CBB(G',z')}=\Y^*\CBB^*s^{\CBB(G',z')}.\]
Define $s^{(G,z,m,\N)}$ as the pull back of $s^{\CBB(G',z')},$ along the vertical map $\CBB.$ Clearly $s^{\Y(G',z')}=\Y^*s^{(G,z)}.$


%


By Observation \ref{obs:correspondence1}, $\oSR^m_{g,k,l}\subseteq\Y(A).$ The 'In particular' cases are now immediate from the definition and Observation \ref{obs:identification_of_bundles}. In the first and third item we use that $\CBB(G,[\K])=\CBB(G_e,[\K_e]),$ while in the second that $\CB_e=\CBB$ in that case.
\end{proof}


The cells $\oCM_{(G,[\K])},$ for graded nodal graphs, also carry canonical orientations.
\begin{definition}
We define orientations for $\oCM_{(G,[\K])}(\pp),~(G,[\K])\in\oSR^m_{g,k,l}$ by
\[
\bar{\mathfrak{o}}_{(G,[\K])}=\prod_{C\in C(G,[\K])}\bar{\mathfrak{o}}_{C},~{\mathfrak{o}}_{(G,[\K])}=\bigwedge_{i\in [l]}dp_i\wedge\bar{\mathfrak{o}}_{(G,[\K])}=\bigwedge_{i\in[l]}\bigwedge_{\K(h)=1, h/s_2=i} d\ell_h,
\]
the wedge over half edges of face $i$ is taken counterclockwise.
\end{definition}

\begin{prop}\label{prop: for_monster_calc}
Let $(G,[\K])\in\oSR^m_{g,k,l},~e\in \text{Br}^{\text{eff}}(G),~(G',[\K'])=\CB_e(G,[\K])\in\oSR^{m+1}_{g,k,l},$ and let $e'$ be the unique edge in $G'$ with two $\CB_e-$preimages.
There is a canonical identification
\[
\partial_e\oCM_{(G,[\K])}\simeq\oCM_{\partial_e (G,[\K])}\simeq \F_{e'},~\partial_e\oCM_{(G,[\K])}(\pp)\simeq\oCM_{\partial_e (G,[\K])}(\mathbf p) \simeq\F_{e'}(\pp),
\]
where the space $\CF_{e'}\to\oCM_{(G',[\K'])}$ is the set of pairs $(\ell,x)$ where $\ell\in\oCM_{(G',[\K'])},~x$ is a point on $e',$ with the natural topology.
Moreover, the orientation on $\partial_e\oCM_{(G,[\K])}(\pp)$ induced from $\oCM_{(G,[\K])}(\pp)$ as in Definition \ref{def:induced_or}, coincides with the orientation
$dx\wedge\mathfrak{o}_{(G',[\K'])},$ on $\CF_{e'},$ where $dx$ is the orientation on the segment $e',$ considered as a segment in the boundary.
\end{prop}
\begin{proof}
The only part which requires an explanation is the statement regarding orientations.
Recall that $\K'$ satisfies $\K(h)=\K'(\CB h)$ for any $h/s_1\neq e.$
It is enough to compare orientations of $\partial_e\oCM_{(G,[\K])}\simeq \F_{e'}G'.$
Suppose $h$ is the legal side of $e,$ that is, the half edge which satisfies $h/s_1=e,\K(h)=1.$
Write $e_{-1}=(s_2^{-1}h)/s_1,e_1=(s_2h)/s_1.$ Then, by recalling the definition of the canonical orientation, Section \ref{sec:or}, we see that the orientation for
$\oCM_{(G,[\K])}$ can be written as $d\ell_{e_{-1}}\wedge d\ell_e\wedge d\ell_{e_1}\wedge O,$ and the orientation on $\oCM_{G'}$ is $d\ell_{e'}\wedge O,$ where $O$ is the wedge of other edge lengths, in some order. Note that $d\ell_{e'}=d\ell_{e_{-1}}+d\ell_{e_1}.$
Now, the induced orientation on the boundary $\partial_e\oCM_{(G,[\K])}$ is given by $d\ell_{e_{-1}}\wedge d\ell_{e_1}\wedge O.$
By considering $\F_{e'} G'$ as the moduli of metrics on the graph obtained from $G'$ by adding a new marked point on $e',$
and with the definition of its orientation, we see that this orientation can be written as $d\ell_{e_{-1}}\wedge d\ell_{e'}\wedge O,$
where $d\ell_{e_{-1}}$ comes from the location of the new point on $f.$
And indeed,
\[
d\ell_{e_{-1}}\wedge d\ell_{e_1}\wedge O =d\ell_{e_{-1}}\wedge d\ell_{e'}\wedge O.
\]
\end{proof}
\begin{cor}\label{cor:or_agree}
The map $\comb:\oCM_{g,k,l}\to\oCM_{g,k,l}^{\text{comb}}$ preserves orientation.
\end{cor}
\begin{proof}
Indeed, by Proposition \ref{prop: for_monster_calc}, we see that the orientations on $\oCM_{g,k,l}^{\text{comb}}$ satisfy the same requirements of Theorem \ref{lem:geometric_orientability}.
The dimension $0$ case can be checked by hand.
\end{proof}

We also have the following corollary of Corollary \ref{cor:G_e orientation1}
\begin{cor}\label{cor:G_e orientation}
For $(G,[\K])\in\oSR^m_{g,k,l}$ and an internal edge $e$ which is not a bridge, the orientations on $\partial_e\oCM_{(G,[\K])}(\mathbf{p})\simeq\partial_e\oCM_{(G_e,[\K_e])}(\mathbf{p}),$ induced as boundaries of $\CM_{(G,[\K])}(\pp),\CM_{(G_e,[\K_e])}(\pp)$ are opposite.
\end{cor}
Corollary \ref{cor:G_e orientation} has an analog for the case $e$ is a boundary loop.
For $(G,[\K])\in\oSR^m_{g,k,l},~e\in\text{Loop}(G),$  write $\FFF^{\text{comb}}_{(G,[\K]),e}$ as the map $\partial_e\oCM_{(G,[\K])}\to \partial_e\oCM_{(G_e,[\K_e])}$ defined in the level of objects by leaving all the metric graph structure, and in particular the edge lengths, invariant, and flipping the lifting in the contracted boundary which corresponds to $e.$ When we write $\FFF^{\text{comb}}$ we mean the union of the maps $\FFF^{\text{comb}}_{(G,[\K]),e}$ over all possible pairs $(G,[\K])\in\oSR^m_{g,k,l},~m\geq 0,~e\in\text{Loop}(G).$
The following is an immediate corollary of the 'Moreover' part of Theorem \ref{lem:geometric_orientability} and Corollary \ref{cor:or_agree}. We will also provide a direct self contained proof of this corollary in Subsection \ref{sec:power_of_2} below.
\begin{cor}\label{cor:G_e orientation_loop}
For $(G,[\K])\in\oSR^m_{g,k,l}$ and $e\in\text{Loop}(G),$ the orientation on $\partial_e\oCM_{(G_e,[\K_e])}(\mathbf{p})$ induced as a boundary of $\CM_{(G_e,[\K_e])}(\pp),$ is opposite to the orientation on it obtained by taking the $\FFF^{\text{comb}}-$push forward of the orientation on $\partial_e\oCM_{(G,[\K])}(\mathbf{p}),$ induced as a boundary of $\CM_{(G,[\K])}(\pp).$\end{cor}

\section{The combinatorial formula}\label{sec:last}
Throughout this section we fix $g,k,l$ and set
\[d=\frac{\dim_\R(\oCM_{g,k,l})}{2}=\frac{3g-3+k+2l}{2}.\] We also write, for $G\in\oSR^m_{g,k,l},$  \[\dim(G)=\frac{\dim_\R(\CM_G)}{2}=\frac{3g-3+k+2l-2m}{2}.\]

In what follows we shall work with the orientations constructed in Subsection \ref{sec:or}. These are the same orientations as the ones constructed in \cite{ST}, by Corollary \ref{cor:or_agree}.
\begin{definition}\label{def:W_and_tilde_W}
For $(G,[\K])\in\oSR^m_{g,k,l}$ define
\[
W_G,\widetilde W_G:\CM_{(G,[\K])}\to\R,
\]
by
\[
W_G(\ell)=\prod_{e\in s_1H^B(G)}\frac{\ell_e^{2m(e)}}{(m(e)+1)!},~\widetilde W_G(\ell)=\prod_{e\in s_1H^B(G)}\frac{\ell_e^{2m(e)}}{m(e)!(m(e)+1)!}.
\]
\end{definition}

\subsection{Iterative integration and the integral form of the combinatorial formula}
Our approach for producing the explicit formula for intersection numbers will be by an iterative process of integration by parts. Recall Definition \ref{def:l-sets} and Notation \ref{nn:comb_sphere}.
Given a $(S,l)-$set, $L:S\to[l],$ for $S\subseteq[d],$ the $t^{th}$ component of $E_L$ is $\CL_{L(t)}.$ Each step of the iterative integration process below, will involve integrating out (the form which corresponds to) one component $\CL_{L(t)},$ for some $t\in S,$ using integration by parts. The integration by parts will produce new boundary terms for the moduli on which we integrate. Only boundary terms that correspond to contracting an effective bridge $e$ may have a non zero contribution which does not cancel. Moreover, in order for such an edge to contribute a non zero contribution, when we integrate out the $t^{th}$ component, the illegal side of the half node obtained by contracting $e$ will have to lie in the face $L(t).$ This is the content of first key lemma, Lemma \ref{lem:simplifying_phi}. In order to be able to state it, we need to add notations, specifically, a notation that will allow us to keep track on which illegal half node corresponds to the $t^{th}$ component of the vector bundle which we integrate out. For this we present the auxiliary notion of decorations. After performing an iteration of integration by parts, the second key lemma, Lemma \ref{lem:big_computation} transforms integrals over the boundaries of the moduli to integrals over the moduli spaces obtained by further forgetting the illegal half node. Theorem \ref{thm:int_form} essentially iterates these lemmas, and uses some other cancellations to obtain a formula for the open intersection numbers as sums of integrals. It is remarkable that this iterative integration process is performed without appealing to a specific canonical multisection, and in some sense this is the key point of the proof. In addition, it gives an alternative proof of the claim that canonical boundary condition give rise to well defined intersection numbers, proven in \cite{PST} for genus $0$ and in \cite{ST} for $g>0.$
\begin{definition}
A \emph{decoration} $D$ of a graph $(G,[\K])\in \oSR^{m}_{g,k,l},$ is a choice of sets $D_h\subseteq [d],$ for any $h\in s_1 H^B,$ which are pairwise disjoint and such that
\[
|D_h| = m(h).
\]
When $e=h/s_1$ we also write $D_e=D_h.$ For a $(S,l)-$set $L,$ a $L-$decoration is a decoration for which
\[
D_h\subseteq L_{i(h)}.
\]

In the next series of claims we shall omit $[\K]$ from the notation of graded graphs, to make notations lighter.

Denote the collection of all decorations of $G$ by $\text{Dec}(G),$ and the collection of all $L-$decorations of $G$ by $\text{Dec}(G,L).$

Let $L(D)$ be the $l-$subset of $L$ given by $L|_{\bigcup_{h\in s_1 H^B}{D_h}},$ so that $L(D)_i=\cup_{i(h)=i}D_h.$

For $(G,[\K])\in\oSR^{m>0}_{g,k,l}$ and a $(G,L)-$decoration $D,$ define the set
\[
\CB^{-1}(G,D)\subseteq \{(G',e',D')|(G',e')\in\CB^{-1}G,D'\in \text{Dec}(G', L)\},
\]
as follows.
$(G',e',D')\in\CB^{-1}(G,D)$ exactly when $(G',e')\in\CB^{-1}G,D'\in \text{Dec}(G', L),$ and for any $e\in E(G')\setminus\{e'\},~D'_e\subseteq D_{\CB e}.$ Note that in this case $L(D')\subseteq L(D),$ and the difference is exactly one element.
\end{definition}
In the language of the paragraph which precedes this definition, $L(D)\setminus L(D')$ is precisely the element $t\in[d]$ which corresponds to the effective bridge $e'$ in the iterative process.

In order to be able to calculate intersection numbers, we must understand the restriction of the forms $\alpha_i,\omega_i$ to the boundary.

Suppose $(G,[\K])\in\oSR^m_{g,k,l},~e\in \text{Br}^{\text{eff}}(G),~h$ is its illegal side, $\K(h)=1,$ and $i\in[l].$
On $\CM_{\partial_e G}(\pp)$ we have two natural representatives for the angular $1-$form,
$\alpha_i^{\partial_e G}=\alpha_i^G|_{\partial_e\CM_G},$ and $\CB^*\alpha_i^{\CB_e G}.$ Similarly, we have two natural choices for the induced $2-$forms, $\omega_i^{\partial_e G}=\omega_i^G|_{\partial_e\CM_{G}},$ and $\CB^*\omega_i^{\CB_e G}.$
\begin{nn}
Write $\beta_i = \beta_i^{\partial_eG} = \alpha_i^{\partial_eG}-\CB^*\alpha_i^{\CB_e G},$
~and $B_i = B_i^{\partial_eG}=\omega_i^{\partial_eG}-\CB^*\omega_i^{\CB_e G}.$
\end{nn}
\begin{obs}\label{obs:change_in_one_form}
With the above notations, if $i\neq i(e),$ then $B_i=\beta_i=0.$
Otherwise we have
\[
p_i^2\beta_i = \ell_{s_2 h}d\ell_{s_2^{-1}h},~~
p_i^2B_i = d\ell_{s_2^{-1}h}\wedge d\ell_{s_2 h}.
\]
\end{obs}
Unlike the forms $\alpha_i,$ the form $\beta_i$ is pulled back from the combinatorial moduli, since it has no angular variables.
\begin{proof}
For $i\neq i(h),$ the forms restricted from $\CM_G$ and those pulled back from $\CM_{\CB_eG}$ are canonically identified.
Suppose $i=i(h),$ we handle $B_i.$ The proof for $\beta_i$ is similar.
$\ell_{e}=0,$ hence also $d\ell_{e}=0$ on $\partial_e\CM_G.$ 
~Hence the only difference between $\omega^{\partial_e G},$ and $\CB^*\omega_i^{\CB_e G}$ is that the former may contains terms with $d\ell_{s_2 h}$ or $d\ell_{s_2^{-1}h},$ while the latter depends only on their sum, by the definition of $\CB_e.$
Choose a good ordering $n$ in the sense of Definition \ref{def:good_order}, such that half edges of the $i^{th}$ face appear first, and some half edge $h'\neq h,s_2 h$ is the first edge in the ordering. One can always find such a half edge. Otherwise, the $i^{th}$ face is bounded by exactly two edges, $h,s_2h,$ which therefore must be a boundary half edge, and in particular $\K(s_2h)=1.$ But then the sum of $\K$ on the $i^{th}$ face is even, which is impossible for a Kasteleyn orientation.

In $\CB_e G$ we choose a good ordering $n'$ for which $h',$ identified as an edge of $\CB_e G,$ is the first half edge. Suppose $s_2^{-1}h$ is the $j^{th}$ half edge in $n,$ so that $h,s_2 h$ are the $j+1^{th},j+2^{th}$ edges. Write $\ell_a$ for $\ell_{n^{-1}(a)}.$
Then,
\begin{align*}
p_i^2\omega_i^G|_{\partial_e\CM_G} &= \sum_{a<b}d\ell_a\wedge d\ell_b \\
&=\sum_{a<b,~a,b\neq j,j+1,j+2}d\ell_a\wedge d\ell_b +\sum_{a<j}d\ell_a\wedge (d\ell_j+d\ell_{j+2})\\
&\quad\quad\quad\quad\quad\quad+\sum_{j+2<a} (d\ell_j+d\ell_{j+2})\wedge d\ell_a + d\ell_j\wedge d\ell_{j+2}\\
&= p_i^2\CB^*\omega_i^{\CB_e G}+d\ell_j\wedge d\ell_{j+2}.
\end{align*}
In the last equality we used the fact that $\ell^{\CB_e G}_{{n^{'-1}(j)}}=\ell_{{n^{-1}(j)}}+\ell_{{n^{-1}(j+2)}},$
and for $a\neq j,~\ell^{\CB_e G}_{{n^{'-1}(a)}}=\ell_{e_{a+w(a)}},$ where $w(a)=0,$ for $a<j,$ and otherwise it is $2.$
\end{proof}
\begin{nn}
Recall Notation \ref{nn:comb_sphere} and Remark \ref{rmk:phi_as_poly}.
For $G,e$ as above, given a $(S,l)-$set $L,$ and $i\in S,$ we define the form $\Phi^i_L$ on the sphere bundle
$S_L\to\partial_e\CM_G$
\[
\Phi^i_L = \Phi(\{r_j\}_{j\in S}, \{\alpha'_j\}_{j\in S},\{\omega'_j\}_{j\in S})=\Phi^{\partial_e G}(\{r_j\}_{j\in S}, \{\alpha'_j\}_{j\in S},\{\omega'_j\}_{j\in S}),
\]
where $\alpha'_j$ is a \emph{copy} of $\CB^*\alpha^{\CB_e G}_{L(j)}$ for $j\neq i,$ and $\alpha'_i=\beta_{L(i)}.$ Similarly, $\omega'_j=\CB^*\omega^{\CB_e G}_{L(j)},$ unless $j=i,$ and then $\omega'_i = B_{L(i)}.$
As usual $\bar{\Phi}^i_L=p^{2L}\Phi^i_L.$ As in Remark \ref{rmk:phi_as_poly}, when $S\subset[d]$ we will also extend the domain of $\Phi^i_L$ by allowing $\sum_{i\in S} r_i^2$ to vary.
\end{nn}

From now until the end of this subsection, we fix a $(d,l)-$set $L,$ and let $E_L$ be the corresponding bundle.
\begin{lemma}\label{lem:simplifying_phi}
Let $s$ be a special canonical multisection of $E_L.$ Take $G\in\oSR^m_{g,k,l}$ arbitrary, $e$ an effective bridge of $G,~h$ its illegal side.
Let $D'$ be a $L-$decoration of $G,$ and write $L'=L(D').$
Then
\[
\int_{\partial_e\CM_{G}(\pp)} s^*(W_G \bar{\Phi}_{L\setminus L'})
= \sum_{j\in (L\setminus L')_{i(h)}}\int_{\partial_e\CM_{G}(\pp)} W_G s^{*}( \bar{\Phi}^j_{L\setminus L'}).
\]
\end{lemma}
It should be noted that different decorations $D',D''$ which determine the same set $L(D')=L(D'')$ will give rise to the same integral. The decorations, as mentioned above, are introduced only in order to keep track on the combinatorics of integrals that will appear in the iterative integration process below.
\begin{proof}
Write $S=\bigcup_{h\in s_1 H^B}{D'_h},$ so that
$L':S\to [l]$ is a restriction of $L:[d]\to[l].$
We first use \eqref{eq:angular} and Notation \ref{nn:comb_sphere} to write $\Phi_{L\setminus L'}$ explicitly.
\begin{align}\label{eq:angular_intersection_theory}
\Phi_{L\setminus L'}(\{{r}_i\}_{i\in S^c},& \{\hat{\alpha}_i\}_{i\in S^c},\{\hat{\omega}_i\}_{i\in S^c}) \\\notag&
=\sum_{k=0}^{|S^c|-1} 2^k k!\sum_{i\in S^c}{r}_i^2\hat{\alpha}_{i}\wedge
\sum_{I\subseteq S^c\setminus\{i\}, |I|=k}\bigwedge_{j\in I}(r_jdr_j\wedge\hat{\alpha}_{j})\bigwedge_{h\notin I\cup\{i\}}\hat{\omega}_{h},
\end{align}
where $\hat\omega_{j}$ is Kontsevich's 2-form $\omega_{L(j)},$ and $\hat{\alpha}_j$ is a copy of Kontsevich's 1-form $\alpha_{L(j)}.$ This is a form of degree $\dim_\R \CM_{\partial_e G}=\dim_\R \CM_{\CB_e G}+1.$ We obtain $\Phi_{L\setminus L'}^i$ by the same formula, after replacing $\hat{\alpha}_i,~\hat\omega_i$ by $\beta_{L(i)},B_{L(i)}$ respectively.

Now, the function $W_G$ does not depend on variables of the fiber of the sphere bundle, hence it can be taken out of the pull-back.
By the definitions of the forms we can write,
\[\hat\alpha_j=\CB^*\hat{\alpha}^{\CB_e G}_{j}+\beta_{L(j)},
~~\hat{\omega}_j=\CB^*\omega^{\CB_e G}_{L(j)}+B_{L(j)},\]
where $\hat{\alpha}^{\CB_e G}_{j}$ is a copy of $\alpha^{\CB_e G}_{L(j)}.$
We now substitute this in $\bar{\Phi}_{L\setminus L'},$ and expand \eqref{eq:angular_intersection_theory} multilinearly.

Write $i=i(h)\in[l].$ Any term containing $\beta_a$ or $B_a$ for $a\neq i,$ will vanish, by Lemma \ref{obs:change_in_one_form}.

Similarly, any term in the expansion that will contain either $\beta_a$ twice, or $B_a$ twice, or $\beta_i$ and $B_i$ once, will vanish, as a consequence of a multiple appearance of $d\ell_{s_2^{-1}h}.$

$s|_{\partial_e\CM_{G}}$ is pulled back from $\CM_{\CB_e G},$ by Proposition \ref{prop:section_on_oSR}.
Now, a term in $s^*\Phi_{L\setminus L'}$ with no $B_i$ or $\beta_i$ is pulled back from $\CM_{\CB_e G}.$
But its degree is $\dim_\R \CM_{\CB_e G}+1.$ 
Thus, it vanishes because of dimensional reasons.

We are left with terms containing a single $\beta_i$ or $B_i.$ These $\beta_i$ or $B_i$ are in fact $\beta_{L(j)}$ or $B_{L(j)}$ for some $j\in S^c$ which is mapped by $L$ to $i,$ meaning $j\in (L\setminus L')_i.$ And the lemma follows.
\end{proof}

The second main lemma we need is the following.
\begin{lemma}\label{lem:big_computation}
Fix $m>0,~G\in \oSR^{m}_{g,k,l},~D\in \text{Dec}(G,L),$ with $L'=L(D).$
Then
\begin{align*}
  \sum_{(G',e',D')\in \CB^{-1}(G,D)}&\int_{\CM_{\partial_{e'}G'}(\pp)}W_{G'}s^{*}(\bar{\Phi}^{\partial_{e'}G'})^{L'\setminus L(D')}_{L\setminus L(D')}\\
  &=\int_{\CM_G(\pp)}W_G\bar{\omega}_{L\setminus L'}+
\int_{\partial\CM_G(\pp)}W_G s^*(\bar{\Phi}^G)_{L\setminus L'}.
\end{align*}
\end{lemma}
Importantly, $\int_{\CM_G(\pp)}W_G\bar{\omega}_{L\setminus L'}$ does not depend on the multisection $s,$ so this lemma pushes the dependence on $s$ to lower dimensional moduli. After iterating, it will allow us to completely remove the dependence of the integrals on $s.$ This phenomenon is expected, from the geometric point of view, since it was proven by \cite{ST,PST} that the intersection numbers should be independent of the specific canonical multisection. And indeed, the lemma is enabled by the properties of canonical multisections, and will not be true for arbitrary, non canonical, boundary conditions.

\begin{proof}
For convenience we treat the case $|\text{Aut}(G)|=1,$ the general case is handled similarly, but notations become more complicated.
Put
\[
E' = \{e\in E(G)|~m(e)>0\}.
\]
Recall Notation \ref{nn:CB_notations}. Suppose $(G',e')\in\CB^{-1} G$ is $\CB^{-1}_{e,a+1}G$ for $e\in E',a+1\in [m(e)].$
Fix $h\in D_e,$ and let
\[
D(G',h) := \{D'| (G',D')\in \CB^{-1}(G,D),~h\notin L(D')\}.
\]
In words, $D(G',h)$ is the set of decorations of $G'$ in $\CB^{-1}(G,D)$ such that the only element of $L'$ that they miss is $h.$ Such decorations are determined by how we split the elements in $D_e\setminus\{h\}$ to sets of sizes $a,m(e)-1-a$ that will decorate the two edges in $\CB^{-1}_{e'}e$, the edges which, after contracting $e'$ and forgetting its illegal side, form $e.$ Thus, $|D(G',h)| = \binom{m(e)-1}{a}.$
Let $e_1=s_2^{-1}e',e_2=s_2e',$ be the two half edges of $G'$ mapped under $\CB_{e'}$ to $e.$ As explained, $m(e_1)=a,~m(e_2)=m(e)-a-1.$ Put $\ell'_e = \ell_{e_1}.$
For fixed $G',h$ we have equality
\[
\int_{\CM_{\partial_{e'}G'}(\pp)}W_{G'}s^{*}\bar{\Phi}^{L'\setminus L(D')}_{L\setminus L(D')}=\int_{\CM_{\partial_{e'}G'}(\pp)}W_{G'}s^{*}\bar{\Phi}^{h}_{L\setminus L(D')},
\]
hence the left hand side of this equation is independent of $D'.$ 
We will now show
\begin{align}\label{eq:monster1}
&\sum_{D'\in D(G',h)}\int_{\CM_{\partial_{e'}G'}(\pp)}W_{G'}s^{*}\bar{\Phi}^{h}_{L\setminus L(D')}\\
\notag
&\quad=\int_{\CM_{G}(\pp)}
\binom{m(e)-1}{a}\left(\prod_{f\in E'\setminus\{e\}}\frac{\ell_{f}^{2m(f)}}{(m(f)+1)!} \right)\\ \nonumber
&\quad\quad\quad\cdot\int_{0}^{\ell_e}\frac{(\ell'_e)^{2a}(\ell_e-\ell'_e)^{2(m(e)-a-1)}}{(a+1)!(m(e)-a)!}(A_{e,h}+B_{e,h}+C_e),
\end{align}
where \begin{align*}
&A_{e,h}=r_h^2(\ell_e-\ell'_e)d\ell'_e\wedge
\sum_{n\geq 0} 2^n n!\sum_{|I|=n,I\subseteq L\setminus L'}\left(\bigwedge_{j\in I}r_jdr_j\wedge\hat{\alpha}_j\right)\wedge\bigwedge_{j\in L\setminus(I\cup L')}\bar{\omega}_{L(j)},\\
&B_{e,h}= r_h dr_h \wedge (\ell_e-\ell'_e)d\ell'_e\wedge\sum_{i\in L\setminus L'}r_i^2\hat{\alpha}_{i}\wedge\sum_{n\geq 0}
2^{(n+1)}(n+1)!\\
&\quad\quad\quad\quad\quad\cdot\sum_{|I|=n, I\subseteq L\setminus (L'\cup\{i\})}
\left(\bigwedge_{j\in I} r_j dr_j \wedge\hat{\alpha}_{j}\right) \wedge\bigwedge_{j\in L\setminus(L'\cup I\cup\{i\})}\bar{\omega}_{L(j)},\\
&C_e=d\ell'_e\wedge d\ell_e\wedge\sum_{i\in L\setminus L'}r_i^2\hat{\alpha}_{i}\wedge\sum_{n\geq 0} 2^n n!
\sum_{|I|=n, I\subseteq L\setminus (L'\cup\{i\})}\left(\bigwedge_{j\in I} r_j dr_j \wedge \hat{\alpha}_{j}\right)\\
&\quad\quad\quad\quad\quad\quad\quad\quad\quad\quad\quad\quad\quad\quad\quad\quad\quad\quad\quad\quad\quad\quad\quad\quad\wedge\bigwedge_{j\in L\setminus(L'\cup I\cup\{i\})}\bar{\omega}_{L(j)}.
\end{align*}
where $\hat\alpha_i$ is a copy of $\bar{\alpha}_{L(i)}.$
Before proving this equation, observe that $A_{e,h},B_{e,h},C_e$ \emph{depend} on the multisection $s$ through the sphere bundle fiber variables $r_i=r_i(s),$ and $\bar{\alpha}_i=\bar{\alpha}_i(s),$ but we omit $s$ from the notations. However, because $s$ is special canonical, it follows from the second item of Proposition \ref{prop:section_on_oSR}, that $s(x,\ell'_e),~x\in\CM_G,~\ell'_e\in[0,\ell_e]$ depends only on $x$ and not on $\ell'_e,$ where we have used the identification of Proposition \ref{prop: for_monster_calc}. Thus, the same is true for the variable $r_i$ and the form $\hat{\alpha_i}.$ Therefore, importantly, $A_{e,h},B_{e,h},C_e$ are \emph{independent} of $a,$ and their only dependence of $\ell'_e,~d\ell'_e$ is through the terms which explicitly involve them.

The last equation follows from the following facts.
First, the multiplicity $\binom{m(e)-1}{a}$ comes from summing over the different decorations $D',$ which all give the same contribution.
Second, the term in $W_{G'}$ for the edge $f\in E'\setminus\{e\}$ is $\frac{\ell_{f}^{2m(f)}}{(m(f)+1)!}.$  The corresponding terms for $e_1,e_2$ are $\frac{(\ell'_e)^{2a}}{(a+1)!},~
\frac{(\ell_e-\ell'_e)^{2(m(e)-a-1)}}{(m(e)-a)!}$ respectively.
Third, Proposition \ref{prop: for_monster_calc} reduces the integration over $\CM_{\partial_{e'}G'}(\pp)$ to the repeated integral obtained by first integrating over $\CM_{G}(\pp),$ and then over the location of the node on the edge $e,$ which is encoded by $\ell'_e.$ This inner integration is precisely the integration $\int_{0}^{\ell_e}$ (with respect to $d\ell'_{e}$).
Next, recall that, with $S=\bigcup_{h\in s_1 H^B}{D_h},$
\begin{align*}
\bar{\Phi}_{L\setminus L'}^h(\{{r}_i\}_{i\in S^c},& \{\hat{\alpha}_i\}_{i\in S^c},\{\hat{\omega}_i\}_{i\in S^c}) \\=\notag&
\sum_{k=0}^{|S^c|-1} 2^k k!\sum_{i\in S^c}{r}_i^2\hat{\alpha}_{i}\wedge
\sum_{I\subseteq S^c\setminus\{i\}, |I|=k}\bigwedge_{j\in I}(r_jdr_j\wedge\hat{\alpha}_{j})\bigwedge_{f\notin I\cup\{i\}}\hat{\omega}_{f},
\end{align*}
where for $j\neq h,$ $\hat\omega_{j}=\bar{\omega}_{L(j)},$ and $\hat{\alpha}_j$ is a copy of $\bar{\alpha}_{L(j)},$ while $\hat\omega_{h}=p_h^2B_{L(h)},~\hat\alpha_h=p_h^2\beta_h.$
Using Lemma \ref{obs:change_in_one_form}, the sum of terms which have $i=h$ in the second summation is precisely $A_{e,h}.$ The sum of terms with $i\neq h$ in which $I$ contains $h$ is $B_{e,h},$ while the remaining terms sum to $C_e.$

We shall use the following proposition.
\begin{prop}\label{prop:id}
\begin{enumerate}
\item\label{prop:id_a}
$
\sum_{a=0}^{m-1}\binom{m-1}{a}\int_0^y\frac{x^{2a}(y-x)^{2(m-a)-1}}{(a+1)!(m-a)!}dx
=\frac{y^{2m}}{(m+1)!}
$
\item\label{prop:id_b}
$
\sum_{a=0}^{m-1}\binom{m-1}{a}\int_0^y\frac{x^{2a}(y-x)^{2(m-a-1)}}{(a+1)!(m-a)!}dx
=\frac{2 y^{2m-1}}{(m+1)!}
$
\end{enumerate}
\end{prop}
Still fixing $e,~h\in D_e,$ we now apply Proposition \ref{prop:id}, the fact that $A_{e,h},B_{e,h},C_e$ are independent of $a,$ and that $r_i,~\hat\alpha_i$ are independent of $\ell'_e,$ to sum Equation \eqref{eq:monster1} over $(G'_{a},e'_a) := \CB^{-1}_{e,a+1}G,$ where $a=0,\ldots,m(e)-1.$ We obtain
\begin{align}\label{eq:monster2}
\sum_{a=0}^{m(e)-1}&\sum_{D'\in D(G'_a,h)}\int_{\CM_{\partial_{e'}G'}(\pp)}W_{G'}s^{*}\Phi^{h}_{L\setminus L(D')}\\=
\notag & \int_{\CM_{G}(\pp)}
\prod_{f\in E'\setminus\{e\}}\frac{\ell_{f}^{2m(f)}}{(m(f)+1)!} \left\{\frac{\ell_e^{2m(e)}}{(m(e)+1)!}(\widetilde{A}_{e,h} +\widetilde{B}_{e,h})+ \frac{2\ell_e^{2m(e)-1}d\ell_e}{(m(e)+1)!}\wedge
Y\right\}
\end{align}
where
\begin{align*}
&\widetilde{A}_{e,h} = r_h^2\sum_{m\geq 0} 2^m m!
\sum_{|I|=m,I\subseteq L\setminus L'}\left(\bigwedge_{j\in I}r_jdr_j\wedge\hat{\alpha}_j\right)\wedge\bigwedge_{j\in L\setminus(I\cup L')}\bar{\omega}_{L(j)}\\
&\widetilde{B}_{e,h} = -r_h dr_h\wedge\sum_{i\in L\setminus L'}r_i^2\hat{\alpha}_{i}\wedge\sum_{m\geq 0}
2^{(m+1)}(m+1)!
\sum_{|I|=m, I\subseteq L\setminus (L'\cup\{i\})}
\left(\bigwedge_{j\in I} r_j dr_j \wedge\hat{\alpha}_{j}\right)\wedge\bigwedge_{j\in L\setminus(L'\cup I\cup\{i\})}\bar{\omega}_{L(j)}\\
&Y
 = \sum_{i\in L\setminus L'}r_i^2\hat{\alpha}_{i}\wedge\sum_{m\geq 0} 2^m m! \sum_{|I|=m, I\subseteq L\setminus (L'\cup\{i\})}
\left(\bigwedge_{j\in I} r_j dr_j \wedge\hat{\alpha}_{j}\right)\wedge\bigwedge_{j\in L\setminus(L'\cup I\cup\{i\})}\bar{\omega}_{L(j)}
\end{align*}

%
%

The next step is to eliminate $r_h$ terms, for $h\in L'.$ For this put
\begin{align*}
&X =
\left(1-\sum_{h\in L\setminus L'}r_h^2\right)\\
&\qquad\cdot\sum_{m\geq 0} 2^m m!\sum_{|I|=m,I\subseteq L\setminus L'}
\left(\bigwedge_{j\in I}r_jdr_j\wedge\hat{\alpha}_{j}\right)\wedge\bigwedge_{j\in L\setminus(I\cup L')}\bar{\omega}_{L(j)}\\
&\qquad+\left(\sum_{h\in L\setminus L'}r_h dr_h\right)\wedge
\sum_{i\in L\setminus (L'\cup\{h\})}r_i^2\hat{\alpha}_{i}\wedge
\sum_{m\geq 0}
2^{(m+1)}(m+1)!\\
&\qquad\qquad\cdot\sum_{|I|=m, I\subseteq L\setminus (L'\cup\{i,h\})}
\left(\bigwedge_{j\in I} r_j dr_j \wedge\hat{\alpha}_{j}\right)\wedge\bigwedge_{j\in L\setminus(L'\cup I\cup\{i\})}\bar{\omega}_{L(j)},
\end{align*}
Then 
since
\begin{equation*}
\sum_{h\in L'}r_h^2 = 1-\sum_{h\in L\setminus L'}r_h^2,~
\sum_{h\in L'}r_h dr_h = -\sum_{h\in L\setminus L'}r_h dr_h,
\end{equation*}
we obtain
\[\sum_{e\in E',h\in D_e}(\widetilde{A}_{e,h}+\widetilde{B}_{e,h})=X.\]
Therefore summing Equation \eqref{eq:monster2} over $e\in E',h\in D_e,$ gives
\begin{align}\label{eq:monster3}
\sum_{(G',e',D')\in\CB^{-1}(G,D)}\int_{\CM_{\partial_{e'}G'}(\pp)}&W_{G'}s^{*}\bar{\Phi}^{L(D)\setminus L(D')}_{L\setminus L(D')}\\
\notag
&\quad=\int_{\CM_{G}(\pp)}
\left(\prod_{f\in E'}\frac{\ell_{f}^{2m(f)}}{(m(f)+1)!} \right)X \\ \nonumber
&\qquad+\int_{\CM_G(\pp)}\left(\sum_{e\in E'}\left(\prod_{f\in E'\setminus\{e\}}\frac{\ell_{f}^{2m(f)}}{(m(f)+1)!}\right)\frac{2m(e)\ell_e^{2m(e)-1}d\ell_e}{(m(e)+1)!} \right)\wedge Y,
\end{align}
where the factor $m(e)$ in the last term comes from the cardinality of $D_e$ and the summation over $h.$ 
Observe that $Y=\Phi_{L\setminus L'},$ where we stress that we do not require $\sum_{h\in L\setminus L'} r_h^2=1,$ as in Remark \ref{rmk:phi_as_poly}. $X$ here is the same as $Z$ there, after substituting $ L\setminus L'$ for $[n],$ $\hat{\alpha}_i$ for $\alpha_i,$ and $\bar{\omega}_{L(i)}$ for $\omega_i.$
Thus, Remark \ref{rmk:phi_as_poly} immediately gives that the right hand side of \eqref{eq:monster3} is
\[
\int_{\CM_G(\pp)}\left\{\prod_{e\in E'} \frac{\ell_e^{2m(e)}}{(m(e)+1)!}\bigwedge_{i\in L\setminus L'}\bar{\omega}_{L(i)}+d\left(
\prod_{e\in E'} \frac{\ell_e^{2m(e)}}{(m(e)+1)!}\bar{\Phi}_{L\setminus L'}
\right)\right\}.
\]
The claim now follows from Stokes' theorem.
\end{proof}
\begin{proof}[Proof of Proposition \ref{prop:id}]
We first prove part \ref{prop:id_b}.
Write
\[
f(x) = \sum_{m=0}^\infty \frac{x^{2m}}{m!(m+1)!}.
\]
The identity we need to prove is equivalent to
\[
(f\ast f)(x) = f'(x),
\]
where $\ast$ is the convolution
\[(f\ast g)(x)=\int_0^x f(y)g(x-y)dy.\]
Using Laplace transform, the last equation is equivalent to
\[
F^2(\lambda) = \lambda F(\lambda)-1,
\]
where
\[
F(\lambda) = \int_0^\infty e^{-\lambda x}f(x)dx
\]
is the Laplace transform of $f.$
Expanding $F$ we obtain
\begin{align}
F = \sum_{m=0}^\infty\frac{1}{m!(m+1)!}&\int_0^\infty e^{-\lambda x}x^{2m}dx=\sum_{m=0}^\infty\frac{(2m)!}{m!(m+1)!}\lambda^{-2m-1}\\
\notag & = \lambda^{-1}\frac{1-\sqrt{1-4\lambda^{-2}}}{2\lambda^{-2}} = \lambda\frac{1-\sqrt{1-4\lambda^{-2}}}{2},
\end{align}
the third equation is a consequence the general binomial formula.
Thus, we are left with verifying that
\[
F^2(\lambda) = \frac{\lambda^2}{2}(1-\sqrt{1-4\lambda^{-2}})-1 = \lambda F(\lambda)-1,
\]
which is straightforward.

The first identity is a consequence of the second. Indeed,
Write
\begin{align*}
&I_m=\sum_{a=0}^{m-1}\binom{m-1}{a}\int_0^y\frac{x^{2a}(y-x)^{2(m-a)-1}}{(a+1)!(m-a)!}dx,\\
&J_m = \sum_{a=0}^{m-1}\binom{m-1}{a}\int_0^y\frac{x^{2a}(y-x)^{2(m-a-1)}}{(a+1)!(m-a)!}dx.
\end{align*}
It suffices to show that \[I_m=\frac{y}{2}J_m.\]
Indeed,
\begin{align}
I_m &= \sum_{a=0}^{m-1}\binom{m-1}{a}\int_0^y\frac{x^{2a}(y-x)^{2(m-a)-1}}{(a+1)!(m-a)!}dx \\
\notag &= y\sum_{a=0}^{m-1}\binom{m-1}{a}\int_0^y\frac{x^{2a}(y-x)^{2(m-a-1)}}{(a+1)!(m-a)!}dx\\ \nonumber
&\quad\quad\quad-\sum_{a=0}^{m-1}\binom{m-1}{a}\int_0^y\frac{x^{2a+1}(y-x)^{2(m-a-1)}}{(a+1)!(m-a)!}dx\\
\notag &= yJ_m-\sum_{a=0}^{m-1}\binom{m-1}{a}\int_0^y\frac{(y-t)^{2a+1}t^{2(m-a-1)}}{(a+1)!(m-a)!}dx \\ \nonumber
&=yJ_m-I_m,
\end{align}
where the second equality follows from opening one $(y-x)$ term, and the third follows from the substitution $t=y-x.$
\end{proof}

In order to be able to write an expression for the open intersection numbers we need the following observation.
\begin{obs}\label{obs:edge_with_nodes}
Suppose $G\in\oSR^{m}_{g,k,l},$ and $e$ an edge with $m(e)>0.$
Then for any decoration $D,$
$$\int_{\partial_e \CM_{G}(\pp)}W_G s^*\bar{\Phi}_{L\setminus L(D)}=0$$
\end{obs}
\begin{proof}
It follows from the definition of $W_G$ that $W_G|_{\CM_{\partial_e G}(\pp)} = 0$ identically.
\end{proof}

We can now state and prove the integral form of the combinatorial formula. We recall that $d=\frac{3g-3+k+2l}{2}.$
\begin{thm}\label{thm:int_form}
Let $L:[d]\to[l]$ be a $(d,l)-$set, with $a_i=|L_i|,$ for $i\in [l].$ Then
\begin{align}\label{eq:int_form}
\pp^{2L}2^{\frac{g+k-1}{2}}
&\langle  \tau_{a_1}\ldots\tau_{a_l}\sigma^k\rangle \\
=\notag &\sum_{G\in \oSRc^{*}_{g,k,l}}\sum_{D\in \text{Dec}(G,L)}\int_{\CM_G(\pp)}W_G\bar{\omega}_{L\setminus L(D)},
\end{align}
where the collection $\oSRc^m_{g,k,l},~m\geq 0$ is defined in Definition \ref{def:nodal2}.
\end{thm}
\begin{proof}
Define
\begin{align*}
&A_m=\sum_{(G,[\K])\in \oSR^{m}_{g,k,l}}\sum_{D\in \text{Dec}(G,L)}\int_{\CM_{(G,[\K])}(\pp)}W_G\bar{\omega}_{L\setminus L(D)}\\
&S_m =\sum_{(G,[\K])\in\oSR^{m}_{g,k,l}}\sum_{D\in \text{Dec}(G,L)}\int_{\partial\CM_{(G,[\K])}(\pp)}W_G s^*\bar{\Phi}_{L\setminus L(D)},
\end{align*}
where $s$ is a nowhere vanishing special canonical multisection.
We will begin by showing that
\begin{equation}\label{eq:1_for_proof}
S_m= A_{m+1}+S_{m+1},
\end{equation}
and that
\begin{equation}\label{eq:2_for_proof}p^{2L}2^{\frac{g+k-1}{2}}
\langle\tau_{a_1}\ldots\tau_{a_l}\sigma^k\rangle =A_0+S_0.\end{equation}

For the first claim, consider $S_m.$ Recall that for any $G,$
\[
\partial\CM_{(G,[\K])}=\bigcup_{[e]\in [E(G)]}\partial_e\oCM_{(G,[\K])}=\bigcup_{[e]\in [E(G)]}\oCM_{\partial_e (G,[\K])}.
\]
Since for different edges the boundary cells intersect in positive codimension, the integral over the union is just the sum over the edges $e$ of the integrals over $\partial_e\oCM_{(G,[\K])}.$

For an edge $e$ which is not a bridge or a boundary loop, by Corollary \ref{cor:G_e orientation}, we know that
$\partial_e \oCM_{(G,[\K])}(\pp)=-\partial_e \oCM_{(G_e,[\K_e])}(\pp)$ considered as oriented orbifolds, with the orientation induced as a boundary.

Now, $\text{Dec}(G,L),\text{Dec}(G_e,L)$ are the same sets, and it is easy to see that
\[
W_G|_{\partial_e\oCM_{(G,[\K])}}=W_{G_e}|_{\partial_e\oCM_{(G_e,[\K_e])}}.
\]
Thus, given a decoration $D,$ and using the first item of Proposition \ref{prop:section_on_oSR},
\[
\int_{\partial_e\oCM_{(G,[\K])}(\pp)}W_G s^*\bar{\Phi}_{L\setminus L(D)} = -\int_{\partial_e\oCM_{(G_e,[\K])}(\pp)}W_{G_e} s^*\bar{\Phi}_{L\setminus L(D)}.
\]

For an effective loop $e,$ the same argument, only with using Corollary \ref{cor:G_e orientation_loop} instead of Corollary \ref{cor:G_e orientation}, and the third item of Proposition \ref{prop:section_on_oSR} instead of the first item, shows that
given a decoration $D,$
\[
\int_{\partial_e\oCM_{(G,[\K])}(\pp)}W_G s^*\bar{\Phi}_{L\setminus L(D)} = -\int_{\partial_e\oCM_{(G_e,[\K])_e}(\pp)}W_{G_e} s^*\bar{\Phi}_{L\setminus L(D)}.
\]
We should note that this is the second place that we use $s$ being special canonical.

If $e$ is a bridge or a boundary loop which is not effective. From Observation \ref{obs:edge_with_nodes}, for any decoration $D$
\[
\int_{\partial_e\CM_{(G,[\K])}(\pp)}W_G s^*\bar{\Phi}_{L\setminus L(D)}=0.
\]

Thus, we can write,
\[
S_m = \sum_{(G,[\K])\in\oSR^{m}_{g,k,l}}\sum_{D\in \text{Dec}(G,L)}\sum_{[e]\in [\text{Br}^{\text{eff}}(G)]}\int_{\CM_{\partial_e (G,[\K])}(\pp)}W_G s^*\bar{\Phi}_{L\setminus L(D)}.
\]
Applying Lemma \ref{lem:simplifying_phi}, we obtain
\begin{align*}
S_m &= \sum_{(G,[\K])\in\oSR^{m}_{g,k,l}}\sum_{D\in \text{Dec}(G,L)}\sum_{[e]\in [\text{Br}^{\text{eff}}(G)]}\sum_{j\in (L\setminus L(D))_{i(e)}}\int_{\CM_{\partial_e (G,[\K])}(\pp)}W_G s^*\bar{\Phi}^j_{L\setminus L(D)}.
\end{align*}

Note that when $e$ is an effective bridge, then $G' =\CB_e (G,[\K])\in \oSR^{m+1}_{g,k,l}$. We should note that this operation is also responsible for the appearance of ghosts components, which result from contracting a boundary edge between two legal boundary tails. In addition, $j\in (L\setminus L(D))_{i(e)}$ induces a single decoration $D'$ of $G',$ which is defined by $(G,D)\in\CB^{-1}(G',D')$ and $j\in L(D').$
Moreover, any $(G',[\K'])\in\oSR^{m+1}_{g,k,l},~D'\in \text{Dec}(G',L)$ is obtained in this way, see Observation \ref{obs:base_of_nodal_graphs}.
Hence, we can apply Lemma \ref{lem:big_computation} and get
\begin{align*}
S_m &= \sum_{(G,[\K])\in\oSR_{g,k,l}^{m+1}}\sum_{D\in \text{Dec}(G,L)}\int_{\CM_{(G,[\K])}(\pp)}W_G\bar{\omega}_{L\setminus L(D)}\\
&\quad+\sum_{(G,[\K])\in\oSR_{g,k,l}^{m+1}}\sum_{D\in \text{Dec}(G,L)}\int_{\partial\oCM_{(G,[\K])}(\pp)}W_G s^*\bar{\Phi}_{L\setminus L(D)} \\
&=A_{m+1}+S_{m+1},
\end{align*}
as claimed.

For the second claim,
using Lemma \ref{lem:combinatorial_translation}, we can write
\begin{align*}
p^{2L}2^{\frac{g+k-1}{2}}
&\langle\tau_{a_1}\ldots\tau_{a_l}\sigma^k\rangle \\
&=\sum_{(G,[\K])\in \oSR^{0}_{g,k,l}}\int_{\CM_{(G,[\K])}(\pp)}\bar{\omega}_L \\
&\quad+\sum_{(G,[\K])\in\oSR^{0}_{g,k,l}}\sum_{[e]\in [\text{Br}(G)\cup\text{Loop}(G)]}\int_{\CM_{\partial_e (G,[\K])}(\pp)}s^*\bar{\Phi}_L.
\end{align*}
Note that this is the non nodal case, so all bridges and boundary loops are effective and the decorations are empty. The cancellation-in-pairs argument used above for the contribution of the integrals over edges which are neither boundary loops nor bridges shows, in particular, that
\[\sum_{(G,[\K])\in\oSR^{0}_{g,k,l}}\sum_{[e]\in [\text{Br}(G)\cup\text{Loop}(G)]}\int_{\CM_{\partial_e (G,[\K])}(\pp)}s^*\bar{\Phi}_L=\sum_{(G,[\K])\in\oSR^{0}_{g,k,l}}\sum_{[e]\in [E(G)]}\int_{\CM_{\partial_e (G,[\K])}(\pp)}s^*\bar{\Phi}_L=S_0,
\]
which, combined with the previous equation, gives \eqref{eq:2_for_proof}.
%

Iterating \eqref{eq:1_for_proof} for $m\geq 0$ and using \eqref{eq:2_for_proof}, we see that the left hand side of Equation \eqref{eq:int_form} is $\sum_{m\geq 0}A_m.$

We now claim
\begin{prop}\label{prop:parity_bdries}
If $G$ is a nodal graph such that on at least one boundary component there is an even total number of boundary marked points and legal nodes on,
then
\[
\int_{\CM_{(G,[\K])}(\pp)}W_G\bar{\omega}_{L\setminus L(D)}=0.
\]
\end{prop}
The proof is given in Subsection \ref{sec:power_of_2}, see Lemma \ref{lem:vanishing_for_even_bdries}.
Thus,
\[
\sum_{m\geq 0}A_m = \sum_{m\geq 0}\sum_{(G,[\K])\in\oSRc^m_{g,k,l}}\sum_{D\in \text{Dec}(G,L)}\int_{\CM_{(G,[\K])}(\pp)}W_G\bar{\omega}_{L\setminus L(D)},
\]
as claimed.
\end{proof}
\begin{obs}
\begin{align*}
|\text{Dec}(G,L)| &= \binom{L_i}{{\{m(e)|e\in E,~i(e)=i\}}}\\
&=\prod_{i\in [l]}\frac{L_i!}{\left(\prod_{\{e\in E|i(e)=i\}}m(e)!\right)(L_i-\sum_{\{e\in E|i(e)=i\}}m(e))!}.
\end{align*}
\end{obs}
Thus, with the above notations we have,
\begin{align*}
&2^{\frac{g+k-1}{2}}\prod_{i\in [l]}p_i^{2a_i}\langle  \tau_{a_1}\ldots\tau_{a_l}\sigma^k\rangle \\
&=\sum_{m\geq 0}\sum_{(G,[\K])\in \oSRc^{m}_{g,k,l}}\left(\prod_{i\in [l]}\binom{a_i}{{\{m(e)|e\in E,~i(e)=i\}}}\right)\int_{\CM_{(G,[\K])}(\pp)}W_G\bar{\omega}_{L\setminus L(D)}\\
&=\sum_{m\geq 0}\sum_{(G,[\K])\in \oSRc^{m}_{g,k,l}}\left(\prod_{i\in [l]}\frac{a_i!}{(a_i-\sum_{\{e\in E|i(e)=i\}}m(e))!}\right)\int_{\CM_{(G,[\K])}(\pp)}\widetilde W_G\bar{\omega}_{L\setminus L(D)},
\end{align*}
where $\widetilde W_G$ is defined in Definition \ref{def:W_and_tilde_W}, and $D\in D(G,L)$ are arbitrary decorations.
Summing over all possible $L,$ and dividing by $d!,$ we get
\begin{align}\label{eq:almost_last}
2^{\frac{g+k-1}{2}}\sum_{\sum a_i = d}&\prod_{i\in [l]}\frac{p_i^{2a_i}}{a_i !}\langle  \tau_{a_1}\ldots\tau_{a_l}\sigma^k\rangle \\
\notag =&\sum_{m\geq 0}\sum_{(G,[\K])\in \oSRc^{m}_{g,k,l}}\int_{\CM_{(G,[\K])}(\pp)}\widetilde W_G\frac{\bar{\omega}^{d-m}}{(d-m)!}
\end{align}
Dimensional reasons give,
\begin{obs}
Let $L'$ be a $l-$set, $(G,[\K])\in\oSRc^*_{g,k,l}.$ Suppose that for some component $C\in C(G,[\K]),$
\[
\dim(C)<\sum_{i\in I(C)}L'_i.
\]
then $\int_{\CM_G}f\omega_{L'}=0,$ for any function $f.$
\end{obs}
Now, $\bar\omega = \sum_{C\in C(G)}\bar{\omega}^C,$ where $\bar{\omega}^C=\sum_{i\in I(C)}\bar{\omega}_i.$ Thus, together with the observation we get,
\begin{cor}\label{cor:int_form}
$\widetilde W_G\frac{\bar{\omega}^{d-m}}{(d-m)!}=\prod_{C\in C(G)}\widetilde W_C\frac{({\bar{\omega}^C)}^{\dim(C)}}{\dim(C)!}.$
Thus,
\begin{align}
\sum_{\sum a_i = d}&\prod_{i\in [l]}\frac{p_i^{2a_i}}{a_i !}2^{\frac{k-1}{2}}\langle  \tau_{a_1}\ldots\tau_{a_l}\sigma^k\rangle \\
\notag &=\sum_{m\geq 0}\sum_{(G,[\K])\in \oSRc^{m}_{g,k,l}}\int_{\CM_{(G,[\K])}(\pp)}\widetilde W_G\prod_{C\in C(G,[\K])}\frac{(\bar{\omega}^C)^{\dim(C)}}{\dim(C)!}\\
\notag &=\sum_{m\geq 0}\sum_{(G,[\K])\in \oSRc^{m}_{g,k,l}}\prod_{C\in C(G,[\K])}\int_{\CM_C}\widetilde W_C\frac{(\bar{\omega}^C)^{\dim(C)}}{\dim(C)!}
\end{align}
\end{cor}
Note that in the above forumla there may appear component $C$ with $\dim(C)=0.$ These are precisely the ghost components and the genus $0$ components with one internal tail and one legal boundary tail.

\subsection{Power of $2$}\label{sec:power_of_2}
The aim of this subsection is to gain a better understanding of the forms $\bigwedge dp_i\wedge\frac{\bar{\omega}^d}{d!},\mathfrak{o}_{(G,[\K])}$ and their ratio.
\begin{definition}
For $(G,[\K])\in\oSR^{*}_{g,k,l}$ define $s(G,[\K]),$ to be the sign of
\[
\bigwedge dp_i\wedge\frac{\bar{\omega}^d}{d!}: \mathfrak{o}_{(G,[\K])}.
\]
For $G\in\oR^*_{g,k,l}$ define
\[
c_{\text{spin}}(G) = \sum_{[\K]\in[\K(G)]}s(G,[\K]).
\]
\end{definition}
\begin{lemma}\label{lem:power_of_2_number_1}
For $G\in\oSR^*_{g,k,l},$
\[\bigwedge dp_i\wedge\frac{\bar{\omega}^d}{d!}: \mathfrak{o}_{(G,[\K])} = s(G,[\K])c_{\text{spin}}(G)2^{|V^I(G)|}.\]
In particular, $c_{\text{spin}}(G)\geq 0.$
\end{lemma}
\begin{proof}
Both the left hand side and the right hand side are multiplicative with respect to taking non-nodal components, by the first statement in \ref{cor:int_form} and the construction of $\mathfrak{o}_{(G,[\K])},$
thus, it is enough to prove the lemma for graphs in $\oSR^{0}_{g,k,l}.$

Recall that any class $[K]$ of Kasteleyn orientations is of size $2^{|V^I(G)|},$ by Lemma \ref{lem:vertex_flip}. In addition,
by Lemma \ref{lem:or_indep}, $\mathfrak{o}_{(G,\K)}$ for different $\K\in[\K]$ are equal. Thus, the lemma is equivalent to the following equation
\begin{equation}\label{eq:simplifying1}
\bigwedge dp_i\wedge\frac{\bar{\omega}^d}{d!}= \sum_{\K\in \K(G)}\mathfrak{o}_{(G,[\K])}.
\end{equation}
Recall $\bar{\omega} = \sum_{i=1}^l\bar{\omega}_i.$ Fix a good ordering $n.$
In order to prove Equation \eqref{eq:simplifying1}, it will be more comfortable to work with new variables $\ell_h,h\in H^I,$ instead of $\ell_e,e\in E.$ Set
\begin{align*}
H_{\K,i}&=\{h\in H_\K| h/s_2=i\},\\
d_{\K,i} &= \frac{|H_{\K,i}|-1}{2},\\
p_{\K,i} &= \sum_{h\in H_{\K,i}} \ell_{h},\\
\bar{\omega}_{\K,i} &= \sum_{h_1,h_2\in H_{\K,i},~n(h_1)<n(h_2)}d\ell_{h_1}\wedge d_{\ell_{h_2}}.
\end{align*}
\begin{rmk}
Note that only $\bar{\omega}_{\K,i}$ depends on the ordering $n.$ For different orders the change in $\bar{\omega}_{\K,i}$ is of the form $dp_{\K,i}\wedge dx,$ where $x$ is a linear combination of $\{d\ell_h\}_h\in{H_{\K,i}}.$ Thus, for any $a$ the form $dp_{\K,i}\wedge\bar{\omega}_{\K,i}^a$ is independent of $n.$
\end{rmk}

Express each $dp_i$ by $\sum_{h\in H_i} d\ell_h,$ and express also each $\bar{\omega}_i$ in the $\{d\ell_h\}_{h\in H^I}$ basis as above.
Our next aim is to show that
\begin{equation}\label{eq:simplifying2}
\bigwedge dp_i\wedge\frac{\bar{\omega}^d}{d!} = \sum_{\K\in\K(G)}\bigwedge_{i\in[l]}dp_{\K,i}\wedge\frac{\bar{\omega}_{\K,i}^{d_{\K,i}}}{d_{\K,i}!} ~(\mod ~I),
\end{equation}
where $I$ is the ideal $(d\ell_h-d\ell_{s_1h})_{h\in H^I}.$
In order to show Equation \eqref{eq:simplifying2} expand
$\bigwedge dp_i\wedge\frac{\bar{\omega}^d}{d!}$
multilinearly, in terms of $\{d\ell_h\}_{h\in H^I},$ without cancellations.
Any monomial which appears in this expression, and contains exactly one of $d\ell_h,d\ell_{s_1h}$ for any $h\in H^I,$ defines a unique Kasteleyn orientation $\K,$ defined by $\K(h)=1$ if and only if $d\ell_h$ appears in the monomial. This is indeed a Kasteleyn orientation since any $h\in s_1H^B$ has $\K(h)=1,$ and for any $i\in[l],$ an odd number of variables of half edges appear, one comes from $dp_i,$ and the others come in pairs via powers of $\bar{\omega}_i.$

It is transparent that any Kasteleyn orientation $\K\in\K(G),$ is generated this way.
Moreover, regrouping all terms which correspond to the same Kasteleyn orientation, and using the identity
\[
\left(\sum_{i=1}^{2m+1}x_i\right)\wedge\frac{(\sum_{i<j}x_i\wedge x_j)^m}{m!}=x_1\wedge x_2\wedge\ldots\wedge x_{2m+1},
\]
we get Equation \eqref{eq:simplifying2}.

The 'In particular' follows from the fact that $\bigwedge dp_i\wedge\frac{\bar{\omega}^d}{d!},~ s(G,[\K])\mathfrak{o}_{(G,[\K])}$ have the same sign.
\end{proof}
\begin{prop}\label{prop:neighboring_c_spin}
For $G\in\oSR^0_{g,k,l},e\notin \text{Br}(G)\cup\text{Loop}(G),$
\[c_{\text{spin}}(G)=c_{\text{spin}}(G_e)\]
\end{prop}
\begin{proof}
It follows from Lemma \ref{lem:power_of_2_number_1} that $$c_{\text{spin}}(G) = \pm\sum_{[\K']\in[\K(G)]}\mathfrak{o}_{(G,[\K'])}:\mathfrak{o}_{(G,[\K])},$$ for any fixed $[\K]\in[\K(G)].$
If $\K,\K'\in\K(G),$ then by the orientability of the moduli, Theorem \ref{thm:orientability}, we see that
\[\mathfrak{o}_{(G,[\K])}:\mathfrak{o}_{(G,[\K'])}=\mathfrak{o}_{(G_e,[\K_e])}:\mathfrak{o}_{(G_e,[\K'_e])},\]
as $(G,[\K]),(G_e,[\K_e])$ and $(G,[\K']),(G_e,[\K'_e])$ parameterize adjacent cells. Thus, $c_{\text{spin}}(G)=\pm c_{\text{spin}}(G_e).$
But $c_{\text{spin}}\geq 0,$ hence the equality.
\end{proof}
\begin{lemma}\label{lem:vanishing_for_even_bdries}
If $G\in\oR^m_{g,k,l}\setminus\oRc^m_{g,k,l},~c_{\text{spin}}=0.$
\end{lemma}
\begin{proof}
Again, as $c_{\text{spin}}$ is multiplicative in non-nodal components, it is enough to consider the case of non nodal graphs. Let $\partial\Sigma_b$ be a boundary with an even number of boundary marked point.
Note that given a surface $\Sigma,$ and a boundary component $\partial\Sigma_b,$ graded spin structures on $\Sigma$ can be partitioned into pairs which differ exactly in the lifting of $\partial\Sigma_b.$ Thus, we can partition $[\K(G)]$ into pairs which differ exactly in the boundary conditions at $\partial\Sigma_b.$
In combinatorial terms, for any pair $\{(G,[\K_1]),(G,[\K_2])\}$ in the partition we can find $\K_1\in[\K_1],\K_2\in[\K_2]$ which agree everywhere, except on edges with exactly one vertex in $\partial\Sigma_b,$ where they disagree. We shall show that $s(G,[\K_1])=-s(G,[\K_2]).$

As a consequence of the Proposition \ref{prop:neighboring_c_spin} $c_{\text{spin}}(G,[\K])=c_{\text{spin}}(G_e,[\K_e]),$\\$G\in\oR^0_{g,k,l},e\notin \text{Br}(G)\cup\text{Loop}(G).$
By performing enough such Feynman moves at boundary edges of $G,$ see Figure \ref{fig:Feynman}, moves (b),(c),
we may assume only one non-boundary edge emanates from $\partial\Sigma_b.$
Let $2a$ denote the number of the boundary marked points on $\partial\Sigma_b.$ Note that $\partial\Sigma_b$ is part of the boundary of a single face, say face $1.$ Let $h,s_1(h)$ be the internal half edges which touch $\partial\Sigma_b.$ Choose a good ordering $n$ on $G,$ so that $n(h)=1,n(h_1)=2,\ldots,n(h_{2a+1})=2a+2,n(s_1h)=2a+3$ where $h_i\in H^I$ are the other half edges on $\partial\Sigma_b.$ This can always be done, possibly after interchanging $h$ and $s_1h.$ Choose any $\K_1\in [\K_1],$ and $\K_2\in[\K_2],$ which differ only in their values at $h,s_1h.$ Thus, the sign difference between $\mathfrak{o}_{(G,[\K_1])}$ and $\mathfrak{o}_{(G,[\K_2])}$ is just $(-1)^{2a+1}=-1,$ since we change only the location of the variable $d\ell_{h/s_1},$ by $2a+1$ spots. As claimed.
\end{proof}
We can now prove Proposition \ref{prop:parity_bdries}
\begin{proof}
By Lemma \ref{lem:power_of_2_number_1} the proposition is equivalent to $c_{\text{spin}}(G)=0.$ But $c_{\text{spin}}(G)=\prod_{C\in C(G)}c_{spin(C)},$ which is $0$ by Lemma \ref{lem:vanishing_for_even_bdries}.
\end{proof}
We can now also prove Corollary \ref{cor:G_e orientation_loop}.
\begin{proof}
As above, it is enough to prove for smooth $G.$ The case where $e$ is a boundary loop is a special case of the graph considered in the proof of Lemma \ref{lem:vanishing_for_even_bdries}, and in particular we see that the orientation expressions for $(G,[\K])$ and $(G_e,[\K_e])$ are opposite. Recall that the map $\FFF^{\comb}$ preserves the edge lengths of all edges, but changes the Kasteleyn orientation to $[\K_e].$
By contracting these orientation expressions with the vector $-\frac{\partial}{\partial\ell_e},$ we see that the induced orientation on $\partial_e\oCM_{(G_e,[\K_e])},$ and the $(\FFF^{\text{comb}})^*-$push forward of the induced orientation on $\partial_e\oCM_{(G,[\K])}$ are opposite.
\end{proof}

\begin{lemma}\label{lem:c_spin}
For $G\in\oRc^{0}_{g,k,l},$ we have
\[
c_{\text{spin}}(G) = 2^{\frac{g+b-1}{2}},
\]
where $g$ is the genus of $G,$ and $b,$ is the number of boundaries. For $G\in\oRc^m_{g,k,l},~c_{\text{spin}}(G)=\prod c_{\text{spin}}(G_i),$ where $G_i$ are the smooth components of $G.$
\end{lemma}
\begin{proof}
Again it is enough to consider non-nodal graphs.
By Lemma \ref{lem:power_of_2_number_1}, 
$c_{\text{spin}}(G)\geq 0.$
By Proposition \ref{prop:neighboring_c_spin} $c_{\text{spin}}(G,[\K])=c_{\text{spin}}(G_e,[\K_e]),$ whenever $G\in\oSRc^0_{g,k,l},e\notin \text{Br}(G)\cup\text{Loop}(G).$
Thus, it is enough to calculate $c_{\text{spin}}$ for the graph $\bar G,$ where $G$ is the graph constructed in Example \ref{ex:macheta}, see Figure \ref{fig:macheta}. We shall work with the notation of that example. We shall order the faces according to their labels, and we choose an ordering $n$ of the edges of face $1$ such that $a_1$ is the first edge.
Choose a Kasteleyn orientation and write
\begin{align*}
\mathfrak{o}_G = W_1\wedge W_2&\wedge\dots W_{g_s}\wedge d\ell_{h_2}\wedge d\ell_{x_2}\ldots d\ell_{h_l}\wedge d\ell_{x_l}\wedge\\
&\wedge d\ell_{e_{1,0}}\ldots\wedge d\ell_{e_{1,k_1}}\wedge R\wedge d\ell_{y_2}\ldots d\ell_{y_l},
\end{align*}
where $W_i$ is the wedge of $d\ell_{a_i},d\ell_{b_i},d\ell_{c_i},d\ell_{d_i},d\ell_{f_i},d\ell_{g_i},$ according to the order induced by $\K,~R$ is the wedge of the remaining variables, according to the ordering.
The ordering $n,$ restricted to the half edges which are involved in $W_i,$ is
\[
a_i,f_i,d_i,\bar g_i, c_i,\bar f_i,b_i, g_i.
\]
There are four possibilities for $\K(\bar f_i),\K(\bar g_i).$ Let $\K_i^0$ denote the set of possibilities with $\K(\bar f_i)\K(\bar g_i)=0.$ Let $K_i^1$ be the singleton made of the remaining possibility.
One can check by hand that the form $W_i$ is constant in $K_i^0,$ and minus that constant in the forth possibility.

The ordering restricted to the remaining edges is
\begin{align*}
{b_{1,2}},& e_{2,k_2+1},{b_{2,3}}, e_{3,k_3+1}\ldots,b_{b-1,b},e_{b,0},e_{b,1},\ldots,e_{b,k_b},\\
&{\bar b_{b-1,b}},e_{b-1,0},e_{b-1,1},\ldots,e_{b-1,k_{b-1}},{\bar b_{b-2,b-1}},e_{b-2,0}\ldots,e_{2, k_2}{\bar b_{1,2}}.
\end{align*}
The only freedom in $\K$ is in the values of $\K(b_{j,j+1}).$ The relative order of these edges is
\[
b_{1,2}, b_{2,3},\ldots,b_{b-1,b},\bar b_{b-1,b},\ldots, \bar b_{1,2}.
\]
Observe that between $b_{j,j+1}$ and $\bar b_{j,j+1}$ in the ordering, there is an even number of half edges.
Thus, different assignments of $\K(b_{j,j+1})$ do not change the orientation $\mathfrak{o}_G.$ There are $2^{b-1}$ such assignments, where $b$ is the number of boundary components.

To summarize, $s(G,[\K])$ depends only on $\sum_i\K(\bar f_i)\K(\bar g_i),$ which is just the parity of the graded spin structure, see Remark \ref{rmk:machetta}, and different parities give rise to different signs. By the calculation in Remark \ref{rmk:machetta} we see that $c_{\text{spin}}(G) = \pm2^{\frac{g-b+1}{2}+b-1},$ but since it cannot be negative we end with
$c_{\text{spin}}(G) = 2^{\frac{g+b-1}{2}}.$
\end{proof}
\begin{rmk}
An analogous power of $2$ appears in \cite{Kont} when one wants to calculate the Laplace transform of the integral combinatorial formula.
The method developed in this paper is also applicable to that calculation.
It shows exactly where this power of $2$ comes from, and how is it connected to spin structures. In fact, our $c_{\text{spin}}$ can be thought as an open analog of the push down of the $r=2-$spin Witten's class to the spinless moduli, see \cite{Witten2}.
\end{rmk}
\begin{cor}\label{cor:power_of_2}
For $G\in\oSR^0_{g,k,l},$
\[
\bigwedge dp_i\wedge\frac{\bar{\omega}^d}{d!}: \mathfrak{o}_{(G,[\K])}=s(G,[\K])2^{|V^I(G)|+\frac{g(G)+b(G)-1}{2}}.
\]
\end{cor}
\subsection{Laplace transform and the combinatorial formula}
As in the closed case, a more compact formula may be obtained after performing a Laplace transform to \ref{cor:int_form}.

Let $\lambda_i$ be the variable dual to $p_i$ and write, for $e=\{h_1,h_2=s_1h_1\},$
\[
\lambda(e)=
\begin{cases}
\frac{1}{\lambda_i+\lambda_j} & i(h_1)=i,i(h_2)=j\\
\frac{1}{m(e)+1}\binom{2m(e)}{m(e)}\lambda_i^{-2m(e)-1} & i(h_1)=i,h_2\in H^B.
\end{cases}
\]
We also define $\widetilde\lambda(e) = \frac{1}{\lambda(e)}$ for an internal edge and $\widetilde\lambda(e) = {\lambda_{i(e)}}$ for a boundary edge of face $i.$

Applying the transform to the left hand side of \ref{cor:int_form} gives
\begin{align*}
\int_{p_1,\ldots,p_l > 0}\bigwedge dp_i& e^{-\sum \lambda_i p_i}\sum_{\sum a_i = d}\prod_{i\in [l]}\frac{p_i^{2a_i}}{a_i !}2^{\frac{g+k-1}{2}}\langle  \tau_{a_1}\ldots\tau_{a_l}\sigma^k\rangle  \\
&=2^{d+\frac{g+k-1}{2}}\sum_{\sum a_i = d}\prod_{i\in[l]}\frac{(2a_i-1)!!}{\lambda_i^{2a_i +1}}\langle  \tau_{a_1}\ldots\tau_{a_l}\sigma^k\rangle,
\end{align*}
where $d=\frac{k+2l+3g-3}{2}.$

Transforming the right hand side leaves us with
\begin{align*}
\sum_{m\geq 0}&\sum_{G\in \oSRc^{m}_{g,k,l}}\int_{p_1,\ldots,p_l > 0}\bigwedge dp_i e^{-\sum \lambda_i p_i}\prod_{C\in C(G,[\K])}\int_{\CM_C}\widetilde W_C\frac{({\bar{\omega}^C)}^{\dim(C)}}{\dim (C)!} \\
& =\sum_{m\geq 0}\sum_{G\in \oSRc^{m}_{g,k,l}}\int_{p_1,\ldots,p_l > 0}\bigwedge dp_i e^{-\sum \widetilde\lambda(e) \ell_e}\prod_{C\in C(G,[\K])}\int_{\CM_C}\widetilde W_C\frac{{(\bar{\omega}^C)}^{\dim (C)}}{\dim (C)!}
\end{align*}
where we have used the fact that the perimeter of a face is the sum of its edges' lengths.

Recall that
\[
\prod_{C\in C(G,[\K])}\widetilde W_C = \prod_{e\in E^B(G)} \frac{\ell_e^{2m(e)}}{(m(e))!(m(e)+1)!}.
\]
By Corollary \ref{cor:power_of_2}, applied to $(G,[\K])\in \oSRc^0_{g,k,l},$
we have
\[
\frac{\left(\bigwedge_{i\in [l]} dp_i\right) \bar{\omega}^d/d!}{\bigwedge_{e\in E(G)} d\ell_e}=s(G,[\K])2^{|V^I(G)|+\frac{g(G)+b(G)-1}{2}},
\]
the variables in the denominator are ordered by $\mathfrak{o}_{(G,[\K])},$ and $|V^I|,g,b$ are the number of internal vertices of $G,$ its genus and the number of boundary components, respectively.
In addition, $\sum_{[\K]\in[\K(G)]}s(G,[\K])=c_{\text{spin}}=2^{\frac{g+b-1}{2}},$ by Lemma \ref{lem:c_spin}.
Moreover, since $\text{Aut}(G)$ acts on $[\K(G)],$ and is sign preserving,
we see that
$$\sum_{[\K]\in[\K(G)]}s(G,[\K])/|\text{Aut}(G)|=\sum_{[\K]\in[\K(G)]/\text{Aut}(G)}s(G,[\K])/|\text{Aut}(G,[\K])|.$$

Thus, for a fixed $G\in\oRc_{g,k,l}^{m},$ summing over ${\text{for}}_{\text{spin}}^{-1}(G)$ using Observation \ref{obs:oR_and_oSR}, and recalling that $\CM_{(G,[\K])}\simeq\R^{E(G)}/|\text{Aut}(G,[\K])|,$ we get
\begin{align*}
&\sum_{[\K]}\frac{1}{|\text{Aut}(G,[\K])|}\int_{p_1,\ldots,p_l > 0}\bigwedge  dp_i e^{-\sum \widetilde\lambda(e) \ell_e}\prod_{C\in C(G,[\K])}\int_{\R^{E(C)}}\widetilde W_C\frac{(\bar{\omega}^C)^{\dim (C)}}{\dim (C)!} \\
&=\frac{\prod_{C\in C(G)}c(C)}{|\text{Aut}(G)|}\prod_{e\in E\setminus E^B}\int_0^\infty e^{-\widetilde\lambda(e)\ell_e}d\ell_e
\prod_{e\in E^B}\int_0^\infty e^{-\widetilde\lambda(e)\ell_e}\frac{\ell_e^{2m(e)}}{m(e)!(m(e)+1)!}d\ell_e \\
&\quad\quad\quad\quad\quad\quad\quad\quad\quad\quad\quad\quad\quad\quad\quad\quad\quad\quad\quad\quad\quad\quad\quad
=\frac{\prod_{C\in C(G)}c(C)}{|\text{Aut}(G)|}\prod_{e\in E} \lambda(e),
\end{align*}
where $c(C)=2^{|V^I(C)|+{g(C)+b(C)-1}}.$ 
Summing over all $G\in\oRc^*_{g,k,l},$
\[
2^{d+\frac{g+k-1}{2}}\sum_{\sum a_i = d}\prod_{i=1}^l\frac{(2a_i-1)!!}{\lambda_i^{2a_i +1}}\langle  \tau_{a_1}\ldots\tau_{a_l}\sigma^k\rangle = \sum_{G\in\oRc^*_{g,k,l}}\frac{\prod_{C\in C(G)}c(C)}{|\text{Aut}(G)|}\prod_{e\in E} \lambda(e).
\]
And theorem \ref{thm:comb_model} is proven.

\begin{open}\label{rmk:open1}
The moduli space $\oCM_{g,k,l}$ is disconnected, and is composed of components which parameterize different topologies, partitions of boundary markings along boundary components and graded structures.
The boundary conditions of \cite{ST,PST} define in fact an intersection number on each such component, and their sum is what we denote in this work by $\langle\tau_{a_1}\ldots\tau_{a_l}\sigma^k\rangle_g.$
Using the techniques presented in this section one can actually calculate all these refined intersection numbers (see \cite{ABT}).
The intersection numbers $\langle\tau_{a_1}\ldots\tau_{a_l}\sigma^k\rangle_g$ are related to the KdV wave function, and therefore satisfy many recursion relations.
A natural question is whether the refined numbers also satisfy interesting recursion relations, and whether they are related to an integrable hierarchy. \cite{ABT} proposes a conjecture in this direction.
\end{open}

\appendix
\section{Properties of the stratification}
\subsubsection{Proposition \ref{prop:streb_gives_graph}}
%

Fix sets $\I,\B,\PP.$
For a stable open ribbon graph $G,$ write $\CM_G=\R^{E(G)}_+/\text{Aut}(G).$ Let $G_{g,\B,(\I,\PP)}$ be the set of all such graphs with boundary markings, internal markings and internal markings of perimeter $0$ being $\B,\I,\PP$ respectively.
We will show that $\Rcomb$ maps $\oRCM_{g,\B,\I\cup\PP}\to\coprod_{G_{g,\B,(\I,\PP)}}\CM_G(\pp),$ surjectively, and that it is $1:1$ on smooth or effective loci.

\begin{st}
\end{st}
An anti holomorphic involution $\varrho$ of a connected stable curve $X$ is \emph{separating} if $X/\varrho$ is a connected orientable stable surface with boundary.
$X^\varrho$ is called the \emph{real} locus. A \emph{half} of $X$ is a stable connected subsurface with boundary $\Sigma\subseteq X$ such that the composition $\Sigma\\hookrightarrow X\to X/\varrho$ is a homeomorphism.

A \emph{doubled $(g,\B,\I\cup\PP)$-surface} is a closed stable marked surface $X,$ with markings $\{x_i\}_{i\in\B},$ $\{z_i,\bar{z}_i\}_{i\in\I\cup\PP}$ together with a separating anti holomorphic involution $\varrho$ and a preferred half $\Sigma$ which satisfies the following
\begin{enumerate}
\item $\forall i,~x_i\in X^\varrho.$
\item $\forall i,~z_i\in \text{int}(\Sigma).$
\end{enumerate}
\begin{obs}\label{obs:ap1}
There is a natural one to one correspondence between open stable $(g,\B,\I\cup\PP-$surfaces $\Sigma$ and doubled $(g,\B,\I\cup\PP)-$surfaces $(X,\varrho,\Sigma),$
given by $\Sigma\to(D(\Sigma),\Sigma)$, where $\Sigma$ is taken as a subset of $D(\Sigma).$
\end{obs}
Note that all components of $X^\varrho$ which are not isolated points are canonically oriented as boundaries of the distinguished half.
\begin{st}
\end{st}
Fix positive $\{p_i\}_{i\in \I}.$ For convenience we denote by $\bar \I,\bar \PP$ the markings of $\bar{z}_i,$ for $i\in \I, \PP.$
We now analyze the image of doubled surfaces $(X,\varrho,\Sigma)$ under the (closed) map $\comb_{\mathbf{q}}$ defined on $\oCM_{g,k+2l},$ where the perimeters $\mathbf{q}$ are defined so that the faces of $z_i,\bar{z}_i,~i\in \I$ have perimeter $p_i$ and the other points are boundary marked points or internal marked with perimeter $0.$
By the construction for closed surfaces, the image is a stable ribbon graph $G$ in the sense of Definition \ref{def:stable_JS}, embedded in $\widetilde X = K_{\B\cup\PP\cup\bar\PP}(X).$
Moreover, $\varrho$ induces an involution, which we also denote by $\varrho,$ on $\widetilde X, G,$ and by Lemma \ref{lem:sym_JS}, $\widetilde X^\varrho\subseteq G.$
Faces and vertices marked by $\I\cup\PP$ are in one distinguished half, $\widetilde\Sigma,$ of $\widetilde X,$ where a half is defined analogously to above.

Write $E^B$ for $\varrho-$invariant edges. Let $H^B$ be their halves which do not agree with the orientation induced by $\widetilde\Sigma.$ Write $V^B$ for $\varrho-$invariant vertices.
Let $V^I$ be vertices in $\text{int}(\widetilde\Sigma)$, and $H^I$ either half edges in $s_1H^B$ or half edges which intersect $\text{int}(\widetilde\Sigma),~E^I=(H^I\setminus s_1H^B)/s_1.$
\begin{obs}
$s_1$ leaves $H^I\cup H^B$ invariant, and that $s_0$ takes $H^I$ to $H^I\cup H^B.$
\end{obs}
Indeed, if there were $h\in H^I,h'\notin H^I\cup H^B,$ with $s_0h=h',$  then there was a common face which contained $h,s_1h'.$ But then this face would intersect both $\text{int}(\widetilde\Sigma),\varrho(\text{int}(\widetilde\Sigma))$, which is impossible.

Let $v$ be a vertex, consider its half edges. The permutation $s_0$ acts on them, and also $\varrho.$ Write $B_v$ for the set of $s_0-$cycles which contain an element of $H^B,$ write $I_v$ for those cycles in $H^I.$
It is easy to see that no $s_0-$cycle contains more than two boundary edges. It follows from the observation that inside a cycle in $B_v$ the half edges are $s_0-$ordered as $h_1,\ldots,h_{2r+2}$ so that 
\[h_1\in s_1H^B,~h_i\in H^I\setminus s_1 H^B,~i\in[r+1]\setminus\{1\},~
h_i=\varrho(h_{i-r-1})
~i\in[2r+2]\setminus[r+1].\]
In particular, $h_{r+2}\in H^B,~h_i\notin (H^I\cup H^B),
~i\in[2r+2]\setminus[r+2].$
Define a permutation $\widetilde s_0$ of $H^I\cup H^B$ which is $s_0$ on $H^I,$ and otherwise, we are in the scenario just described, $\widetilde s_0 h_{r+2}=h_1.$

Define new marking assignments, $f^{I},f^{B},f^{\PPP}$ as follows. $f^{I}$ maps $i\in \I$ to the face which contains $z_i,~f^{B}$ maps $i\in \B$ to the vertex $x_i$ is mapped to. $f^{\PPP}$ is defined similarly.

Recall Notation \ref{nn:isotopy_classes}. Define $\widetilde \D(g,I,B)$ to be the set of isotopy types of smooth doubled $(g,I,B)-$surfaces. Write $\widetilde \D(g,I)=\D(g,I).$
Clearly there exists a canonical identification $\alpha:\widetilde \D(g,I,B)\simeq \D(g,I,B) .$

We can enrich the graph $(G,\varrho)$ with a defect function $d$ on $V^I\cup V^B$ defined as follows. Let $v\in V^I\cup V^B$ be a vertex, consider its preimage $X_v$ in $X.$
If $X_v$ is not a point, then it is a pointed nodal surface, doubled in case $v\in V^B,$ and otherwise just a usual closed one, without $z_i,\bar{z}_i$ for $i\in \I.$
Some of the special points of $X_v$ correspond to nodes whose two halves belong to $X_v.$ Smooth $X_v$ along these nodes.
~There is a unique topological way to perform the smoothing process on a doubled surface, which is consistent with the choice of a half and such that the resulting surface is doubled.
~Define $d(v) \in \D(g(v),I_v\cup(f^{\PPP})^{-1}(v), B_v\cup(f^{B})^{-1}(v))$ to be the class of the smoothed $X_v$ in the doubled case. Otherwise $d(v)$ is the unique element in $\D(g(v),I_v\cup(f^{\PPP})^{-1}(v)).$

The ribbon graph $G,$ together with the involution $\varrho,$ and the \emph{doubled data}, which consists of the sets $H^I,H^B,V^I,V^B,$ and the maps $d,f^{I},f^{B},f^{\PPP}$ is called a \emph{doubled ribbon graph}.
We see that any doubled surface, together with perimeters as above, is associated with a doubled graph. Call this association $\text{Dcomb}$
It now follows from definitions that
\begin{obs}\label{obs:ap2}
There is a canonical bijection $\text{Half}$ between doubled $(g,\B,(\I,\PP))-$metric ribbon graphs, and open $(g,\B,(\I,\PP))-$metric ribbon graphs.
$\text{Half}(G)$ is the graph spanned by $H^I,H^B,V^I,V^B,$ permutations $\widetilde s_0,s_1,$ maps $f^{I},f^{B},f^{\PPP},$ the same genus defect of $G$ and topological defect $\alpha(d).$

$\text{Half}(G)$ is embedded in $\widetilde \Sigma,$ which, after defining the corresponding defects, is exactly $K_{\B,\PP}\Sigma.$
\end{obs}
Thus, by Observations \ref{obs:ap1},\ref{obs:ap2}, for any $\Sigma\in\oRCM_{g,k,l}$ and perimeters $\pp,$ the symmetric JS differential indeed defines a stable open ribbon graph with perimeters $\pp$ embedded in $K_{\B,\PP}\Sigma.$
\begin{st}
\end{st}
We now show that
\begin{prop}
 $\Rcomb:\oRCM_{g,\B,\I\cup\PP}\times\R^{\I}\to\coprod_{G_{g,\B,(\I,\PP)}}\CM_G,$ is a surjection, and in the smooth case, or more generally when unmarked components are not adjacent and form a moduli of dimension $0,$ it is in fact a bijection on its image. 
\end{prop}
This proposition is true in the closed case.
By the above construction, it will be enough to show these properties for $\text{Dcomb}.$
By the closed theory, from the doubled metric graph $(G,\ell)$ one can reconstruct the unique surface with extra structure $\widetilde X,$ in which it embeds, including the complex structure on its marked components. Write $\mathbf{q}$ for the set of perimeters of faces of $G.$ It is evident that the perimeters of faces $i,\bar{i}$ are the same.
The involution on $(G,\ell)$ lifts to an involution on $\widetilde X.$
For any singular point $v\in\widetilde X,$ which corresponds the vertex $v$ of the graph, any $s_0-$cycle $\widetilde v$ of half edges corresponds a new marked point labelled $\widetilde v$ in the normalization of $\widetilde\Sigma.$
We define a surface $X$ is follows.
For a singular $v,$ if $v\in V^B,$ replace $v$ by a doubled surface $\Sigma_v,$ in the isotopy class $d(v).$ For a singular $v\in V^I,$ replace $v,\varrho(v)$ by two conjugate closed surfaces $\Sigma_v,\bar{\Sigma}_v,~\Sigma_v$ is in the class of $d(v).$ Note that $\Sigma_v$ is not necessarily stable.
Let $\Sigma_1,\ldots,\Sigma_r$ be the marked components of $\widetilde\Sigma.$ Define
\[X=\text{Stab}((\coprod{X_i}\cup\coprod {X_v})/\sim)\]
where the $\sim$ identifies a marked point in some $\Sigma_v$ which corresponds to a $s_0-$cycle $\widetilde v$ with the corresponding point in some $\Sigma_i.$
$\text{Stab}$ is the stabilization map which contracts an unstable component to a point.

One can easily extend $\varrho$ and the choice of a half to $X,$ and \\$\text{Dcomb}(X,\mathbf{q})=(G,\ell),$ where $\mathbf{q}$ is the set of perimeters.

In the smooth or the more general case described in the statement, we have no freedom in the reconstruction of $X.$

\bibliographystyle{amsabbrvc}
\bibliography{bibli}

\end{document}